\documentclass[10pt]{article}

\usepackage{amssymb,amsmath,amsthm}
\usepackage{times,a4wide}
\usepackage{pifont}
\usepackage{tikz}
\usetikzlibrary{positioning,arrows,fit,calc}
\usetikzlibrary{decorations.pathmorphing}
\usetikzlibrary{shapes,patterns}
\usepackage{url}

\usepackage{lmodern}
\usepackage[T1]{fontenc}

\newtheorem{theorem}{Theorem}

\newtheorem{example}[theorem]{Example}
\newtheorem{lemma}[theorem]{Lemma}
\newtheorem{corollary}[theorem]{Corollary}
\newtheorem{proposition}[theorem]{Proposition}

\newcommand{\nb}[1]{\hbox to 0pt {\textcolor{red}{\bf !}}\marginpar{\parbox{16mm}{\raggedright\scriptsize \textcolor{red}{#1}}}}
\newcommand{\nz}[1]{\hbox to 0pt {\textcolor{blue}{\bf !}}\marginpar{\parbox{16mm}{\raggedright\scriptsize\raggedright\textcolor{blue}{#1}}}}
\newcommand{\meg}[1]{\hbox to 0pt {\textcolor{green}{\bf !}}\marginpar{\parbox{16mm}{\raggedright\scriptsize\raggedright\textcolor{green}{#1}}}}

\newcommand{\q}{{\boldsymbol{q}}}

\newcommand{\T}{{\cal T}}
\newcommand{\A}{{\cal A}}
\newcommand{\D}{{\cal D}}
\newcommand{\M}{{\cal M}}
\newcommand{\C}{{\cal C}}
\newcommand{\I}{{\cal I}}
\newcommand{\R}{\mathfrak{D}}
\newcommand{\can}{\C_{\T, \A}}
\newcommand{\ind}{\mathsf{ind}}
\newcommand{\dom}{\mathsf{dom}}
\newcommand{\degree}{\mathsf{deg}}
\newcommand{\restr}{\upharpoonright}
\newcommand{\dD}{\partial D}
\newcommand{\dDp}{\partial D'}
\newcommand{\dDpp}{\partial D''}
\newcommand{\rpred}{G}
\def\t{\mathfrak{t}}
\newcommand{\ti}{\t_\mathsf{i}}
\newcommand{\tr}{\t_\mathsf{r}}
\newcommand{\rew}{G}

\newcommand{\od}{\boldsymbol{d}}
\newcommand{\twi}{\boldsymbol{t}}
\newcommand{\nlf}{\boldsymbol{\ell}}

\newcommand{\omq}{{\ensuremath{\boldsymbol{Q}}}}

\newcommand{\cir}{\boldsymbol{C}}
\newcommand{\wid}{\mathsf{w}}
\newcommand{\dep}{\mathsf{d}}
\newcommand{\sdep}{\mathsf{sd}}

\newcommand{\LOGCFL}{\ensuremath{\mathsf{LOGCFL}}}

\newcommand{\OWL}{\textsl{OWL\,2}}
\newcommand{\OWLQL}{\textsl{OWL\,2\,QL}}

\newcommand{\FO}{\text{FO}}
\newcommand{\PE}{\text{PE}}
\newcommand{\NDL}{\text{NDL}}

\newcommand{\ACz}{{\ensuremath{\mathsf{AC}^0}}}

\newcommand{\NCone}{{\ensuremath{\mathsf{NC}^1}}}

\newcommand{\NP}{\ensuremath{\mathsf{NP}}}
\newcommand{\PTime}{\mathsf{P}}

\newcommand{\Ppoly}{\mathsf{P}/\mathsf{poly}}
\newcommand{\NLpoly}{\ensuremath{\mathsf{NL}}/\ensuremath{\mathsf{poly}}}
\newcommand{\NL}{\ensuremath{\mathsf{NL}}}

\newcommand{\vars}{\mathsf{var}}

\newcommand{\avec}[1]{\boldsymbol{#1}}

\newcommand{\gfmn}{\mathcal{G}}
\newcommand{\exs}{{\scriptscriptstyle\exists}}

\tikzset{bpoint/.style={circle,inner sep=0pt,minimum size=1.5mm,fill=black,draw=black},
	wpoint/.style={circle,inner sep=0pt,minimum size=1.5mm,fill=white,draw=black,thick}}

\newcommand{\nd}{t}
\newcommand{\tpd}{\avec{w}}
\newcommand{\tpr}{\avec{s}}

\newcommand{\rni}{\ensuremath{\boldsymbol{R}_\T}}

\newcommand{\twords}{\boldsymbol{W}_{\!\T}}
\newcommand{\sqset}{\mathfrak{Q}}

\newcommand{\pr}{\textit{pDepth-}\textsc{TreeOMQ}}
\newcommand{\hs}{$p$-{\sc HittingSet}}

\renewcommand{\P}{P_{+}}
\newcommand{\N}{P_{-}}
\newcommand{\V}{P_{0}}

\newcommand{\Fa}{\mathsf{f}}
\newcommand{\Tr}{\mathsf{t}}

\newcommand{\blpr}{\textit{pLeaves-}\textsc{TreeOMQ}}
\newcommand{\partclique}{{\sc PartitionedClique}}
\def\w{\mathsf{w}}

\newcommand{\bdObtwCQ}{\textmd{\textsf{OMQ}}(\od,\twi,\infty)}
\newcommand{\bdOblCQ}{\textmd{\textsf{OMQ}}(\od,1,\nlf)}
\newcommand{\blCQ}{\textmd{\textsf{OMQ}}(\infty,1,\nlf)}

\definecolor{RoyalBlue}{rgb}{0.0, 0.14, 0.4}
\definecolor{Green}     {rgb}{0.13,0.55,0.13}
\definecolor{OrangeRed}{rgb}{1.0, 0.27, 0.0}
\definecolor{Orange}{rgb}{1.0, 0.5, 0.0}
\definecolor{Violet}{rgb}{0.56, 0.27, 0.52}
\definecolor{Cerulean}{rgb}{0.0, 0.48, 0.65}
\definecolor{Yellow}{rgb}{1.0, 1.0, 0.0}
\definecolor{YellowOrange}{rgb}{1.0, 0.83, 0.0}

\begin{document}

\title{The Complexity of Ontology-Based Data Access with OWL\,2\,QL and Bounded Treewidth Queries}
\author{
Meghyn Bienvenu$^1$ \and Stanislav Kikot$^2$ \and Roman Kontchakov$^2$ \and Vladimir V. Podolskii$^3$ \and
Vladislav Ryzhikov$^4$ \and Michael Zakharyaschev$^2$\\[12pt]
 \multicolumn{1}{p{.7\textwidth}}{\centering\normalsize{%
 ${}^1$ CNRS \& University of Montpellier, France\\
%
${}^2$ Birkbeck, University of London, UK\\
%
${}^3$ Steklov Mathematical Institute \& National Research University\\Higher School of Economics, Moscow, Russia\\
${}^4$ Free University of Bozen-Bolzano, Italy
}}}
\date{}

\maketitle

\begin{abstract}
Our concern is the overhead of answering \OWLQL{} on\-to\-lo\-gy-mediated queries (OMQs) in ontology-based data access compared to evaluating their underlying tree-shaped and bounded treewidth conjunctive queries (CQs). We show that OMQs with bounded-depth onto\-logies have nonrecur\-sive datalog (NDL) rewritings that can be constructed and evaluated in \LOGCFL{} for combined complexity, even in \NL{} if their CQs are tree-shaped with a bounded number of leaves, and so incur no overhead in complexity-theoretic terms. For OMQs with arbitrary ontologies and bounded-leaf CQs,  NDL-rewritings are constructed and evaluated in \LOGCFL.
We show experimentally feasibility and scalability of our rewritings compared to previously proposed NDL-rewritings.
On the negative side, we prove that answering OMQs with tree-shaped CQs is not fixed-parameter trac\-table if the
ontology depth or  the number of leaves in the CQs is regarded as the parameter, and that answering OMQs with a fixed ontology (of infinite depth)  is \NP-complete for tree-shaped and \LOGCFL{} for bounded-leaf CQs. Moreover, we construct an ontology $\T$ (of infinite depth) such that answering OMQs $(\T,\q)$ with tree-shaped CQs $\q$ is W[1]-hard if the number of leaves in $\q$ is regarded as the parameter.
\end{abstract}

\section{Introduction}

\begin{figure*}[t]%
\tikzset{cmplx/.style={draw,thick,rounded corners,inner sep=0mm}}%
\hspace*{-0.4em}\begin{tikzpicture}[xscale=0.5,yscale=0.6]
\draw[thick] (0.4,0.6) rectangle +(8.2,-6.2);
\node[rotate=90] at (-0.7,-3) {\scriptsize ontology depth};
\begin{scope}[ultra thin]
\draw (0.4,0) -- +(8.2,0); \node at (0,0) {\scriptsize 0};
\draw (0.4,-1) -- +(8.2,0); \node at (0,-1) {\scriptsize 1};
\draw (0.4,-2) -- +(8.2,0); \node at (0,-2) {\scriptsize 2};
\draw (0.4,-3) -- +(8.2,0); \node at (0,-3) {\scriptsize \dots};
\draw (0.4,-4) -- +(8.2,0); \node at (0,-4) {\scriptsize $\od$};
\draw (0.4,-5) -- +(8.2,0); \node at (0,-5) {\scriptsize $\infty$};
\draw (1,0.6) -- +(0,-6.2); \node at (1,1) {\scriptsize 2};
\draw (2,0.6) -- +(0,-6.2); \node at (2,1) {\scriptsize \dots};
\draw (3,0.6) -- +(0,-6.2); \node at (3,1) {\scriptsize $\nlf$};
\draw (4,0.6) -- +(0,-6.2); \node at (4,1) {\scriptsize $\infty$};
\draw (5,0.6) -- +(0,-6.2); \node at (5,1) {\scriptsize 2};
\draw (6,0.6) -- +(0,-6.2); \node at (6,1) {\scriptsize \dots};
\draw (7,0.6) -- +(0,-6.2); \node at (7,1) {\scriptsize $\twi$};
\draw (8,0.6) -- +(0,-6.2); \node at (8,1) {\scriptsize$\infty$};
\end{scope}
\draw[thin] (0.5,1.4) -- ++(0,0.3) -- ++(4,0) -- ++(0,-0.3);
\node at (2.5,1.5) {\scriptsize number of leaves};
\node at (6.5,1.8) {\scriptsize treewidth};
\node at (2.5,2) {\scriptsize trees};
\node [fill=gray!5,cmplx,fill opacity=0.9,fit={(3.4,-4.4) (0.6,0.4)}]  {\NL};
\node [fill=gray!40,cmplx,fill opacity=0.9,fit={(3.6,0.4) (7.4,-4.4)}]  {\LOGCFL};
\node [fill=black,cmplx,fit={(8.4,-5.4) (7.6,0.4)}]  {};
\node [fill=black,cmplx,fit={(3.6,-4.6) (8.4,-5.4)}]  {\hspace*{5em}\raisebox{-1ex}{\textcolor{white}{\NP}}};
\node [fill=gray!40,cmplx,fill opacity=0.9,fit={(0.6,-4.6) (3.4,-5.4)}]  {\raisebox{-1.5ex}{\LOGCFL}};
\node at (-0.5,1.7) {\footnotesize (a)};
\end{tikzpicture}%
\hfill
\begin{tikzpicture}[xscale=0.9,yscale=0.6]
\draw[thick] (0.4,0.6) rectangle +(8.2,-6.2);
\begin{scope}[ultra thin]
\draw (0.4,0) -- +(8.2,0); \node at (0.2,0) {\scriptsize 0};
\draw (0.4,-1) -- +(8.2,0); \node at (0.2,-1) {\scriptsize 1};
\draw (0.4,-2) -- +(8.2,0); \node at (0.2,-2) {\scriptsize 2};
\draw (0.4,-3) -- +(8.2,0); \node at (0.2,-3) {\scriptsize \dots};
\draw (0.4,-4) -- +(8.2,0); \node at (0.2,-4) {\scriptsize $\od$};
\draw (0.4,-5) -- +(8.2,0); \node at (0.2,-5) {\scriptsize $\infty$};
\draw (1,0.6) -- +(0,-6.2); \node at (1,1) {\scriptsize 2};
\draw (2,0.6) -- +(0,-6.2); \node at (2,1) {\scriptsize \dots};
\draw (3,0.6) -- +(0,-6.2); \node at (3,1) {\scriptsize $\nlf$};
\draw (4,0.6) -- +(0,-6.2); \node at (4,1) {\scriptsize $\infty$};
\draw (5,0.6) -- +(0,-6.2); \node at (5,1) {\scriptsize 2};
\draw (6,0.6) -- +(0,-6.2); \node at (6,1) {\scriptsize \dots};
\draw (7,0.6) -- +(0,-6.2); \node at (7,1) {\scriptsize $\twi$};
\draw (8,0.6) -- +(0,-6.2); \node at (8,1) {\scriptsize $\infty$};
\end{scope}
\draw[thin] (0.5,1.4) -- ++(0,0.3) -- ++(4,0) -- ++(0,-0.3);
\node at (2.5,1.5) {\scriptsize number of leaves};
\node at (6.5,1.8) {\scriptsize treewidth};
\node at (2.5,2) {\scriptsize trees {\tiny (treewidth 1)}};
\node [fill=gray!25,cmplx,fill opacity=0.9,fit={(0.545,-1.6) (3.455,-5.4)}]
{\raisebox{-6ex}{\begin{tabular}{c}poly NDL\\[0pt] no poly PE\\[0pt]
\footnotesize poly FO\\[-3pt]\scriptsize iff\\[-5pt]\scriptsize \NLpoly  $\,\subseteq\,$ \NCone\end{tabular}}};
\node [fill=gray!50,cmplx,fill opacity=0.9,fit={(3.545,-1.6) (7.455,-4.4)}]
{\raisebox{-6ex}{\begin{tabular}{c}poly NDL\\[0pt] no poly PE\\[0pt]
\footnotesize poly FO\ \  \scriptsize iff\hspace*{4em}\\[-4pt]\scriptsize\hspace*{3em}\textsf{LOGCFL/poly} $\!\subseteq\!$ \NCone\end{tabular}}};
\node [fill=black,cmplx,fit={(3.545,-4.6) (8.455,-5.4)}]
{\hspace*{-1.8em}\raisebox{-9pt}{\textcolor{white}{\bfseries\begin{tabular}{c}\small no poly NDL\,{\scriptsize\&}\,PE
\end{tabular}}}};
\node [fill=black,cmplx,fit={(7.545,-1.6) (8.455,-5.4)}]
{\raisebox{-7ex}{\textcolor{white}{\tabcolsep=0pt\small\bfseries\begin{tabular}{c}
poly\\FO\\\scriptsize iff\\[-1pt]{\textmd{\textsf{NP{\tiny\!/\!poly}}}}\\[-2ex]\rotatebox{-90}{$\subseteq$}\\[4pt]\NCone{}
\end{tabular}}}};
\node [fill=gray!5,cmplx,fill opacity=0.9,fit={(0.545,0.4) (8.455,-0.4)}]  {\raisebox{-1.5ex}{poly $\Pi_2$-PE}};
\node [fill=gray!5,cmplx,fill opacity=0.9,fit={(0.545,-0.6) (4.455,-1.4)}]  {\raisebox{-1.5ex}{poly $\Pi_4$-PE}};
\node [fill=gray!5,cmplx,fill opacity=0.9,fit={(4.545,-0.6) (7.455,-1.4)}]  {\raisebox{-1.5ex}{poly PE}};
\node [fill=gray!25,cmplx,fill opacity=0.9,fit={(7.545,-0.6) (8.455,-1.4)}] { };
\node[inner sep=0pt] (test) at (9.9,-1.5) {\begin{tabular}{c}poly NDL\\[0pt]no poly PE\\[0pt]\footnotesize poly FO\\[-3pt] \scriptsize iff\\[-5pt]\scriptsize\NLpoly  $\,\subseteq\!$ \NCone\end{tabular}};
\draw (8,-1) -- (test);
\node at (0,1.7) {\footnotesize (b)};
\end{tikzpicture}
\caption{OMQ answering in \OWLQL{} (a) combined complexity and (b) the size of rewritings.}
\label{pic:results}
\end{figure*}

The main aim of ontology-based data access (\!OBDA\!) \cite{PLCD*08,Lenzerini13} is to facilitate access to complex data for non-expert end-users.
The ontology, given by a logical theory $\T$, provides a unified conceptual view of one or more data sources, so
the users do not have to know the actual structure of the data and can formulate their queries in the vocabulary of the ontology, which is connected to the data schema by a mapping $\M$. The instance $\M(\D)$ obtained by applying $\M$ to a given dataset $\D$ is interpreted under the open-world assumption, and additional facts can be \emph{inferred} using the domain knowledge provided by the ontology. A certain answer to a query $\q(\avec{x})$ over $\D$ is any tuple of constants $\avec{a}$ such that $\T,\M(\D) \models \q(\avec{a})$.
OBDA is closely related to querying incomplete databases under (ontological) constraints, data integration \cite{DBLP:books/daglib/0029346}, and data exchange \cite{DBLP:books/cu/ArenasBLM2014}.

In the classical approach to OBDA~\cite{CDLLR07,PLCD*08}, the computation of certain answers is reduced to standard data\-base query evaluation: given an ontology-mediated query (OMQ) $\omq =(\T,\q(\avec{x}))$, one constructs a first-order (FO) query~$\q'(\avec{x})$, called a rewriting of $\omq$, such that, for all datasets $\mathcal{D}$ and mappings $\mathcal{M}$,
\begin{equation}\label{reduction}
\T,\mathcal{M}(\mathcal{D}) \models \q(\avec{a}) \quad \text{ iff } \quad \I_{\mathcal{M}(\mathcal{D})} \models \q'(\avec{a}),
\end{equation}
where $\I_{\mathcal{M}(\mathcal{D})}$ is the FO-structure
comprised of the atoms in $\mathcal{M}(\mathcal{D})$.
When the form of $\M$ is appropriately restricted (e.g., $\M$ is a GAV mapping), one can further unfold $\q'(\avec{x})$ using $\M$ to obtain an FO-query that can be evaluated directly over the original dataset $\D$ (so there is no need to materialise $\mathcal{M}(\mathcal{D})$).

For reduction \eqref{reduction} to hold for all OMQs, it is necessary to restrict the expressivity of~$\T$ and~$\q$. The {\sl DL-Lite} family of description logics \cite{CDLLR07} was specifically designed to ensure \eqref{reduction} for OMQs with conjunctive queries (CQs) $\q$.
Other ontology languages with this property
include linear and sticky tuple-generating dependencies (tgds)~\cite{DBLP:journals/ws/CaliGL12,DBLP:journals/ai/CaliGP12}, and the \OWLQL{} profile \cite{profiles} of the W3C-standardised Web Ontology Language \OWL{}, the focus of this work.
Like many other ontology languages, \OWLQL{} admits only unary and binary predicates,
but arbitrary relational instances can be queried due to the mapping.
Various types of FO-rewritings $\q'(\avec{x})$ have been developed and implemented for the preceding
languages~\cite{PLCD*08,DBLP:conf/dlog/Perez-UrbinaMH09,KR10our,DBLP:conf/kr/RosatiA10,DBLP:conf/cade/ChortarasTS11,DBLP:conf/aaai/EiterOSTX12,DBLP:conf/esws/Rosati12,DBLP:conf/kr/KikotKZ12,DBLP:conf/icde/GottlobOP11,kyrie2,DBLP:journals/semweb/KonigLMT15},
and a few mature OBDA systems have emerged, including pioneering MASTRO~\cite{DBLP:journals/semweb/CalvaneseGLLPRRRS11}, commercial Stardog~\cite{Perez-Urbina12} and Ultrawrap~\cite{DBLP:conf/semweb/SequedaAM14}, and the Optique platform~\cite{optique} with the
query answering engine Ontop~\cite{DBLP:conf/semweb/Rodriguez-MuroKZ13,DBLP:conf/semweb/KontchakovRRXZ14}.

Our concern here is the  overhead of OMQ an\-swering---i.e., checking whether the left-hand side of \eqref{reduction} holds---compared to evaluating the underlying CQs. At first sight,
there is no apparent difference between the two problems when viewed through the lens of computational complexity: OMQ answering is in $\ACz$ for data complexity by~\eqref{reduction} and \NP-complete for combined complexity~\cite{CDLLR07}, which in both cases corresponds to the complexity of evaluating CQs in the relational setting.
Further analysis revealed, however, that answering OMQs is already \NP-hard for combined complexity when the underlying CQs are tree-shaped (acyclic)~\cite{DBLP:conf/dlog/KikotKZ11}, which sharply contrasts with the well-known \LOGCFL-comple\-teness of evaluating bounded treewidth CQs~\cite{DBLP:conf/vldb/Yannakakis81,DBLP:journals/tcs/ChekuriR00,DBLP:conf/icalp/GottlobLS99}.
This surprising difference motivated a systematic investigation of the combined complexity of OMQ answering along two dimensions: (\emph{i}) the query topology (treewidth $\twi$ of CQs, and the number $\nlf$ of leaves in tree-shaped CQs), and (\emph{ii}) the existential depth $\od$ of ontologies (i.e., the length of the longest chain of labelled nulls in the chase on any data). The resulting landscape, displayed in Fig.~\ref{pic:results}~(a) (under the assumption that datasets are given as RDF graphs and $\M$ is the identity)~\cite{CDLLR07,DBLP:conf/dlog/KikotKZ11,LICS14,DBLP:conf/lics/BienvenuKP15}, indicates three tractable cases:
\begin{description}
\item[$\bdObtwCQ$:] ontologies of depth $\le \od$ coupled with CQs of treewidth $\le \twi$ (for fixed $\od,\twi$);
\item[$\bdOblCQ$:] ontologies of  depth $\le \od$  with tree-shaped CQs with $\le\nlf$ leaves (for fixed $\od, \nlf$);
\item[$\blCQ$:] ontologies of arbitrary depth and tree-shaped CQs with $\le\nlf$ leaves (for fixed $\nlf$).
\end{description}
Observe in particular that when the ontology depth is bounded by a fixed constant, the complexity of OMQ answering is precisely the same as for evaluating the underlying CQs.
If we place no restriction on the ontology, then tractability of tree-shaped queries can be recovered by bounding the number of leaves, but we have
\LOGCFL\  rather than the expected \NL.

While the results in Fig.~\ref{pic:results}(a) appear to answer the question of the additional cost incurred by adding an \OWLQL{} ontology, they only tell part of the story. Indeed, in the context of classical rewriting-based OBDA \cite{PLCD*08}, it is not the abstract complexity of OMQ answering that matters, but the cost  of computing and evaluating OMQ rewritings.
Fig.~\ref{pic:results}(b) summarises what is known  about the size of positive existential (PE), nonrecursive datalog (NDL) and FO-rewritings~\cite{DBLP:conf/icalp/KikotKPZ12,DBLP:journals/ai/GottlobKKPSZ14,LICS14,DBLP:conf/lics/BienvenuKP15}.
Thus, we see, for example, that PE-rewritings for OMQs from $\bdObtwCQ$
can be of super-polynomial size, and so are not computable and evaluable in polynomial time, even though Fig.~\ref{pic:results}(a) shows that such OMQs can be answered in \LOGCFL. The same concerns  $\bdOblCQ$ and $\blCQ$, which
can be answered in \NL{} and \LOGCFL, respectively, but do not enjoy polynomial-size PE-rewritings.
Moreover, our experiments show that standard rewriting engines exhibit exponential behaviour on OMQs drawn from $\textsf{OMQ}(1,1,2)$ lying in the intersection of the three tractable classes.

Our first aim is to show that the positive complexity results in Fig.~\ref{pic:results}(a) can in fact be achieved using query rewriting.
To this end, we develop \NDL-rewritings for the three tractable cases
that can be computed and evaluated by algorithms of optimal combined complexity.
In theory, such algorithms are known to be space efficient and highly parallelisable.
We demonstrate practical efficiency of our optimal NDL-rewritings by comparing them with the NDL-rewritings produced by Clipper~\cite{DBLP:conf/aaai/EiterOSTX12}, Presto~\cite{DBLP:conf/kr/RosatiA10} and Rapid~\cite{DBLP:conf/cade/ChortarasTS11}, using a sequence of OMQs from the class \textsf{OMQ}(1,1,2).

Our second aim is to understand the contribution of the ontology depth and the number of leaves in tree-shaped CQs to the complexity of OMQ answering. (As follows from Fig.~\ref{pic:results} (a), if these parameters are unbounded, this problem is harder than evaluating the underlying CQs unless $\LOGCFL = \NP$.) Unfortunately, it turns out that answering OMQs with ontologies of finite depth and tree-shaped CQs is not fixed-parameter tractable if either the ontology depth or the number of leaves in CQs is regarded as a parameter. More precisely, we prove that the problem is $W[2]$-hard in the former case and $W[1]$-hard in the latter. We also construct an ontology $\T$ (of infinite depth) such that answering OMQs $(\T,\q)$ with tree-shaped CQs $\q$ is W[1]-hard if the number of leaves in $\q$ is regarded as the parameter. These results suggest that the ontology depth and the number of leaves are inherently in the exponent of the size of the input
in any OMQ answering algorithm.

Finally, we revisit the \NP{}- and \LOGCFL-hardness results for OMQs with tree-shaped CQs.
The known \NP{} and \LOGCFL{} lower bounds
have been established using sequences $(\T_n,\q_n)$ of OMQs, where the depth of $\T_n$ grows with $n$~\cite{DBLP:conf/dlog/KikotKZ11,DBLP:conf/lics/BienvenuKP15}. One might thus hope to make answering OMQs with tree-shaped CQs easier by restricting the ontology signature, size, or even by fixing the whole ontology, which is very relevant
for applications as a typical OBDA scenario has users posing different queries over the same ontology. Our third main result is that this is not the case: we present ontologies $\T_\dag$ and $\T_\ddag$ of infinite depth such that
answering OMQs $(\T_\dag, \q)$ with tree-shaped $\q$ and $(\T_\ddag, \q)$ with linear $\q$ is \NP- and \LOGCFL-hard for query complexity, respectively.
We also show that no algorithm can construct FO-rewritings of the OMQs $(\T_\dag, \q)$ in polynomial time unless $\PTime = \NP$, even though polynomial-size FO-rewritings of these OMQs do exist.

The paper is organised as follows.
We begin in Section 2 by introducing the \OWLQL{} ontology language
and key notions like OMQ answering and query rewriting.
In Section 3, we first identify fragments of NDL which can be evaluated in
\LOGCFL\ or \NL, and then we use these results to develop NDL-rewritings of
optimal combined complexity for the three tractable cases.
Section 4 concerns the parameterised complexity of OMQ answering with tree-shaped CQs.
For ontologies of finite depth, we show $W[2]$-hardness (resp.\ $W[1]$-hardness) when
the ontology depth (resp.\ number of leaves) is taken as the parameter.
For the infinite depth case, we show in Section 5 that \NP{}-hardness applies even for a fixed ontology.
The final section of the paper presents preliminary experiments comparing our new rewritings to those
produced by existing rewriting engines and discusses possible directions for future work.


\section{Preliminaries}\label{sec:prelims}

An \OWLQL{} \emph{ontology} (\emph{TBox} in description logic), $\T$, is a finite set of sentences (\emph{axioms}) of the forms
\begin{align*}
& \forall x\, (\tau(x) \to \tau'(x)), &
& \forall x\, (\tau(x) \land \tau'(x) \to \bot ),\\
& \forall xy\, (\varrho(x,y) \to \varrho'(x,y)), &
& \forall xy\, (\varrho(x,y) \land \varrho'(x,y) \to \bot),\\
& \forall x\, \varrho(x,x), &
& \forall x\, (\varrho(x,x) \to \bot),
\end{align*}
where $\tau(x)$ and $\varrho(x,y)$
are defined, using unary predicates $A$ and binary predicates $P$, by the grammars
\begin{align*}
 \tau(x) \ &::= \ \top \ \mid \ A(x) \ \mid \ \exists y\,\varrho(x,y) , \\
\varrho(x,y) \ &::= \ \top \ \mid \ P(x,y) \ \mid \ P(y,x).
%
\end{align*}
When writing ontology axioms, we omit the universal quantifiers and 
denote by $\rni$ the set of binary predicates $P$ occurring in $\T$ and their inverses $P^-$, assuming that $P^{--} = P$.
For every $\varrho\in\rni$, we take a fresh unary predicate $A_\varrho$ and add $A_\varrho(x) \leftrightarrow \exists y\, \varrho(x,y)$ to $\T$
 (where, as usual, $\varphi \leftrightarrow\psi$ is  an abbreviation for $\varphi \rightarrow \psi$ and $\psi \rightarrow \varphi$).  The resulting ontology is said to be in \emph{normal form}, and we assume, without loss of generality, that all our ontologies are in normal form.

A \emph{data instance}, $\A$, is a finite set of unary or binary ground atoms (called an \emph{ABox} in description logic).
We denote by $\ind(\A)$ the set of individual constants in $\A$ and
write $\varrho(a,b)\in \A$ if $P(a,b)\in \A$ and $\varrho= P$, or $P(b,a)\in\A$ and $\varrho = P^-$. We say that $\A$ is \emph{complete} for an ontology $\T$ if $\T,\A \models S(\avec{a})$ implies $S(\avec{a}) \in \A$, for any ground atom $S(\avec{a})$ with $\avec{a} \subseteq \ind(\A)$.\!\footnote{If the meaning is clear from the context, we use set-theoretic notation for lists.}

A \emph{conjunctive query} (CQ) $\q(\avec{x})$ is a formula of the form $\exists \avec{y}\, \varphi(\avec{x}, \avec{y})$, where $\varphi$ is a conjunction of atoms $S(\avec{z})$ all of whose variables are among $\vars(\q) = \avec{x} \cup \avec{y}$.
We assume, without loss of generality, that CQs contain no constants. We often regard a CQ as the set of its atoms.
With every CQ $\q$, we associate its \emph{Gaifman graph} $\gfmn$ whose vertices are the variables of $\q$ and whose edges are the  pairs $\{u,v\}$ such that $P(u,v)\in\q$, for some $P$. We call $\q$  \emph{connected} if $\gfmn$ is connected, \emph{tree-shaped} if $\gfmn$ is a tree, and \emph{linear} if $\gfmn$ is a tree with two leaves.

An \emph{ontology-mediated query} (OMQ) is a pair $\omq(\avec{x}) = (\T,\q(\avec{x}))$, where $\T$ is an ontology and $\q(\avec{x})$ a CQ. A tuple $\avec{a} \subseteq \ind (\mathcal{A})$ is a \emph{certain answer} to $\omq(\avec{x})$ over a data instance $\mathcal{A}$ if  $\mathcal{I} \models \q(\avec{a})$ for all models $\mathcal{I}$ of $\T$ and $\A$; in this case we write \mbox{$\T,\A \models \q(\avec{a})$}. If $\avec{x} = \emptyset$, then a certain answer to $\omq$ over $\A$ is `yes' if $\T,\A \models \q$ and `no' otherwise. The \emph{OMQ answering problem} (for a class of OMQs) is to decide whether $\T,\A \models \q(\avec{a})$ holds, given an OMQ $\omq(\avec{x})$ (in the class), $\A$ and $\avec{a} \subseteq \ind(\A)$.
If $\T$, $\q(\avec{x})$, and $\A$ are regarded as input, we speak about \emph{combined complexity} of OMQ answering; if $\A$ and $\T$ are regarded as fixed, we speak about \emph{query complexity}.

Every consistent \emph{knowledge base} (KB) $(\T,\A)$ has a \emph{canonical model} (or \emph{chase} in database  theory)~\cite{Abitebouletal95} $\mathcal{C}_{\T,\A}$ with the property that \mbox{$\T,\A \models \q(\avec{a})$} iff $\mathcal{C}_{\T,\A} \models \q(\avec{a})$, for all CQs $\q(\avec{x})$ and $\avec{a} \subseteq \ind(\A)$. In our constructions, we use the following definition of $\C_{\T,\A}$, where without loss of generality we assume that $\T$ contains no binary predicates $P$ with $\T \models \forall xy \, P(x,y)$.
The domain, $\Delta^{\mathcal{C}_{\T,\A}}$, consists of $\ind(\A)$ and the \emph{witnesses} (or \emph{labelled nulls}) of the form $w = a \varrho_1 \dots \varrho_n$, for $n \geq 1$, such that
\begin{itemize}
\item[--] $a \in \ind(\A)$ and $\T, \A \models \exists y\, \varrho_1 (a,y)$;
\item[--] $\T\not\models \varrho_i(x,x)$, for $1 \le i \le n$;
\item[--] $\T \models \exists x\, \varrho_i(x,y) \to \exists z\, \varrho_{i+1}(y,z)$ but $\T \not \models \varrho_i(x,y) \to \varrho_{i+1}(y,x)$, for $1 \leq i < n$.
\end{itemize}
We denote by $\twords$ the set of words $\varrho_1 \dots \varrho_n \in \rni^*$ satisfying the last two conditions. Every $a \in \ind(\A)$ is interpreted in $\mathcal{C}_{\T,\A}$ by itself, and unary and binary predicates are interpreted as follows:
\begin{itemize}
\item[--] $\mathcal{C}_{\T,\A} \models A(u)$ iff either $u \in \ind(\A)$ and \mbox{$\T,\A \models A(u)$}, or $u = w\varrho$ with $\T \models \exists y\,\varrho(y,x) \to A(x)$;
\item[--] $\mathcal{C}_{\T,\A} \models P(u,v)$ iff one of the three conditions holds: (\emph{i}) $u,v \in \ind(\A)$ and $\T,\A \models P(u,v)$; (\emph{ii})~$u=v$ and $\T\models P(x,x)$; (\emph{iii}) $\T \models \varrho(x,y) \to P(x,y)$ and either $v = u\varrho$ or $u = v\varrho^-$.
\end{itemize}
We say that $\T$ is \emph{of depth} 0 if it does not contain any axioms with $\exists$ on the right-hand side, excepting the normalisation axioms\footnote{This somewhat awkward definition of depth 0 ontologies is due to the use of normalisation axioms, which may introduce unnecessary words on length 1 in $\twords$.}\!.
Otherwise, we say that $\T$ is \emph{of depth} $0 < \od< \infty$ if $\od$ is the maximum length of the words in $\twords$,
and it is \emph{of depth} $\infty$ if $\twords$ is infinite. 
(Note that the depth of $\T$ is computable in \NL;
cf.~\cite{DBLP:conf/icalp/GogaczM14,DBLP:conf/pods/CalauttiGP15} for related results on chase termination for tgds.)

An FO-formula $\q'(\avec{x})$, possibly with equality, is an \emph{FO-rewriting of an OMQ $\omq(\avec{x}) = (\T,\q(\avec{x}))$} if, for \emph{any} data instance $\A$ and any tuple $\avec{a} \subseteq \ind(\A)$,
\begin{equation}\label{def:rewriting}
\T, \A \models \q(\avec{a}) \qquad \text{iff} \qquad \I_\A \models \q'(\avec{a}),
\end{equation}
where $\I_\A$ is the FO-structure over the domain $\ind(\A)$ such that $\I_\A \models S(\avec{a})$ iff $S(\avec{a}) \in \A$, for any ground atom $S(\avec{a})$. If $\q'(\avec{x})$ is a positive existential formula, we call it a \emph{PE-rewriting of $\omq(\avec{x})$}. A PE-rewriting whose matrix is a $\mathsf{\Pi}_k$-formula (with respect to $\land$ and $\lor$) is called a $\mathsf{\Pi}_k$-\emph{rewriting}.
The size $|\q'|$ of $\q'$ is the number of symbols in it.

We also consider rewritings in the form of nonrecursive datalog  queries.
A \emph{datalog program}, $\Pi$, is a finite set of Horn clauses
$\forall \avec{z}\, (\gamma_0 \leftarrow \gamma_1 \land \dots \land \gamma_m)$,
where each $\gamma_i$ is an atom $Q(\avec{y})$ with $\avec{y} \subseteq  \avec{z}$ or an equality $(z = z')$ with $z,z'\in \avec{z}$. (As usual, we omit $\forall \avec{z}$ from clauses.) The atom $\gamma_0$ is the \emph{head} of the clause, and $\gamma_1,\dots,\gamma_m$ its  \emph{body}. All variables in the head must occur in the body, and $=$ can only occur in the body.
The predicates in the heads of clauses in $\Pi$ are \emph{IDB predicates},  the rest (including~$=$) \emph{EDB predicates}.
A predicate $Q$ \emph{depends} on $P$ in $\Pi$ if $\Pi$ has a clause with $Q$ in the head and $P$ in the body. $\Pi$ is a \emph{nonrecursive datalog} (NDL) \emph{program} if the (directed)  \emph{dependence graph} of the  dependence relation is acyclic.

An \emph{NDL query} is a pair $(\Pi,G(\avec{x}))$, where $\Pi$ is an NDL program and $G(\avec{x})$ a predicate. A tuple $\avec{a} \subseteq \ind(\A)$ is an \emph{answer to $(\Pi,G(\avec{x}))$ over} a data instance $\A$ if $G(\avec{a})$ holds in the first-order structure with domain $\ind(\A)$ obtained by closing $\A$ under the clauses in $\Pi$; in this case we write $\Pi, \A \models G(\avec{a})$. The problem of checking whether  $\avec{a}$ is an answer to $(\Pi,G(\avec{x}))$ over $\A$ is called the \emph{query evaluation problem}.
The \emph{depth} of $(\Pi,G(\avec{x}))$ is the length, $\dep(\Pi,G)$, of the longest directed path in the dependence graph for $\Pi$ starting from $G$.
NDL queries are \emph{equivalent} if they have exactly the same answers over any data instance.

An NDL query $(\Pi,G(\avec{x}))$ is an \emph{\NDL-rewriting of an OMQ $\omq(\avec{x}) = (\T,\q(\avec{x}))$ over complete data instances} in case $\T,\A\models \q(\avec{a})$ iff  $\Pi,\mathcal{A} \models G(\avec{a})$, for any complete $\A$ and any $\avec{a} \subseteq \ind (\mathcal{A})$.
Rewritings \emph{over arbitrary data instances} are defined by dropping the completeness condition.
Given an \NDL-rewriting $(\Pi, G(\avec{x}))$ of $\omq(\avec{x})$ over complete data instances, we denote by $\Pi^*$ the result of replacing each predicate $S$ in $\Pi$ with a fresh IDB predicate $S^*$ of the same arity and adding the clauses
\begin{align*}
A^*(x) & \leftarrow \tau(x), && \text{ if } \ \T \models \tau(x) \to A(x),\\
P^*(x,y) & \leftarrow \varrho(x,y), && \text{ if } \ \T \models \varrho(x,y) \to P(x,y),\\
P^*(x,x) & \leftarrow \top(x), &&\text{ if } \ \T\models P(x,x),
\end{align*}
where $\top(x)$ is an EDB predicate for the active domain \cite{DBLP:journals/ai/KaminskiNG16}.
Clearly, $(\Pi^*,G(\avec{x}))$ is an \NDL-rewriting  of $\omq(\avec{x})$ over arbitrary data instances and $|\Pi^*| \leq |\Pi| + |\T|^2$.

Finally, we remark that, without loss of generality, we can (and will) assume that our ontologies $\T$ do not contain $\bot$. Indeed, we can always incorporate into rewritings subqueries that check whether the left-hand side of an axiom with $\bot$ holds and output all tuples of constants if this is the case~\cite{DBLP:journals/ws/CaliGL12}.


%

\section{Optimal NDL-Rewritings\label{sec:3}}

To construct theoretically optimal NDL-rewritings for OMQs in the three tractable  classes, we first identify two types of \NDL{} queries whose 
evaluation problems are in \NL{} and \LOGCFL{} for combined complexity.


\subsection{NL and LOGCFL fragments of NDL}\label{sec3}

To simplify the analysis of non-Boolean NDL queries, it is convenient to regard certain variables 
as parameters to be instantiated with constants from the candidate answer.
Formally,  an NDL query $(\Pi,G(x_1,\dots,x_n))$ is called
\emph{ordered} if each of its IDB predicates $Q$ comes with fixed variables $x_{i_1},\dots,x_{i_k}$ ($1 \le i_1 < \dots < i_k \le n$), called the \emph{parameters of} $Q$,  such that
(\emph{i}) every occurrence of $Q$ in $\Pi$ is of
the form $Q(y_1,\dots,y_m,x_{i_1},\dots,x_{i_k})$,
(\emph{ii}) the parameters of $G$ are $x_1,\dots,x_n$,
and
(\emph{iii})
parameters of the head of every clause include all the parameters of the predicates in the body.
Observe that Boolean NDL queries are trivially ordered.
The \emph{width} $\wid(\Pi,G)$ of an ordered $(\Pi,G)$ is the maximal number of non-parameter variables in a clause of $\Pi$.
\begin{example}\em
The NDL query $(\Pi,G(x))$, where
\begin{equation*}
\Pi = \{\,G(x) \leftarrow R(x,y)\land Q(x), \ Q(x) \leftarrow R(y,x)\,\},
\end{equation*}
is ordered with parameter $x$ and width 1 (the conditions do not restrict the EDB predicate $R$).
Replacing $Q(x)$ by $Q(y)$ in the first clause yields a query that is not ordered in view of (\emph{i}). A further swap of $Q(x)$ in the second clause with $Q(y)$ would satisfy (\emph{i}) but not~(\emph{iii}).
\end{example}

As all the NDL-rewritings we construct are ordered, with their parameters being the answer variables, from now on we only consider ordered NDL queries.

Given an NDL query $(\Pi,G(\avec{x}))$, a data instance $\A$ and a tuple $\avec{a}$ with $|\avec{x}| = |\avec{a}|$, the \emph{$\avec{a}$-grounding  $\Pi_{\A}^{\avec{a}}$ of $\Pi$ on $\A$} is the set of ground clauses obtained by first replacing each parameter in $\Pi$ by the corresponding constant from $\avec{a}$,
and then performing the standard grounding~\cite{DBLP:journals/csur/DantsinEGV01} of $\Pi$ using the constants from $\A$. The size of $\Pi_{\A}^{\avec{a}}$ is bounded by $|\Pi|\cdot |\A|^{\wid(\Pi,G)}$, and so we can check whether $\Pi,\mathcal{A} \models G(\avec{a})$ holds in time $\textit{poly}(|\Pi|\cdot |\A|^{\wid(\Pi,G)})$.

\subsubsection{Linear NDL in NL}

An NDL program is \emph{linear}~\cite{Abitebouletal95} if the body of its every  clause contains at most one IDB predicate.
\begin{theorem}\label{linear-nl}
For any $\wid > 0$, evaluation of linear NDL queries of width $\le \wid$ is \NL-complete for combined complexity.
\end{theorem}
\begin{proof}
Let $(\Pi,G(\avec{x}))$ be a linear NDL query.
Deciding whether $\Pi, \A \models G(\avec{a})$ is reducible to finding a path to $G(\avec{a})$ from a certain set $X$ in the grounding graph $\mathfrak{G}$ constructed as follows. The vertices of $\mathfrak{G}$ are the IDB atoms of $\Pi_{\A}^{\avec{a}}$, and
$\mathfrak{G}$ has an edge from $Q(\avec{c})$ to $Q'(\avec{c}')$
iff $\Pi_\A^{\avec{a}}$ contains 
$Q'(\avec{c}')\leftarrow Q(\avec{c})\land S_1(\avec{c}_1)\land \dots \land S_k(\avec{c}_k)$
with $S_i(\avec{c}_i) \in \A$, for $1\le i \le k$ (we assume $\A$ contains all  $c=c$, for $c\in\ind(\A)$).
The set $X$ consists of all vertices $Q(\avec{c})$ with IDB predicates $Q$ being
of in-degree 0 in the dependency graph of $\Pi$ for which there is a clause
$Q(\avec{c}) \leftarrow S_1(\avec{c}_1)\land \dots \land S_k(\avec{c}_k)$ in $\Pi_\A^{\avec{a}}$
with $S_i(\avec{c}_i) \in \A$ \mbox{($1\le i \le k$)}.
Bounding the width of $(\Pi,G)$ ensures that
$\mathfrak{G}$ is of polynomial size and can be constructed by a deterministic Turing machine with read-only input, write-once output and logarithmic-size work tapes.
\end{proof}

The transformation $^*$ of NDL-rewritings over complete data instances into NDL-rewritings over arbitrary data instances does not preserve linearity. A more involved construction is given in the proof of the following:
\begin{lemma}\label{linear-arbitrary}
Fix any $\wid > 0$. There is an $\mathsf{L}^\NL$-transducer that, for any linear \NDL-rewriting $(\Pi, G(\avec{x}))$ of an OMQ $\omq(\avec{x})$ over complete data instances with $\wid(\Pi,G)\leq \wid$, computes a linear NDL-rewriting $(\Pi',G(\avec{x}))$ of $\omq(\avec{x})$ over arbitrary data instances such that $\wid(\Pi',G) \leq \wid+1$.
\end{lemma}

We note that a possible increase of the width by 1 is due to
the `replacement' of unary atoms $A(z)$ by binary atoms $\varrho(y,z)$ whenever $\T\models \exists y \,\varrho(y,z) \to A(z)$.

\subsubsection{Skinny NDL in LOGCFL}

The complexity class \LOGCFL{} can be defined using
\emph{nondeterministic auxiliary pushdown automata} (NAuxPDAs)~\cite{DBLP:journals/jacm/Cook71}, which are nondeterministic Turing machines with an additional work tape constrained to operate as a pushdown store. Sudborough~\cite{sudborough78} proved that \LOGCFL{} coincides with the class of problems that are solved by NAuxPDAs  in logarithmic space and polynomial time (the space on the pushdown tape is not subject to the logarithmic  bound).
It is known that
\LOGCFL{} can equivalently be defined  in terms of logspace-uniform
families of semi-unbounded fan-in circuits (where \textsc{or}-gates have arbitrarily many inputs, and \textsc{and}-gates two inputs) of polynomial size and logarithmic depth. Moreover,  there is an algorithm that, given such a circuit $\cir$, computes the output using an NAuxPDA
in logarithmic space in the size of $\cir$ and exponential time in the depth of $\cir$~\cite[pp.~392--397]{DBLP:journals/jcss/Venkateswaran91}.

Similarly to the restriction on the circuits for \LOGCFL, we call an NDL query $(\Pi,G)$ \emph{skinny} if the body of any  clause in $\Pi$ has at most two atoms.
\begin{lemma}\label{prop:SkinnyNDLEvaluation}
For any skinny $(\Pi, G(\avec{x}))$ and any data instance~$\A$, query evaluation can be done by an NAuxPDA in space $\log |\Pi| + \wid(\Pi,G) \cdot \log |\A|$ and time $2^{O(\dep(\Pi,G))}$.
\end{lemma}
\begin{proof}
Using the atoms of the grounding $\Pi_{\A}^{\avec{a}}$ as gates and inputs, we define a monotone Boolean circuit $\cir$ as follows:
its output is $G(\avec{a})$; for every atom $\gamma$ in the head of a clause in
 $\Pi_{\A}^{\avec{a}}$, we take an \textsc{or}-gate whose output is~$\gamma$ and inputs are the bodies of the clauses with head~$\gamma$; for every such body, we take an \textsc{and}-gate whose inputs are the atoms in the body. We set  input  $\gamma$ to 1 iff $\gamma \in \A$. Clearly, $\cir$ is a semi-unbounded fan-in circuit of depth $O(\dep(\Pi,G))$ with $O(|\Pi| \cdot\nolinebreak |\A|^{\wid(\Pi,G)})$ gates.
Having observed that our $\cir$ can be computed by a deterministic logspace Turing machine, we conclude that the query evaluation problem can be solved by an NAuxPDA in the required space and time.
\end{proof}

Observe that Lemma~\ref{prop:SkinnyNDLEvaluation} holds for NDL queries with any \emph{bounded} number of atoms, not only two. In the rewritings we propose in Sections~\ref{sec:boundedtw} and \ref{sec:boundedleaf}, however, the number of atoms in the clauses is not bounded by a constant.
We  require the following notion to generalise skinny programs.
A function $\nu$ from the predicate names in $\Pi$ to $\mathbb{N}$ is called a \emph{weight function for} an NDL query $(\Pi, G(\avec{x}))$ if
\begin{equation*}
\nu(Q) > 0 \quad\text{ and }\quad \nu(Q) \ \ge \ \nu(P_1) + \dots + \nu(P_k),
\end{equation*}
for any  clause $Q(\avec{z}) \leftarrow P_1(\avec{z}_1) \land \dots \land P_k(\avec{z}_k)$ in $\Pi$. Note that $\nu(P)$ can be $0$ for an EDB predicate $P$. To illustrate, we consider NDL queries with the  following dependency graphs:\\[3pt]
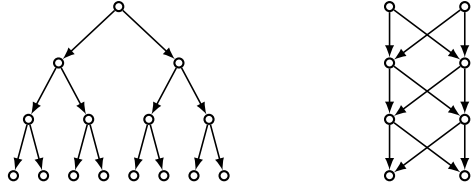
\centerline{%
\begin{tikzpicture}[nd/.style={circle, draw, thick, inner sep=0mm, minimum size=1.2mm},yscale=0.75,>=latex]
\node[nd] (a000) at (0,0) {};
\node[nd] (a001) at (0.4,0) {};
\node[nd] (a010) at (0.8,0) {};
\node[nd] (a011) at (1.2,0) {};
\node[nd] (a100) at (1.6,0) {};
\node[nd] (a101) at (2,0) {};
\node[nd] (a110) at (2.4,0) {};
\node[nd] (a111) at (2.8,0) {};
\node[nd] (a00) at (0.2,1) {};
\node[nd] (a01) at (1,1) {};
\node[nd] (a10) at (1.8,1) {};
\node[nd] (a11) at (2.6,1) {};
\node[nd] (a0) at (0.6,2) {};
\node[nd] (a1) at (2.2,2) {};
\node[nd] (a) at (1.4,3) {};
\begin{scope}[semithick]
\draw[<-] (a000) -- (a00);
\draw[<-] (a001) -- (a00);
\draw[<-] (a010) -- (a01);
\draw[<-] (a011) -- (a01);
\draw[<-] (a100) -- (a10);
\draw[<-] (a101) -- (a10);
\draw[<-] (a110) -- (a11);
\draw[<-] (a111) -- (a11);
\draw[<-] (a00) -- (a0);
\draw[<-] (a01) -- (a0);
\draw[<-] (a10) -- (a1);
\draw[<-] (a11) -- (a1);
\draw[<-] (a0) -- (a);
\draw[<-] (a1) -- (a);
\node[nd] (b00) at (5,0) {};
\node[nd] (b01) at (6,0) {};
\node[nd] (b10) at (5,1) {};
\node[nd] (b11) at (6,1) {};
\node[nd] (b20) at (5,2) {};
\node[nd] (b21) at (6,2) {};
\node[nd] (b30) at (5,3) {};
\node[nd] (b31) at (6,3) {};
\draw[->] (b10) -- (b00);
\draw[->] (b11) -- (b00);
\draw[->] (b10) -- (b01);
\draw[->] (b11) -- (b01);
\draw[->] (b20) -- (b10);
\draw[->] (b21) -- (b10);
\draw[->] (b20) -- (b11);
\draw[->] (b21) -- (b11);
\draw[->] (b30) -- (b20);
\draw[->] (b31) -- (b20);
\draw[->] (b30) -- (b21);
\draw[->] (b31) -- (b21);
\end{scope}
\end{tikzpicture}}\\[4pt]
The NDL on the left has a weight function bounded by the number of predicates, and so, such weight functions are linear in the size of the query; intuitively, this function corresponds to the number of directed paths from a vertex to the leaves. In contrast, any NDL query with the dependency graph on the right can only have a weight function whose values (numbers of paths) are exponential.
Also observe that linear NDL queries have weight functions bounded by $1$.

We show, using Huffman coding, that any NDL query $(\Pi,G(\avec{x}))$ can be transformed into an equivalent skinny NDL query whose depth increases linearly in addition to the logarithms of  the weight function and the number $\textsf{e}_\Pi$ of EDB predicates in a clause. We call the minimum  (over possible weight functions $\nu$) value of $2\dep(\Pi,G) + \log \nu (G) + \log \textsf{e}_\Pi$ the \emph{skinny depth} of $(\Pi,G)$ and denote it by $\sdep(\Pi,G)$.

\begin{lemma}\label{thm:NDLToSkinny}
Any NDL query $(\Pi, G(\avec{x}))$
is equivalent to a skinny NDL query  $(\Pi',G(\avec{x}))$ such that $|\Pi'| = O(|\Pi|^2)$, $\dep(\Pi',G) \le \sdep(\Pi,G)$, and $\wid(\Pi',G) \le \wid(\Pi,G)$.
\end{lemma}
\begin{proof}
Let $\nu$ be a weight function such that $\sdep(\Pi,G) = 2\dep(\Pi,G) + \log \nu (G) + \log \textsf{e}_\Pi$.  Without loss of generality, we will assume that $\nu(E) = 0$, for EDB predicates $E$. First, we split  clauses into their  EDB and IDB components: each $Q(\avec{z}) \leftarrow \varphi(\avec{z}')$ is replaced by $Q(\avec{z}) \leftarrow Q_E(\avec{z}_E) \land Q_I(\avec{z}_E')$ and $Q_\alpha(\avec{z}_\alpha) \leftarrow \varphi_\alpha(\avec{z}_\alpha')$, for $\alpha \in \{ E, I\}$, where $Q_E$ and $Q_I$ are fresh predicates, and $\varphi_E(\avec{z}_E')$ and $\varphi_I(\avec{z}_I')$ are conjunctions of the EDB and IDB predicates in $\varphi$, respectively. The depth of the resulting NDL query $(\Pi_*,G(\avec{x}))$ is $2\dep(\Pi,G)$. Next, each clause $Q_E(\avec{z}_E) \leftarrow \varphi_E(\avec{z}_E')$ in $\Pi_*$ is replaced by $\leq \textsf{e}_\Pi - 1$ clauses with at most two atoms in the body, which results in an NDL query of depth not exceeding $2\dep (\Pi,G) + \log \textsf{e}_\Pi$. In the rest of the proof, we concentrate on the part $\Pi_\dagger$ of $\Pi_*$ comprising clauses that have predicates $Q$ and $Q_I$ in their heads (thus making the $Q_E$ EDB predicates). The weight function for $(\Pi_\dagger, G(\avec{x}))$ is obtained by extending $\nu$ as follows: we set $\nu(Q_I) = \nu(Q)$ and $\nu(Q_E) = 0$, for each $Q$.

Next, by induction on $\dep(\Pi_\dagger,G)$, we show that there is an equivalent skinny NDL query $(\Pi_\dagger',G(\avec{x}))$ of the required size and width and such that $\dep(\Pi'_\dagger,G) \leq \dep(\Pi_\dagger,G) + \log \nu(G)$. We take $\Pi'_\dagger=\Pi_\dagger$ if $\dep(\Pi_\dagger,G)=0$.
Otherwise, let  $\psi$ be a clause of the form $G(\avec{z}) \leftarrow P_1(\avec{z}_1) \land \dots \land
P_k(\avec{z}_k)$  in $\Pi_\dagger$, for $k > 2$.
Since, by the construction of $\Pi_\dagger$, if a clause has an EDB predicate, then $k = 2$. So,
the $P_i$ are IDB predicates and $\nu(G) \geq \nu(P_i) > 0$.
Suppose that, for each $i$ ($1\le i \le k$), we have an NDL query $(\Pi'_i,P_i)$ equivalent to $(\Pi_\dagger, P_i)$  with
\begin{equation}\label{eq:depPj}
\dep(\Pi'_i,P_i) ~\le~ \dep(\Pi_\dagger,P_i) + \log \nu(P_i) ~\le~ \dep(\Pi_\dagger,G) - 1 + \log \nu(P_i).
\end{equation}
Construct the Huffman tree \cite{huf52} for the alphabet $\{1,\dots,k\}$, where the frequency of $i$ is $\nu(P_i)/\nu(G)$. For example, for $\nu(G) = 39$, $\nu(P_1) = 15$, $\nu(P_2) = 7$, $\nu(P_3) = 6$, $\nu(P_4) = 6$ and $\nu(P_5) = 5$, we obtain the following tree:\\[5pt]
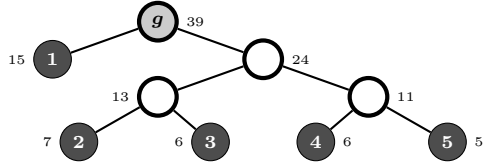
\centerline{%
\begin{tikzpicture}[nd/.style={circle,draw,color=white,draw=black,fill=black!70,inner sep=1pt,minimum size=5mm},yscale=1,xscale=1.4]\scriptsize
\node[nd,fill=gray!40,ultra thick,label=right:{\tiny 39}] (G) at (0,0) {\textcolor{black}{$\boldsymbol{g}$}};
\node[nd,label=left:{\tiny 15}] (P1) at (-1,-0.5) {$\boldsymbol{1}$};
\node[nd,fill=white,ultra thick,label=right:{\tiny 24}] (i2) at (1,-0.5) {};
\node[nd,fill=white,ultra thick,label=left:{\tiny 13}] (i3) at (0,-1) {};
\node[nd,fill=white,ultra thick,label=right:{\tiny 11}] (i4) at (2,-1) {};
\node[nd,label=left:{\tiny 7}] (P2) at (-0.75,-1.6) {$\boldsymbol{2}$};
\node[nd,label=left:{\tiny 6}] (P3) at (0.5,-1.6) {$\boldsymbol{3}$};
\node[nd,label=right:{\tiny 6}] (P4) at (1.5,-1.6) {$\boldsymbol{4}$};
\node[nd,label=right:{\tiny 5}] (P5) at (2.75,-1.6) {$\boldsymbol{5}$};
\begin{scope}[thick]
\draw (G) -- (P1);
\draw (G) -- (i2);
\draw (i2) -- (i3);
\draw (i2) -- (i4);
\draw (i3) -- (P2);
\draw (i3) -- (P3);
\draw (i4) -- (P4);
\draw (i4) -- (P5);
\end{scope}
\end{tikzpicture}
}
\\[5pt]
In general, the Huffman tree is a binary tree with $k$ leaves $1,\dots,k$, a root $g$ and $k-2$ internal nodes and such that the length of the path from $g$ to any leaf $i$ is bounded by $\lceil \log(\nu(G)/\nu(P_i))\rceil$. For each internal node $v$ of the tree, we take a predicate $P_v(\avec{z}_v)$, where $\avec{z}_v$ is the union of $\avec{z}_u$ for all descendants $u$ of~$v$; for the root~$g$, we take $P_g(\avec{z}_g) = G(\avec{z})$. Let $\Pi'_\psi$ be the extension of the union of the $\Pi'_i$ ($1\leq i \leq k$) with clauses $P_v(\avec{z}_{v}) \leftarrow P_{u_1}(\avec{z}_{u_1}) \land P_{u_2}(\avec{z}_{u_2})$, for each $v$ with immediate successors  $u_1$ and $u_2$.
The number of the new clauses is $k-1$.
By~\eqref{eq:depPj}, we  have:
\begin{multline*}
\dep(\Pi'_\psi,G) \le \max\nolimits_i \{ \lceil \log (\nu(G) / \nu(P_i)) \rceil + \dep(\Pi'_i,P_i) \}\\
\hspace*{4mm}\le
\max\nolimits_i \{ \log (\nu(G) / \nu(P_i)) + \dep(\Pi_\dagger,G) + \log \nu(P_i) \} \ \  = \ \
 \dep(\Pi_\dagger,G) + \log \nu(G).
\end{multline*}
Let $\Pi'_\dagger$ be the result of applying this transformation to
each clause in $\Pi_\dagger$  with head $G(\avec{z})$ and more than two atoms in the body.

Finally, we add to $\Pi_\dagger'$ the clauses with the $Q_E$ predicates and denote the result by $\Pi'$. It is readily seen that $(\Pi',G)$ is as required; in particular, $|\Pi'|= O(|\Pi|^2)$.
\end{proof}

We now use Lemmas~\ref{prop:SkinnyNDLEvaluation} and~\ref{thm:NDLToSkinny} to obtain the following:
\begin{theorem}\label{thm:LOGCFLdatalog}
For every $c > 0$ and $\wid > 0$, evaluation of NDL queries $(\Pi, G(\avec{x}))$
of width at most~$\wid$ and such that $\sdep(\Pi,G) \le c \log|\Pi|$
is in \LOGCFL{} for combined complexity.
\end{theorem}

We say that a class of OMQs is \emph{skinny-reducible} if, for some fixed $c > 0$ and $\wid > 0$, there is an $\textsf{L}^{\smash{\LOGCFL}}$-transducer that, given any  OMQ $\omq(\avec{x})$ in the class, computes its NDL-rewriting $(\Pi,G(\avec{x}))$ over complete data instances
such that $\sdep(\Pi,G)\le c\log|\Pi|$ and $\wid(\Pi,G) \le \wid$.
Theorem~\ref{thm:LOGCFLdatalog} and the transformation $^*$ give the following:
\begin{corollary}\label{cor:weight}
For any skinny-reducible class, the OMQ answering problem is in \LOGCFL{} for combined complexity.
\end{corollary}

In the following subsections, we will exploit the results obtained above to construct optimal \NDL-rewritings for the three classes of tractable OMQs. Appendix A.6 gives concrete examples of our rewritings.


\subsection{LOGCFL rewritings for OMQ({$\displaystyle\od,\, \twi,\, \infty$})}\label{sec:boundedtw}

Recall~(see, e.g.,~\cite{DBLP:series/txtcs/FlumG06}) that a \emph{tree decomposition} of an undirected graph $\gfmn=(V,E)$ is a pair $(T,\lambda)$, where $T$ is an (undirected) tree and $\lambda$ a function from the nodes of $T$ to $2^V$ such that
\begin{itemize}
\item[--] for every $v \in V$, there exists a node $\nd$ with $v \in \lambda(\nd)$;

\item[--] for every $e \in E$, there exists a node $\nd$ with $e \subseteq \lambda(\nd)$;

\item[--] for every $v \in V$, the nodes $\{\nd\mid v \in \lambda(\nd)\} $ induce a connected subgraph of~$T$ (called a \emph{subtree} of $T$).
\end{itemize}
We call the set $\lambda(\nd) \subseteq V$ a \emph{bag for} $\nd$. The \emph{width} of $(T, \lambda)$ is $\max_{\nd\in T} |\lambda(\nd)| - 1$. The \emph{treewidth of a graph} $\gfmn$ is the minimum width over all tree decompositions of $\gfmn$. The \emph{treewidth of a CQ} is the treewidth of its Gaifman graph.

\begin{example}\em\label{ex:rewriting:1}
Consider the CQ $\q(x_0, x_7)$ depicted below (black nodes represent answer variables):\\[10pt]
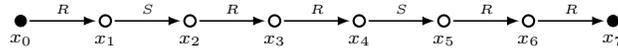
\centerline{%
\begin{tikzpicture}[>=latex,xscale=0.75]\scriptsize
\node[bpoint,label=below:{$x_0$}] (v0) at (0,0) {};
\node[wpoint,label=below:{$x_1$}] (v1) at (1.5,0) {};
\node[wpoint,label=below:{$x_2$}] (v2) at (3,0) {};
\node[wpoint,label=below:{$x_3$}] (v3) at (4.5,0) {};
\node[wpoint,label=below:{$x_4$}] (v4) at (6,0) {};
\node[wpoint,label=below:{$x_5$}] (v5) at (7.5,0) {};
\node[wpoint,label=below:{$x_6$}] (v6) at (9,0) {};
\node[bpoint,label=below:{$x_7$}] (v7) at (10.5,0) {};
\begin{scope}[semithick,shorten >= 1pt, shorten <= 1pt]\tiny
\draw[->] (v0) to node[above] {$R$} (v1);
\draw[->] (v1)to node[above] {$S$}  (v2);
\draw[->] (v2) to node[above] {$R$}  (v3);
\draw[->] (v3) to node[above] {$R$}  (v4);
\draw[->] (v4) to node[above] {$S$}  (v5);
\draw[->] (v5) to node[above] {$R$}  (v6);
\draw[->] (v6) to node[above] {$R$}  (v7);
\end{scope}
\end{tikzpicture}}\\[10pt]
Its natural tree decomposition of treewidth 1 is based on the chain  $T$ of 7~vertices shown as bags below:\\[10pt]
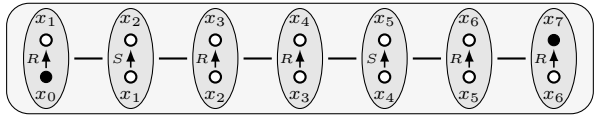
\centerline{%
\begin{tikzpicture}[>=latex,yscale=0.7,xscale=0.75]\scriptsize
\draw[rounded corners=3mm,fill=gray!7] (-0.7,-0.55) rectangle (9.7,1.55);
\foreach \x in {0,1.5,3,4.5,6,7.5,9} {
\draw[fill=gray!20,thin] (\x,0.5) ellipse (0.4 and 0.95);
}
\foreach \x in {0,1.5,3,4.5,6,7.5} {
\draw[thick] (\x+0.5,0.5) -- +(0.5,0);
}
\node[bpoint,label=below:{$x_0$}] (v0) at (0,0.15) {};
\node[wpoint,label=above:{$x_1$}] (v1p) at (0,0.85) {};
\node[wpoint,label=below:{$x_1$}] (v1) at (1.5,0.15) {};
\node[wpoint,label=above:{$x_2$}] (v2p) at (1.5,0.85) {};
\node[wpoint,label=below:{$x_2$}] (v2) at (3,0.15) {};
\node[wpoint,label=above:{$x_3$}] (v3p) at (3,0.85) {};
\node[wpoint,label=below:{$x_3$}] (v3) at (4.5,0.15) {};
\node[wpoint,label=above:{$x_4$}] (v4p) at (4.5,0.85) {};
\node[wpoint,label=below:{$x_4$}] (v4) at (6,0.15) {};
\node[wpoint,label=above:{$x_5$}] (v5p) at (6,0.85) {};
\node[wpoint,label=below:{$x_5$}] (v5) at (7.5,0.15) {};
\node[wpoint,label=above:{$x_6$}] (v6p) at (7.5,0.85) {};
\node[wpoint,label=below:{$x_6$}] (v6) at (9,0.15) {};
\node[bpoint,label=above:{$x_7$}] (v7p) at (9,0.85) {};
\begin{scope}[semithick,shorten >= 1pt, shorten <= 1pt]\tiny
\draw[->] (v0) to node[left] {$R$} (v1p);
\draw[->] (v1)to node[left] {$S$}  (v2p);
\draw[->] (v2) to node[left] {$R$}  (v3p);
\draw[->] (v3) to node[left] {$R$}  (v4p);
\draw[->] (v4) to node[left] {$S$}  (v5p);
\draw[->] (v5) to node[left] {$R$}  (v6p);
\draw[->] (v6) to node[left] {$R$}  (v7p);
\end{scope}
\end{tikzpicture}%
}%
\end{example}

In this section, we prove the following:
\begin{theorem}\label{thm:logcfl-1}
For any fixed $\od \ge 0$ and $\twi \ge 1$, the class \textup{$\bdObtwCQ$} is skinny-reducible.
\end{theorem}

In a nutshell, we split recursively a given CQ $\q$ into sub-CQs $\q_D$ based on subtrees $D$ of the tree decomposition  of $\q$, and combine their rewritings into a rewriting of $\q$. To guarantee compatibility of these rewritings, we use `boundary conditions' $\tpd$ that describe the types of points on the boundaries of the $\q_D$ and, for each possible boundary condition $\tpd$, we define recursively a fresh IDB predicate $\rpred^{\tpd}_D$. We now formalise the construction and illustrate it using the CQ from Example~\ref{ex:rewriting:1}.

Fix a connected CQ $\q(\avec{x})$ and a tree decomposition $(T, \lambda)$ of its Gaifman graph $\gfmn = (V,E)$.
Let $D$ be a subtree of $T$. The \emph{size} of $D$ is the number of nodes in it.
A node $\nd$ of $D$ is called \emph{boundary} if $T$ has an edge $\{\nd,\nd'\}$ with $\nd'\notin D$. The \emph{degree} $\degree(D)$ of $D$ is the number of its boundary nodes ($T$ itself  is the only subtree of $T$ of degree  $0$).  We say that a node $\nd$ \emph{splits} $D$ into subtrees $D_1,\dots,D_k$ if the $D_i$ partition $D$ without~$\nd$: each node of $D$ except $\nd$ belongs to exactly one $D_i$.
\begin{lemma}[\cite{DBLP:conf/lics/BienvenuKP15}]\label{l:6.8}
Let $D$ be a subtree of $T$ of size $n > 1$.
If $\degree(D) =2$, then there is a node $\nd$ splitting $D$ into subtrees of size $\leq n/2 $ and degree~$\leq 2$ and, possibly, one subtree of size $<n-1$ and degree~$1$.
If $\degree(D) \leq 1$, then there is $\nd$ splitting $D$ into subtrees of size $\leq n/2 $ and degree $\leq 2$.
\end{lemma}

In Example~\ref{ex:rewriting:1}, $t$ splits $T$ into $D_1$ and $D_2$ as follows:\\[10pt]
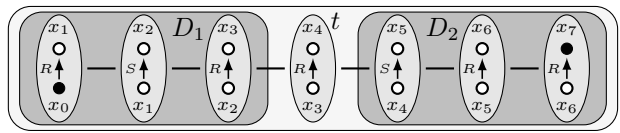
\centerline{%
\begin{tikzpicture}[>=latex,yscale=0.75,xscale=0.75]\scriptsize
\draw[rounded corners=3mm,fill=gray!7] (-0.9,-0.6) rectangle (9.9,1.6);
\draw[rounded corners=3mm,fill=gray!50] (-0.7,-0.5) rectangle (3.7,1.5);
\node at (2.3,1.2) {\normalsize $D_1$};
\draw[rounded corners=3mm,fill=gray!50] (5.3,-0.5) rectangle (9.7,1.5);
\node at (6.8,1.2) {\normalsize $D_2$};
\node at (4.9,1.35) {\normalsize $t$};
\foreach \x in {0,1.5,3,4.5,6,7.5,9} {
\draw[fill=gray!20,thin] (\x,0.5) ellipse (0.4 and 0.95);
}
\foreach \x in {0,1.5,3,4.5,6,7.5} {
\draw[thick] (\x+0.5,0.5) -- +(0.5,0);
}
\node[bpoint,label=below:{$x_0$}] (v0) at (0,0.15) {};
\node[wpoint,label=above:{$x_1$}] (v1p) at (0,0.85) {};
\node[wpoint,label=below:{$x_1$}] (v1) at (1.5,0.15) {};
\node[wpoint,label=above:{$x_2$}] (v2p) at (1.5,0.85) {};
\node[wpoint,label=below:{$x_2$}] (v2) at (3,0.15) {};
\node[wpoint,label=above:{$x_3$}] (v3p) at (3,0.85) {};
\node[wpoint,label=below:{$x_3$}] (v3) at (4.5,0.15) {};
\node[wpoint,label=above:{$x_4$}] (v4p) at (4.5,0.85) {};
\node[wpoint,label=below:{$x_4$}] (v4) at (6,0.15) {};
\node[wpoint,label=above:{$x_5$}] (v5p) at (6,0.85) {};
\node[wpoint,label=below:{$x_5$}] (v5) at (7.5,0.15) {};
\node[wpoint,label=above:{$x_6$}] (v6p) at (7.5,0.85) {};
\node[wpoint,label=below:{$x_6$}] (v6) at (9,0.15) {};
\node[bpoint,label=above:{$x_7$}] (v7p) at (9,0.85) {};
\begin{scope}[semithick,shorten >= 1pt, shorten <= 1pt]\tiny
\draw[->] (v0) to node[left] {$R$} (v1p);
\draw[->] (v1)to node[left] {$S$}  (v2p);
\draw[->] (v2) to node[left] {$R$}  (v3p);
\draw[->] (v3) to node[left] {$R$}  (v4p);
\draw[->] (v4) to node[left] {$S$}  (v5p);
\draw[->] (v5) to node[left] {$R$}  (v6p);
\draw[->] (v6) to node[left] {$R$}  (v7p);
\end{scope}
\end{tikzpicture}}

\smallskip

We define recursively a set $\R$ of subtrees of $T$, a binary `predecessor' relation $\prec$ on $\R$,  and a function $\sigma$ on $\R$ indicating the splitting node. We begin by adding $T$ to $\R$. Take any $D\in \R$ that has not been split yet. If $D$  is of size~1, then $\sigma(D)$ is the only node of $D$. Otherwise, by Lemma~\ref{l:6.8}, we find a node $\nd$ in $D$ that splits it into $D_1,\dots,D_k$. We set $\sigma(D) = t$ and, for  $1\leq i\leq k$, add  $D_i$ to $\R$ and set $D_i \prec D$; then, we apply the procedure recursively to each of $D_1,\dots,D_k$.
In Example~\ref{ex:rewriting:1} with $\nd$ splitting $T$, we have $\sigma(T) = t$, $D_1 \prec T$ and
$D_2 \prec T$.

For each $D\in\R$, we recursively define a set of atoms
\begin{equation*}
\q_D \ \ = \ \ \bigl\{S(\avec{z}) \in \q \mid \avec{z} \subseteq \lambda(\sigma(D)) \bigr\} \ \cup \ \bigcup_{D' \prec D} \q_{D'}.
\end{equation*}
By the definition of tree decomposition, $\q_T = \q$. Denote by $\avec{x}_D$ the subset of  $\avec{x}$ that occurs in $\q_D$. In Example~\ref{ex:rewriting:1}, $\avec{x}_{T} = \{x_0, x_7\}$,  $\avec{x}_{D_1} = \{x_0\}$
and $\avec{x}_{D_2} = \{x_7\}$.
Let $\dD$  be the union of all $\lambda(\nd) \cap\lambda(\nd')$ for
boundary nodes $\nd$ of $D$ and its neighbours $\nd'$ in $T$ \emph{outside}~$D$.
In our example, $\partial T=\emptyset$, $\partial D_1 =\{x_3\}$ and $\partial D_2 =\{x_4\}$.

Let $\T$ be an ontology of depth $\le \od$. A \emph{type} is a partial map $\avec{w}$ from
$V$ to $\twords$; its domain is denoted by $\dom(\avec{w})$. The unique partial type with $\dom(\avec{\varepsilon}) = \emptyset$ is denoted by $\avec{\varepsilon}$.
We use types to represent how variables are mapped into $\can$, with $\avec{w}(z)=w$ indicating that $z$ is mapped to an element of the form $a w$ (for some $a \in \ind(\A)$), and
with $\avec{w}(z)=\varepsilon$ that $z$ is mapped to an individual constant.
We say that a type $\avec{w}$ is \emph{compatible} with
a bag $\nd$ if, for all $y,z\in\lambda(\nd)\cap\dom(\avec{w})$, we have
\begin{itemize}
\item[--]  if $z \in \avec{x}$, then $\avec{w}(z) = \varepsilon$;
\item[--] if $A(z) \in \q$, then either $\avec{w}(z)=\varepsilon$ or
$\avec{w}(z)= w \varrho $ with $\T \models \exists y\,\varrho (y,x) \to  A(x)$;
\item[--] if $P(y, z)\in \q$, then one of the three conditions holds:
(\emph{i}) $\avec{w}(y) = \avec{w}(z) =\varepsilon$; (\emph{ii}) $\avec{w}(y) = \avec{w}(z)$ and \mbox{$\T\models P(x,x)$}; (\emph{iii}) $\T\models \varrho(x,y) \to P(x,y)$ and either $\avec{w}(z) = \avec{w}(y) \varrho$ or
$\avec{w}(y) = \avec{w}(z) \varrho^-$.
\end{itemize}

In the sequel we abuse notation and use sets of variables in place of sequences assuming that they are ordered in some (fixed) way. For example, we use $\avec{x}_D$ for a tuple of variables in the set $\avec{x}_D$ (ordered in some way). Also, given a tuple $\avec{a}\in \ind(\A)^{|\avec{x}_D|}$  and $x\in\avec{x}_D$, we write $\avec{a}(x)$ to refer to the component of $\avec{a}$ that corresponds to $x$ (that is, the component with the same index).

We now define an NDL-rewriting of $\omq(\avec{x}) = (\T,\q(\avec{x}))$. For any $D\in\R$ and type $\tpd$ with $\dom (\tpd)=\dD$, let $\rpred^{\tpd}_D(\dD, \avec{x}_D)$ be a fresh IDB predicate with parameters~$\avec{x}_D$ (note that $\dD$ and $\avec{x}_D$ may be not disjoint). For each type
$\tpr$ with $\dom(\tpr) = \lambda(\sigma(D))$ such that $\tpr$ is compatible with $\sigma(D)$ and agrees with $\tpd$ on their common domain, the NDL program $\Pi_{\omq}^{\textsc{Log}}$ contains
\begin{equation*}
\rpred^{\tpd}_D(\dD, \avec{x}_D)  \leftarrow  \mathsf{At}^{\tpr} \ \land
\bigwedge_{D' \prec D} \rpred^{(\tpr\cup\tpd) \restr\dDp}_{D'}(\dDp,\avec{x}_{D'}),
\end{equation*}
where $(\tpr\cup\tpd) \restr\dDp$ is
the restriction of the union \mbox{$\tpr\cup\tpd$} to $\dDp$ (since
$\dom(\tpr\cup\tpd)$  covers $\dDp$,
the domain of
the restriction is $\dDp$),
and $\mathsf{At}^{\tpr}$ is the conjunction of
\begin{itemize}
\item[(a)] $A(z)$, for  $A(z)\in\q$ with $\tpr(z) = \varepsilon$, and $P(y, z)$, for  $P(y,z)\in \q$ with $\tpr(y) = \tpr(z) = \varepsilon$;
\item[(b)] $y =  z$, for $P(y,z)\in \q$ with  $\tpr(y) \ne\varepsilon$ or $\tpr(z) \ne \varepsilon$;
\item[(c)] $A_\varrho(z)$, for $z$ with $\tpr(z) = \varrho w$, for some $w$.
\end{itemize}
The conjuncts in~(a) ensure that atoms all of whose variables are assigned $\varepsilon$
hold in the data instance. The conjuncts in~(b) ensure that if one variable in a binary atom is not mapped to $\varepsilon$,
then the images of both its variables share the same initial individual.
Finally, the conjuncts in~(c) ensure that if a variable is to be mapped to $a\varrho w$,
then  $a\varrho w$ is indeed in the domain of $\can$.
\begin{example}\em\label{ex:rewriting:2}
With the query in Example~\ref{ex:rewriting:1}, consider now the following ontology $\T$:
\begin{align*}
P(x,y) & \to  S(x,y), \quad  & A_P(x) & \leftrightarrow \exists y \, P(x,y),\\
P(x,y) & \to R(y,x), \quad &
A_{P^-}(x) & \leftrightarrow \exists y \, P(y,x)
\end{align*}
(the remaining normalisation axioms are omitted).
Since $\lambda(\nd) = \{ x_3, x_4\}$, there are two types compatible with~$\nd$ that can contribute to the rewriting:
$\tpr_1 = \{ x_3\mapsto\nolinebreak \varepsilon, \ x_4\mapsto \varepsilon\}$
and
$\tpr_2 = \{ x_3 \mapsto \varepsilon, \ x_4\mapsto P^-\}$.
So we have
$\mathsf{At}^{\tpr_1}  = R(x_3, x_4)$ and
$\mathsf{At}^{\tpr_2}  = A_{P^-}(x_4) \land (x_3 = x_4)$.
Thus, the predicate $\rpred^{\avec{\varepsilon}}_{T}$ is defined by two clauses with the head $\rpred^{\avec{\varepsilon}}_{T} (x_0,x_7)$ and the following bodies:
\begin{align*}
& \rpred^{x_3\mapsto\varepsilon}_{D_1} \!(x_3, x_0)\land R(x_3, x_4)\land \rpred^{x_4\mapsto\varepsilon}_{D_2}\!(x_4, x_7),\\
& \rpred^{x_3\mapsto\varepsilon}_{D_1} \!(x_3,x_0)\land A_{P^-}(x_4)\land (x_3 =  x_4)\land\rpred^{x_4\mapsto P^-}_{D_2}\!\! (x_4, x_7),
\end{align*}
for $\tpr_1$ and~$\tpr_2$, respectively.
Although \mbox{$\{ x_3 \mapsto P, \ x_4\mapsto \varepsilon\}$} is also compatible with $\nd$, its predicate $\rpred^{x_3\mapsto P}_{D_1}$ will have no definition in the rewriting, and hence can be omitted. The same is true of the other compatible  types
$\{x_3 \mapsto \varepsilon, \ x_4\mapsto R \}$ and $\{x_3 \mapsto R^-, \ x_4\mapsto \varepsilon \}$.
\end{example}

By induction on $\prec$, one can now show that $(\Pi^{\textsc{Log}}_\omq, \rpred^{\avec{\varepsilon}}_T)$ is a rewriting of $\omq(\avec{x})$; see Appendix~\ref{appA2} for details.

%

Now fix $\od$ and $\twi$, and consider $\omq(\avec{x}) = (\T,\q(\avec{x}))$ from $\bdObtwCQ$. Let $T$ be a tree decomposition of $\q$ of tree\-width~$\leq \twi$; we may assume without loss of generality that $T$ has at most $|\q|$ nodes. We take the following weight function:  $\nu(\rpred^{\tpd}_D) = |D|$, where $|D|$ is the size of $D$, that is, the number of nodes in it. Clearly, $\nu(\rpred^{\avec{\varepsilon}}_T) \le |\omq|$.  By Lemma~\ref{l:6.8}, we have
\begin{align*}
& \wid(\Pi^{\smash{\textsc{Log}}}_\omq,\rpred^{\avec{\varepsilon}}_T)\le \max_{D} |\dD \cup \lambda(\sigma(D))| \le 3(\twi + 1),\\
&\sdep(\Pi^{\smash{\textsc{Log}}}_\omq,\rpred^{\avec{\varepsilon}}_T) \le 4 \log |T| + 2\log |\omq| \le 6 \log |\omq|.
\end{align*}
Since $|\R| \le |T|^2$ and
there are at most $|\T|^{2\od(\twi+1)}$ options for~$\tpd$, there are polynomially
many predicates $\rpred^{\tpd}_D$, and so $\Pi^{\smash{\textsc{Log}}}_\omq$ is of polynomial size. Thus, by Corollary~\ref{cor:weight}, the constructed NDL-rewriting over arbitrary data instances  can be evaluated in \LOGCFL{}. Finally, we note that a tree decomposition of treewidth $\leq \twi$ can be computed using an $\textsf{L}^{\smash{\LOGCFL}}$-transducer~\cite{DBLP:conf/icalp/GottlobLS99}, and so the NDL-rewri\-ting can also be constructed by an $\textsf{L}^{\smash{\LOGCFL}}$-transducer.

The obtained NDL-rewriting shows that answering OMQs $(\T,\q(\avec{x}))$ with $\T$ of finite depth $\od$ and $\q$ of tree\-width $\twi$ over any data instance $\A$ can be done  in time
\begin{equation}\label{eq:time1}
\textit{poly}(|\T|^{\od \twi},\,|\q|,\,|\A|^{\twi}).
\end{equation}
Indeed, we can evaluate $(\Pi_\omq^{\smash{\textsc{Log}}},\rpred^{\avec{\varepsilon}}_T(\avec{x}))$ in time polynomial in $|\Pi_\omq^{\smash{\textsc{Log}}}|$ and $|\A|^{\wid(\Pi^{\smash{\textsc{Log}}}_{\omq},\rpred^{\avec{\varepsilon}}_T)}$, which are bounded by a polynomial in
$|\T|^{2\od(\twi+1)}$, $|\q|$ and $|\A|^{2(\twi+1)}$.


\subsection{NL rewritings for OMQ($\displaystyle\od,\, 1,\, \nlf$)}\label{sec:5}

\begin{theorem}\label{thm:ndl}
Let $\od \ge 0$ and $\nlf \ge 2$ be fixed. There is an $\mathsf{L}^{\smash{\NL}}$-transducer that, given an OMQ $\omq = (\T,\q(\avec{x}))$ in \textup{$\bdOblCQ$}, constructs its polynomial-size linear NDL-rewriting of width $\le 2\nlf$.
\end{theorem}

Let $\T$ be an ontology of finite depth $\od$,
and let $\q(\avec{x})$ be a tree-shaped CQ with at most $\nlf$ leaves.
Fix one of the variables of $\q$ as root, and let $M$ be the maximal distance to a leaf from the root.
For $0 \leq n \leq M$, let
$\avec{z}^n$ denote the set of all variables of $\q$ at distance $n$
from the root; clearly, $|\avec{z}^n| \le \nlf$.
We call the $\avec{z}^n$ \emph{slices} of $\q$ and observe that
they satisfy the following: for every
$P(z,z') \in \q$ with $z \neq z'$, there exists $n < M$ such that
\begin{equation*}
\text{either } z\in \avec{z}^n \text{ and } z'\in \avec{z}^{n+1} \ \ \text{ or } \ \ z'\in\avec{z}^n \text{ and } z\in \avec{z}^{n+1}.
\end{equation*}
For $0 \leq n \leq M$, let $\q_n(\avec{z}^n_\exs, \avec{x}^n)$ be the query
consisting of all atoms $S(\avec{z})$ of $\q$ such that $\avec{z} \subseteq \bigcup_{n \leq k \leq M} \avec{z}^k$,
where
$\avec{x}^n$ is the subset of $\avec{x}$ that occurs in $\q_n$ and
$\avec{z}^n_\exs = \avec{z}^n \setminus \avec{x} $.

By a \emph{type for slice} $\avec{z}^n$, we mean
a total map $\avec{w}$ from $\avec{z}^n$ to $\twords$.
Analogously to Section~\ref{sec:boundedtw},
we define the notions of types compatible with slices.
Specifically,  we call $\avec{w}$ \emph{locally compatible} with $\avec{z}^{n}$ if for every $z \in \avec{z}^n$:
\begin{itemize}
\item[--]  if $z \in \avec{x}$, then $\avec{w}(z) = \varepsilon$;
\item[--] if $A(z) \in \q$, then either $\avec{w}(z)= \varepsilon$ or
$\avec{w}(z)= w\varrho$ with $\T\models \exists y\,\varrho(y,x) \to A(x)$;
\item[--] if $P(z, z) \in \q$, then either $\avec{w}(z) = \varepsilon$ or $\T\!\models\! P(x,x)$.
\end{itemize}
If $\tpd, \tpr$ are types for $\avec{z}^{n}$ and $\avec{z}^{n+1}$, respectively,
then we say $(\tpd, \tpr)$ \emph{is compatible} with
$(\avec{z}^{n}, \avec{z}^{n+1})$ if
 $\avec{w}$ is locally compatible with $\avec{z}^{n}$,
$\tpr$ is locally compatible with $\avec{z}^{n+1}$,
\begin{itemize}
\item[--] for every
$P(z, z') \in \q$ with $z\in\avec{z}^n$ and \mbox{$z'\in\avec{z}^{n+1}$}, one of the three condition holds:
$\tpd(z) =\tpr(z') = \varepsilon$, or $\tpd(z) =\tpr(z')$ with $\T\models P(x,x)$, or \mbox{$\T\models \varrho(x,y) \to P(x,y)$} with either
$\tpr(z')= \tpd(z) \varrho$ or
$\tpd(z) = \tpr(z') \varrho^-$.
\end{itemize}

Consider the \NDL{} program $\Pi_\omq^{\textsc{Lin}}$ defined as follows.
For every $0 \leq n < M$
and every pair of types $(\avec{w}, \tpr)$ that is compatible with $(\avec{z}^n, \avec{z}^{n+1})$, we include the clause
\begin{equation*}
G^{\avec{w}}_{n}(\avec{z}^{n}_\exs, \avec{x}^n) \leftarrow
\mathsf{At}^{\avec{w}\cup\tpr}(\avec{z}^n,\avec{z}^{n+1})  \land G^{\tpr}_{n+1}(\avec{z}^{n+1}_\exs, \avec{x}^{n+1}),
\end{equation*}
where $\avec{x}^n$ are the parameters of
$G^{\avec{w}}_{n}$ and $\mathsf{At}^{\avec{w}\cup\tpr}(\avec{z}^n,\avec{z}^{n+1})$
is the conjunction of atoms~(a)--(c) as defined in Section~\ref{sec:boundedtw}, for the union $\avec{w}\cup\tpr$.
For every type $\avec{w}$ locally compatible with $\avec{z}^M$, we include the clause
\begin{equation*}
G^{\avec{w}}_{M}(\avec{z}^{M}_\exs, \avec{x}^M) \leftarrow \mathsf{At}^{\avec{w}}(\avec{z}^M).
\end{equation*}
(Recall that $\avec{z}^M$ is a disjoint union of $\avec{z}^M_\exs$ and $\avec{x}^M$.) We use $G$ with parameters $\avec{x}$ as the goal predicate and include
 $G(\avec{x}) \leftarrow G^{\avec{w}}_{0}(\avec{z}^{0}_\exs, \avec{x})$
for every predicate $G^{\avec{w}}_{0}$ occurring in the head of one of the preceding clauses.

By induction on $n$, we show in Appendix~\ref{appA3} that  $(\Pi_\omq^{\textsc{Lin}},G(\avec{x}))$
is a rewriting of $(\T, \q(\avec{x}))$ over complete data instances.
%
%
%
It should be clear that $\Pi^{\textsc{Lin}}_\omq$ is a linear NDL program
of  width $\le 2 \nlf$ and containing $\le |\q|\cdot |\T|^{\smash{2\od\nlf}}$ predicates. Moreover,
it takes only logarithmic space to store a type~$\avec{w}$, which allows us to show
that
 $\Pi^{\textsc{Lin}}_\omq$ can be computed
by an $\mathsf{L}^{\smash{\NL}}$-transducer. We apply Lemma~\ref{linear-arbitrary} to obtain
an NDL-rewriting for arbitrary data instances,
and then use Theorem~\ref{linear-nl} to conclude that the resulting
program can be evaluated in \NL.

The obtained  NDL-rewriting shows  that answering OMQs $(\T,\q(\avec{x}))$
with $\T$ of finite depth $\od$ and tree-shaped $\q$ with $\nlf$ leaves over any data $\A$ can be done  in time
\begin{equation}\label{eq:time2}
\textit{poly}(|\T|^{\od\nlf},\, |\q|,\, |\A|^{\nlf} ).
\end{equation}
Indeed, $(\Pi_\omq^{\smash{\textsc{Lin}}},\rpred(\avec{x}))$ can be evaluated  in time polynomial
in $|\Pi^{\smash{\textsc{Lin}}}_{\omq}|$ and $|\A|^{\wid(\Pi^{\smash{\textsc{Lin}}}_{\omq},G)}$, which are bounded by a polynomial in $|\T|^{2\od\nlf}$, $|\q|$ and $|\A|^{2\nlf}$.

\subsection{LOGCFL rewritings for OMQ($\displaystyle\infty,\,1,\,\nlf$)}\label{sec:boundedleaf}

Unlike the previous two classes, answering OMQs in $\blCQ$ can be  harder---\LOGCFL-complete---than evaluating their CQs, which can be done in \NL.
\begin{theorem}\label{thm:logcfl-2}
For any fixed $\nlf \ge 2$, \textup{$\blCQ$} is skinny-reducible.
\end{theorem}

For OMQs with bounded-leaf CQs and ontologies of unbounded depth,
our rewriting uses the notion of tree witness~\cite{DBLP:conf/kr/KikotKZ12}.
Consider an OMQ $\omq(\avec{x}) = (\T,\q(\avec{x}))$.
Let  $\t = (\tr, \ti)$ be a pair of disjoint sets of variables in $\q$ such that  $\ti \ne\emptyset$ but $\ti\cap \avec{x} = \emptyset$.
Set
\begin{equation*}
\q_\t \ = \ \bigl\{\, S(\avec{z}) \in \q \mid \avec{z} \subseteq \tr\cup \ti \text{ and } \avec{z}\not\subseteq \tr\,\bigr\}.
\end{equation*}
If $\q_\t$ is a minimal subset of $\q$ containing every atom of~$\q$ with a variable from $\ti$ and such that there is a homomorphism $h \colon \q_\t  \to \C_{\T, \{A_\varrho(a)\}}$ with $h^{-1}(a) = \tr$, we call $\t$ a \emph{tree witness for $\omq(\avec{x})$ generated by~$\varrho$}. Intuitively, $\t$ identifies a minimal subset of $\q$ that can be mapped to the tree-shaped part of the canonical model consisting of labelled nulls: the variables in $\tr$ are mapped to an individual constant, say, $a$, at the root of a tree and the~$\ti$ are mapped to the labelled nulls of the form $a w$, for some $w\in\twords$ that begins with $\varrho$.
Note that the same tree witness can be generated by different~$\varrho$.

The logarithmic-depth NDL-rewriting for OMQs from $\blCQ$  is based on the following observation:
\begin{lemma}[\cite{LICS14}]\label{PrepMiddleVertex}
Every tree $T$ of size $n$ has a node splitting it into subtrees of size
$\leq\! \lceil n/2 \rceil$.
\end{lemma}

Let $\omq(\avec{x}_0)=(\T, \q_0(\avec{x}_0))$ be an OMQ with a tree-shaped CQ.
We will repeatedly apply Lemma~\ref{PrepMiddleVertex} to decompose the CQ into smaller and smaller subqueries.
Formally, for a tree-shaped CQ $\q$, we denote by $z_\q$ a vertex in the Gaifman graph $\gfmn$ of $\q$
that satisfies the condition of Lemma~\ref{PrepMiddleVertex};
if $|\vars(\q)| = 2$ and $\q$ has at least one existentially quantified variable, then
we assume that $z_\q$ is such.
Let $\sqset$ be the smallest set
that contains $\q_0(\avec{x}_0)$ and the following CQs,
for every $\q(\avec{x}) \in \sqset$ with existentially quantified variables:\label{page:decomposition}
\begin{itemize}
\item[--] for each $z_i$ adjacent to $z_\q$ in $\gfmn$, the CQ $\q_i(\avec{x}_i)$
comprising all binary atoms with both $z_i$ and $z_\q$,
and all atoms whose variables cannot reach $z_\q$ in $\gfmn$ without passing by $z_i$,
where $\avec{x}_i$ is the set of variables in $\avec{x} \cup \{z_\q\}$ that occur in $\q_i$;
\item[--] for each tree witness $\t$ for $(\T,\q(\avec{x}))$ with
$\tr\neq \emptyset$ and $z_\q\in\ti$, the CQs $\q_1^\t(\avec{x}_1^\t), \dots, \q_k^\t(\avec{x}_k^\t)$ that correspond
to the connected components of the set of atoms of $\q$ that are not in $\q_\t$, where each $\avec{x}_i^\t$
is the set of variables in $\avec{x} \cup \tr$ that occur in  $\q_i^\t$.
\end{itemize}
The two cases are depicted below:\\[10pt]
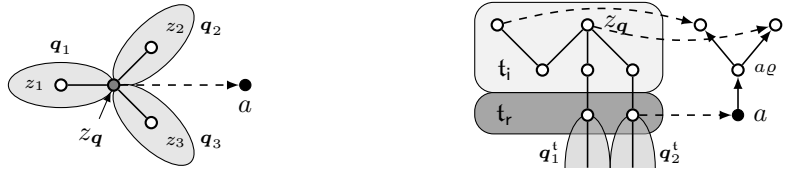
\centerline{%
\begin{tikzpicture}[>=latex]%
\draw[fill=gray!20] (-0.7,0) ellipse (0.7 and 0.3);
\draw[fill=gray!20,rotate=45] (0.75,0) ellipse (0.7 and 0.3);
\draw[fill=gray!20,rotate=-45] (0.75,0) ellipse (0.7 and 0.3);
\node at (-0.7,0.5) {\scriptsize $\q_1$};
\node at (1.3,0.75) {\scriptsize $\q_2$};
\node at (1.3,-0.75) {\scriptsize $\q_3$};
\node[wpoint,fill=gray] (zq) at (0,0) {};
\node (zql) at ($(zq)+(-0.3,-0.7)$) {$z_\q$};
\draw[->,thin] (zql) -- (zq);
\node[wpoint,label=left:{\scriptsize $z_1$}] (z1) at (-0.7,0) {};
\node[wpoint,label=above right:{\scriptsize $z_2$}] (z2) at (0.5,0.5) {};
\node[wpoint,label=below right:{\scriptsize $z_3$}] (z3) at (0.5,-0.5) {};
\begin{scope}[semithick]
\draw (zq) -- (z1);
\draw (zq) -- (z2);
\draw (zq) -- (z3);
\end{scope}
\node[bpoint,label=below:{$a$}] (a) at (1.75,0) {};
\draw[semithick,dashed,->] (zq) -- (a);
\begin{scope}[xshift=23mm]
\draw[rounded corners=2.5mm,fill=gray!70] (2.5,-0.65) rectangle +(2.5,0.55);
\draw[rounded corners=2.5mm,fill=gray!10] (2.5,-0.1) rectangle +(2.5,1.2);
\node at (2.9,-0.4) {$\tr$};
\node at (2.9,0.2) {$\ti$};
\begin{scope}
\clip (3,-1.1) rectangle +(2,1);
\draw[fill=gray!50,fill opacity=0.5] (4.6,-1.1) ellipse (0.3 and 0.7);
\draw[fill=gray!50,fill opacity=0.5] (4,-1.1) ellipse (0.3 and 0.7);
\end{scope}
\node at (3.5,-0.9) {\scriptsize $\q_1^\t$};
\node at (5.1,-0.9) {\scriptsize $\q_2^\t$};
\node[wpoint,label=right:{$z_\q$}] (zqp) at (4,0.8) {};
\node[wpoint] (zi6) at (2.8,0.8) {};
\node[wpoint] (zi1) at (3.4,0.2) {};
\node[wpoint] (zi2) at (4,0.2) {};
\node[wpoint] (zi3) at (4.6,0.2) {};
\node[wpoint] (zr4) at (4,-0.4) {};
\node[wpoint] (zr5) at (4.6,-0.4) {};
\begin{scope}[semithick]
\draw (zqp) -- (zi1);
\draw (zqp) -- (zi2);
\draw (zqp) -- (zi3);
\draw (zi1) -- (zi6);
\draw (zi2) -- (zr4);
\draw (zi3) -- (zr5);
\draw (zr4) -- ++(0,-0.7);
\draw (zr5) -- ++(0,-0.7);
\end{scope}
\node[bpoint,label=right:{$a$}] (ap) at (6,-0.4) {};
\node[wpoint,label=right:{\scriptsize$a\varrho$}] (ar) at (6,0.2) {};
\node[wpoint] (ar1) at (6.5,0.8) {};
\node[wpoint] (ar2) at (5.5,0.8) {};
\begin{scope}[semithick,->]
\draw (ap) -- (ar);
\draw (ar) -- (ar1);
\draw (ar) -- (ar2);
\end{scope}
\draw[semithick,dashed,->] (zr5) -- (ap);
\draw[semithick,dashed,->,out=-15,in=-165] (zqp) to (ar1);
\draw[semithick,dashed,->,out=15,in=165] (zi6) to (ar2);
\end{scope}
\end{tikzpicture}
}\\%
Note that $\tr\ne\emptyset$ ensures that part of the query without $\q_\t$ is mapped onto individual constants.

The NDL program $\Pi^{\textsc{Tw}}_{\omq}$ uses IDB predicates $\rew_\q(\avec{x})$, for $\q(\avec{x}) \in \sqset$,
whose parameters are the variables in $\avec{x}_0$ that occur in $\q(\avec{x})$. For each $\q(\avec{x}) \in \sqset$, if it has no existentially quantified variables, then we include the clause $\rew_\q(\avec{x}) \leftarrow \q(\avec{x})$.
Otherwise, we include the clause
\begin{equation*}
\rew_\q(\avec{x}) \ \ \ \leftarrow
\bigwedge_{S(\avec{z})\in\q, \ \avec{z}\subseteq\{z_\q\}}\hspace*{-2em} S(\avec{z})
 \ \ \  \land \bigwedge_{1\leq i\leq n} \rew_{\q_i}(\avec{x}_i),
\end{equation*}
where $\q_1(\avec{x}_1), \ldots, \q_n(\avec{x}_n)$ are the subqueries induced by the neighbours of $z_\q$ in $\gfmn$, and, for each tree witness $\t$ for $(\T,\q(\avec{x}))$ with $\tr\neq \emptyset$ and $z_\q \in\ti$ and for every~$\varrho$ generating $\t$, the following clause
\begin{equation*}
\rew_\q(\avec{x}) \ \ \ \leftarrow A_\varrho(z_0)\  \land \bigwedge_{z \in \tr\setminus \{z_0\}} \hspace*{-0.5em}(z=z_0) \ \
\land \bigwedge_{1\leq i \leq k} \rew_{\q_i^\t}(\avec{x}_i^\t),
\end{equation*}
where $z_0$ is any variable in $\tr$ and $\q_1^\t, \dots, \q_k^\t$ are the connected components of $\q$ without $\q_\t$.
Finally, if $\q_0$ is Boolean, then we include clauses
$\rew_{\q_0} \leftarrow A(x)$ for all unary predicates~$A$ such that $\T, \{A(a)\} \models \q_0$.

The program $\Pi^{\textsc{Tw}}_{\omq}$ is inspired by a similar construction from~\cite{LICS14}. By adapting the proof,
we can show that $(\Pi^{\textsc{Tw}}_{\omq}, \rew_{\q_0}(\avec{x}_0))$ is indeed a rewriting; see Appendix~\ref{appA4}.
%


Now fix $\nlf > 1$ and consider $\omq(\avec{x})=(\T, \q_0(\avec{x}))$ from the class $\blCQ$.
The size of the program $\Pi^{\smash{\textsc{Tw}}}_{\omq}$ is polynomially bounded in $|\omq|$ since $\q_0$ has $\smash{O(|\q_0|^{\nlf})}$ tree witnesses
and tree-shaped subqueries.
It is readily seen that
the function $\nu$ defined by setting $\nu (\rew_{\q}) = |\q|$, for each $\q\in\sqset$,
is a weight function for $(\Pi^{\smash{\textsc{Tw}}}_{\omq},\rew_{\q_0}(\avec{x}))$ with $\nu (\rew_{\q_0}) \leq |\omq|$.
Moreover, by Lemma~\ref{PrepMiddleVertex}, $\dep(\Pi^{\smash{\textsc{Tw}}}_{\omq},\rew_{\q_0}) \le \log \nu(\rew_{\q_0}) +1$; and clearly, $\wid(\Pi^{\smash{\textsc{Tw}}}_{\omq},\rew_{\q_0}) \le \nlf + 1$.
By Corollary~\ref{cor:weight}, the obtained NDL-rewritings can be
evaluated in \LOGCFL. Finally, we note that since the number of leaves is bounded,
it is in $\NL$ to decide whether a vertex satisfies the conditions of Lemma~\ref{PrepMiddleVertex},
and in \LOGCFL{} to decide whether $\T, \{A(a)\} \models \q_0$ \cite{DBLP:conf/lics/BienvenuKP15}
or whether a (logspace) representation of a possible tree witness is
indeed a tree witness. This allows us to show that $(\Pi^{\smash{\textsc{Tw}}}_{\omq},\rew_{\q_0}(\avec{x}))$ can be generated by an $\mathsf{L}^{\smash{\LOGCFL}}$-transducer.

It also follows that answering OMQs $(\T,\q(\avec{x}))$ with a tree-shaped CQ with $\nlf$ leaves over any data instance~$\A$ can be done in time
\begin{equation}\label{eq:time3}
\textit{poly}(|\T|,|\q|^{\nlf},|\A|^{\nlf}).
\end{equation}
Indeed, $(\Pi_\omq^{\smash{\textsc{Tw}}},\rpred(\avec{x}))$ can be evaluated  in time polynomial
in $|\Pi^{\smash{\textsc{Tw}}}_{\omq}|$ and $|\A|^{\wid(\Pi^{\smash{\textsc{Tw}}}_{\omq},G)}$, which are bounded by polynomials in $|\T|$, $|\q|^{\nlf}$ and $|\A|^{\nlf}$, respectively.


\section{Parameterised complexity\label{sec:param}}

The upper bounds~\eqref{eq:time1} and~\eqref{eq:time3} for the time required to evaluate NDL-rewritings of OMQs from $\mathsf{OMQ}(\od,1,\infty)$ and $\mathsf{OMQ}(\infty,1,\nlf)$ contain $\od$ and $\nlf$ in the exponent of $|\T|$ and $|\q|$. Moreover, if we allow $\od$ and $\nlf$ to grow while keeping CQs tree-shaped, the combined complexity of OMQ answering will jump to \NP; see Fig.~\ref{pic:results}(a). In this section, we regard $\od$ and $\nlf$ as parameters and show that answering tree-shaped OMQs is not fixed-parameter tractable.

\subsection{Ontology Depth}\label{sec:DepthAsPar}

Consider the following problem \pr:\\[4pt]
{\tabcolsep=3pt\hspace*{-3pt}\begin{tabular}{ll}
\textbf{Instance:} & an OMQ $\omq = (\T,\q)$ with $\T$ of finite depth and tree-shaped Boolean CQ $\q$.\\
\textbf{Parameter:} & the depth of $\T$.\\
\textbf{Problem:} & decide whether $\T,\{A(a)\} \models \q$.
\end{tabular}}
\begin{theorem}\label{thm:w2-hard}
\pr{} is $W[2]$-hard.
\end{theorem}
\begin{proof}
The proof is by reduction of the problem \hs{}, which is known to be $W[2]$-complete~\cite{DBLP:series/txtcs/FlumG06}:\\[4pt]
{\tabcolsep=3pt\hspace*{-3pt}\begin{tabular}{ll}
\textbf{Instance:} & a hypergraph $H=(V,E)$ and $k \in \mathbb N$.\\
\textbf{Parameter:} & $k$.\\
\textbf{Problem:} & decide whether there is $A \subseteq V$ such that $|A| = k$ and $e \cap A \neq\emptyset$, for every $e \in E$.
\end{tabular}}\\
(Such a set $A$ of vertices is called a \emph{hitting set of size $k$}.)
Suppose that $H=(V,E)$ is a hypergraph with vertices $V = \{v_1, \dots, v_n\}$ and hyperedges $E = \{e_1, \dots, e_m\}$. Let $\T^k_H$ be the (normal form of an) ontology with the following axioms, for $1 \le l \le k$:
\begin{align*}
V_i^{l-1}(x) & \to \exists z\, \bigl(P(z,x) \land V_{i'}^l(z)\bigr),
&& \text{for } 0 \leq i < i' \le n,\\
V_i^l(x) & \to 
E^l_j(x), && \text{for } v_i \in e_j,\ e_j\in E,\\
E^l_j(x) & \to \exists z \, \bigl(P(x,z) \land E^{l-1}_j(z)\bigr),\hspace*{-0.5em} && \text{for } 1 \le j \le m.
\end{align*}
Let $\q^k_H$ be a tree-shaped Boolean CQ with the following atoms, for $1 \le j \le m$: 
\begin{equation*}
P(y,z^{k-1}_j), \quad P(z_j^l,z_j^{l-1})  \text{ for } 1 \le l < k, \quad  \text{ and } E_j^0(z_j^0).
\end{equation*}
The first axiom of $\T^k_H$ generates a tree of depth $k$, with branching ranging from $n$ to $1$, such that
the points $w$ of level $k$ are labelled with subsets $X \subseteq V$ of size~$k$ that are read off the path from the root to $w$.
The CQ $\q^k_H$ is a star with rays corresponding to the hyperedges of $H$. The second and third axioms generate `pendants'
ensuring that, for any hyperedge $e$, the central point of the CQ can be mapped to a point with a label $X$
iff $X$ and $e$ have a common vertex. The canonical model of $(\T^2_H, \{V_0^0(a)\})$ and the CQ $\q^2_H$, for $H = (V,\{e_1,e_2,e_3\})$ with $V = \{ 1,2,3\}$, $e_1 = \{1,3\}$, $e_2 = \{ 2,3\}$  and $e_3=\{1,2\}$, is shown below:\\
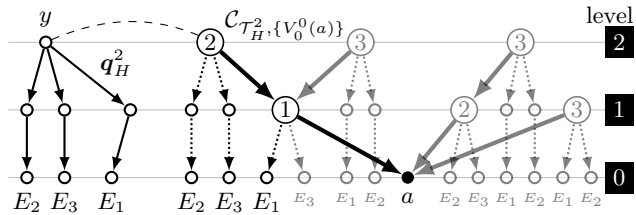
\centerline{\begin{tikzpicture}[>=latex,xscale=1.25,yscale=0.9,vnode/.style={draw,circle,fill=white,semithick,inner sep=0pt,minimum size=3.5mm}]\small
\node at (-1.2,2.25) {$\C_{\T^2_H, \{V_0^0(a)\}}$};
\node at (-3,1.7) {$\q^2_H$};
\node at (2.25,2.4) {level};
\draw[ultra thin,gray!50,thin] (-4.15,0) -- ++(6.35,0); \node[white,fill=black] at (2.35,0) {\small 0};
\draw[ultra thin,gray!50,thin] (-4.15,1) -- ++(6.35,0); \node[white,fill=black] at (2.35,1) {\small 1};
\draw[ultra thin,gray!50,thin] (-4.15,2) -- ++(6.35,0); \node[white,fill=black] at (2.35,2) {\small 2};
\node[bpoint,label=below:{$a$}] (a) at (0.1,0) {}; 
\node[vnode] (a1) at (-1.2,1) {1};
\node[vnode,opacity=0.5,fill opacity=1] (a2) at (0.7,1) {\color{black!50}2};
\node[vnode,opacity=0.5,fill opacity=1] (a3) at (1.9,1) {\color{black!50}3};
\node[vnode] (a12) at (-2,2) {2};
\node[vnode,opacity=0.5,fill opacity=1] (a13) at (-0.4,2) {\color{black!50}3};
\node[vnode,opacity=0.5,fill opacity=1] (a23) at (1.3,2) {\color{black!50}3};
\begin{scope}[ultra thick]
\draw[<-] (a) -- (a1);
\draw[<-,opacity=0.5] (a) -- (a2);
\draw[<-,opacity=0.5] (a) -- (a3);
\draw[<-] (a1) -- (a12);
\draw[<-,opacity=0.5] (a1) -- (a13);
\draw[<-,opacity=0.5] (a2) -- (a23);
\end{scope}
\node[wpoint,label=below:{$E_1$}] (c1) at (-1.4,0) {};
\node[wpoint,opacity=0.5,label=below:{\color{black!50}\tiny$E_3$}] (c3) at (-1,0) {};
\node[wpoint] (c201) at (-2.2,1) {};
\node[wpoint] (c203) at (-1.8,1) {};
\node[wpoint,label=below:{$E_2$}] (c21) at (-2.2,0) {};
\node[wpoint,label=below:{$E_3$}] (c23) at (-1.8,0) {};
\node[wpoint,opacity=0.5] (c301) at (-0.55,1) {};
\node[wpoint,opacity=0.5] (c303) at (-0.25,1) {};
\node[wpoint,label=below:{\color{black!50}\tiny$E_1$},opacity=0.5] (c31) at (-0.55,0) {};
\node[wpoint,label=below:{\color{black!50}\tiny$E_2$},opacity=0.5] (c33) at (-0.25,0) {};
\node[wpoint,label=below:{\color{black!50}\tiny$E_1$},opacity=0.5] (cd1) at (1.75,0) {};
\node[wpoint,label=below:{\color{black!50}\tiny$E_2$},opacity=0.5] (cd3) at (2.05,0) {};
\node[wpoint,opacity=0.5] (ce1) at (1.15,1) {};
\node[wpoint,opacity=0.5] (ce3) at (1.45,1) {};
\node[wpoint,label=below:{\color{black!50}\tiny$E_1$},opacity=0.5] (ce01) at (1.15,0) {};
\node[wpoint,label=below:{\color{black!50}\tiny$E_2$},opacity=0.5] (ce03) at (1.45,0) {};
\node[wpoint,label=below:{\color{black!50}\tiny$E_3$},opacity=0.5] (d01) at (0.85,0) {};
\node[wpoint,label=below:{\color{black!50}\tiny$E_2$},opacity=0.5] (d03) at (0.55,0) {};
\begin{scope}[thick,densely dotted]
\draw[<-] (c1) -- (a1);
\draw[<-,opacity=0.5] (c3) -- (a1);
\draw[<-] (c21) -- (c201);
\draw[<-] (c23) -- (c203);
\draw[<-] (c201) -- (a12);
\draw[<-] (c203) -- (a12);
\draw[<-,opacity=0.5] (c31) -- (c301);
\draw[<-,opacity=0.5] (c33) -- (c303);
\draw[<-,opacity=0.5] (c301) -- (a13);
\draw[<-,opacity=0.5] (c303) -- (a13);
\draw[<-,opacity=0.5] (cd1) -- (a3);
\draw[<-,opacity=0.5] (cd3) -- (a3);
\draw[<-,opacity=0.5] (ce1) -- (a23);
\draw[<-,opacity=0.5] (ce3) -- (a23);
\draw[<-,opacity=0.5] (ce01) -- (ce1);
\draw[<-,opacity=0.5] (ce03) -- (ce3);
\draw[<-,opacity=0.5] (d01) -- (a2);
\draw[<-,opacity=0.5] (d03) -- (a2);
\end{scope}
\node[wpoint,label=above:{$y$}]  (y) at (-3.75,2) {};
\node[wpoint]  (z11) at (-3.95,1) {};
\node[wpoint]  (z12) at (-3.55,1) {};
\node[wpoint]  (z13) at (-2.85,1) {};
\node[wpoint,label=below:{$E_2$}]  (z01) at (-3.95,0) {};
\node[wpoint,label=below:{$E_3$}]  (z02) at (-3.55,0) {};
\node[wpoint,label=below:{$E_1$}]  (z03) at (-3.05,0) {};
\begin{scope}[thick,black]
\draw[<-] (z11) -- (y);
\draw[<-] (z12) -- (y);
\draw[<-] (z13) -- (y);
\draw[<-] (z01) -- (z11);
\draw[<-] (z02) -- (z12);
\draw[<-] (z03) -- (z13);
\end{scope}
\draw[dashed,out=30,in=150] (y) to (a12);
\end{tikzpicture}}\\
Points \raisebox{-2pt}{\begin{tikzpicture}\node[draw,circle,fill=white,semithick,inner sep=0pt,minimum size=3.5mm] at (0,0) {\small $i$};\end{tikzpicture}}\ at level $l$ belong to $V_i^l$.
In Appendix \ref{AppB.1} we prove that $\T^k_H, \{V_0^0(a)\} \models \q^k_H$ iff $H$ has a hitting set of size $k$.
In the example above, $\{1,2\}$ is a hitting set of size $2$, which corresponds to the homomorphism from $\q^2_H$ into the part of
$\C_{\T^2_H, \{V_0^0(a)\}}$ shown in black.
\end{proof}

By Theorem~\ref{thm:logcfl-1}, OMQs $(\T,\q)$ from $\mathsf{OMQ}(\od,1,\infty)$ can be answered (via NDL-rewriting) over a data instance $\A$ in time $\textit{poly}(|\T|^{\od},|\q|,|\A|)$. Theorem~\ref{thm:w2-hard} shows that no algorithm can do this in time $f(\od) \cdot \textit{poly}(|\T|,|\q|,|\A|)$, for any computable function $f$,  unless $W[2] = \text{FPT}$.


\subsection{Number of Leaves}

Next we consider the problem \blpr{}:\\[3pt]
{\tabcolsep=3pt\hspace*{-3pt}\begin{tabular}{ll}
\textbf{Instance:} & an OMQ $\omq = (\T,\q)$ with $\T$ of finite depth and tree-shaped Boolean CQ $\q$.\\
\textbf{Parameter:} & the number of leaves in $\q$.\\
\textbf{Problem:} & decide whether $\T,\{A(a)\} \models \q$.
\end{tabular}}
\begin{theorem}\label{leaves-param-w1}
\blpr\ is $W[1]$-hard.
\end{theorem}
\begin{proof}
The proof is by reduction of the following $W[1]$-complete  \partclique\ problem~\cite{DBLP:journals/tcs/FellowsHRV09}:\\[4pt]
{\tabcolsep=3pt\hspace*{-2pt}\begin{tabular}{ll}
\textbf{Instance:} & a graph $G=(V,E)$ whose vertices are partitioned into $p$ sets $V_1, \ldots, V_p$.\\
\textbf{Parameter:} & $p$, the number of partitions.\\
\textbf{Problem:} & decide whether $G$ has a clique of size $p$ containing one vertex from each $V_i$.
\end{tabular}}\\[4pt]
Consider a graph $G=(V,E)$ with $V=\{v_1, \ldots, v_M\}$
partitioned into
$V_1, \dots, V_p$.
The ontology $\T_G$ will create a tree rooted at $A(a)$ whose every
branch corresponds to selecting one vertex from each $V_i$.
Each branch has length $(p \cdot 2M) + 1$ and consists of $p$ `blocks' of length $2M$,
plus an extra edge at the end (used for padding). 
Each block corresponds to an enumeration of $V$, with
positions $2j$ and $2j+1$ being associated with $v_j$.
In the $i$th block of a branch, we will select a vertex $v_{j_i}$ from $V_i$
by marking the positions $2j_i$ and $2j_i+1$
with the binary predicate~$S$;
we also mark the positions of the neighbours of $v_{j_i}$ in $G$ with  the
predicate $Y$. We use the unary predicate $B$ to mark the end of the $p$th block (square nodes in the picture below).
The left side of the picture illustrates the construction for $p=3$, where
$V_1=\{v_1, v_2\}$, $V_2=\{v_3\}$, $V_3=\{v_4,v_5\}$,
and $E=\{\{v_1,v_3\}, \{v_3,v_5\}\}$.\\
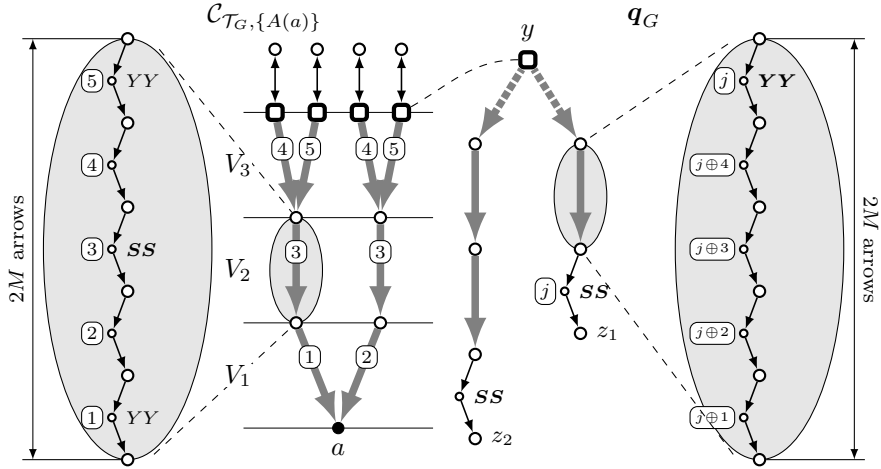
\centerline{%
\begin{tikzpicture}[>=latex,swpoint/.style={circle,inner sep=0pt,minimum size=1mm,fill=white,draw=black,thick},
Bpoint/.style={rectangle,rounded corners=0.5mm,inner sep=0pt,minimum size=2mm,fill=white,draw=black,ultra thick},
lbl/.style={draw,thin,black,rectangle,rounded corners=1mm,inner sep=2pt,fill=white},
yscale=1.4,xscale=1.4]
\draw[fill=gray!20] (-0.7,1.5) ellipse (.25 and 0.5);
\begin{scope}[xshift=2mm]
\draw[fill=gray!20] (-2.5,1.7) ellipse (0.8 and 2);
\end{scope}
\begin{scope}[ultra thin]
\draw (-1.2,0) -- (0.6,0);
\draw (-1.2,1) -- (0.6,1);
\draw (-1.2,2) -- (0.6,2);
\draw (-1.2,3) -- (0.6,3);
\end{scope}
\node[bpoint,label=below:{$a$}] (a) at (-0.3,0) {};
\node[wpoint] (v1) at (-0.7,1) {};
\node[wpoint] (v2) at (0.1,1) {};
\node[wpoint] (v13) at (-0.7,2) {};
\node[wpoint] (v23) at (0.1,2) {};
\node[Bpoint] (v134) at (-0.9,3) {};
\node[Bpoint] (v135) at (-0.5,3) {};
\node[Bpoint] (v234) at (0.3,3) {};
\node[Bpoint] (v235) at (-0.1,3) {};
\node[wpoint] (v134p) at (-0.9,3.6) {};
\node[wpoint] (v135p) at (-0.5,3.6) {};
\node[wpoint] (v234p) at (0.3,3.6) {};
\node[wpoint] (v235p) at (-0.1,3.6) {};
\begin{scope}[line width=1mm,black!50]
\draw[<-] (a) to node[pos=0.7,lbl] {\scriptsize $1$} (v1);
\draw[<-] (a) to node[pos=0.7,lbl] {\scriptsize $2$} (v2);
\draw[<-] (v1) to node[pos=0.7,lbl] {\scriptsize $3$} (v13);
\draw[<-] (v13) to node[pos=0.7,lbl] {\scriptsize $4$} (v134);
\draw[<-] (v13) to node[pos=0.7,lbl] {\scriptsize $5$} (v135);
\draw[<-] (v2) to node[pos=0.7,lbl] {\scriptsize $3$} (v23);
\draw[<-] (v23) to node[pos=0.7,lbl] {\scriptsize $5$} (v234);
\draw[<-] (v23) to node[pos=0.7,lbl] {\scriptsize $4$} (v235);
\end{scope}
\begin{scope}[semithick,black]
\draw[<->] (v134) -- (v134p);
\draw[<->] (v135) -- (v135p);
\draw[<->] (v234) -- (v234p);
\draw[<->] (v235) -- (v235p);
\end{scope}
\begin{scope}[xshift=2mm]
\node[wpoint] (l0) at (-2.5,-0.3) {};
\node[swpoint,label=right:{\scriptsize $YY$}] (n0) at (-2.65,0.1) {}; \node[left=0.5mm of n0,lbl] {\scriptsize $1$};
\node[wpoint] (l1) at (-2.5,0.5) {};
\node[swpoint] (n1) at (-2.65,0.9) {};  \node[left=0.5mm of n1,lbl] {\scriptsize $2$};
\node[wpoint] (l2) at (-2.5,1.3) {};
\node[swpoint,label=right:{\scriptsize $\boldsymbol{SS}$}] (n2) at (-2.65,1.7) {};  \node[left=0.5mm of n2,lbl] {\scriptsize $3$};
\node[wpoint] (l3) at (-2.5,2.1) {};
\node[swpoint] (n3) at (-2.65,2.5) {};  \node[left=0.5mm of n3,lbl] {\scriptsize $4$};
\node[wpoint] (l4) at (-2.5,2.9) {};
\node[swpoint,label=right:{\scriptsize $YY$}] (n4) at (-2.65,3.3) {};  \node[left=0.5mm of n4,lbl] {\scriptsize $5$};
\node[wpoint] (l5) at (-2.5,3.7) {};
\end{scope}
\begin{scope}[semithick]
\draw[<-] (l0) -- (n0);
\draw[<-] (n0) -- (l1);
\draw[<-] (l1) -- (n1);
\draw[<-] (n1) -- (l2);
\draw[<-] (l2) -- (n2);
\draw[<-] (n2) -- (l3);
\draw[<-] (l3) -- (n3);
\draw[<-] (n3) -- (l4);
\draw[<-] (l4) -- (n4);
\draw[<-] (n4) -- (l5);
\end{scope}
\draw[dashed] ($(l0)+(0.25,0.05)$)-- (v1);
\draw[dashed] ($(l5)+(0.3,-0.10)$)-- (v13);
\node[fill=white,inner sep=0pt] at (-1.25,0.5) {$V_1$};
\node[fill=white,inner sep=0pt] at (-1.25,1.5) {$V_2$};
\node[fill=white,inner sep=0pt] at (-1.25,2.5) {$V_3$};
\draw[thin] (l0) -- +(-1,0);
\draw[thin] (l5) -- +(-1,0);
\draw[<->] ($(l0)+(-0.9,0)$) to node[sloped,above] {\footnotesize $2M$ arrows} ($(l5)+(-0.9,0)$);
\node at (-1,3.9) {$\C_{\T_G,\{A(a)\}}$};
\begin{scope}[xshift=2mm]
\draw[fill=gray!20] (1.8,2.2) ellipse (.25 and 0.5);
\draw[fill=gray!20] (3.5,1.7) ellipse (0.8 and 2);
\node at (2.4,3.9) {$\q_G$};
\node[Bpoint,label=above:{$y$}] (y) at (1.3,3.5) {};
\node[wpoint] (s2) at (0.8,2.7) {};
\node[wpoint] (s1) at (1.8,2.7) {};
\node[wpoint] (s20) at (0.8,1.7) {};
\node[wpoint] (s10) at (1.8,1.7) {};
\node[wpoint] (s200) at (0.8,0.7) {};
\node[swpoint,label=right:{\scriptsize $\boldsymbol{SS}$}] (s200p) at (0.65,0.3) {};
\node[wpoint,label=right:{\small $z_2$}] (s200e) at (0.8,-0.1) {};
\node[swpoint,label=right:{\scriptsize $\boldsymbol{SS}$}] (s10p) at (1.65,1.3) {}; \node[left=0.5mm of s10p,lbl] {\scriptsize $j$};
\node[wpoint,label=right:{\small $z_1$}] (s10e) at (1.8,0.9) {};
\begin{scope}[line width=1mm,black!50]
\draw[<-,densely dotted] (s1) -- (y);
\draw[<-,densely dotted] (s2) -- (y);
\draw[<-] (s10) -- (s1);
\draw[<-] (s20) -- (s2);
\draw[<-] (s200) -- (s20);
\end{scope}
\draw[dashed,out=180,in=30] (y) to (v234);
\begin{scope}[semithick]
\draw[<-] (s200e) -- (s200p);
\draw[<-] (s200p) -- (s200);
\draw[<-] (s10e) -- (s10p);
\draw[<-] (s10p) -- (s10);
\end{scope}
\node[wpoint] (w0) at (3.5,-0.3) {};
\node[swpoint] (z0) at (3.35,0.1) {}; \node[left=0.5mm of z0,lbl] {\tiny $j\!\oplus\!1$};
\node[wpoint] (w1) at (3.5,0.5) {};
\node[swpoint] (z1) at (3.35,0.9) {}; \node[left=0.5mm of z1,lbl] {\tiny $j\!\oplus\!2$};
\node[wpoint] (w2) at (3.5,1.3) {};
\node[swpoint] (z2) at (3.35,1.7) {}; \node[left=0.5mm of z2,lbl] {\tiny $j\!\oplus\!3$};
\node[wpoint] (w3) at (3.5,2.1) {};
\node[swpoint] (z3) at (3.35,2.5) {}; \node[left=0.5mm of z3,lbl] {\tiny $j\!\oplus\!4$};
\node[wpoint] (w4) at (3.5,2.9) {};
\node[swpoint,label=right:{\scriptsize $\boldsymbol{YY}$}] (z4) at (3.35,3.3) {}; \node[left=0.5mm of z4,lbl] {\scriptsize $j$};
\node[wpoint] (w5) at (3.5,3.7) {};
\begin{scope}[semithick]
\draw[<-] (w0) -- (z0);
\draw[<-] (z0) -- (w1);
\draw[<-] (w1) -- (z1);
\draw[<-] (z1) -- (w2);
\draw[<-] (w2) -- (z2);
\draw[<-] (z2) -- (w3);
\draw[<-] (w3) -- (z3);
\draw[<-] (z3) -- (w4);
\draw[<-] (w4) -- (z4);
\draw[<-] (z4) -- (w5);
\end{scope}
\draw[dashed] ($(w0)+(-0.25,0.05)$)-- (s10);
\draw[dashed] ($(w5)+(-0.3,-0.05)$)-- (s1);
\draw[thin] (w0) -- +(1,0);
\draw[thin] (w5) -- +(1,0);
\draw[<->] ($(w5)+(0.9,0)$) to node[sloped,above] {\footnotesize $2M$ arrows} ($(w0)+(0.9,0)$);
\end{scope}
\end{tikzpicture}
}\\
Since vertices are enumerated in the same order in every block,
to check whether the selected vertex $v_{j_i}$ for $V_i$ is a neighbour of the vertices selected from $V_{i+1}, \ldots, V_p$,
it suffices to check that positions $2j_i$ and $2j_i+1$ in blocks $i+1, \ldots, p$
are marked $YY$.
Moreover, the distance between the positions of a vertex in consecutive blocks is always $2M-2$.
The idea is thus to construct a CQ $\q_G$ (right side of the picture)
which, starting from a variable labelled $B$ (mapped to the end of a $p$th block),
splits into $p-1$ branches,
with the $i$th branch checking for
a sequence of $i$ evenly-spaced $YY$ markers leading to an $SS$ marker.
The distance from 
the end of the $p$th block (marked $B$)
to the positions $2j_i$ and $2j_i+1$ in the $p$th
block (where the first $YY$ should occur) depends on the choice of $v_{j_i}$. We thus add
an outgoing
edge at the end of the $p$th block, which can be navigated in both directions, to be able to `consume' any \emph{even number}
of query atoms preceding the first~$YY$.



The Boolean CQ $\q_G$ looks as follows
(for readability, we use atoms with star-free regular expressions):
\begin{equation*}
B(y) \land  \bigwedge_{1 \leq i < p} \bigl(U^{\smash{2M-2}} \cdot (YY \cdot U^{\smash{2M-2}})^{i} \cdot S S \bigr) (y, z_i),
\end{equation*}
and the ontology $\T_G$ contains the following axioms:
\begin{align*}
A(x) & \rightarrow \exists y \, L_j^1(x,y), && \text{for } v_j\in V_1,\\ 
\exists z\,L_j^k(z,x) & \rightarrow \exists y \, L_j^{k+1}(x,y),\hspace*{-0.2em} && \text{for } 1 \leq k < 2M, \ v_j \in V,\\
\exists z\,L_j^{2M}(z,x) & \rightarrow \exists y \, L_{j'}^1(x,y), && \text{for } v_j\in V_i,\  v_{j'} \in V_{i+1},\\
L_j^k(x,y)&  \rightarrow S(y,x), && \text{for }k\in\{2j, 2j+1\},\\
L_j^k(x,y) & \rightarrow Y(y,x), && \text{for }\{v_j,v_{j'}\} \in E \text{ and } k\in \{2j', 2j' + 1\},\\
L_j^k(x,y) & \rightarrow U(y,x),&& \text{for } 1 \leq k \leq 2M,\ v_j \in V,\\
\exists z\,L_j^{2M}(z,x) & \rightarrow B(x), &&\text{for } v_j\in V_p,\\ 
B(x) & \rightarrow \exists y \, \bigl(U(x,y) \land U(y,x)\bigr).\hspace*{-5em}
\end{align*}
We prove in the appendix that $\T_G, \{A(a)\} \models \q_G$ iff $G$ has a clique 
containing one vertex from each set $V_i$.
\end{proof}

By~\eqref{eq:time3}, OMQs $(\T,\q)$ from $\mathsf{OMQ}(\infty,1,\nlf)$ can be answered (via NDL-rewriting) over a data instance $\A$ in time $\textit{poly}(|\T|,|\q|^{\nlf},|\A|^{\nlf})$. Theorem~\ref{leaves-param-w1} shows that no algorithm can do this in time $f(\nlf) \cdot \textit{poly}(|\T|,|\q|,|\A|)$, for any computable function $f$,  unless $W[1] = \text{FPT}$.

One may consider various other types of parameters that can hopefully reduce the complexity of OMQ answering. Obvious candidates are the size of ontology, the size of ontology signature or the number of role inclusions in ontologies. (Indeed, it is shown in \cite{DBLP:conf/ijcai/BienvenuOSX13} that in the absence of role inclusions,
tree-shaped OMQ answering is tractable.) Unfortunately, bounding any of these parameters does not make OMQ answering easier,
as we establish in Section \ref{sec:fixed} that already one \emph{fixed} ontology makes the problem \NP-hard for tree-shaped CQs and \LOGCFL-hard for linear ones.


\section{OMQs with a Fixed Ontology\label{sec2:fixed}}\label{sec:fixed}

In a typical OBDA scenario~\cite{DBLP:conf/semweb/KharlamovHJLLPR15},  users are provided with an ontology in a familiar signature (developed by a domain expert)
with which they formulate their queries.
Thus, it is of interest to identify the complexity of answering tree-shaped OMQs $(\T,\q)$ with a fixed $\T$ of infinite depth (see~Fig.~\ref{pic:results}).
Surprisingly,  we show that the problem is \NP-hard even when both $\T$ and $\A$ are fixed (in the database setting, answering tree-shaped CQs is in \LOGCFL\ for combined complexity).
\begin{theorem}\label{thm:NP:query}
There is an ontology $\T_\dag$ such that answering OMQs of the form $(\T_\dag,\q)$ with Boolean tree-shaped CQs $\q$ is \NP-hard for query complexity.
\end{theorem}
\begin{proof}
The proof is by reduction of SAT. Given a CNF $\varphi$ with variables $p_1,\dots,p_k$ and clauses $\chi_1,\dots,\chi_m$, take a Boolean CQ $\q_\varphi$ with $A(y)$ and, for $1 \le j \le m$, the following atoms with $z^k_j = y$:
\begin{align*}
& \P(z^l_j,z^{l-1}_j), && \text{ if }  p_l \text{ occurs in } \chi_j \text{ positively},\\
& \N(z^l_j, z^{l-1}_j), && \text{ if } p_l \text{ occurs in } \chi_j \text{ negatively},\\
& \V(z^l_j, z^{l-1}_j), && \text{ if } p_l \text{ does not occur in } \chi_j,\\
& B_0(z^0_j).
\end{align*}
Thus, $\q_\varphi$ is a star with centre $A(y)$ and $m$ rays encoding the $\chi_j$ by the binary predicates $\P$, $\N$ and $\V$.
Let $\T_\dag$ be an ontology with the axioms
\begin{align*}
& A(x) \to \exists y \, \bigl(\P(y,x) \land \V(y,x) \land B_{-}(y) \land A(y)\bigr),\\
& \hspace*{1em}B_{-}(y)  \to \exists x'\, \bigl(\N(y,x') \land B_0(x')\bigr),\\
& A(x) \to \exists y \, \bigl(\N(y,x) \land \V(y,x) \land B_{+}(y)  \land A(y)\bigr),\\
& \hspace*{1em}B_{+}(y)  \to \exists x'\, \bigl(\P(y,x') \land B_0(x') \bigr),\\
& B_0(x) \to \exists y \, \bigl(\P(x,y) \land \N(x,y) \land \V(x,y) \land B_0(y)\bigr).
\end{align*}
Intuitively, $(\T_\dag, \{A(a)\})$ generates an infinite binary tree whose nodes of depth $n$ represent all $2^n$ truth assignments to $n$ propositional variables. The CQ $\q_\varphi$ can  only be mapped along a branch of this tree towards its root $a$, with the image of $y$, the centre of the star, giving a satisfying assignment for $\varphi$. Each non-root node of the tree also starts an infinite `sink' branch of $B_0$-nodes, where the remainder of the ray for $\chi_j$ can be mapped as soon as one of its literals is satisfied.
We show in Appendix~\ref{app:thm:NP:query} that $\T_\dag,\{A(a)\} \models \q_\varphi$ iff $\varphi$ is satisfiable. To illustrate, the CQ $\q_\varphi$ for $\varphi = (p_1 \lor p_2) \land \neg p_1$ and a fragment of the canonical model $\C_{\T_\dag,\{A(a)\}}$ are shown below:\\[10pt]
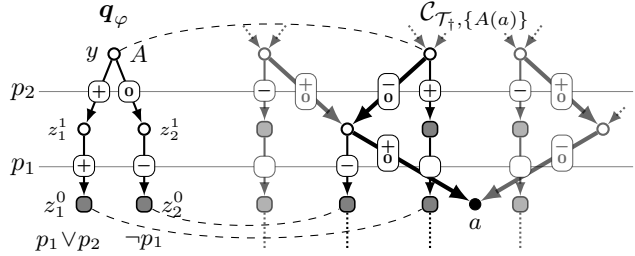
\centerline{\begin{tikzpicture}[>=latex,
man/.style={draw,thin,fill=white,rectangle,rounded corners=1mm,inner sep=0pt,minimum height=2.8mm,minimum width=2.8mm,fill opacity=1},
maw/.style={draw,thin,fill=white,rectangle,rounded corners=1mm,inner sep=0pt,minimum height=4.8mm,minimum width=2.8mm,fill opacity=1},
spoint/.style={draw,semithick,fill=gray,rectangle,rounded corners=0.7mm,inner sep=0pt,minimum height=2mm,minimum width=2mm}
]\small
\newcommand{\dl}[2]{\tabcolsep=0pt\tiny\bf\begin{tabular}{c}#2\\[-1pt]#1\end{tabular}}
%
\draw[ultra thin,gray] (-1,-1.5) -- +(7.9,0);
\draw[ultra thin,gray] (-1,-2.5) -- +(7.9,0);
\node at (-1.2,-1.5) {\normalsize $p_2$};
\node at (-1.2,-2.5) {\normalsize $p_1$};
\node[wpoint,label=left:{$y$},label=right:{$A$}] (y) at (0,-1) {};
\node[wpoint,label=left:{\scriptsize$z_1^1$}] (z11) at (-0.4,-2) {};
\node[wpoint,label=right:{\scriptsize$z_2^1$}] (z21) at (0.4,-2) {};
\node[spoint,label=left:{$z_1^0$}] (z12) at (-0.4,-3) {};
\node[spoint,label=right:{$z_2^0$}] (z22) at (0.4,-3) {};
\begin{scope}[thick]
\draw[->] (y) to node[midway,man] {\tiny\bf +} (z11);
\draw[->] (z11) to node[midway,man] {\tiny\bf +} (z12);
\draw[->] (y) to node[midway,man] {\tiny\bf 0} (z21);
\draw[->] (z21) to node[midway,man] {\tiny$\boldsymbol{-}$} (z22);
\end{scope}
\begin{scope}\small
\node[rotate=0] at (-0.6,-3.5) {$p_1\!\lor\!p_2$};
\node[rotate=0] at (0.4,-3.5) {$\neg p_1$};
\end{scope}
\begin{scope}[xshift=8mm]
\node[bpoint,label=below:{$a$}] (a) at (4,-3) {};
\node[wpoint] (an) at (2.3,-2) {};
\node[wpoint,opacity=0.6] (ap) at (5.7,-2) {};
\node[wpoint,opacity=0.6] (ann) at (1.2,-1) {};
\node[wpoint] (anp) at (3.4,-1) {};
\node[wpoint,opacity=0.6] (apn) at (4.6,-1) {};
%
%
\begin{scope}[ultra thick]
\draw[->] (an) to node[pos=0.3,maw] {\dl{0}{+}}(a);
\draw[->,opacity=0.6] (ap) to node[pos=0.3,maw] {\color{black!60}\dl{0}{$\boldsymbol{-}$}}(a);
\draw[->,opacity=0.6] (ann) to node[midway,maw] {\color{black!60}\dl{0}{+}}(an);
\draw[->] (anp) to node[midway,maw] {\dl{0}{$\boldsymbol{-}$}}(an);
\draw[->,opacity=0.6] (apn) to node[midway,maw] {\color{black!60}\dl{0}{+}}(ap);
\end{scope}
\node[spoint,opacity=0.6] (zann) at ($(ann)+(0,-1)$) {};
\node[spoint,opacity=0.6] (zann1) at ($(zann)+(0,-1)$) {};
\node[spoint] (zanp) at ($(anp)+(0,-1)$) {};
\node[spoint] (zanp1) at ($(zanp)+(0,-1)$) {};
\node[spoint] (zan) at ($(an)+(0,-1)$) {};
\node[spoint,opacity=0.6] (zapn) at ($(apn)+(0,-1)$) {};
\node[spoint,opacity=0.6] (zapn1) at ($(zapn)+(0,-1)$) {};
\begin{scope}[thick, densely dotted]
\draw[opacity=0.6] (zann1) -- +(0,-0.6);
\draw (zanp1) -- +(0,-0.6);
\draw[opacity=0.6] (zapn1) -- +(0,-0.6);
\draw (zan) -- +(0,-0.6);
\draw[<-,opacity=0.6] (ap) -- +(0.3,0.3);
\draw[<-,opacity=0.6] (ann) -- +(0.3,0.4);
\draw[<-,opacity=0.6] (ann) -- +(-0.3,0.4);
\draw[<-,opacity=0.6] (anp) -- +(0.3,0.4);
\draw[<-,opacity=0.6] (anp) -- +(-0.3,0.4);
\draw[<-,opacity=0.6] (apn) -- +(0.3,0.4);
\draw[<-,opacity=0.6] (apn) -- +(-0.3,0.4);
\end{scope}
\begin{scope}[thick]
\draw[->] (an) to node[solid,midway,man] {\tiny\bf$\boldsymbol{-}$} (zan);
\draw[->,opacity=0.6] (ann) to node[solid,midway,man] {\tiny\bf$\boldsymbol{-}$} (zann);
\draw[->,opacity=0.6] (zann) to node[solid,midway,man] {} (zann1);
\draw[->] (anp) to node[solid,midway,man] {\tiny\bf+} (zanp);
\draw[->] (zanp) to node[solid,midway,man] {} (zanp1);
\draw[->,opacity=0.6] (apn) to node[solid,midway,man] {\tiny\bf$\boldsymbol{-}$} (zapn);
\draw[->,opacity=0.6] (zapn) to node[solid,midway,man] {} (zapn1);
\end{scope}
\node[fill=white,inner sep=1pt] at (4,-0.5) {\normalsize $\C_{\T_\dag,\{A(a)\}}$};
\end{scope}
\node at (0,-0.5) {\normalsize $\q_\varphi$};
\begin{scope}[dashed]
\draw[out=30,in=150,looseness=0.6] (y) to (anp);
\draw[out=-30,in=-150,looseness=0.6] (z22) to (zan);
\draw[out=-30,in=-150,looseness=0.6] (z12) to (zanp1);
\end{scope}
\end{tikzpicture}}\\[10pt]
Here, \begin{tikzpicture}\node[draw,semithick,fill=gray,rectangle,rounded corners=0.7mm,inner sep=0pt,minimum height=2mm,minimum width=2mm] at (0,0) {};\end{tikzpicture} are the points in $B_0$ and the labels on arrows indicate the subscripts of the binary predicates $P$ (the empty label means all three: $+$, $-$ and $0$); predicates $A$, $B_+$, $B_-$ are not shown in $\C_{\T_\dag,\{A(a)\}}$.
\end{proof}

The proof above uses OMQs $\omq_\varphi = (\T_\dag,\q_\varphi)$ over a data instance with a single individual constant. Thus:
\begin{corollary}\label{cor:no-poly}
No polynomial-time algorithm can construct \FO- or \NDL-rewritings for the OMQs $\omq_\varphi$ 
unless $\PTime = \NP$.
\end{corollary}
\begin{proof}
Indeed, if a polynomial-time algorithm could find a rewriting $\q'_\varphi$ of $\omq_\varphi$, then we would be able to check whether $\varphi$ is satisfiable in polynomial time by evaluating $\q'_\varphi$ over the data instance $\{A(a)\}$.
\end{proof}

Curiously enough, Corollary~\ref{cor:no-poly} can be complemented with the following theorem:
\begin{theorem}\label{thm:FOpoly}
The $\omq_\varphi$\! have
polynomial \FO-rewritings.
\end{theorem}
\begin{proof}
Define $\q'_\varphi$ as the FO-sentence
\begin{equation*}
\forall xy \, \big((x=y) \land A(x) \land  \varphi^*\big) \ \lor \
\exists xy \, \big((x\ne y) \land \q^*_\varphi(x,y)  \big),
\end{equation*}
where $\varphi^*$ is $\top$ if $\varphi$ is satisfiable and $\bot$ otherwise, and $\q^*_\varphi(x,y)$ is the polynomial-size FO-rewriting of $\omq_\varphi$ over data with \emph{at least 2} constants~\cite[Corollary~14]{DBLP:journals/ai/GottlobKKPSZ14}.  Recall that the proof of Theorem~\ref{thm:NP:query}  shows that, if $\A$ has a single constant, $a$, and there is a homomorphism from $\q_\varphi$ to $\C_{\T_\dag,\A}$, then $A(a)\in \A$ and $\varphi$ is satisfiable.
Thus, the first disjunct of $\q'_\varphi$ is an FO-rewriting of $\omq_\varphi$ over data instances with a single constant; the case of at least 2 constants follows from~\cite[Corollary~14]{DBLP:journals/ai/GottlobKKPSZ14}.
\end{proof}

Whether the OMQs $\omq_\varphi$ have a polynomial-size \PE- or \NDL-rewritings remains open.
We have only managed to construct a modification $\bar \q_\varphi(x)$ of $\q_\varphi$ with the following interesting properties (details are given in Appendix~\ref{app:no-ql-rewritings}).
Let $\mathfrak T$ be the class of data instances representing finite binary trees with root $a$ whose edges are labelled with $\P$ and $\N$, and some of whose leaves are labelled with $B_0$.
Let $\mathcal{QL}$ be any query language such that, for every $\mathcal{QL}$-query $\Phi(x)$ and every $\A \in \mathfrak T$, the answer to $\Phi(a)$ over $\A$ can be computed in time polynomial in $|\Phi|$ and $|\A|$. Typical examples of $\mathcal{QL}$ are modal-like languages such as certain fragments of XPath~\cite{Koch:2006:PQT:1142351.1142382} or description logic instance queries~\cite{BCMNP03}.
\begin{theorem}\label{no-ql-rewritings}
The OMQs $(\T_\dag,\bar \q_\varphi(x))$ do not have polynomial-size rewritings in $\mathcal{QL}$ unless $\NP \subseteq \Ppoly$.
\end{theorem}

To our surprise, Theorem~\ref{no-ql-rewritings} is not applicable to PE.\!\footnote{This result might be known but we could not  find it in the literature, and so provide a proof in Appendix~\ref{app:trees-data-np}.}
\begin{theorem}\label{trees-data-np}
Evaluating \PE-queries over trees in $\mathfrak T$ is \NP-hard.
\end{theorem}

Finally, we consider bounded-leaf CQs (whose evaluation is \NL-complete in the database setting) with fixed ontology and data.
\begin{theorem}\label{fixed-logcfl}
There is an ontology $\T_\ddag$ such that answering OMQs of the form $(\T_\ddag,\q)$ with Boolean linear CQs $\q$ is \LOGCFL-hard for query complexity.
\end{theorem}

The proof is by reduction of the recognition problem for the hardest \LOGCFL\ language $\mathcal{L}$~\cite{DBLP:journals/siamcomp/Greibach73,Sudborough:1975:NTC:321906.321913}. We construct an ontology $\T_\ddag$ and a logspace transducer that converts the words $w$ in the alphabet of $\mathcal{L}$ to linear CQs $\q_w$ such that $w \in \mathcal{L}$ iff $\T_\ddag,\{A(a)\} \models \q_w$.

Finally, it is possible to strengthen the result of Theorem~\ref{leaves-param-w1} in the following way.  
Given an ontology $\T$, we define the problem $\blpr[\T]$ as follows:

{\tabcolsep=3pt\hspace*{-3pt}\begin{tabular}{ll}
\textbf{Instance:} & a tree-shaped Boolean CQ $\q$.\\
\textbf{Parameter:} & the number of leaves in $\q$.\\
\textbf{Problem:} & decide whether $\T,\{A(a)\} \models \q$.
\end{tabular}}

\begin{theorem}\label{fixed-w1}
There is an ontology $\T_\Box$ such that $\blpr[\T_\Box]$ is $W[1]$-hard.
\end{theorem}
Observe that Theorem~\ref{fixed-w1} does not imply Theorem~\ref{leaves-param-w1} directly, because $\T_\Box$ is not of finite depth; however, the proof of the former theorem (see Appendix~\ref{fixed-w1}) can be modified to obtain the latter.


\section{Experiments \& Conclusions}
The main positive result of this paper is the development of theoretically optimal NDL-rewritings for three classes $\bdObtwCQ$, $\bdOblCQ$, $\blCQ$ of OMQs. It was known that answering such OMQs is tractable, but the proofs employed
elaborate algorithms tailored for each of the three cases. We have shown that the optimal complexity can be achieved \emph{via NDL-rewriting}, thus reducing OMQ answering to standard query evaluation.
This result is practically relevant as many user queries are tree-shaped (see, e.g., \cite{PicalausaV:sparql-2011} for evidence in the RDF setting),
and indeed, recent tools for query formulation over ontologies (like~\cite{soylu2016}) produce tree-shaped CQs.   Moreover, the majority of important real-world \OWL{} ontologies are of finite depth; see
\cite{DBLP:journals/jair/GrauHKKMMW13} for statistics. In the context of OBDA, \OWLQL{} ontologies are often built
starting from the database schemas (bootstrapping \cite{DBLP:conf/semweb/Jimenez-RuizKZH15a}), which typically do not contain  cycles such as `every manager is managed by a manager.\!'  For example, the NPD FactPages ontology,\!\footnote{http://sws.ifi.uio.no/project/npd-v2/} designed to facilitate querying the datasets of the Norwegian Petroleum Directorate, is of depth 5.

The starting point of our research was the observation that standard query rewriting systems tend to produce suboptimal rewritings of the OMQs in these three classes. This is obviously so for UCQ-rewriters~\cite{PLCD*08,DBLP:conf/dlog/Perez-UrbinaMH09,DBLP:conf/cade/ChortarasTS11,DBLP:conf/icde/GottlobOP11,kyrie2,DBLP:journals/semweb/KonigLMT15}. However, this is also true of more elaborate PE-rewriters (which use disjunctions inside conjunctions)~\cite{DBLP:conf/semweb/Rodriguez-MuroKZ13,DBLP:conf/ijcai/Thomazo13} whose rewritings in theory can be of superpolynomial size; see~Fig.~\ref{pic:results}(b). Surprisingly, even NDL-rewriters such as Clipper~\cite{DBLP:conf/aaai/EiterOSTX12}, Presto~\cite{DBLP:conf/kr/RosatiA10} and Rapid~\cite{DBLP:conf/cade/ChortarasTS11} do not fare much better in practice. To illustrate,
we generated three sequences of OMQs in the class \textsf{OMQ}(1,\,1,\,2)  (lying in the intersection of $\bdObtwCQ$, $\bdOblCQ$ and $\blCQ$) with the  ontology from Example~\ref{ex:rewriting:2} and
linear CQs of up to 15 atoms as in Example~\ref{ex:rewriting:1} (which are associated with words from
$\{R, S\}^*$).  By~Fig.~\ref{pic:results}(a), answering these OMQs can be done in \NL. The barcharts in Fig.~\ref{fig:rewritigns} show the number of clauses in their NDL-rewritings produced by Clipper, Presto and Rapid, as well as by our algorithms \textsc{Lin}, \textsc{Log} and \textsc{Tw} from Sections~\ref{sec:boundedtw}--\ref{sec:boundedleaf}, respectively.
The first three NDL-rewritings display a clear exponential growth, with Clipper and Rapid failing to produce rewritings for longer CQs. In contrast, our rewritings grow linearly in accord with theory.

We evaluated the rewritings over a few randomly generated data instances using off-the-shelf datalog engine RDFox~\cite{DBLP:conf/semweb/NenovPMHWB15}.
The experiments (details are in the appendix) show that
our rewritings are usually executed faster than those produced by Clipper, Presto and Rapid.

\begin{figure*}
\centerline{%
\begin{tikzpicture}[xscale=0.55,yscale=0.92]
\def\ysc{0.025}
\def\cwidth{0.12}
%
\begin{scope}[yshift=24mm]
\node[draw,rounded corners=1mm,inner sep=6pt,fill=gray!20] at (-4,1.3) {$RRSRSRSRRSRRSSR$};
\clip (-0.5,-1) rectangle (16,310*\ysc);
\foreach \x in {1,...,15} {
\node at (\x,-0.27) {\scriptsize \x};
}
\draw (0.7,0) -- ++(0,150*\ysc);
\foreach \y in {10,25,50,100,200} {
\draw (0.7,\y*\ysc) -- +(-0.2,0);
\draw[ultra thin,gray] (0.7,\y*\ysc) -- +(14.8,0);
\node at (0.2,\y*\ysc) {\scriptsize \y};
}
\foreach \x/\v in {1/2,2/2,3/3,4/4,5/6,6/10,7/10,8/14,9/15,10/16,11/16,12/21,13/24,14/25,15/22} { 
\filldraw[pattern=north east lines, pattern color=black] (\x-2*\cwidth,0) -- ++(0,\v*\ysc) -- ++(\cwidth,0) -- (\x-\cwidth,0);
}
\foreach \x/\v in {1/2,2/5,3/8,4/11,5/14,6/17,7/20,8/23,9/26,10/29,11/32,12/35,13/38,14/41,15/44} { 
\filldraw (\x-\cwidth,0) -- ++(0,\v*\ysc) -- ++(\cwidth,0) -- (\x,0);
}
\foreach \x/\v in {1/1,2/2,3/5,4/8,5/12,6/16,7/20,8/24,9/27,10/32,11/36,12/40,13/45,14/47,15/51} {
\filldraw[pattern=north west lines, pattern color=black] (\x,0) -- ++(0,\v*\ysc) -- ++(\cwidth,0) -- (\x+\cwidth,0);
}
\foreach \x/\v in {1/1,2/1,3/2,4/3,5/5,6/7,7/10,8/13,9/13,10/26,11/39,12/39,13/-5,14/-5,15/-5} {
\filldraw[fill=gray!50,ultra thin] (\x+\cwidth,0) -- ++(0,\v*\ysc) -- ++(\cwidth,0) -- (\x+2*\cwidth,0);
}
\foreach \x/\v in {1/1,2/1,3/2,4/3,5/5,6/7,7/11,8/16,9/16,10/44,11/72,12/126,13/241,14/-5,15/-5} {
\filldraw[fill=gray!20,ultra thin] (\x+2*\cwidth,0) -- ++(0,\v*\ysc) -- ++(\cwidth,0) -- (\x+3*\cwidth,0);
}
\foreach \x/\v in {1/5,2/5,3/14,4/19,5/24,6/33,7/49,8/77,9/77,10/203,11/329,12/329,13/959,14/959,15/2723} {
\filldraw[fill=gray,ultra thin] (\x+3*\cwidth,0) -- ++(0,\v*\ysc) -- ++(\cwidth,0) -- (\x+4*\cwidth,0);
}
\filldraw[white,draw=white,decorate, decoration={coil,segment length=12pt,aspect=0}] (0,300*\ysc) rectangle ++(17,1);
\draw (0.5,0) -- (15.5,0);
\end{scope}
%
%
%
\begin{scope}[xshift=-65mm,yshift=48mm]
\clip (-5,-1) rectangle (21,210*\ysc);
\fill[white] (-0.1,-0.4) rectangle (15.5,130*\ysc);
\foreach \x in {1,...,15} {
\node at (\x,-0.27) {\scriptsize \x};
}
\draw (0.7,0) -- ++(0,100*\ysc);
\foreach \y in {10,25,50,100} {
\draw (0.7,\y*\ysc) -- +(-0.2,0);
\draw[ultra thin,gray] (0.7,\y*\ysc) -- +(14.8,0);
\node at (0.2,\y*\ysc) {\scriptsize \y};
}
\draw[ultra thin,gray] (10,200*\ysc) -- +(5.5,0);
\node[draw,rounded corners=1mm,inner sep=6pt,fill=gray!20] at (1,1.9) {$SRRSSRSRSRRSRRS$};
%
\foreach \x/\v in {1/1,2/2,3/3,4/5,5/6,6/7,7/14,8/8,9/10,10/17,11/20,12/23,13/25,14/27,15/29} {
\filldraw[pattern=north east lines, pattern color=black] (\x-2*\cwidth,0) -- ++(0,\v*\ysc) -- ++(\cwidth,0) -- (\x-\cwidth,0);
}
\foreach \x/\v in {1/2,2/5,3/8,4/11,5/14,6/17,7/20,8/23,9/26,10/29,11/32,12/35,13/38,14/41,15/44} { 
\filldraw (\x-\cwidth,0) -- ++(0,\v*\ysc) -- ++(\cwidth,0) -- (\x,0);
}
\foreach \x/\v in {1/1,2/4,3/5,4/8,5/10,6/15,7/18,8/21,9/27,10/33,11/37,12/42,13/46,14/51,15/52} {
\filldraw[pattern=north west lines, pattern color=black] (\x,0) -- ++(0,\v*\ysc) -- ++(\cwidth,0) -- (\x+\cwidth,0);
}
\foreach \x/\v in {1/1,2/2,3/2,4/4,5/4,6/8,7/11,8/18,9/24,10/34,11/43,12/56,13/-5,14/-5,15/-5} {
\filldraw[fill=gray!50,ultra thin] (\x+\cwidth,0) -- ++(0,\v*\ysc) -- ++(\cwidth,0) -- (\x+2*\cwidth,0);
}
\foreach \x/\v in {1/1,2/2,3/2,4/4,5/4,6/8,7/11,8/24,9/35,10/63,11/100,12/302,13/-5,14/-5,15/-5} {
\filldraw[fill=gray!20,ultra thin] (\x+2*\cwidth,0) -- ++(0,\v*\ysc) -- ++(\cwidth,0) -- (\x+3*\cwidth,0);
}
\foreach \x/\v in {1/5,2/14,3/14,4/23,5/23,6/39,7/57,8/96,9/183,10/356,11/356,12/1028,13/1712,14/1712,15/5108} {
\filldraw[fill=gray,ultra thin] (\x+3*\cwidth,0) -- ++(0,\v*\ysc) -- ++(\cwidth,0) -- (\x+4*\cwidth,0);
}
\draw (0.5,0) -- (15.5,0);
\begin{scope}
\clip (8,0) rectangle +(3,10);
\filldraw[white,draw=white,decorate, decoration={coil,segment length=12pt,aspect=0}] (8,115*\ysc) rectangle ++(9,5);
\end{scope}
\filldraw[white,draw=white,decorate, decoration={coil,segment length=12pt,aspect=0}] (8,210*\ysc) rectangle ++(9,5);
\end{scope}
%
\begin{scope}[xshift=-115mm,yshift=79mm]
\node[draw,rounded corners=1mm,inner sep=6pt,fill=gray!20] at (4.5,1.8) {$SRRRRRSRSRRRRRR$};
\clip (-0.5,-1) rectangle (16,95*\ysc);
\foreach \x in {1,...,15} {
\node at (\x,-0.27) {\scriptsize \x};
}
\draw (0.7,0) -- ++(0,85*\ysc);
\foreach \y in {10,25,50,100} {
\draw (0.7,\y*\ysc) -- +(-0.2,0);
\draw[ultra thin,gray] (0.7,\y*\ysc) -- +(14.8,0);
\node at (0.2,\y*\ysc) {\scriptsize \y};
}
\foreach \x/\v in {1/1,2/2,3/3,4/3,5/4,6/4,7/7,8/7,9/10,10/11,11/14,12/18,13/20,14/16,15/15} { 
\filldraw[pattern=north east lines, pattern color=black] (\x-2*\cwidth,0) -- ++(0,\v*\ysc) -- ++(\cwidth,0) -- (\x-\cwidth,0);
}
\foreach \x/\v in {1/2,2/5,3/8,4/11,5/14,6/17,7/20,8/23,9/26,10/29,11/32,12/35,13/38,14/41,15/44} { 
\filldraw (\x-\cwidth,0) -- ++(0,\v*\ysc) -- ++(\cwidth,0) -- (\x,0);
}
\foreach \x/\v in {1/1,2/4,3/5,4/5,5/8,6/10,7/13,8/16,9/22,10/27,11/29,12/33,13/35,14/36,15/37} {
\filldraw[pattern=north west lines, pattern color=black] (\x,0) -- ++(0,\v*\ysc) -- ++(\cwidth,0) -- (\x+\cwidth,0);
}
%
\foreach \x/\v in {1/1,2/2,3/2,4/2,5/2,6/2,7/4,8/6,9/10,10/14,11/14,12/14,13/-5,14/-5,15/-5} {
\filldraw[fill=gray!50,ultra thin] (\x+\cwidth,0) -- ++(0,\v*\ysc) -- ++(\cwidth,0) -- (\x+2*\cwidth,0);
}
\foreach \x/\v in {1/1,2/2,3/2,4/2,5/2,6/2,7/4,8/7,9/13,10/26,11/26,12/26,13/30,14/31,15/30} {
\filldraw[fill=gray!20,ultra thin] (\x+2*\cwidth,0) -- ++(0,\v*\ysc) -- ++(\cwidth,0) -- (\x+3*\cwidth,0);
}
\foreach \x/\v in {1/5,2/14,3/14,4/14,5/14,6/14,7/23,8/26,9/29,10/50,11/83,12/83,13/83,14/83,15/83} {
\filldraw[fill=gray,ultra thin] (\x+3*\cwidth,0) -- ++(0,\v*\ysc) -- ++(\cwidth,0) -- (\x+4*\cwidth,0);
}
\draw (0.5,0) -- (15.5,0);
\filldraw[white,draw=white,decorate, decoration={coil,segment length=20pt,aspect=0}] (13,95*\ysc) rectangle ++(3,1);
\end{scope}
\begin{scope}[xshift=-5mm,yshift=20mm]\footnotesize
\filldraw[pattern=north east lines, pattern color=black] (-10.7,0.5) rectangle +(\cwidth,20*\ysc);
\node[minimum height=7mm] at (-10.4,0.1) {\smash{\textsc{Tw}}};
\filldraw (-8.9,0.5) rectangle +(\cwidth,20*\ysc);
\node[minimum height=7mm] at (-8.6,0.1) {\smash{\textsc{Lin}}};
\filldraw[pattern=north west lines, pattern color=black] (-7.1,0.5) rectangle +(\cwidth,20*\ysc);
\node[minimum height=7mm] at (-6.9,0.1) {\smash{\textsc{Log}}};
\filldraw[fill=gray!50,ultra thin]  (-5.3,0.5) rectangle +(\cwidth,20*\ysc);
\node[minimum height=7mm] at (-4.9,0.1) {\smash{Rapid}};
\filldraw[fill=gray!20,ultra thin]  (-3.5,0.5) rectangle +(\cwidth,20*\ysc);
\node[minimum height=7mm] at (-2.9,0.1) {\smash{Clipper}};
\filldraw[fill=gray,ultra thin]  (-1.7,0.5) rectangle +(\cwidth,20*\ysc);
\node[minimum height=7mm] at (-1.1,0.1) {\smash{Presto}};
\end{scope}
\end{tikzpicture}%
}%
\caption{The size of NDL-rewritings produced by different algorithms.}\label{fig:rewritigns}
\end{figure*}

The version of RDFox we used did not seem to take advantage of the structure of the \NL{}/\LOGCFL{} rewritings by simply materialising all the predicates without using magic sets or optimising programs before execution. It would be interesting to see whether the nonrecursiveness and  parallelisability of our rewritings can be utilised to produce efficient execution plans.
One could also
investigate whether our rewritings can be
efficiently implemented using views in standard DBMSs.

Our rewriting algorithms are based on the same idea: pick a point splitting the given CQ into sub-CQs, rewrite the sub-CQs recursively, and then formulate rules that combine the resulting rewritings. The difference between the algorithms is in the choice of the splitting points, which determines the execution plans for OMQs and has a big impact on their performance.
The experiments show that none of the three splitting strategies  systematically outperforms
the others.
This suggests that 
execution times may be dramatically improved by employing 
an `adaptable' splitting strategy that would work similarly to query execution
planners in DBMSs and use statistical information about the relational tables to generate efficient NDL programs.
For example, one could first define a `cost function' on some set of alternative rewritings that roughly estimates their evaluation  time and then construct a rewriting minimising this function. Such a  performance-oriented approach was introduced and exploited in \cite{DBLP:journals/pvldb/BursztynGM16}, where the target language for OMQ rewritings was joins of UCQs (unions of CQs).
Other optimisation techniques for removing redundant rules or sub-queries from rewritings~\cite{DBLP:conf/kr/RosatiA10,DBLP:conf/semweb/Rodriguez-MuroKZ13,gottlob2014query,DBLP:journals/semweb/KonigLMT15}
or exploiting the emptiness of certain predicates
\cite{venetis2016rewriting} are also relevant here.
In the context of OBDA with relational databases and mappings, integrity constraints \cite{DBLP:conf/esws/Rosati12,DBLP:conf/dlog/Rodriguez-MuroKZ13} and the structure of mappings \cite{di2013optimizing} are particularly important for optimisation.


Having observed that (\emph{i}) the ontology depth and (\emph{ii}) the number of leaves in tree-shaped CQs occur in the exponent of our upper bounds for the complexity of OMQ answering algorithms, we regarded (\emph{i}) and (\emph{ii}) as parameters and investigated the parameterised complexity of the OMQ answering problem. We proved that the problem is $W[2]$-hard in the former case and $W[1]$-hard in the latter (it remains open whether these lower bounds are tight). Furthermore, we established that answering OMQs with a fixed ontology (of infinite depth) is \NP-complete for tree-shaped CQs and \LOGCFL-complete for linear CQs, which dashed hopes of taming intractability by restricting the ontology size, signature, etc. One remaining open problem is whether answering OMQs with a fixed ontology and tree-shaped CQs is fixed-parameter tractable if the number of leaves is regarded as the parameter.

A more general avenue
for future research is to extend the study of succinctness and optimality of rewritings to suitable ontology languages with predicates of higher-arity, such as linear and sticky tgds.

\section{Acknowledgements}

This work was supported by the French ANR  grant 12-JS02-007-01 `PAGODA: Practical Algorithms for Onto\-logy-Based Data Access', the UK EPSRC grant EP/M012670 `iTract: Islands of Tractability in Onto\-logy-Based Data Access', the Russian Foundation for Basic Research grant MK-7312.2016.1, and the Russian Academic Excellence Project 5-100.
We thank the developers of Clipper and Rapid for making their systems freely available and Riccardo Rosati for the opportunity to conduct experiments with Presto.


%

\appendix

\section{Proofs for Section~\ref{sec:3}}

\subsection{Lemma~\ref{linear-arbitrary}}

\indent\textsc{Lemma~\ref{linear-arbitrary}}.
{\it Fix any $\wid > 0$. There is an $\mathsf{L}^\NL$-transducer that, for any linear \NDL-rewriting $(\Pi,G(\avec{x}))$ of an OMQ $\omq(\avec{x})$ over complete data instances with $\wid(\Pi,G)\leq \wid$, computes a linear NDL-rewriting $(\Pi',G(\avec{x}))$ of $\omq(\avec{x})$ over arbitrary data instances such that $\wid(\Pi',G)\leq \wid+1$.}
\begin{proof}
Let $(\Pi, G(\avec{x}))$ be a linear \NDL-rewriting of the OMQ $\omq(\avec{x}) = (\T,\q(\avec{x}))$ over complete data instances such that $\wid(\Pi,G)\leq \wid$.
We will replace every clause $\lambda$ in $\Pi$ by a set of clauses $\lambda^*$ defined as follows.
Suppose $\lambda$ is of the form
\begin{equation*}
Q(\avec{z}) \gets I \land \textit{EQ} \land E_1 \land \ldots \land E_n,
\end{equation*}
where $I$ is the only IDB body atom in $\lambda$, $\textit{EQ}$ contains all equality body atoms, and $E_1, \ldots, E_n$ are the EDB body atoms not involving equality.
For every atom $E_i$, we define a set $\upsilon(E_i)$ of atoms by taking
\begin{align*}
\upsilon(E_i) & =
\bigl\{B(z) \mid \T\models B(x) \to A(x) \bigr\} \cup {} \\
& \hspace*{1.3em} \bigl\{ \varrho(y_i,z) \mid \T\models \exists y\,\varrho(y,x) \to A(x) \bigr\}, && \text{ if } E_i= A(z),\\
\upsilon(E_i) & = \bigl\{\varrho(z,z') \mid \T\models \varrho(x,y)\to  P(x,y)\bigr\}, && \text{ if } E_i = P(z,z'),
\end{align*}
where $y_i$ is a fresh variable not occurring in $\lambda$;
we assume $P^-(z,z')$ coincides with $P(z',z)$, for all binary predicates $P$.
Intuitively, $\upsilon(E_i)$ captures all atoms that imply $E_i$ with respect to $\T$.
Then $\lambda^*$ consists of the following clauses:
\begin{align*}
Q_0(\avec{z}_0) & \gets I,\\
Q_i(\avec{z}_i) & \gets Q_{i-1}(\avec{z}_{i-1}) \land E_i',\ \text{ for  } 1 \leq i \leq n \text{ and  }E_i' \in \upsilon(E_i),\\
Q(\avec{z}) & \gets Q_n(\avec{z}_n) \land \textit{EQ},
\end{align*}
where $\avec{z}_i$ is the restriction of $\avec{z}$ to variables occurring in  $I$ if $i = 0$ and in  $Q_{i-1}(\avec{z}_{i-1})$ and $E_i'$ except for $y_i$ if $i > 0$ (note that $\avec{z}_n = \avec{z}$).
Let $\Pi'$ be the program obtained from $\Pi$ by replacing each clause $\lambda$ by the set of clauses~$\lambda^*$.
By construction, $\Pi'$ is a linear NDL program and its width cannot exceed $\wid(\Pi,G)+1$ (the possible increase of $1$ is due to
the replacement of unary atoms $A(z)$ by binary atoms $\varrho(y_i,z)$).

We now argue that $(\Pi', G(\avec{x}))$ is a rewriting of $\omq(\avec{x})$ over arbitrary data instances.
It can be easily verified that $(\Pi', G(\avec{x}))$ is equivalent to $(\Pi'', G(\avec{x}))$, where
NDL program $\Pi''$ is obtained from $\Pi$ by replacing each clause
\mbox{$Q(\avec{z}) \gets I \land \textit{EQ} \land E_1 \land \ldots \land E_n$} by the (possibly exponentially larger) set of clauses of the form
\begin{equation*}
Q(\avec{z}) \gets I \land \textit{EQ} \land E_1' \land \ldots \land E_n',
\end{equation*}
for all $E_i' \in \upsilon(E_i)$ and $1 \leq i \leq n$.
It thus suffices to show that  $(\Pi'', G(\avec{x}))$ is a rewriting of $\omq(\avec{x})$ over arbitrary data instances.

First suppose that $\T, \A \models \q(\avec{a})$, where $\A$ is an arbitrary data instance.
Let $\A'$ be the complete data instance obtained from $\A$ by adding the ground atoms:
\begin{align*}
P(a,b) &~~\text{ if }~~ \varrho(a,b) \in \A \text{ and } \T\models \varrho(x,y) \to P(x,y);\\
A(a) &~~\text{ if }~~ B(a) \in \A \text{ and } \T\models B(x) \to A(x);\\
A(a) &~~\text{ if }~~ \varrho(a,b) \in \A \text{ and } \T\models \exists y\,\varrho(y,x) \to A(x).
\end{align*}
(We write $\varrho(a,b)\in\A$ for $P(a,b)\in\A$ if $\varrho=P$ and for $P(b,a)$ if $\varrho = P^-$.)
Clearly, $\T, \A' \models \q(\avec{a})$, so we must have $\Pi, \A' \models G(\avec{a})$.
A simple inductive argument (on the order of derivation of ground atoms)
shows that whenever a clause \mbox{$Q(\avec{z}) \gets I \land \textit{EQ}\land E_1 \land \ldots \land E_n$}
is applied using a substitution $\avec{c}$ for the variables in the body to derive $Q(\avec{c}(\avec{z}))$ using $\Pi$,
we can find a corresponding clause \mbox{$Q(\avec{z}) \gets I \land \textit{EQ}\land E_1' \land \ldots \land E_n'$} and a substitution $\avec{c}'$ extending $\avec{c}$ (on the fresh variables $y_i$)
that allows us to derive $Q(\avec{c}'(\avec{z}))$ using $\Pi''$. Indeed,
\begin{itemize}
\item[--]
if $E_i=A(z)$, then $A(\avec{c}(z)) \in \A'$, so there must exist either a unary ground atom $B(\avec{c}(z)) \in \A$ such that \mbox{$\T\models B(x) \to A(x)$}
or a binary ground atom \mbox{$\varrho(a,\avec{c}(z)) \in \A$}, for some $a\in\ind(\A)$, such that \mbox{$\T\models \exists y\,\varrho(y,x)\to A(x)$}; in the latter case, we set $\avec{c}'(y_i) = a$;
\item[--]
similarly, if \mbox{$E_i=P(z,z')$}, then there must exist a binary ground atom \mbox{$\varrho(\avec{c}(z),\avec{c}(z')) \in \A$} such that $\T\models \varrho(x,y) \to P(x,y)$.
\end{itemize}
It then suffices to choose \mbox{$Q(\avec{z}) \gets I \land \textit{EQ} \land E_1' \land \ldots \land E_n'$}
with atoms $E_i'$ whose form match that of the ground atoms in $\A$ corresponding to $E_i$.

For the converse direction, it suffices to observe that $\Pi\subseteq \Pi''$.

To complete the proof, we note that it is in \NL{} to decide whether an atom belongs to $\upsilon(E_i)$,
and thus we can construct the program $\Pi'$ by means of an $\mathsf{L}^\NL$-transducer.
\end{proof}

\subsection{Theorem~\ref{thm:LOGCFLdatalog}}

Next, we combine the transformation in Lemma~\ref{thm:NDLToSkinny} with the established complexity in Lemma~\ref{prop:SkinnyNDLEvaluation} to obtain the combined complexity upper bound:

\smallskip

\textsc{Theorem}~\ref{thm:LOGCFLdatalog}. \ \ {\em
For every $c > 0$ and $\wid > 0$, evaluation of NDL queries $(\Pi, G(\avec{x}))$
of width at most~$\wid$ and such that $\sdep(\Pi,G) \le c \log|\Pi|$
is in \LOGCFL{} for combined complexity.}

\begin{proof}
By Lemma~\ref{thm:NDLToSkinny}, $(\Pi,G)$ is equivalent to a skinny  NDL query $(\Pi',G)$ such that $|\Pi'| = O(|\Pi|^2)$, $\wid(\Pi',G) \le \wid$, and $\dep(\Pi',G) \le \sdep(\Pi, G)$. By
Lemma~\ref{prop:SkinnyNDLEvaluation}, query evaluation for $(\Pi',G)$ over $\A$ is   done by an NAuxPDA in space $\log |\Pi'| + \wid(\Pi',G) \cdot \log |\A| = O(\log |\Pi|+ \log |\A|)$ and time $2^{O(\dep(\Pi',G))}  \le |\Pi|^{O(1)}$.
\end{proof}

\subsection{\textsc{Log}-rewritings}\label{appA2}

\begin{lemma}
For any complete data instance $\A$, any $D \in \R$,  any type $\tpd$ with $\dom(\tpd) = \dD$ and any tuples \mbox{$\avec{b} \in \ind(\A)^{|\dD|}$} and $\avec{a}\in \ind(\A)^{|\avec{x}_D|}$, we have \mbox{$\Pi^{\textsc{Log}}_\omq,\A \models \rpred^{\tpd}_D(\avec{b}, \avec{a})$} iff there is a homomorphism \mbox{$h\colon \q_D \to \can$} such that
\begin{equation}\label{eq3}
h(x) = \avec{a}(x), \ \ \  \text{ for } x\in \avec{x}_D,\quad
\text{ and } \quad h(z)  = \avec{b}(z) \tpd(z), \ \ \  \text{ for  } z\in \dD.
\end{equation}
\end{lemma}
\begin{proof}
$(\Rightarrow)$
The proof is by induction on $\prec$. For the basis of induction, let $D$ be of size~1.
By the definition of $\Pi^{\textsc{Log}}_\omq$, there exists
a type $\tpr$  such that  \mbox{$\dom(\tpr) = \lambda(\sigma(D))$} and $\tpd$ agrees with $\tpr$ on $\dD$ and a respective tuple
\mbox{$\avec{c} \in \ind(\A)^{|\lambda(\sigma(D))|}$} such that $\avec{c}(z) = \avec{b}(z)$,
for all \mbox{$z \in \dD$},  and $\avec{c}(x) = \avec{a}(x)$, for all $x\in \avec{x}_D$, and
$\Pi^{\textsc{Log}}_\omq,\A \models \mathsf{At}^{\tpr}(\avec{c})$. Then, for any atom $S(\avec{z}) \in \q_D$, we have $\avec{z}\subseteq\lambda(\sigma(D))$, whence $\can \models S(h(\avec{z}))$ as $\tpd$ agrees with $\tpr$ on $\dD$.

\smallskip

For the inductive step, suppose that we have $\Pi^{\textsc{Log}}_\omq,\A\models \rpred^{\tpd}_D(\avec{b}, \avec{a})$. By the definition of $\Pi^{\textsc{Log}}_\omq$, there exists
a type $\tpr$  such that  $\dom(\tpr) = \lambda(\sigma(D))$ and $\tpd$ agrees with $\tpr$ on their common domain and a respective tuple
$\avec{c} \in \ind(\A)^{|\lambda(\sigma(D))|}$ such that $\avec{c}(z) = \avec{b}(z)$,
for all $z \in \dD$,  and $\avec{c}(x) = \avec{a}(x)$, for all $x\in \avec{x}_D$, and
\begin{equation*}
\Pi^{\textsc{Log}}_\omq,\A \models \mathsf{At}^{\tpr}(\avec{c}) \land
\bigwedge_{D' \prec D} \rpred^{(\tpr\cup\tpd) \restr\dDp}_{D'}(\avec{b}_{D'},\avec{a}_{D'}),
\end{equation*}
where $\avec{b}_{D'}$ and $\avec{a}_{D'}$ are the restrictions of
$\avec{b} \cup \avec{c}$ to $\dDp$ and  of $\avec{a}$ to
$\avec{x}_{D'}$, respectively.
By the induction hypothesis, for any $D'\prec D$, there is a  homomorphism $h_{D'}\colon\q_{D'}\to\can$ such that~\eqref{eq3} is satisfied.

\smallskip

Let us show that the $h_{D'}$ agree on common variables. Suppose that $z$ is shared by $\q_{D'}$ and
$\q_{D''}$ for $D' \prec D$ and $D'' \prec D$.  By
the definition of tree decomposition,
for every $z \in V$, the nodes $\{\nd\mid z \in \lambda(\nd)\} $ induce a connected subtree of~$T$, and so
\mbox{$z \in \lambda(\sigma(D)) \cap \lambda(\nd') \cap \lambda(\nd'')$}, where
$\nd'$ and $\nd''$ are the unique neighbours of $\sigma(D)$ lying in $D'$ and $D''$, respectively.
Since
$\tpd'=(\tpd\cup\tpr) \restr\dDp$ and $\tpd''=(\tpd\cup\tpr) \restr\dDpp$ are the restrictions of
$\tpd \cup\tpr$, we have
$\tpd'(z)  = \tpd''(z)$.
This implies that
\begin{equation*}
h_{D'}(z)  = \avec{c}(z) \tpd'(z) = \avec{c}(z) \tpd''(z) = h_{D''}(z).
\end{equation*}

Now we define $h$ on every $z$ in $\q_D$ by taking
\begin{equation*}
h(z) = \begin{cases}
h_{D'}(z) & \text{if }  z \in\lambda(t),\\ &\hspace*{1em} \text{ for }t\in D' \text{ and } \ D' \prec D,\\
\avec{c}(z)\cdot (\tpd\cup\tpr)(z), & \text{if } z \in\lambda(\sigma(D)).	
\end{cases}
\end{equation*}
If follows that $h$ is well defined, $h$ satisfies~\eqref{eq3} and that $h$ is a homomorphism
from $\q_D$ to $\can$.
Indeed, take an atom $S(\avec{z}) \in \q_D$. Then either $\avec{z}\subseteq\lambda(\sigma(D))$,
in which case $\can \models S(h(\avec{z}))$ since $\tpd$ is compatible
with $\sigma(D)$ and $\Pi^{\textsc{Log}}_\omq, \A \models \mathsf{At}^{\tpr}(\avec{c})$, or
$S(\avec{z}) \in \q_{D'}$ for some $D' \prec D$, in which case we use the fact
that $h$ extends a homomorphism $h_{D'}$.

\bigskip

$(\Leftarrow)$
The proof is by induction on $\prec$. Fix $D$ and $\tpd$ such that $|\tpd| = |\dD|$.
Take tuples $\avec{b} \in \ind(\A)^{|\dD|}$ and $\avec{a} \in \ind(\A)^{|\avec{x}_D|}$, and
a homomorphism \mbox{$h\colon\q_D\to\can$} satisfying~\eqref{eq3}.
Define a type $\tpr$ and a tuple \mbox{$\avec{c} \in \ind(\A)^{|\lambda(\sigma(D))|}$} by taking, for all $z\in \lambda(\sigma(D))$,
\begin{equation*}
\tpr(z) = w \ \text{ and } \  \avec{c}(z) = a, \ \ \  \text{ if } h(z) = a w, \text{ for  } a\in\ind(\A).
\end{equation*}
By definition, $\dom(\tpr) = \lambda(\sigma(D))$ and, by~\eqref{eq3}, $\tpr$ and $\tpd$ agree on the common domain.
For the inductive step, for each $D'\prec D$,  let
$h_{D'}$ be the restriction of $h$ to $\q_{D'}$ and let $\avec{b}_{D'}$ and and $\avec{a}_{D'}$  be the restrictions of
$\avec{b} \cup \avec{c}$ to $\dDp$ and  of $\avec{a}$ to
$\avec{x}_{D'}$, respectively.
By the inductive hypothesis, $\Pi^{\textsc{Log}}_\omq, \A \models \rpred^{\tpd'}_{D'}( \avec{b}_{D'},\avec{a}_{D'})$. (This argument is not needed for the basis of induction.)
Since $h$ is a homomorphism, we have
$\Pi^{\textsc{Log}}_\omq,\A \models \mathsf{At}^{\tpr}(\avec{c})$,
whence, $\Pi^{\textsc{Log}}_\omq,\A\models  \rpred^{\tpd}_D(\avec{b}, \avec{a})$.
\end{proof}

It follows that answering OMQs $\omq(\avec{x}) = (\T,\q(\avec{x}))$ with $\T$ of finite depth $\od$ and $\q$ of treewidth $\twi$ over any data instance $\A$ can be done  in time
\begin{equation}\tag{\ref{eq:time1}}
\textit{poly}(|\T|^{\od \twi},\,|\q|,\,|\A|^{\twi}).
\end{equation}
Indeed, we can evaluate $(\Pi_\omq^{\smash{\textsc{Log}}},\rpred^{\avec{\varepsilon}}_T(\avec{x}))$ in time polynomial in $|\Pi_\omq^{\smash{\textsc{Log}}}|$ and $|\A|^{\wid(\Pi^{\smash{\textsc{Log}}}_{\omq},\rpred^{\avec{\varepsilon}}_T)}$, which are bounded by a polynomial in
$|\T|^{2\od(\twi+1)}$, $|\q|$ and $|\A|^{2(\twi+1)}$.

\subsection{\textsc{Lin}-rewritings}\label{appA3}

\begin{lemma}
For any complete data instance $\A$, any predicate $G^{\tpd}_n$,
any $\avec{a}\in \ind(\A)^{|\avec{x}^n|}$ and $\avec{b} \in \ind(\A)^{|\avec{z}^n_\exists|}$,
we have
$\Pi^{\textsc{Lin}}_\omq,\A \models G^{\tpd}_n(\avec{b}, \avec{a})$  iff there is a homomorphism
$h\colon \q_n \to \can$ such that
\begin{equation}\label{nl-rewriting-eq}
h(x) = \avec{a}(x), \ \ \ \text{ for } x\in \avec{x}^n,\quad
\text{ and } \quad h(z) = \avec{b}(z) \tpd(z), \ \ \ \text{ for  } z\in\avec{z}^n_\exs.
\end{equation}
\end{lemma}
\begin{proof}
The proof is by induction on $n$.

\smallskip

For the base case ($n=M$), first suppose that we have
$\Pi^{\textsc{Lin}}_\omq,\A \models G^{\avec{w}}_M(\avec{b}, \avec{a})$.
The only rule in $\Pi^{\textsc{Lin}}_\omq$ with head predicate $G^{\tpd}_M$ is $G^{\tpd}_{M}(\avec{z}^{M}_\exs, \avec{x}^M) \leftarrow \mathsf{At}^{\tpd}(\avec{z}^M)$ with $\avec{z}^M = \avec{z}^M_\exs \uplus \avec{x}^M$, which is equivalent to
\begin{equation}\label{rule-M}
G^{\tpd}_{M}(\avec{z}^{M}_\exs, \avec{x}^M) \leftarrow \bigwedge_{z \in \avec{z}^M} \Bigl(\bigwedge_{\substack{A(z) \in \q\\ \tpd(z) = \varepsilon}} \!\!\!\!A(z)
\,\, \land \!\! \bigwedge_{\substack{P(z, z) \in \q\\ \tpd(z) = \varepsilon}} \!\!\!\! \!\! P(z, z) \,\,\land \
\bigwedge_{\substack{\tpd(z) = \varrho w}} \!\!\!A_\varrho(z)\Bigr).
\end{equation}
So the body of this rule must be satisfied when $\avec{b}$ and $\avec{a}$ are substituted for $\avec{z}^M_\exs$ and $\avec{x}^M$ respectively.
Moreover, by local compatibility of $\avec{w}$ with $\avec{z}^M$, we know that $\avec{w}(x) = \varepsilon$ for every $x \in \avec{x}^M$.
It follows that
\begin{itemize}
\item[--] $A(\avec{a}(x)) \in \A$ for every $A(x) \in \q$ such that $x \in \avec{x}^M$;
\item[--] $A(\avec{b}(z)) \in \A$ for every $A(z) \in \q$ such that $z \in \avec{z}^M_\exs$ and $\tpd(z) = \varepsilon$;
\item[--] $P(\avec{a}(x),\avec{a}(x)) \in \A$ for every $P(x, x) \in \q$ such that $x \in \avec{x}^M$;
\item[--] $P(\avec{b}(z),\avec{b}(z)) \in \A$ for every $P(z, z) \in \q$ such that $z \in \avec{z}^M_\exs$ and $\tpd(z) = \varepsilon$;
\item[--] $A_\varrho(z) \in \A$ for every $z\in \avec{z}^M$ with $\tpd(z) = \varrho w$.
\end{itemize}
Now let $h^M$ be the unique mapping from $\avec{z}^M$ to $\Delta^{\can}$ satisfying~\eqref{nl-rewriting-eq}. First note that $h^M$ is well-defined, since by the last item, if $\tpd(z) = \varrho w$, then
we have \mbox{$A_\varrho(z) \in \A$} and $\varrho w \in \twords$, so $\avec{b}(z) \varrho w$ belongs to $\Delta^{\can}$.
To show that $h^M$ is a homomorphism of $\q_M$ into $\can$, first recall that the atoms of $\q_M$ are of two types: $A(z)$ or $P(z, z)$,
with $z \in \avec{z}^M$.
Take some $A(z) \in \q_M$. If $\tpd(z) = \varepsilon$, then we immediately obtain either
$A(h^M(z))=A(\avec{a}(z)) \in \A$ or  $A(h^M(z))=A(\avec{b}(z)) \in \A$, depending on whether $z \in \avec{z}^M_\exs$ or in $\avec{x}^M$.
Otherwise, if $\tpd(z) \neq \varepsilon$, then the local compatibility of $\tpd$ with $\avec{z}^M$ means that the final letter $\varrho$ in $\tpd(z)$ is such that $\T\models \exists y\,\varrho(y,x)\to  A(x)$, hence $h^M(z) = \avec{b}(z) \tpd(z) \in A^{\can}$.
Finally, suppose that $P(z, z) \in \q$. The local compatibility of $\tpd$ with $\avec{z}^M$ ensures that either $\tpd(z) = \varepsilon$ or $\T\models P(x,x)$. In the former case,
we have either $P(\avec{a}(z),\avec{a}(z)) \in \A$ or $P(\avec{b}(z),\avec{b}(z)) \in \A$, depending again on whether
$z \in \avec{z}^M_\exs$ or $z \in\avec{x}^M$. In the latter case, $(h^M(z),h^M(z))\in P^{\can}$.

\smallskip

For the other direction, $(\Leftarrow)$, of the base case, suppose that the mapping $h^M$ given by~\eqref{nl-rewriting-eq} defines a homomorphism from $\q_M$ into $\can$.
We therefore have:
\begin{itemize}
\item[--] $\avec{a}(x) \in A^{\can}$ for every $A(x) \in \q$ with $x \in \avec{x}^M$;
\item[--] $\avec{b}(z) \tpd(z) \in A^{\can}$ for every $A(z) \in \q$ with $z \in \avec{z}^M_\exs$;
\item[--] $(\avec{a}(x),\avec{a}(x)) \in P^{\can}$ for every $P(x, x) \in \q$ such that $x \in \avec{x}^M$;
\item[--] $(\avec{b}(z),\avec{b}(z)) \in P^{\can}$ for every $P(z, z) \in \q$ such that $z \in \avec{z}^M_\exs$;
\item[--]  $\T, \A \models \exists y\,\varrho(\avec{b}(z), y)$ for every $z \in \avec{z}^M_\exs$  with $\tpd(z) = \varrho w$ (for otherwise $\avec{b}(z)\tpd(z)$
would not belong to the domain of $\can$).
\end{itemize}
The first two items, together with completeness of the data instance $\A$, ensure that all atoms in
\begin{equation*}
\bigl\{A(z) \mid A(z) \in \q, z\in\avec{z}^M, \tpd(z) = \varepsilon\bigr\}
\end{equation*}
are present in $\A$ when $\avec{b}$ and $\avec{a}$  substituted for $\avec{z}^M_\exs$ and $\avec{x}^M$, respectively.
The third and fourth items, again together with completeness of $\A$, ensure the presence of the atoms in
\begin{equation*}
\bigl\{P(z, z) \mid P(z, z) \in \q, z\in\avec{z}^M, \tpd(z) = \varepsilon\bigr\}.
\end{equation*}
Finally, the fifth item plus completeness of $\A$ ensure that $\A$ contains all atoms in
\begin{equation*}
\{A_\varrho(z) \mid z\in\avec{z}^M, \tpd(z) = \varrho w\}.
\end{equation*}
It follows that the body of the unique rule for $G^{\tpd}_{M}$
is satisfied when $\avec{b}$ and $\avec{a}$ are substituted for $\avec{z}^M_\exs$ and $\avec{x}^M$ respectively, and thus
$\Pi^{\textsc{Lin}}_\omq,\A \models G^{\tpd}_M(\avec{b}, \avec{a})$.

\bigskip

For the induction step, assume that the statement has been shown to hold for all $n \leq k+1 \leq M$,
and let us show that it holds when $n=k$. For the first direction, $(\Rightarrow)$, suppose $\Pi_\omq^{\textsc{Lin}},\A \models G^{\tpd}_k(\avec{b}, \avec{a})$.
It follows that there exists a pair of types $(\tpd, \tpr)$ compatible with $(\avec{z}^k, \avec{z}^{k+1})$
and an assignment $\avec{c}$ of individuals from $\A$ to the variables in $\avec{z}^k \cup \avec{z}^{k+1}$
such that $\avec{c}(x)=\avec{a}(x)$
for all \mbox{$x \in (\avec{z}^k \cup \avec{z}^{k+1})\cap \avec{x}$}, and $\avec{c}(z)=\avec{b}(z)$ for all $z \in \avec{z}^k_\exs$, and
such that every atom in the body of the clause
\begin{equation*}
G^{\tpd}_{k}(\avec{z}^{k}_\exs, \avec{x}^k) \leftarrow
\mathsf{At}^{\tpd\cup\tpr}(\avec{z}^k,\avec{z}^{k+1})  \land G^{\tpr}_{k+1}(\avec{z}^{k+1}_\exs, \avec{x}^{k+1})
\end{equation*}
is entailed from $\Pi_\omq^\textsc{Lin},\A$ when the individuals in $\avec{c}$ are substituted for $\avec{z}^k \cup \avec{z}^{k+1}$.
Recall that  $\mathsf{At}^{\tpd\cup\tpr}(\avec{z}^k,\avec{z}^{k+1})$ is the conjunction of the following atoms, for $z,z'\in \avec{z}^k\cup\avec{z}^{k+1}$:
\begin{itemize}
\item[--] $A(z)$, if $A(z) \in \q$ and $(\tpd\cup\tpr)(z) = \varepsilon$,
\item[--] $P(z, z')$,  if $P(z, z') \in \q$ and $(\tpd\cup\tpr)(z) = (\tpd\cup\tpr)(z') = \varepsilon$,
\item[--] $z = z'$, if $P(z, z') \in \q$ and either $(\tpd\cup\tpr)(z) \neq \varepsilon$  or $(\tpd\cup\tpr)(z') \neq \varepsilon$,
\item[--] $A_\varrho(z)$, if $(\tpd\cup\tpr)(z)$ is of the form $\varrho w$.
\end{itemize}
In particular, we have $\Pi_\omq^\textsc{Lin},\A \models  G^{\tpr}_{k+1}(\avec{c}(\avec{z}^{k+1}_\exs), \avec{c}(\avec{x}^{k+1}))$.
By the induction hypothesis, there exists a homomorphism $h^{k+1}\colon \q_{k+1} \to \can$ such that
$h^{k+1}(z) = \avec{c}(z) \tpr(z)$ for every $z \in \avec{z}^{k+1}_\exs \cup \avec{x}^{k+1}$.
Define a mapping $h^k$ from $\vars(\q_k)$ to $\Delta^{\can}$ by setting $h^k(z)=h^{k+1}(z)$ for every variable $z\in \vars(\q_{k+1})$,
setting $h^k(x) = \avec{a}(x)$ for every $x \in \avec{z}^k \cap \avec{x}$, and setting
$h^k(z)=\avec{b}(z) \tpd(z)$ for every $z \in \avec{z}^k$.
Using the same argument as was used in the base case,
we can show that $h^k$ is well-defined. For atoms from $\q_k$ involving only variables from $\q_{k+1}$, we can
use the induction hypothesis to conclude that they are satisfied under $h^k$,
and for atoms only involving variables from $\avec{z}^k$, we can argue as in the base case.
It thus remains to handle role atoms that contain one variable from $\avec{z}^k$ and one variable from $\avec{z}^{k+1}$.
Consider such an atom $P(z, z') \in \q_k$, for $z\in\avec{z}^k$ and $z'\in\avec{z}^{k+1}$.
If $\tpd(z) = \tpr(z') = \varepsilon$, then the atom $P(z, z')$ appears in the body of the clause we are considering.
It follows that $\Pi_\omq^{\textsc{Lin}},\A \models  P(\avec{c}(z), \avec{c}(z'))$,
hence $(\avec{c}(z), \avec{c}(z')) \in P^{\can}$.
It then suffices to note that $\avec{c}$ agrees with $\avec{a}$ and $\avec{b}$ on the variables in $\avec{z}^k$.
Next suppose that either $\tpd(z) \neq \varepsilon$ or $ \tpr(z') \neq \varepsilon$.
It follows that the clause body contains $z = z'$, hence $\avec{c}(z) = \avec{c}(z')$.
As $(\tpd, \tpr)$ is compatible with $(\avec{z}^k, \avec{z}^{k+1})$, one of the following must hold:
either
\begin{itemize}
\item[(a)] $\tpr(z')= \tpd(z)$ and $\T\models P(x,x)$
\item[(b)] or $\T\models \varrho(x,y) \to P(x,y)$ and either $\tpr(z')= \tpd(z) \varrho$ or
$\tpd(z) = \tpr(z') \varrho^-$.
\end{itemize}
We give the argument in the case where $z \in \avec{z}^k_\exs$ (the argument is entirely similar if $z \in \avec{x}^k$). If (a) holds, then
\begin{equation*}
(h^k(z), h^k(z'))= (\avec{b}(z) \tpd(z), \avec{c}(z') \tpr(z'))=
(\avec{b}(z) \tpd(z), \avec{c}(z') \tpd(z))\in P^{\can}
\end{equation*}
since $\T\models P(x,x)$ and $\avec{c}(z') =\avec{c}(z) = \avec{b}(z)$.
If the first option of~(b) holds, then
\begin{equation*}
(h^k(z), h^k(z'))= (\avec{b}(z) \tpd(z), \avec{c}(z') \tpr(z'))=
(\avec{b}(z) \tpd(z), \avec{c}(z') \tpd(z) \varrho)\in P^{\can}
\end{equation*}
since $\T\models \varrho(x,y) \to P(x,y)$ and $\avec{c}(z') =\avec{c}(z) = \avec{b}(z)$.
If the second option of~(b) holds,
then
\begin{equation*}
(h^k(z), h^k(z'))= (\avec{b}(z) \tpd(z), \avec{c}(z') \tpr(z'))=
(\avec{b}(z) \tpr(z') \varrho^-, \avec{c}(z') \tpr(z'))\in P^{\can}
\end{equation*}
since $\T\models \varrho(x,y) \to P(x,y)$.

\bigskip

For the converse direction, $(\Leftarrow)$, of the induction step, let $\tpd$ be a type that is locally compatible with $\avec{z}^k$,
 let  $\avec{a}\in \ind(\A)^{|\avec{x}^k|}$, $\avec{b} \in \ind(\A)^{|\avec{z}^k_\exists|}$, and
let $h^k\colon \q_k \to \can$ be a homomorphism satisfying
\begin{equation}\label{indstep-sec5}
h^k(x) = \avec{a}(x), \ \  \ \text{ for } x\in \avec{x}^k, \quad \text{ and }\quad
h^k(z) = \avec{b}(z) \tpd(z), \ \ \ \text{ for  } z\in\avec{z}^k_\exs.
\end{equation}
We let $\avec{c}$ for $\avec{z}^{k+1}$ be defined by setting $\avec{c}(z)$ equal to the unique individual $c$ such that $h(z)$ is of the form $c w$ (for some $w \in \twords$),
and let $\tpr$ be the unique type for $\avec{z}^{k+1}$ satisfying $h(z) = \avec{c}(z) \tpr(z)$ for every $z \in \avec{z}^{k+1}$; in other words,  we obtain $\avec{s}(z)$ from $h(z)$ by omitting the initial individual name $\avec{c}(z)$.
Note that since $\avec{x}^{k+1} \subseteq \avec{x}^{k}$, we have $\avec{a}(x)=\avec{c}(x)$ for every $x\in \avec{x}^{k+1}$.
It follows from the fact that $h^k$ is a homomorphism that $\tpr$ is locally compatible with $\avec{z}^{k+1}$
and that, for every role atom
$P(z, z') \in \q_k$ with $z\in\avec{z}^k$ and $z'\in\avec{z}^{k+1}$, one of the following holds:
(\emph{i}) $\tpd(z) =\tpr(z')= \varepsilon$,
(\emph{ii}) $\tpd(z) =\tpr(z')$ and $\T\models P(x,x)$,
(\emph{iii}) $\T\models \varrho(x,y) \to P(x,y)$ and either $\tpr(z')= \tpd(z) \varrho$ or
$\tpd(z) = \tpr(z') \varrho^-$.
Thus, the pair of types $(\tpd, \tpr)$ is compatible with $(\avec{z}^{k}, \avec{z}^{k+1})$,
and so the following rule appears in $\Pi_\omq^{\textsc{Lin}}$:
\begin{equation*}
G^{\tpd}_{k}(\avec{z}^{k}_\exs, \avec{x}^k) \leftarrow
\mathsf{At}^{\tpd\cup\tpr}(\avec{z}^k,\avec{z}^{k+1})  \land G^{\tpr}_{k+1}(\avec{z}^{k+1}_\exs, \avec{x}^{k+1}),
\end{equation*}
where we recall that $\mathsf{At}^{\tpd\cup\tpr}(\avec{z}^k,\avec{z}^{k+1})$ is the conjunction of the following atoms, for $z,z'\in \avec{z}^k\cup\avec{z}^{k+1}$:
\begin{itemize}
\item[--] $A(z)$, if $A(z) \in \q$ and $(\tpd\cup\tpr)(z) = \varepsilon$,
\item[--] $P(z, z')$,  if $P(z, z') \in \q$ and $(\tpd\cup\tpr)(z) = (\tpd\cup\tpr)(z') = \varepsilon$,
\item[--] $z = z'$, if $P(z, z') \in \q$ and either $(\tpd\cup\tpr)(z) \neq \varepsilon$  or $(\tpd\cup\tpr)(z') \neq \varepsilon$,
\item[--] $A_\varrho(z)$, if $(\tpd\cup\tpr)(z)$ is of the form $\varrho w$.
\end{itemize}
It follows from Equation~\eqref{indstep-sec5} and the fact that $h^k$ is a homomorphism that each of the ground atoms obtained
by taking an atom from $\mathsf{At}^{\tpd\cup\tpr}(\avec{z}^k,\avec{z}^{k+1})$ and substituting
 $\avec{a}$, $\avec{b}$, and $\avec{c}$ for  $\avec{x}^k$, $\avec{z}^k_\exs$ and $\avec{z}^{k+1}$, respectively, is present in $\A$.
By applying the induction hypothesis to the predicate $G^{\tpr}_{k+1}$ and the homomorphism $h^{k+1}\colon \q_{k+1} \to \can$
obtained by restricting $h^k$ to $\vars(\q_{k+1} )$, we obtain that  $\Pi_\omq^{\textsc{Lin}},\A \models G^{\tpr}_{k+1}(\avec{c}(\avec{z}^{k+1}_\exs), \avec{a}(\avec{x}^{k+1}))$.
Since for the considered substitution, all body atoms are entailed, we can conclude that \mbox{$\Pi_\omq^{\textsc{Lin}},\A \models G^{\tpd}_k( \avec{b},\avec{a})$}.
\end{proof}

It follows that answering OMQs $\omq(\avec{x})=(\T,\q(\avec{x}))$ with $\T$ of finite depth $\od$ and tree-shaped $\q$ with $\nlf$ leaves over any data instance $\A$ can be done  in time
\begin{equation}\tag{\ref{eq:time2}}
\textit{poly}(|\T|^{\od\nlf},\, |\q|,\, |\A|^{\nlf} ).
\end{equation}
Indeed, $(\Pi_\omq^{\smash{\textsc{Lin}}},\rpred(\avec{x}))$ can be evaluated  in time polynomial
in $|\Pi^{\smash{\textsc{Lin}}}_{\omq}|$ and $|\A|^{\wid(\Pi^{\smash{\textsc{Lin}}}_{\omq},G)}$, which are bounded by a polynomial in $|\T|^{2\od\nlf}$, $|\q|$ and $|\A|^{2\nlf}$.

\subsection{\textsc{Tw}-rewritings}\label{appA4}

\begin{lemma}
For any OMQ \mbox{$\omq(\avec{x}_0)=(\T, \q_0(\avec{x}_0))$} with a tree-shaped CQ, any complete data instance~$\A$, any $\q(\avec{x}) \in \sqset$ and $\avec{a}\in\ind(\mathcal{A})^{|\avec{x}|}$, we have
$\Pi^{\textsc{Tw}}_{\omq},\A \models \rew_\q(\avec{a})$ iff there exists a homomorphism $h\colon \q \to \can$ such that $h(\avec{x})= \avec{a}$.
\end{lemma}
\begin{proof}
An inspection of the definition of the set~$\sqset$ shows that every $\q(\avec{x}) \in \sqset$ is a tree-shaped query having at least one answer variable,
with the possible exception of the original query $\q_0(\avec{x}_0)$, which may be Boolean.

Just as we did for subtrees in Section~\ref{sec:boundedtw}, we associate a binary relation on the queries in $\sqset$ by setting
$ \q'(\avec{x}') \prec \q(\avec{x})$ whenever $\q'(\avec{x}')$ was introduced when applying one of the two decomposition conditions on p.~\pageref{page:decomposition} to $\q(\avec{x})$.
The proof is by induction on the subqueries in $\sqset$, according to $\prec$.
We will start by establishing the statement
for all queries in $\sqset$ other than $\q_0(\avec{x}_0)$,
and afterwards, we will complete the proof by giving an argument for $\q_0(\avec{x}_0)$.

\bigskip

For the basis of induction, take some $\q(\avec{x}) \in \sqset$ that is minimal in the ordering induced by $\prec$, which means that
$\vars(\q) = \avec{x}$. Indeed, if there is an existentially quantified variable, then the first decomposition rule will give rise to a `smaller' query (in particular, if $|\vars(\q)| = 2$, then although the `smaller' query may have the same atoms, the selected existential variable will become an answer variable).
For the first direction, $(\Rightarrow)$, suppose that $\Pi^{\textsc{Tw}}_{\omq},\A \models \rew_\q(\avec{a})$. By definition, $\rew_\q(\avec{x}) \gets \q(\avec{x})$ is the only clause   with head predicate~$\rew_\q$. Thus, all atoms in the ground CQ $\q(\avec{a})$ are present in $\A$, and hence the desired homomorphism exists.
For the converse direction, $(\Leftarrow)$, suppose there is a homomorphism $h\colon \q(\avec{x}) \to \can$ such that $h(\avec{x}) = \avec{a}$.
It follows that every atom in the ground CQ $\q(\avec{a})$ is entailed from $\T, \A$.
Completeness of $\A$ ensures that all of the ground atoms in $\q(\avec{a})$ are present in $\A$, and thus we can apply the clause $\rew_\q(\avec{x}) \gets \q(\avec{x})$
to derive $\rew_\q(\avec{a})$.

\bigskip

For the induction step, let $\q(\avec{x}) \in \sqset$ with $\vars(\q)\ne\avec{x}$ and
suppose that the claim holds for all $\q'(\avec{x}')  \in  \sqset$ with
$ \q'(\avec{x}') \prec \q(\avec{x})$. For the first direction, $(\Rightarrow)$, suppose $\Pi^{\textsc{Tw}}_{\omq},\A \models \rew_\q(\avec{a})$.
There are two cases, depending on which type of clause was used to derive $\rew_\q(\avec{a})$.
\begin{itemize}
\item Case 1: $\rew_\q(\avec{a})$ was derived by an application of the following clause:
\begin{equation*}
\rew_\q(\avec{z}) \leftarrow \hspace*{-0.7em}\bigwedge_{A(z_\q)\in \q}\hspace*{-1em} A(z_\q) \ \ \land\hspace*{-0.7em} \bigwedge_{P(z_\q,z_\q) \in \q} \hspace*{-1.5em}P(z_\q,z_\q)  \ \ \land \bigwedge_{1\leq i\leq n}\hspace*{-0.5em} \rew_{\q_i}(\avec{x}_i),
\end{equation*}
where $\q_1(\avec{x}_1), \ldots, \q_n(\avec{x}_n)$ are the subqueries induced by the neighbours of $z_\q$ in the Gaifman graph $\gfmn$ of $\q$.
Then there exists a substitution $\avec{c}$ for the variables in the body of this rule that coincides with $\avec{a}$ on $\avec{z}$ and is
such that the ground atoms obtained by applying $\avec{c}$ to the variables in the body are all entailed from $\Pi^{\textsc{Tw}}_{\omq},\A$.
In particular,  $\Pi^{\textsc{Tw}}_{\omq},\A  \models \rew_{\q_i}(\avec{c}(\avec{x}_i))$
for every $1 \leq i \leq n$. We can apply the induction hypothesis to the $\q_i(\avec{x}_i)$ to obtain homomorphisms $h_i\colon \q_i \to \can$ such that  $h_i(\avec{x}_i)=\avec{c}(\avec{x}_i)$.
Let $h$ be the mapping from $\vars(\q)$ to $\Delta^{\can}$ defined by taking
$h(z)=h_i(z)$, for $z \in \vars(\q_i)$.
Note that $h$ is well-defined since \mbox{$\vars(\q)=\bigcup_{i = 1}^n\vars(\q_i)$}, and
the $\q_i$ have no variable in common other than $z_\q$, which is sent to $\avec{c}(z_\q)$ by every $h_i$.
To see why $h$ is a homomorphism from~$\q$ to $\can$, observe that
\begin{equation*}
\q \ \ = \ \ \bigcup_{i =1}^n \q_i \ \ \cup \ \ \bigl\{A(z_\q)\in \q\bigr\} \ \ \cup \ \ \bigl\{P(z_\q,z_\q) \in \q\bigr\}.
\end{equation*}
By the definition of~$h$, all atoms in $\bigcup_{i = 1}^n \q_i$ hold under $h$.
If $A(z_\q)\in \q$, then $A(\avec{c}(z_\q))$ is entailed from $\Pi^{\textsc{Tw}}_{\omq},\A$, and hence is present in $\A$.
Similarly, we can show that for every $P(z_\q,z_\q) \in \q$, the ground atom $P(\avec{c}(z_\q),\avec{c}(z_\q))$ belongs to $\A$.
It follows that all of these atoms hold in $\can$ under~$h$. Finally, we recall that $\avec{c}$ coincides with $\avec{a}$
on $\avec{x}$, so we have $h(\avec{x})=\avec{a}$, as required.

\medskip

\item Case 2: $\rew_\q(\avec{a})$ was derived by an application of the following clause, for a tree witness $\t$ for $(\T,\q(\avec{x}))$ generated by $\varrho$ with $\tr\neq \emptyset$ and $z_\q\in\ti$:
\begin{equation*}
\rew_\q(\avec{x}) \leftarrow A_\varrho(z_0)\
\land \hspace*{-0.5em}\bigwedge_{z\in \tr\setminus \{z_0\}}\hspace*{-0.5em} (z=z_0) \ \ \land \  \bigwedge_{1\leq i \leq k} \rew_{\q_i^\t}(\avec{x}_i^\t),
\end{equation*}
where $\q_1^\t, \dots, \q_k^\t$ are the connected components of $\q$ without  $\q_\t$ and $z_0$ is some variable in $\tr$.
There must exist a substitution $\avec{c}$ for the variables in the body of this rule that coincides with $\avec{a}$ on $\avec{x}$ and is
such that the ground atoms obtained by applying $\avec{c}$ to the variables in the body are all entailed from $\Pi^{\textsc{Tw}}_{\omq},\A$.
In particular, for every $1 \leq i \leq k$, we have $\Pi^{\textsc{Tw}}_{\omq},\A \models \rew_{\q_i^\t}(\avec{c}(\avec{x}_i^\t))$.
We can apply the induction hypothesis to the $\q_i^\t(\avec{z}_i^\t)$ to find homomorphisms $h_1, \ldots, h_k$ of $\q_1^\t, \ldots, \q_k^\t$ into $\can$
such that $h_i(\avec{x}_i^\t)=\avec{c}(\avec{x}_i^\t)$.
Since $\t$ is a tree witness for $(\T,\q(\avec{x}))$ generated by $\varrho$, there exists a homomorphism $h_\t$ of $\q_\t$ into
$\C_{\T,\{A_\varrho(a)\}}$ with $\tr = h_\t^{-1}(a)$ and such that $h_\t(z)$ begins by $a \varrho$ for every $z \in \ti$.
Now take $z_0 \in \tr$ such that $A_\varrho(z_0)$ is the atom in the clause body  (recall that $\tr \neq \emptyset$), and so $\Pi^{\textsc{Tw}}_{\omq},\A \models A_\varrho(\avec{c}(z_0))$,
which means that $A_\varrho(\avec{c}(z_0))$ must appear in $\A$. It follows that for every element in $ \C_{\T, \{A_\varrho(a)\}}$
of the form $a \varrho w$,  there exists a corresponding element $\avec{c}(z_0) \varrho w$ in $\Delta^{\can}$.
We now define a mapping $h$ from $\vars(\q)$ to $\Delta^{\can}$ as follows:
\begin{equation*}
h(z) = \begin{cases}
h_i(z), & \text{ for every } z \in \vars(\q_i^\t),\\
\avec{c}(z_0) \varrho w, & \text{ if } z \in \ti \text{ and } h_\t(z)= a \varrho w,\\
\avec{c}(z_0) & \text{ if } z \in \tr.
\end{cases}
\end{equation*}
Every variable in $\vars(\q)$ occurs  in $\tr\cup\ti$ or in exactly one of the $\q_i^\t$, and so is assigned a unique value by $h$. Note that although $\tr\cap\vars(\q_i^\t)$ is not necessarily empty, due to the equality atoms, we have $h(z) = h(z')$, for all $z,z'\in\tr$, and so the function is well-defined.
We claim that $h$ is a homomorphism from $\q$ into $\can$. Clearly, the atoms occurring in some $\q_i^\t$ are preserved under $h$. Now consider some unary atom $A(z)$ with $z \in \ti$. Then $h(z)=\avec{c}(z_0) \varrho w$, where $h_\t(z)= a \varrho w$.
Since $h_\t$ is a homomorphism, we know that $w$ ends with a role $\sigma$ such that $\T\models \exists y\,\sigma(y,x)\to A(x)$.
It follows that $h(z)$ also ends with $\sigma$, and thus $h(z) \in A^{\can}$.
Next, consider a binary atom $P(z,z')$, where at least one of $z$ and $z'$ belongs to $\ti$.
As $h_\t$ is a homomorphism, either
\begin{itemize}
\item[--] $\T\models \sigma(x,y)\to P(x,y)$, for some $\sigma$, such that  $h_\t(z') = h_\t(z) \sigma$
or $h_\t(z)= h_\t(z') \sigma^-$,
\item[--] or $\T\models P(x,x)$ and $h_\t(z') = h_\t(z)$.
\end{itemize}
We also know that $\avec{c}(z)=\avec{c}(z_0)$ for all $z \in \tr$, hence $h(z)=h(z_0)$ for all $z\in \tr$.
It follows  that in the former case we have $h(z') = h(z) \sigma$
or $h(z)= h(z') \sigma^-$ with $\T\models \sigma(x,y)\to P(x,y)$. In the latter case, we have $h(z') = h(z)$ with $\T\models P(x,x)$. Thus, $P(z,z')$ is preserved under  $h$. Finally, since $\avec{c}$ coincides with $\avec{a}$ on $\avec{x}$, we have $h(\avec{x})=\avec{a}$.
\end{itemize}

\bigskip

For the converse direction, $(\Leftarrow)$, of the induction step, suppose that $h$ is a homomorphism of $\q$ into $\can$ such that $h(\avec{x})=\avec{a}$.
There are two cases to consider, depending on where $h$ maps the `splitting' variable $z_\q$.
\begin{itemize}
\item Case 1: $h(z_\q) \in \ind(\A)$. Let $\q_1(\avec{x}_1), \ldots, \q_n(\avec{x}_n)$ be the subqueries of $\q(\avec{x})$
induced by the neighbours of $z_\q$ in $\gfmn$. Recall that $\avec{x}_i$ consists of $z_\q$ and the variables in $\vars(\q_i)\cap\avec{x}$.
By restricting $h$ to $\vars(\q_i)$, we obtain, for each $1 \leq i \leq n$, a homomorphism of
$\q_i(\avec{x}_i)$ into $\can$ that maps $z_\q$ to $h(z_\q)$ and $\vars(\q_i)\cap \avec{x}$ to $\avec{a}(\vars(\q_i)\cap\avec{x})$. Consider $\avec{a}^*$ defined by taking  $\avec{a}^*(x)= \avec{a}(x)$ for every $x \in  \vars(\q_i)\cap \avec{x}$
and $\avec{a}^*(z_\q)=h(z_\q)$. By the induction hypothesis, for every $1 \leq i \leq n$, we have
$\Pi^{\textsc{Tw}}_{\omq},\A \models \rew_{\q_i}(\avec{a}^*(\avec{x}_i))$. Next, since $h$ is a homomorphism, we must have $h(z_\q) \in A^{\can}$ whenever $A(z_\q) \in \q$
and $(h(z_\q), h(z_\q)) \in P^{\can}$ whenever $P(z_\q,z_\q) \in \q$. Since $\A$ is a complete data instance,  $A(h(z_\q)) \in \A$
for every $A(z_\q) \in \q$ and $P(h(z_\q),h(z_\q))$ for every $P(z_\q,z_\q) \in \q$. We have thus shown that, under the substitution $\avec{a}^*$,  every atom in the body of the clause
\begin{equation*}
\rew_\q(\avec{z}) \leftarrow \hspace*{-0.7em}\bigwedge_{A(z_\q)\in \q}\hspace*{-1em} A(z_\q) \ \ \land\hspace*{-0.7em} \bigwedge_{P(z_\q,z_\q) \in \q} \hspace*{-1.5em}P(z_\q,z_\q)  \ \ \land \bigwedge_{1\leq i\leq n}\hspace*{-0.5em} \rew_{\q_i}(\avec{x}_i),
\end{equation*}
is entailed from $\Pi^{\textsc{Tw}}_{\omq},\A$. It follows that we must also have $\Pi^{\textsc{Tw}}_{\omq},\A \models \rew_\q(\avec{a})$.

\medskip

\item Case 2: $h(z_\q) \notin \ind(\A)$. Then $h(z_\q)$ is of the form $b \varrho w$, for some $\varrho$.
Let $V$ be the smallest subset of $\vars(\q)$ that contains $z_\q$
and satisfies the following closure property:
\begin{itemize}
\item[--] if $z \in V$,
$h(z) \notin \ind(\A)$ and $\q$ contains an atom with $z$ and $z'$, then $z' \in V$.
\end{itemize}
Let $V'$ consist of all variables $z$ in $V$
such that $h(z) \notin \ind(\A)$. We observe that $h(z)$ begins by $b \varrho$ for every $z \in V'$ and $h(z)=b$
for every \mbox{$z \in V \setminus V'$}.
Define $\q_V$ as the CQ comprising all atoms in $\q$
whose variables are in $V$ and which contain at least one variable from $V'$; the answer variables of $\q_V$ are $V\setminus V'$.
By replacing the initial $b$ by $a$ in the mapping $h$, we obtain a homomorphism $h_V$ of $\q_V$ into $ \C_{\T,\{A_\varrho(a)\}}$ with $V\setminus V' = h_V^{-1}(a)$.
It follows that $\t=(\tr,\ti)$ with $\tr= V \setminus V'$ and $\ti =V'$ is a tree witness for $(\T,\q(\avec{x}))$ generated by $\varrho$ (and $\q_\t = \q_V$).
Moreover, $\tr\neq \emptyset$ because $\q$ has at least one answer variable.
This means that the program $\Pi^{\textsc{Tw}}_{\omq}$ contains the following clause
\begin{equation*}
\rew_\q(\avec{x}) \leftarrow A_\varrho(z_0)\
\land \hspace*{-0.5em}\bigwedge_{z\in \tr\setminus \{z_0\}}\hspace*{-0.5em} (z=z_0) \ \ \land \  \bigwedge_{1\leq i \leq k} \rew_{\q_i^\t}(\avec{x}_i^\t),
\end{equation*}
where $\q_1^\t, \dots, \q_k^\t$ are the connected components of $\q$ without $\q_\t$ and $z_0\in\tr$.
Recall that the query $\q_i^\t$ has answer variables $\avec{x}_i^\t = \vars(\q^\t_i)\cap (\avec{x} \cup \tr)$.
Let $\avec{a}^*$ be the substitution for $\avec{x} \cup \tr$ such that \mbox{$\avec{a}^*(x)=\avec{a}(x)$} for $x \in \avec{x}$
and $\avec{a}^*(z)=h(z)$ for $z \in \tr$.
Then, for every \mbox{$1 \leq i \leq k$}, there exists a homomorphism $h_i$ from $\q_i^\t$ to $\can$ such that
\mbox{$h_i(x) = \avec{a}^*(x)$} for every $x \in \avec{x}_i^\t$. By the induction
hypothesis, we obtain \mbox{$\Pi^{\textsc{Tw}}_{\omq}, \A \models  \rew_{\q_i^\t}(\avec{a}^*(\avec{x}_i^\t))$}.
Next, since $h(z)=b$ for every $z \in \tr$, we have $\avec{a}^*(z)=\avec{a}^*(z')$ for every
$z,z' \in \tr$. Moreover, the presence of the element $b \varrho$ in $\can$ means that
$\T, \A \models A_\varrho(b)$. Since $\A$ is a complete data instance, we have \mbox{$A_\varrho(b) \in \A$.}
It follows that under the substitution $\avec{a}^*$, all atoms in the body of the clause under consideration
are entailed by $\Pi^{\textsc{Tw}}_{\omq}, \A $. Therefore, we must also have \mbox{$\Pi^{\textsc{Tw}}_{\omq}, \A  \models \rew_\q(\avec{a})$}.
\end{itemize}

\medskip

We have thus shown the lemma for all queries $\sqset$ other than $\q_0(\avec{x}_0)$.
Let us now turn to $\q_0(\avec{x}_0)$.

\smallskip

For the first direction, $(\Rightarrow)$, suppose $\Pi^{\textsc{Tw}}_{\omq},\A \models \rew_{\q_0}(\avec{a})$.
There are four cases, depending on which type of clause was used to derive $\rew_{\q_0}(\avec{a})$.
We skip the first three cases, which are identical to those considered in the base case and induction step,
and focus instead on the case in which $\rew_{\q_0}(\avec{a})$ was derived using a clause of the form
$\rew_{\q_0} \leftarrow A(x)$ with $A$ a unary predicate such that $\T, \{A(a)\} \models \q_0$.
In this case, there must exist some $b \in \ind(\A)$
such that $\T, \A \models A(b)$. By completeness of~$\A$, we obtain $A(b) \in \A$.
Since $\T, \{A(a)\} \models \q_0$, we get $\T, \A \models \q_0$, which implies the
existence of a homomorphism from $\q_0$ into $\can$.

\smallskip

For the converse direction, $(\Leftarrow)$, suppose that there is a homomorphism $h\colon \q_0 \to \can$ such that $h(\avec{x}_0)=\avec{a}$.
We focus on the case in which $\q_0$ is Boolean  ($\avec{x}_0 = \emptyset$) and none of the variables in $\q_0$ is mapped to an  individual constant (the other cases
can be handled exactly as in the induction basis and induction step). In this case, there must
exist an individual constant $b$ and some~$\varrho$ such that $h(z)$ begins by $b \varrho$ for every $z \in \vars(\q_0)$. It follows that $\T, \{A_\varrho(a)\} \models \q_0$, since the mapping $h'$
defined by setting $h'(z)=a \varrho w$ whenever $h(z) = b \varrho w$ is a homomorphism from $\q_0$ to $\C_{\T, \{A_\varrho(a)\}}$.
It follows that $\Pi^{\textsc{Tw}}_\omq$ contains the clause $\rew_{\q_0} \gets A_\varrho(x)$. Since $b \varrho$ occurs in $\Delta^{\can}$, we have $\T, \A \models A_\varrho(b)$.
By completeness of~$\A$, $A_\varrho(b) \in \A$, and so by applying the clause $\rew_{\q_0} \gets A_\varrho(x)$, we obtain $\Pi^{\textsc{Tw}}_\omq,\A \models \rew_{\q_0}$.
\end{proof}

\subsection{Rewritings Zoo}\label{zoo}

In this section, we put together the rewritings from Sections~\ref{sec:boundedtw}--\ref{sec:boundedleaf} for the OMQ given in Examples~\ref{ex:rewriting:1} and~\ref{ex:rewriting:2}.

Consider the CQ $\q(x_0, x_7)$ depicted below (black nodes represent answer variables)\\[10pt]
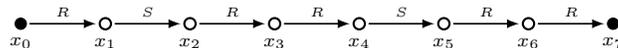
\centerline{%
\begin{tikzpicture}[>=latex,xscale=0.75]\scriptsize
\node[bpoint,label=below:{$x_0$}] (v0) at (0,0) {};
\node[wpoint,label=below:{$x_1$}] (v1) at (1.5,0) {};
\node[wpoint,label=below:{$x_2$}] (v2) at (3,0) {};
\node[wpoint,label=below:{$x_3$}] (v3) at (4.5,0) {};
\node[wpoint,label=below:{$x_4$}] (v4) at (6,0) {};
\node[wpoint,label=below:{$x_5$}] (v5) at (7.5,0) {};
\node[wpoint,label=below:{$x_6$}] (v6) at (9,0) {};
\node[bpoint,label=below:{$x_7$}] (v7) at (10.5,0) {};
\begin{scope}[semithick,shorten >= 1pt, shorten <= 1pt]\tiny
\draw[->] (v0) to node[above] {$R$} (v1);
\draw[->] (v1)to node[above] {$S$}  (v2);
\draw[->] (v2) to node[above] {$R$}  (v3);
\draw[->] (v3) to node[above] {$R$}  (v4);
\draw[->] (v4) to node[above] {$S$}  (v5);
\draw[->] (v5) to node[above] {$R$}  (v6);
\draw[->] (v6) to node[above] {$R$}  (v7);
\end{scope}
\end{tikzpicture}}\\[10pt]
and the following ontology $\T$ in normal form:
\begin{align*}
P(x,y) &\to  S(x,y), \quad & P(x,y) &\to  R(y,x),\\[4pt]
A_P(x) &\leftrightarrow \exists y \, P(x,y), \quad & A_{P^-}(x) & \leftrightarrow \exists y \, P(y,x), \\
A_R(x) &\leftrightarrow \exists y \, R(x,y),  \quad & A_{R^-}(x) &\leftrightarrow \exists y \, R(y,x),\\
A_S(x) & \leftrightarrow \exists y \, S(x,y) \quad & A_{S^-}(x) & \leftrightarrow \exists y \, S(y,x). 
\end{align*}

\subsubsection{UCQ rewriting}

The 9 CQs below form a UCQ rewriting of the OMQ $\omq(x_0,x_7) = (\T,\q(x_0,x_7))$ over complete data instances given as an NDL program with goal predicate $G$:
\begin{align*}
 G(x_0, x_7) \leftarrow & [R(x_0, x_1) \land S(x_1, x_2)\land R(x_2, x_3)] \land {} \\
 & [R(x_3, x_4) \land S(x_4, x_5) \land R(x_5, x_6)] \land R(x_6, x_7),\\
 G(x_0, x_7) \leftarrow & [A_{P^-}(x_0) \land R(x_0, x_3)] \land {} \\
  & [R(x_3, x_4) \land S(x_4, x_5) \land R(x_5, x_6)] \land R(x_6, x_7),\\
 G(x_0, x_7) \leftarrow & [R(x_0, x_3) \land  A_P(x_3)] \land {} \\
 & [R(x_3, x_4) \land S(x_4, x_5) \land R(x_5, x_6)] \land R(x_6, x_7),\\
 G(x_0, x_7) \leftarrow & [R(x_0, x_1) \land S(x_1, x_2) \land R(x_2, x_3)] \land {} \\
 &  [A_{P^-}(x_3) \land R(x_3, x_6)] \land R(x_6, x_7),\\
 G(x_0, x_7) \leftarrow & [R(x_0, x_1) \land S(x_1, x_2) \land R(x_2, x_3)] \land {} \\
  & [R(x_3, x_6) \land A_P(x_6)] \land R(x_6, x_7),\\
 G(x_0, x_7) \leftarrow & [A_{P^-}(x_0) \land R(x_0, x_3)] \land{}\\
 &  [A_{P^-}(x_3) \land R(x_3, x_6)] \land R(x_6, x_7),\\
 G(x_0, x_7) \leftarrow & [A_{P^-}(x_0) \land R(x_0, x_3)] \land{}\\
 & [R(x_3, x_6) \land A_P(x_6)] \land R(x_6, x_7),\\
 G(x_0, x_7) \leftarrow & [R(x_0, x_3) \land A_P(x_3)] \land{} \\
 & [A_{P^-}(x_3) \land R(x_3, x_6)] \land R(x_6, x_7),\\
 G(x_0, x_7) \leftarrow & [R(x_0, x_3) \land A_P(x_3)] \land{} \\
 & [R(x_3, x_6) \land A_P(x_6)] \land R(x_6, x_7).
\end{align*}
We note that a UCQ rewriting over all data instances would in addition
contain variants of the CQs above with each of the predicates $R$ and $S$ replaced by $P$ (with arguments swapped appropriately).

The UCQ rewriting above can be obtained by transforming the following PE-formula into UCQ form:
\begin{align*}
& \bigl[\bigl(R(x_0, x_1) \land S(x_1, x_2) \land R(x_2, x_3)\bigr)\\  & \hspace*{1em}\lor \bigl(A_{P^-}(x_0) \land R(x_0, x_3)\bigr) \lor \bigl(R(x_0, x_3) \land  A_P(x_3)\bigr)\bigr] \\
\land \ \ \ & \bigl[\bigl(R(x_3, x_4) \land S(x_4, x_5) \land R(x_5, x_6)\bigr)\\  & \hspace*{1em}\lor \bigl(A_{P^-}(x_3) \land R(x_5, x_6)\bigr) \lor \bigl(R(x_3, x_6) \land  A_P(x_6)\bigr)\bigr]\\
\land \ \ \ & R(x_6,x_7).
\end{align*}
(Intuitively, each of the two sequences $RSR$ in the query can be derived in three possible ways: from $RSR$, from $A_{P^-}R$ and from $RA_P$).

\subsubsection{\textsc{Log}-rewriting}

As explained in Example~\ref{ex:rewriting:2}, we split $T$ into $D_1$ and $D_2$ and obtain two rules:
\begin{align*}
G_T^{\boldsymbol{\varepsilon}}(x_0,x_7) & \leftarrow G_{D_1}^{x_3 \mapsto \varepsilon}(x_3,x_0) \land R(x_3,x_4) \land G_{D_2}^{x_4\mapsto \varepsilon}(x_4,x_7),\\
G_T^{\boldsymbol{\varepsilon}}(x_0,x_7) & \leftarrow G_{D_1}^{x_3 \mapsto \varepsilon}(x_3,x_0) \land A_{P^-}(x_4) \land (x_3 = x_4) \land   G_{D_2}^{x_4\mapsto P^-}\!\!(x_4,x_7).
\end{align*}
Next, we split each of $D_1$ and $D_2$ into single-atom subqueries, which yields the following rules:
\begin{align*}
G_{D_1}^{x_3 \mapsto \varepsilon}(x_3,x_0) &\leftarrow (x_0= x_1) \land A_{P^-}(x_1) \land (x_1 = x_2) \land R(x_2, x_3),\\
G_{D_1}^{x_3 \mapsto \varepsilon}(x_3,x_0) &\leftarrow R(x_0,x_1) \land (x_1 = x_2) \land A_P(x_2) \land (x_2 = x_3),\\
G_{D_1}^{x_3 \mapsto \varepsilon}(x_3,x_0) &\leftarrow R(x_0, x_1) \land S(x_1, x_2) \land R(x_2, x_3),\\[10pt]
G_{D_2}^{x_4\mapsto \varepsilon}(x_4,x_7) &\leftarrow (x_4 = x_5) \land A_P(x_5) \land (x_5 = x_6) \land R(x_6, x_7),\\
G_{D_2}^{x_4\mapsto \varepsilon}(x_4,x_7) &\leftarrow S(x_4, x_5) \land R(x_5, x_6) \land R(x_6, x_7),\\[10pt]
G_{D_2}^{x_4\mapsto P^-}(x_4,x_7) & \leftarrow A_{P^-}(x_4)\land (x_4 = x_5) \land R(x_5,x_6) \land R(x_6, x_7).
\end{align*}
Note that in each case we consider only those types that give rise to predicates that have definitions in the rewriting.
The resulting NDL rewriting with goal $G_T^{\boldsymbol{\varepsilon}}$ consists of 8~rules. Note, however, that the rewriting illustrated above is a slight simplification of the definition given in Section~\ref{sec:boundedtw}: here, for the leaves of the tree decomposition, we directly use the atoms $\mathsf{At}^{\tpr}$ instead of including a rule $\rpred^{\tpd}_D(\dD, \avec{x}_D)  \leftarrow  \mathsf{At}^{\tpr}$ in the rewriting. This simplification clearly does not affect the width of the NDL query or the choice of weight function.

\subsubsection{\textsc{Lin}-rewriting}

We assume that $x_0$ is the root, which makes $x_7$ the only leaf of the query. (Note that
we could have chosen another variable, say $x_3$, as the root, with $x_0$ and $x_7$ the two leaves.) So, the top-level rule is
\begin{align*}
G(x_0,x_7) &\leftarrow G^{x_0 \mapsto \varepsilon}_0(x_0,x_7).
\end{align*}
We then move along the query and consider the variables $x_1$, $x_2$ and $x_3$. The possible ways of mapping these variables to the canonical model give rise to the following 7 rules:
\begin{align*}
G^{x_0 \mapsto \varepsilon}_0(x_0,x_7) &\leftarrow R(x_0,x_1) \land P^{x_1\mapsto \varepsilon}_1(x_1,x_7),\\
G^{x_0 \mapsto \varepsilon}_0(x_0,x_7) &\leftarrow (x_0 = x_1) \land A_{P^-}(x_1) \land  G^{x_1 \mapsto P^-}_1\!\!(x_1,x_7),\\[8pt]
G^{x_1 \mapsto \varepsilon}_1(x_1,x_7) &\leftarrow S(x_1,x_2) \land G^{x_2 \mapsto \varepsilon}_2(x_2,x_7),\\
G^{x_1 \mapsto \varepsilon}_1(x_1,x_7) &\leftarrow (x_1 = x_2) \land A_P(x_2) \land G^{x_2\mapsto P}_2(x_2,x_7),\\
G^{x_1 \mapsto P^-}_1(x_1,x_7) & \leftarrow A_{P^-}(x_1)  \land (x_1 = x_2) \land G^{x_2 \mapsto \varepsilon}_2(x_2,x_7),\\[8pt]
G^{x_2\mapsto \varepsilon}_2(x_2,x_7) &\leftarrow R(x_2,x_3) \land G^{x_3\mapsto \varepsilon}_3(x_3,x_7),\\
G^{x_2\mapsto P}_2(x_2,x_7) & \leftarrow  A_P(x_2) \land (x_2 = x_3) \land G^{x_3 \mapsto \varepsilon}_3(x_3,x_7).
\end{align*}
Next, we move to the variables $x_4$, $x_5$ and $x_6$, which give similar 7 rules:
\begin{align*}
G^{x_3 \mapsto \varepsilon}_3(x_3,x_7) &\leftarrow R(x_3,x_4) \land P^{x_4\mapsto \varepsilon}_4(x_4,x_7),\\
G^{x_3 \mapsto \varepsilon}_3(x_3,x_7) &\leftarrow (x_3 = x_4) \land A_{P^-}(x_4) \land  G^{x_4 \mapsto P^-}_4\!\!(x_4,x_7),\\[8pt]
G^{x_4 \mapsto \varepsilon}_4(x_4,x_7) &\leftarrow S(x_4,x_5) \land G^{x_5 \mapsto \varepsilon}_5(x_5,x_7),\\
G^{x_4 \mapsto \varepsilon}_4(x_4,x_7) &\leftarrow (x_4 = x_5) \land A_P(x_5) \land G^{x_5\mapsto P}_5(x_5,x_7),\\
G^{x_4 \mapsto P^-}_4(x_4,x_7) & \leftarrow A_{P^-}(x_4)  \land (x_4 = x_5) \land G^{x_5 \mapsto \varepsilon}_5(x_5,x_7),\\[8pt]
G^{x_5\mapsto \varepsilon}_5(x_5,x_7) &\leftarrow R(x_5,x_6) \land G^{x_6\mapsto \varepsilon}_6(x_6,x_7),\\
G^{x_5\mapsto P}_5(x_5,x_7) & \leftarrow  A_P(x_2) \land (x_5 = x_6) \land G^{x_6 \mapsto \varepsilon}_6(x_6,x_7).
\end{align*}
Finally, the last variable can only be mapped to a constant in the data instance, which yields a single rule:
\begin{align*}
G^{x_6 \mapsto \varepsilon}_6(x_6,x_7) &\leftarrow R(x_6,x_7).
\end{align*}
Note that, like in the previous case, we consider only those types that give rise to predicates with definitions (and ignore the dead-ends in the construction).

\subsubsection{\textsc{Tw}-rewriting}

We begin by splitting the query roughly in the middle, that is, we choose $x_3$ and consider two subqueries:
\begin{align*}
\q_{03}(x_0,x_3) & = \exists x_1x_2\,\bigl(R(x_0,x_1) \land S(x_1,x_2) \land R(x_2,x_3)\bigr)\\
&  \text{ and } \\
\q_{37}(x_3,x_7) & = \exists x_4x_5x_6\,\bigl(R(x_3,x_4) \land S(x_4,x_5) \land{}\\
& \hspace*{10em} R(x_5,x_6) \land R(x_6,x_7)\bigr).
\end{align*}
Since there is no tree witness $\t$ for $(\T,\q(x_0,x_7))$ that contains $x_3$ in $\ti$, we have only one top-level rule:
\begin{align*}
G_{07}(x,y) &\leftarrow G_{03} (x_0, x_3)\land G_{37} (x_3, x_7).
\end{align*}
Next, we focus on $\q_{03}$ and choose $x_1$ as the splitting variable. In this case, there is a tree witness $\t^1$ with $\ti^1 = \{x_1\}$ and $\tr^1 = \{x_0,x_2\}$, and so we obtain two rules for $G_{03}$:
\begin{align*}
G_{03} (x_0,x_3) &\leftarrow R(x_0, x_1) \land G_{13} (x_1, x_3),\\
G_{03} (x_0,x_3) &\leftarrow A_{P^-} (x_0) \land (x_0 = x_2) \land R(x_2, x_3).
\end{align*}
The subquery $\q_{13}(x_1,x_3) = \exists x_2 \bigl(S(x_1,x_2)\land R(x_2,x_3)\bigr)$ contains two atoms and is split at $x_2$. Since there is a tree witness $\t^2$ for $(\T,\q_{13}(x_1,x_3))$ with $\ti^2 = \{x_2\}$ and $\tr^2 = \{x_1,x_3\}$, we obtain two rules:
\begin{align*}
G_{13} (x_1,x_3) &\leftarrow S (x_1, x_2) \land R (x_2, x_3),\\
G_{13} (x_1,x_3) &\leftarrow A_P(x_1) \land (x_1 = x_3).
\end{align*}
By applying the same procedure to $\q_{37}(x_3,x_7)$, we get the following five rules:
\begin{align*}
G_{37} (x_3,x_7) &\leftarrow G_{35} (x_3, x_5) \land G_{57} (x_5, x_7),\\
G_{37} (x_5,x_7) &\leftarrow R(x_3, x_4) \land A_P(x_4) \land (x_4 = x_6) \land R (x_6, x_7),\\
G_{35} (x_3,x_5) &\leftarrow R (x_3, x_5) \land  S (x_5, x_7),\\
G_{35} (x_3,x_5) &\leftarrow A_{P^-} (x_3) \land (x_3 = x_5),\\
G_{57} (x_3,x_5) &\leftarrow R (x_3, x_4) \land R (x_4, x_7).
\end{align*}
Note that the rewriting illustrated above is slightly simpler than the definition in Section~\ref{sec:boundedleaf}: here, we directly use the atoms of $\q(\avec{x})$ instead of including a rule $G_\q(\avec{x}) \leftarrow \q(\avec{x})$, for each $\q(\avec{x})$ without existentially quantified variables. This simplification clearly does not affect the width of the NDL query and the choice of weight function.

\section{Proofs for Section~\ref{sec:param}}

\subsection{Theorem~\ref{thm:w2-hard}}\label{AppB.1}

\indent\textsc{Theorem \ref{thm:w2-hard}.} {\it
\pr{} is $W[2]$-hard.}

\begin{proof}

We show that $\T^k_H, \{V_0^0(a)\} \models \q^k_H$ iff $H$ has a hitting set of size $k$. Denote by $\mathcal C$ the canonical model of $(\T^k_H, \{V_0^0(a)\})$. For convenience of reference to the points of the canonical model
we assume that $\T^k_H$ contains the following axioms:
\begin{align*}
V_i^{l-1}(x)  &\to \exists z\, \upsilon_{i'}^l(x,z) \text{ and }\\
\upsilon_{i'}^l(x,z) &\to P(z,x) \land V_{i'}^l(z), &&\text{for } 0 \leq i < i' \le n,\\
V_i^l(x) &\to E^l_j(z), &&\text{for } v_i \in e_j,\ e_j\in E,\\
E^l_j(x)  &\to \exists z \, \eta^l_j(x,z) \text{ and }\\
\eta^l_j(x,z) &\to P(x,z) \land E^{l-1}_j(z),&&  \text{for } 1 \le j \le m.
\end{align*}
We show that $\mathcal{C} \models \q^k_H$ iff  $H$ has a hitting set of size~$k$.

\smallskip

$(\Rightarrow)$
Suppose $h \colon \q_H^k \to \mathcal{C}$ is a homomorphism.
Note that $\mathcal{C}$ satisfies the following properties: (i) $w \in E^0_j$ iff
$w = a \upsilon^1_{i_1} \upsilon^2_{i_2} \dots \upsilon^s_{i_s}\eta^s_{j}\eta^{s-1}_{j}\dots\eta^{1}_{j}$ where
$v_{j_s} \in e_j$ and (ii) all points in $\Delta^\mathcal{C}$ have at most one $P$-predecessor.
By starting with some $E_j^0$ atom and applying first (i) and then iterating (ii), we conclude
that $h(y) = a \upsilon^1_{i_1} \dots \upsilon^k_{i_k}$ for some
$1 \leq i_1 < i_2 < \dots i_k \le n$. We claim that $\{v_{i_1}, v_{i_2}, \dots, v_{i_k}\}$ is  a hitting set in $H$.  Indeed, for every branch $j$ of $\q_H^k$, there is $1 \leq s \leq k$ such that this branch is mapped on $\mathcal{C}$ in the following way:
\begin{align*}
  h(z_j^l) &= a \upsilon^1_{i_1} \upsilon^2_{i_2} \dots \upsilon^l_{i_l},&& s \leq l \leq k-1,\\
  h(z_j^l) &= a \upsilon^1_{i_1} \upsilon^2_{i_2} \dots \upsilon^s_{i_s}\eta^s_{j}\eta^{s-1}_{j}\dots\eta^{l+1}_{j},&& 0 \leq l < s,
\end{align*}
with $v_{i_s} \in e_j$. This can be shown by induction on $l$ from $0$ to $k-1$ using
(i) to prove the base of induction and (ii) to prove the induction step.
Therefore, for every $j$, there exists $s$ such that $v_{i_{s}} \in e_j$.

\bigskip

$(\Leftarrow)$ Suppose $\{v_{i_1}, v_{i_2}, \dots, v_{i_k}\}$ is
a hitting set in $H$. We construct a homomorphism $h$
from $\q_H^k$ to $\mathcal{C}$. First, we set $h(y) = a \upsilon^1_{i_1} \dots \upsilon^k_{i_k}$.
Then, for each $1 \le j\le m$, we find $s$ such $v_{i_s} \in e_j$ and define $h$ as follows:
\begin{align*}
  h(z_j^l) &= a \upsilon^1_{i_1} \upsilon^2_{i_2} \dots \upsilon^l_{i_l},&& s \leq l \leq k-1,\\
  h(z_j^l) &= a \upsilon^1_{i_1} \upsilon^2_{i_2} \dots \upsilon^s_{i_s}\eta^s_{j}\eta^{s-1}_{j}\dots\eta^{l+1}_{j},&& 0 \leq l < s.
\end{align*}
It should be clear that $h$ is indeed a homomorphism.
\end{proof}

\subsection{Theorem~\ref{leaves-param-w1}}

\indent\textsc{Theorem \ref{leaves-param-w1}.} {\it
\blpr\ is $W[1]$-hard.}
\begin{proof}
We prove that $\T_G, \{A(a)\} \models \q_G$ iff $G$ has a clique
containing one vertex from each set $V_i$.

We start with some preliminaries. First note we assume that the final axiom in $\T_G$ (which uses the syntactic sugar $\land$)
is actually given by the following three axioms (where $P$ is a fresh binary predicate):
\begin{align*}
B(x) &\rightarrow \exists y \, P(x,y),\\
P(x,y) &\rightarrow U(x,y),\\
P(x,y) &\rightarrow U(y,x).
\end{align*}
To simplify notation, we will abbreviate $\C_{\T_G, \{A(a)\}}$ by $\C$, and for every $1 \leq j \leq M$,
we let $\w(v_j)=L_j^1 L_j^2 \ldots L_j^{2M}$.
Observe that for every $v_{j_1} \in V_1, v_{j_2} \in V_2, \ldots, v_{j_p} \in V_p$, the element
$a \w(v_{j_1}) \w(v_{j_2}) \ldots \w(v_{j_p})$ belongs to $\Delta^{\C}$.
Further, observe that if $a w \in \Delta^{\C}$ with $|w|=2M \cdot p$,
then there exist $v_{j_1} \in V_1, v_{j_2} \in V_2, \ldots, v_{j_p} \in V_p$
such that $w=\w(v_{j_1}) \w(v_{j_2}) \ldots \w(v_{j_p})$.

\smallskip

$(\Rightarrow)$ Suppose that $\T_G, \{A(a)\} \models \q_G$,
and let $h$ be a homomorphism of $\q_G$ into $\C$.
Note that because of the atom $B(y)$, the variable
$y$ must be sent by $h$ to an element occurring at the end of the $p$th block. As noted above,
every such element takes the form $$a \w(v_{j_1}) \w(v_{j_2}) \ldots \w(v_{j_p})$$
where $v_{j_1} \in V_1, v_{j_2} \in V_2, \ldots, v_{j_p} \in V_p$.
We claim that $\{v_{j_1}, \dots, v_{j_p}\}$ is a clique in
$G$. To see why, consider the $i$th branch of $\q_G$, compactly represented as follows:
$$\bigl(U^{\smash{2M-2}} \cdot (YY \cdot U^{\smash{2M-2}})^{i} \cdot S S \bigr) (y, z_i)$$
By examining the axioms, we see that starting from the first occurrence of $YY$, every $U$ and $Y$ atom
takes us one step closer to $a$ (prior to the first $YY$, we may go back and forth on the extra $P$-edge leaving
from $h(y)$). It follows that $SS$ must be mapped within the $p\text{-}i$th block of the selected branch,
and since $S$ is present only at positions $2j_{p-i}$ and $2j_{p-i}+1$ of the block,
we must have $h(z_i)= a \w(v_{j_1}) \ldots  \w(v_{j_{p-i-1}}) L_{j_{p-i}}^1 \ldots L_{j_{p-i}}^{2 j_{p-i}-1}$.
As the distance
between consecutive occurrences of $YY$
(and between the final $YY$ and the $SS$) is $2M-2$, it follows that all $YY$ blocks occur at positions
$2j_{p-i}$ and $2j_{p-i}+1$ of blocks $p-i+1, \ldots, p$, which implies that
$v_{j_{p-i+1}}, \ldots, v_{j_p}$ are neighbours of $v_{j_i}$ in $G$.
Since $\q_G$ contains branches for every $1 \leq i < p$,
the selected vertices $v_{j_{1}}, \ldots, v_{j_p}$ are all neighbours in $G$,
and $G$ contains a clique with the required properties.

\smallskip

$(\Leftarrow)$ Suppose that $v_{j_1} \in V_1, \dots, v_{j_p} \in V_p$ form a clique.
We construct a homomorphism $h$ of $\q_G$ into~$\C$.
First, set $h(y)=aw$ where $w=\w(v_{j_1} )\w(v_{j_2} )\dots\w(v_{j_p} )$ and observe
that the atom $B(y)$ is satisfied by this assignment.
We will use $w[\ell, \ell']$ to
denote the subword of $w$ beginning with the $\ell$th symbol of $w$ and ending with the $\ell'$th symbol
(note that $w=|2M\cdot p|$, so $w=w[1,2M\cdot p]$).
Next, consider the $i$th branch of the query, which connects $y$ to~$z_i$, and
let $y_0, y_1, \dots y_{2M(i+1)}$ be the variables lying between $y$ and $z_i$ with
$y_0 = y$ and $z_i =  y_{2M(i+1)}$.
For $0 \leq k \leq 2 j_{p-i}$, we set $h(y_k)=h(y)$ if $k$ is even,
and set $h(y_k)=h(y) P$ otherwise. Observe that because $P$ is included in both $U$ and $U^-$,
we satisfy all binary atoms between variables from $\{y_0, \ldots, y_{2 j_{p-i}}\}$.
For $2 j_{p-i} < k \leq 2M(i+1)$, we set $$h(y_k)=a w[1,2M \cdot p - (k-2 j_{p-i})].$$
Note that, in particular, this yields
\begin{align*}
h(y_{2M(i+1)-2}) & = aw[1, 2M(p-i-1) + 2j_{p-i}+2],\\
h(y_{2M(i+1)-1}) & =aw[1,2M(p-i-1) + 2j_{p-i}+1],\\
h(y_{2M(i+1)}) & =aw[1,2M(p-i-1) + 2j_{p-i}],
\end{align*}
so the final two $S$-atoms in the branch are satisfied by~$h$.
It is easy to see that all $U$-atoms between variables from $y_{2 j_{p-i}}, \ldots, y_{2M(i+1)}$ are also satisfied.
Finally, using the fact that vertices $v_{j_{p-i+1}}, \ldots , v_{j_p}$ are neighbours of $v_{j_{p-i}}$,
we can show that all of the $Y$-atoms in the $i$th branch are satisfied by $h$. As we have constructed
a homomorphism from $\q_G$ into $\C$, we can conclude $\T_G, \{A(a)\} \models \q_G$.
\end{proof}

\section{Proofs for Section~\ref{sec2:fixed}}

\begin{figure*}[t]\centering
\begin{tikzpicture}[>=latex,
man/.style={draw,thin,fill=white,rectangle,rounded corners=1mm,inner sep=0pt,minimum height=2.8mm,minimum width=2.8mm,fill opacity=1},
maw/.style={draw,thin,fill=white,rectangle,rounded corners=1mm,inner sep=0pt,minimum height=4.8mm,minimum width=2.8mm,fill opacity=1},
spoint/.style={draw,semithick,fill=gray,rectangle,rounded corners=0.7mm,inner sep=0pt,minimum height=2mm,minimum width=2mm},
xscale=0.9
]\small
\newcommand{\dl}[2]{\tabcolsep=0pt\tiny\bf\begin{tabular}{c}#2\\[-1pt]#1\end{tabular}}
%
\draw[ultra thin,gray] (-2.5,-1.5) -- +(7.9,0);
\draw[ultra thin,gray] (-2.5,-2.5) -- +(7.9,0);
\draw[ultra thin,gray] (-2.5,-.5) -- +(7.9,0);
\draw[ultra thin,gray] (-2.5,.5) -- +(7.9,0);

\node at (-3,.5) {\normalsize $p_4$};
\node at (-3,-.5) {\normalsize $p_3$};
\node at (-3,-1.5) {\normalsize $p_2$};
\node at (-3,-2.5) {\normalsize $p_1$};

\node[wpoint,label=left:{$y$}] (y) at (0,1) {};
\node[wpoint,label=right:{$\!z_1^3$}] (z13) at (-1.5,0) {};
\node[wpoint,label=right:{$\!z_2^3$}] (z23) at (-0.5,0) {};
\node[wpoint,label=right:{$\!z_1^2$}] (z12) at (-1.5,-1) {};
\node[wpoint,label=right:{$\!z_2^2$}] (z22) at (-0.5,-1) {};
\node[wpoint,label=right:{$\!z_3^3$}] (z33) at (.5,0) {};
\node[wpoint,label=right:{$\!z_4^3$}] (z43) at (1.5,0) {};
\node[wpoint,label=right:{$\!z_3^2$}] (z32) at (.5,-1) {};
\node[wpoint,label=right:{$\!z_4^2$}] (z42) at (1.5,-1) {};
\node[wpoint,label=right:{$\!z_1^1$}] (z11) at (-1.5,-2) {};
\node[wpoint,label=right:{$\!z_2^1$}] (z21) at (-0.5,-2) {};
\node[wpoint,label=right:{$\!z_1^0$}] (z10) at (-1.5,-3) {};
\node[wpoint,label=right:{$\!z_2^0$}] (z20) at (-0.5,-3) {};
\node[wpoint,label=right:{$\!z_3^1$}] (z31) at (.5,-2) {};
\node[wpoint,label=right:{$\!z_4^1$}] (z41) at (1.5,-2) {};
\node[wpoint,label=right:{$\!z_3^0$}] (z30) at (.5,-3) {};
\node[wpoint,label=right:{$\!z_4^0$}] (z40) at (1.5,-3) {};
\node[wpoint,label=right:{$\!z_1^{-1}$}] (z1-1) at (-1.5,-4) {};
\node[wpoint,label=right:{$\!z_2^{-1}$}] (z2-1) at (-.5,-4) {};
\node[wpoint,label=right:{$\!z_3^{-1}$}] (z3-1) at (.5,-4) {};
\node[wpoint,label=right:{$\!z_4^{-1}$}] (z4-1) at (1.5,-4) {};
\node[spoint,label=right:{$\!z_1^{-2}$}] (z1-2) at (-1.5,-5) {};
\node[spoint,label=right:{$\!z_2^{-2}$}] (z2-2) at (-.5,-5) {};
\node[spoint,label=right:{$\!z_3^{-2}$}] (z3-2) at (.5,-5) {};
\node[spoint,label=right:{$\!z_4^{-2}$}] (z4-2) at (1.5,-5) {};

\node[wpoint,label=right:{$\!y^3$}] (y3) at (3,0) {};
\node[wpoint,label=right:{$\!y^2$}] (y2) at (3,-1) {};
\node[wpoint,label=right:{$\!y^1$}] (y1) at (3,-2) {};
\node[wpoint,label=right:{$\!x$}] (x) at (3,-3) {};

\begin{scope}[thick]
\draw[->] (y) to node[midway,man] {\tiny\bf +} (z13);
\draw[->] (y) to node[midway,man] {\tiny\bf +} (z23);
\draw[->] (y) to node[midway,man] {\tiny\bf 0} (z33);
\draw[->] (y) to node[midway,man] {\tiny$\boldsymbol{-}$} (z43);
\draw[->] (z13) to node[midway,man] {\tiny$\boldsymbol{-}$} (z12);
\draw[->] (z23) to node[midway,man] {\tiny$\boldsymbol{-}$} (z22);
\draw[->] (z33) to node[midway,man] {\tiny\bf 0} (z32);
\draw[->] (z43) to node[midway,man] {\tiny$\boldsymbol{-}$} (z42);
\draw[->] (z12) to node[midway,man] {\tiny\bf 0} (z11);
\draw[->] (z22) to node[midway,man] {\tiny\bf 0} (z21);
\draw[->] (z32) to node[midway,man] {\tiny\bf 0} (z31);
\draw[->] (z42) to node[midway,man] {\tiny\bf 0} (z41);
\draw[->] (z11) to node[midway,man] {\tiny\bf +} (z10);
\draw[->] (z21) to node[midway,man] {\tiny\bf 0} (z20);
\draw[->] (z31) to node[midway,man] {\tiny\bf +} (z30);
\draw[->] (z41) to node[midway,man] {\tiny\bf 0} (z40);
\draw[->] (y) to node[midway,man] {\tiny\bf 0} (y3);
\draw[->] (y3) to node[midway,man] {\tiny\bf 0} (y2);
\draw[->] (y2) to node[midway,man] {\tiny\bf 0} (y1);
\draw[->] (y1) to node[midway,man] {\tiny\bf 0} (x);
\draw[->] (z10) to node[midway,man] {\tiny$\boldsymbol{-}$} (z1-1);
\draw[->] (z1-1) to node[midway,man] {\tiny$\boldsymbol{-}$} (z1-2);
\draw[->] (z20) to node[midway,man] {\tiny$\boldsymbol{-}$} (z2-1);
\draw[->] (z2-1) to node[midway,man] {\tiny\bf +} (z2-2);
\draw[->] (z30) to node[midway,man] {\tiny\bf +} (z3-1);
\draw[->] (z3-1) to node[midway,man] {\tiny$\boldsymbol{-}$} (z3-2);
\draw[->] (z40) to node[midway,man] {\tiny\bf +} (z4-1);
\draw[->] (z4-1) to node[midway,man] {\tiny\bf +} (z4-2);

\end{scope}
\begin{scope}\small
\node[rotate=0] at (-1.5,1.6) {$\chi_1$};
\node[rotate=0] at (-.5,1.6) {$\chi_2$};
\node[rotate=0] at (.5,1.6) {$\chi_3$};
\node[rotate=0] at (1.5,1.6) {$\chi_4$};
\end{scope}
\begin{scope}[xshift=5cm]
\node[bpoint,label=below:{$a$}] (a) at (4,-3) {};
\node[wpoint] (an) at (2.3,-2) {};
\node[wpoint,opacity=0.6] (ap) at (5.7,-2) {};
\node[wpoint,opacity=0.6] (ann) at (1.2,-1) {};
\node[wpoint] (anp) at (3.4,-1) {};
\node[wpoint,opacity=0.6] (apn) at (4.6,-1) {};
\node[wpoint] (app) at (7,-1) {};

%
%
\begin{scope}[ultra thick]
\draw[->] (an) to node[pos=0.3,maw] {\dl{0}{+}}(a);
\draw[->] (ap) to node[pos=0.3,maw] {\dl{0}{$\boldsymbol{-}$}}(a);
\draw[->] (ann) to node[midway,maw] {\dl{0}{+}}(an);
\draw[->] (anp) to node[midway,maw] {\dl{0}{$\boldsymbol{-}$}}(an);
\draw[->] (apn) to node[midway,maw] {\dl{0}{+}}(ap);
\draw[->] (app) to node[midway,maw] {\dl{0}{$\boldsymbol{-}$}}(ap);

\end{scope}
\node[spoint,opacity=0.6] (zann) at ($(ann)+(0,-1)$) {};
\node[spoint,opacity=0.6] (zann1) at ($(zann)+(0,-1)$) {};
\node[spoint,opacity=0.6] (zanp) at ($(anp)+(0,-1)$) {};
\node[spoint,opacity=0.6] (zanp1) at ($(zanp)+(0,-1)$) {};
\node[spoint,opacity=0.6] (zan) at ($(an)+(0,-1)$) {};
\node[spoint,opacity=0.6] (zapn) at ($(apn)+(0,-1)$) {};
\node[spoint,opacity=0.6] (zapn1) at ($(zapn)+(0,-1)$) {};
\node[spoint,opacity=0.6] (zap) at ($(ap)+(0,-1)$) {};
\node[spoint,opacity=0.6] (zapp) at ($(app)+(0,-1)$) {};
\node[spoint,opacity=0.6] (zapp1) at ($(zapp)+(0,-1)$) {};

\begin{scope}[thick, densely dotted]
\draw[opacity=0.6] (zann1) -- +(0,-0.4);
\draw[opacity=0.6] (zanp1) -- +(0,-0.3);
\draw[opacity=0.6] (zapn1) -- +(0,-0.3);
\draw[opacity=0.6] (zan) -- +(0,-0.4);
\draw[opacity=0.6] (zapp1) -- +(0,-0.4);
\draw[opacity=0.6] (zap) -- +(0,-0.4);

\draw[<-,opacity=0.6] (app) -- +(0.3,0.3);
\draw[<-,opacity=0.6] (app) -- +(-0.3,0.3);
\draw[<-,opacity=0.6] (ann) -- +(0.3,0.4);
\draw[<-,opacity=0.6] (ann) -- +(-0.3,0.4);
\draw[<-,opacity=0.6] (anp) -- +(0.3,0.4);
\draw[<-,opacity=0.6] (anp) -- +(-0.3,0.4);
\draw[<-,opacity=0.6] (apn) -- +(0.3,0.4);
\draw[<-,opacity=0.6] (apn) -- +(-0.3,0.4);
\end{scope}
\begin{scope}[thick]
\draw[->,opacity=0.6] (an) to node[solid,midway,man] {\tiny\bf$\boldsymbol{-}$} (zan);
\draw[->,opacity=0.6] (ann) to node[solid,midway,man] {\tiny\bf$\boldsymbol{-}$} (zann);
\draw[->,opacity=0.6] (zann) to node[solid,midway,man] {} (zann1);
\draw[->,opacity=0.6] (anp) to node[solid,midway,man] {\tiny\bf+} (zanp);
\draw[->,opacity=0.6] (zanp) to node[solid,midway,man] {} (zanp1);
\draw[->,opacity=0.6] (apn) to node[solid,midway,man] {\tiny\bf$\boldsymbol{-}$} (zapn);
\draw[->,opacity=0.6] (zapn) to node[solid,midway,man] {} (zapn1);
\draw[->,opacity=0.6] (ap) to node[solid,midway,man] {\tiny\bf+} (zap);
\draw[->,opacity=0.6] (app) to node[solid,midway,man] {\tiny\bf+} (zapp);
\draw[->,opacity=0.6] (zapp) to node[solid,midway,man] {} (zapp1);
\end{scope}
\node[fill=white,inner sep=1pt] at (4,1.5) {\normalsize $\C_{\T_\dag,\A^{\avec{\alpha}}_m}$};

\node[wpoint] (a-n) at (2.5,-4) {};
\node[wpoint] (a-nn) at (1.5,-5) {};
\node[spoint] (a-np) at (3.5,-5) {};

\node[wpoint] (a-p) at (5.5,-4) {};
\node[spoint] (a-pn) at (4.5,-5) {};
\node[wpoint] (a-pp) at (6.5,-5) {};

\draw[->] (a) to node[midway,man] {\tiny$\boldsymbol{-}$} (a-n);
\draw[->] (a-n) to node[midway,man] {\tiny$\boldsymbol{-}$} (a-nn);
\draw[->] (a-n) to node[midway,man] {\tiny\bf +} (a-np);

\draw[->] (a) to node[midway,man] {\tiny\bf +} (a-p);
\draw[->] (a-p) to node[midway,man] {\tiny$\boldsymbol{-}$} (a-pn);
\draw[->] (a-p) to node[midway,man] {\tiny\bf +} (a-pp);

\node at (4,-4) {\normalsize $\A^{\avec{\alpha}}_m$};

\end{scope}
\node at (-3,1.5) {\normalsize $\bar\q_\varphi(x)$};
\begin{scope}[dashed]
\draw[out=-30,in=-150,looseness=0.6] (x) to (a);
\end{scope}
\end{tikzpicture}
\caption{Example of $\bar\q_\varphi(x)$ and $\C_{\T_\dag, \A^{\avec{\alpha}}_m}$ for $\varphi = \chi_1 \land \dots \land \chi_4$ with $\chi_1 = (p_1 \lor \neg p_3 \lor p_4)$, $\chi_2 = (\neg p_3\land p_4)$, $\chi_3 = p_1$, $\chi_4=(\neg p_3\lor \neg p_4)$ and $\avec{\alpha} = (0,1,1,0)$}
\label{fig:size-proof}
\end{figure*}

\subsection{Theorem~\ref{thm:NP:query}\label{app:thm:NP:query}}

\indent\textsc{Theorem \ref{thm:NP:query}.} {\it
There is an ontology $\T_\dag$ such that answering OMQs of the form $(\T_\dag,\q)$ with Boolean tree-shaped CQs $\q$ is \NP-hard for query complexity.}

\begin{proof}
We assume that $\T_\dag$ consists of the following axioms:
\begin{align*}
& A(x) \to \exists y \, \upsilon_+(x,y)\\
&\upsilon_+(x,y) \to  \P(y,x) \land \V(y,x) \land B_{-}(y) \land A(y),\\
& \hspace*{1em}B_{-}(x)  \to \exists y\, \eta_-(x,y)\\
& \hspace*{1em}\eta_-(x,y) \to  \N(x,y) \land B_0(y),\\
& A(x) \to \exists y \, \upsilon_-(x,y)\\
&\upsilon_-(x,y) \to \N(y,x) \land \V(y,x) \land B_{+}(y)  \land A(y),\\
& \hspace*{1em}B_{+}(x)  \to \exists y\, \eta_+(x,y)\\
& \hspace*{1em} \eta_+(x,y) \to \P(x,y) \land B_0(y),\\
& B_0(x) \to \exists y \, \eta_0(x,y)\\
& \eta_0(x,y) \to  \P(x,y) \land \N(x,y) \land \V(x,y) \land B_0(y).
\end{align*}

Let $\C$ be the canonical model of  $(\T_\dag,\{A(a)\})$. We prove that $\C \models \q_\varphi$ iff $\varphi$ is satisfiable.

\smallskip

$(\Rightarrow)$ Suppose  $h$ is a homomorphism from $\q_\varphi$ to $\C$ and $h(z^k_j) = h(y) = a \varrho_1 \dots \varrho_n$, for some roles $\varrho_l$.
Since $A(y) \in \q_\varphi$, it follows that $\varrho_l \in \{\upsilon_+, \upsilon_-\}$. Moreover, because of the structure of $\C$, without
any loss of generality we may assume that $n=k$.
Define a valuation $\nu \colon \{p_1,\dots,p_k\} \to \{\Tr,\Fa\}$ by taking $\nu(p_l) = \Tr$ if $\varrho_l = \upsilon_-$, $\nu(p_l) = \Fa$,
if $\varrho_l = \upsilon_+$. We claim that $\nu$ makes $\varphi$ true. To verify that the clause $\chi_j$ is satisfied, consider a number $1 \leq s \leq k$, such that the $j$th branch of the query is mapped on $\mathcal{C}$ in the following way:
\begin{align*}
  h(z_j^l) = a \varrho_1 \dots \varrho_l,&& s \leq l \leq k,\\
  h(z_j^l) = a \varrho_1 \dots \varrho_s \gamma_1 \dots \gamma_{s-l},&& 0 \leq l < s,
\end{align*}
for some roles $\gamma_1 \dots \gamma_{s-l}$ with $\gamma_1 \in \{\eta_-, \eta_+\}$ and $\gamma_i =  \eta_0$ for $2 \le i \le s-l$.
Such $s$ and the roles $\gamma_i$ exist, because
the $P$-atoms in $\C$ are directed towards the root if they cover $\upsilon$-atoms, and away from the root if they cover $\eta$-atoms
($s \ge 1$ since $B_0(z_j^0) \in \q_\varphi$).
 Clearly, $\T_\dag \models \gamma_1(x,y) \to P_{+}(x,y)$ iff $\rho_s = \upsilon_-$ iff $\nu(p_s) = \Tr$ and
 $\T_\dag \models \gamma_1(x,y) \to P_{-}(x,y)$ iff $\rho_s = \upsilon_+$ iff $\nu(p_s) = \Fa$.
 It follows that either $\P(z^s_{j},z^{s-1}_j) \in \q_\varphi$ and $\nu(p_s) =  \Tr$, or $\N(z^s_{j},z^{s-1}_j) \in \q_\varphi$ and $\nu(p_s) =  \Fa$. In either case, $\chi_j$ contains a literal with $p_s$ satisfied by $\nu$.

\medskip

$(\Leftarrow)$ Suppose a valuation $\nu \colon \{p_1,\dots,p_k\} \to \{\Tr,\Fa\}$ satisfies $\varphi$. Consider the sequence of roles $\varrho_1 \dots \varrho_k$, such that for $1\le l \le k$ we have $\varrho_l = \upsilon_+$, if $\nu(p_l) = \Fa$, and $\varrho_l = \upsilon_-$, if $\nu(p_l) = \Tr$. We claim that there exists a homomorphism $h$ from $q_\varphi$ to $\C$. First, let $h(y) = a \varrho_1 \dots \varrho_k$. To map the $j$th branch of the query, consider the maximal $1 \leq s \leq k$, such that a $p_s$-literal (positive or negative) makes $\chi_j$ true. Set
\begin{align*}
  h(z_j^l) = a \varrho_1 \dots \varrho_l,&& s \leq l \leq k-1,\\
  h(z_j^l) = a \varrho_1 \dots \varrho_s \gamma_1 \dots \gamma_{s-l},&& 0 \leq l < s,
\end{align*}
where $\gamma_1 = \eta_+$ if $p_s$ occurs positively, $\gamma_1 = \eta_-$ if $p_s$ occurs negatively and  $\gamma_{i} = \eta_0$ for $i \ge 2$. That $z_j^l$, for $s \leq l \leq k-1$, are mapped correctly follows from the maximality of $s$. That $z_j^l$ is mapped correctly for $l = s-1$ follows from the fact that $p_s$ occurs in $\chi_j$ positively iff $\P(z^s_{j},z^{s-1}_j) \in \q_\varphi$ iff $\nu(p_s) = \Tr$ iff $\varrho_s = \upsilon_-$ iff $\gamma_1 = \eta_+$ (similarly for negative $p_s$). Finally, $z_j^l$ is mapped correctly for $0 \leq l < s-1$ since the sequence of roles $\gamma_2 \dots \gamma_{s-l}$ can embed any $\P$, $\N$, or $\V$ roles, and $B_0$ concept. Thus, $h$ is a homomorphism from $\q_\varphi$ to $\C$.
\end{proof}

\subsection{Theorem~\ref{no-ql-rewritings}\label{app:no-ql-rewritings}}

We need several intermediate results and definitions before we present the proof in the end of the section. Suppose $\varphi$ is a propositional formula in CNF having $k$ variables $p_1,\dots,p_k$ and $m$ clauses $\chi_1,\dots,\chi_m$. We assume that $m = 2^\ell$. 
We associate with every such $\varphi$ a CQ $\bar\q_\varphi(x)$ with one answer variable $x$ and the following atoms, where $1\le j \le m$, $1 \le l \le k$, and $z^k_j = y^k$:
\begin{align*}
& \V(y^1,x),\dots,\V(y^{k},y^{k-1}),\hspace*{-5em}\\
& \P(z_j^l,z_j^{l-1}) && \text{ if } \chi_j  \text{ contains } p_l,\\
& \N(z_j^l,z_j^{l-1}), && \text{ if } \chi_j \text{ contains } \neg p_l,\\
& \V(z_j^l,z_j^{l-1}), && \text{ if } \chi_j \text{ contains no occurrence of } p_l.
\end{align*}
Then, for $0 \leq l \leq \ell-1$,
\begin{align*}
& \N(z_j^{-l}, z_j^{-l-1}), &&  \text{ if the $l$th bit of $(j-1)_2$ is $0$},\\
& \P(z_j^{-l}, z_j^{-l-1}), && \text{ if the $l$th bit of $(j-1)_2$ is $1$},\\
& B_0(z_j^{-\ell}).
\end{align*}
See an example in Fig.~\ref{fig:size-proof}.
For any $\avec{\alpha} \in \{0,1\}^m$, define a data instance $\A_m^{\avec{\alpha}}$ as the full binary tree of depth $\ell$ (and so $m=2^\ell$ leaves) on the binary predicates $\N$ (for the left child) and $\P$ (for the right child); $\A_m^{\avec{\alpha}}$ contains  $A(a)$ for the root $a$ of the tree and, for every $i$th leaf $b_i$ of the tree, $B_0(b_i) \in \A_m^{\avec{\alpha}}$ iff $\avec{\alpha}_i = 1$.

Denote by $f_\varphi \colon \{0,1\}^m \to \{0,1\}$ the \emph{monotone}  function such that $f_\varphi(\avec{\alpha}) = 1$ iff the CNF $\varphi^{-\avec{\alpha}}$, which is obtained from $\varphi$ by removing all conjuncts $\chi_i$ with $\avec{\alpha}_i = 1$, is satisfiable. It is readily checked that we have
\begin{lemma}
For any $\avec{\alpha} \in \{0,1\}^m$,
$$
\T_\dag,\A_m^{\avec{\alpha}} \models \bar\q_\varphi(a) \quad \text{iff} \quad f_\varphi(\avec{\alpha}) = 1.
$$
\end{lemma}

Let $\mathcal{QL}$ be any query language such that, for any $\mathcal{QL}$-query $\Phi(x)$ and any $\A_m^{\avec{\alpha}}$, the answer to $\Phi(a)$ over $\A_m^{\avec{\alpha}}$ can be computed in time $\textit{poly}(|\Phi|,m)$.

\begin{theorem}
The OMQ $(\T_\dag, \bar\q_\varphi(x))$ does not have a polynomial-size rewriting in $\mathcal{QL}$ unless $\NP \subseteq \Ppoly$.
\end{theorem}
\begin{proof}
Take any sequence of CNFs $\varphi_n$ of polynomial size in $n$ such that $f_{\varphi_n}$ is \NP-hard~\cite[Sec.~3]{DBLP:journals/ai/GottlobKKPSZ14}. Suppose there is a $\mathcal{QL}$-rewriting $\Phi_n$ of $(\T_\dag, \bar\q_\varphi(x))$ of polynomial size. By adapting the proof of $\PTime \subseteq \Ppoly$~\cite[Theorem~6.6]{Arora&Barak09} to the algorithm that checks $\A_m^{\avec{\alpha}} \models \Phi_n(a)$, we obtain a sequence of polynomial-size circuits computing $f_{\varphi_n}$, from which $\NP \subseteq \Ppoly$.
\end{proof}

\subsection{Theorem~\ref{trees-data-np}\label{app:trees-data-np}}

\indent\textsc{Theorem \ref{trees-data-np}.} {\it Evaluating \PE-queries over trees in $\mathfrak T$ is \NP-hard.}

\bigskip

More precisely, we are going to prove:
\begin{theorem}
The evaluation problem for PE-queries over data instances of the form $\A_m^{\avec{\alpha}}$ is $\NP$-hard.
\end{theorem}

\begin{proof}
Let $\varphi_k$, $k \ge 1$, be the 3-CNF with all possible $m = O(k^3)$ clauses of $k$ variables. Without loss of generality, we will assume that the number of clauses in $\varphi_k$  is actually $m = 2^\ell$, for some $\ell$. We construct a PE-query $\q_m(x)$ such that, for any $\avec{\alpha} \in \{0,1\}^m$, we have $\A_m^{\avec{\alpha}} \models \q_m(a)$ iff the CNF $\varphi^{-\avec{\alpha}}_k$ is satisfiable, and the size of $\q_m$ is polynomial in $m$ (and $k$).

The query $\q_m(x)$ takes the form
$$
\q_m(x) = \exists \avec{z} \, \big(\avec{r}(x,\avec{z}) \land \avec{s}(x,\avec{z}) \land \avec{t}(x,\avec{z})\big),
$$
where the subqueries (without quantified variables) $\avec{r}$, $\avec{s}$ and  $\avec{t}$ and the variables $\avec{z}$ are defined as follows.
Among the variables $\avec{z}$, there are variables $z_1, \dots, z_m$ corresponding to the leaves of $\A_m^{\avec{\alpha}}$, variables $x_1, \ldots, x_k$ corresponding to the propositional variables of $\varphi_k$, and variables $x_1^\prime, \ldots, x_k^\prime$ corresponding to their negations (there are other auxiliary variables which will be introduced later on).

Now we will describe the subqueries $\avec{r}, \avec{s}, \avec{t}$ of $\q_m$. The subquery $\avec{r}$ expresses that the variables $z_1, \ldots, z_m$ indeed correspond to the clauses of $\varphi_k$; it takes the form
$
\avec{r} = \bigwedge_{i=1}^{m} \avec{r}_i.
$
Each $\avec{r}_i$ corresponds to a leaf of $\A_m^{\avec{\alpha}}$. Consider a path from the root $a$ to this $i$th leaf. Let $P_1, \ldots,P_\ell$ be the sequence of labels on the edges of this path, that is, each $P_i$ is either $\N$ or $\P$. Then
$$
\avec{r}_{i} = P_1(x,y_i^1) \land P_2(y_i^1,y_i^2) \land \ldots \land P_\ell(y_i^{\ell-1},z_i),
$$
where $y_i^1,\ldots, y_i^{\ell-1}$ are variables among $\avec{z}$.

\newcommand{\np}{P_{\pm}}

The subquery $\avec{s}$ encodes that the variables $x_1\ldots, x_k$ and $x_1^\prime, \ldots, x_k^\prime$ correspond to an arbitrary Boolean assignment. It is of the form
$
\avec{s} = \bigwedge_{i=1}^{k} \avec{s}_i,
$
and each $\avec{s}_i$ is the following:
\begin{multline*}
 \np(x,u_i^1) \land \np(u_i^1,u_i^2) \land \dots \land
\np(u_i^{\ell-2},u_i^{\ell-1})\land{}  \\
 \big[\left( \np(u_i^{\ell-1},x_i) \land \np(x_i^\prime, u_i^{\ell-1}) \land B_0(x_i)\right) \lor{} \\
 \left( \np(u_i^{\ell-1},x_i^\prime) \land \np(x_i, u_i^{\ell-1}) \land B_0(x_i^\prime)\right) \big],
\end{multline*}
where $u_i^1,\ldots, u_i^{\ell-1}$ are variables among $\avec{z}$ and \mbox{$\np(x,y) = \N(x,y)\lor \P(x,y)$}.

The last subquery $\avec{t}$ encodes that the assignment given by $x_1,\ldots,x_k$ and $x_1^\prime, \ldots, x_k^\prime$ satisfies the CNF given by $z_1,\ldots, z_m$.
The formula $\avec{t}$ has the following form:
$
\avec{t} = \bigwedge_{i=1}^m \avec{t}_i.
$
Suppose the clause $z_i$ is a disjunction of literals $l_{i,1}, l_{i,2}$ and $l_{i,3}$, where each $l_{i,n}$ is among $x_1,\ldots,x_k$ and $x_1^\prime, \ldots, x_k^\prime$.
Then
$$
\avec{t}_i = B_0(z_i) \lor B_0(l_{i,1}) \lor B_0(l_{i,2}) \lor B_0(l_{i,3}).
$$

It is easy to see that $\q_m$ is satisfiable over a given $\A_m^{\avec{\alpha}}$ iff $\A_m^{\avec{\alpha}}$ corresponds to a satisfiable 3-CNF $\varphi_k^{-\avec{\alpha}}$. Thus we have reduced the 3-SAT problem to the problem of evaluating $\q_m$ over $\A_m^{\avec{\alpha}}$. Since 3-SAT is $\NP$-complete, we thus have shown $\NP$-hardness of our query evaluation problem.
\end{proof}

\subsection{Theorem~\ref{fixed-logcfl}\label{app:fixed-logcfl}}

\indent\textsc{Theorem \ref{fixed-logcfl}.} {\it There is an ontology $\T_\ddag$ such that answering OMQs of the form $(\T_\ddag,\q)$ with Boolean linear CQs $\q$ is \LOGCFL-hard for query complexity.}
\begin{proof}
  Our proof encodes the hardest \LOGCFL{} language $\mathcal L$~\cite{DBLP:journals/siamcomp/Greibach73} as formulated in~\cite{Sudborough:1975:NTC:321906.321913}.    The language $\mathcal{L}$ enjoys the following property:
  for every language $\mathcal{L}'$ over the alphabet $\Sigma'$ in \LOGCFL{}, there exists a logspace transducer $\tau$ converting words over $\Sigma'$ to the words over the alphabet $\Sigma$ of $\mathcal{L}$ in the sense that $w \in \mathcal{L}'$ iff $\tau(w) \in \mathcal{L}$. We construct an ontology $\T_\ddag$ and a logspace transducer that converts the words $w\in \Sigma^*$ to linear Boolean CQs $\q_w$ such that
  \begin{equation*}
  w \in \mathcal{L} \ \ \ \text{ iff } \ \ \ \T_\ddag,\{A(a)\} \models \q_w.
  \end{equation*}

To explain the construction, we begin with a simpler context-free language. Let $\Sigma_0 = \{a_1, b_1, a_2, b_2 \}$ be an alphabet and $B_0$ be the context-free language generated by the following grammar:
 \begin{align*}
 S & \to SS, & S & \to \epsilon, &
 S & \to a_1 S b_1, & S & \to a_2 S b_2.
\end{align*}
With each word $w = c_0\dots c_n$ over $\Sigma_0$ we associate conjunction $\gamma_w(u_0,v_0,\dots,u_n,v_n,u_{n+1})$ of the following atoms:
    \begin{equation*}
    R_{c_0}(u_0, v_0), S_{c_0}(v_0, u_1), R_{c_1}(u_1, v_1), S_{c_1}(v_1, u_2),\dots,
     R_{c_n}(u_n, v_n), S_{c_n}(v_n, u_{n+1}),
    \end{equation*}
where $R_c$ and $S_c$ are binary predicates, for $c\in \Sigma_0$.
Let $\T_0$ contain the following axioms, for $i = 1,2$:
\begin{equation}
\label{eq:LB:main}
D(x)  \to \exists y\,\bigl(R_{a_i}(x,y) \land S_{b_i}(y,x) \land
 \exists z\,\bigl(S_{a_i}(y,z) \land R_{b_i}(z,y)\land D(z)\bigr)\bigr).
\end{equation}
An initial part of the canonical model of $(\T_0,\{A(a), D(a)\})$ encoded by these axioms is shown below:\\
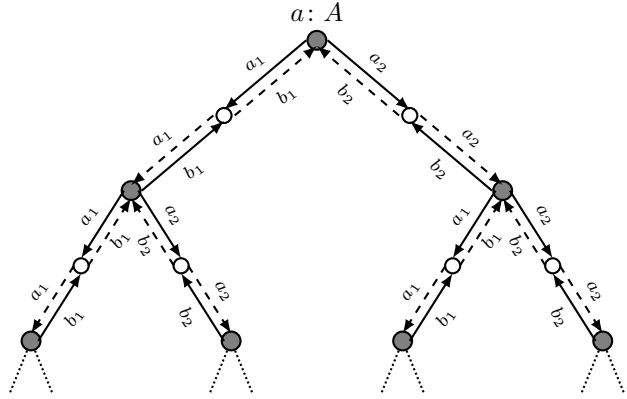
\centerline{\begin{tikzpicture}[xscale=0.95,nd/.style={inner sep=0pt,minimum size=2.5mm,thick,draw,circle,fill=gray},
nd2/.style={inner sep=0pt,minimum size=2mm,thick,draw,circle,fill=white},>=latex]
\node[nd,label=above:{$a\colon A$}] (a) at (0,0) {};
\node[nd2] at (-1.3,-1) (c1) {};
\node[nd] at (-2.6,-2) (w1) {};
\node[nd2] at (1.3,-1) (c2) {};
\node[nd] at (2.6,-2) (w2) {};
\node[nd2] at (-3.3,-3) (c11) {};
\node[nd] at (-4,-4) (w11) {};
\node[nd2] at (-1.9,-3) (c12) {};
\node[nd] at (-1.2,-4) (w12) {};
\node[nd2] at (1.9,-3) (c21) {};
\node[nd] at (1.2,-4) (w21) {};
\node[nd2] at (3.3,-3) (c22) {};
\node[nd] at (4,-4) (w22) {};
\draw[->,thick] ($(a)-(0.15,0)$) to node[midway,above,sloped] {\scriptsize $a_1$} ($(c1)+(0,0.075)$);
\draw[->,thick,dashed] ($(c1)-(0.15,0)$) to node[midway,above,sloped] {\scriptsize $a_1$} ($(w1)+(0,0.075)$);
\draw[->,thick] ($(w1)+(0.15,0)$) to node[midway,below,sloped] {\scriptsize $b_1$} ($(c1)+(0,-0.075)$);
\draw[->,thick,dashed] ($(c1)+(0.15,0)$) to node[midway,below,sloped] {\scriptsize $b_1$} ($(a)+(0,-0.075)$);
\draw[->,thick] ($(a)+(0.15,0)$) to node[midway,above,sloped] {\scriptsize $a_2$} ($(c2)+(0,0.075)$);
\draw[->,thick,dashed] ($(c2)+(0.15,0)$) to node[midway,above,sloped] {\scriptsize $a_2$} ($(w2)+(0,0.075)$);
\draw[->,thick] ($(w2)+(-0.15,0)$) to node[midway,below,sloped] {\scriptsize $b_2$} ($(c2)+(0,-0.075)$);
\draw[->,thick,dashed] ($(c2)+(-0.15,0)$) to node[midway,below,sloped] {\scriptsize $b_2$} ($(a)+(0,-0.075)$);
\draw[->,thick] ($(w1)-(0.1,0)$) to node[midway,above,sloped] {\scriptsize $a_1$} ($(c11)+(0,0.1)$);
\draw[->,thick,dashed] ($(c11)-(0.1,0)$) to node[midway,above,sloped] {\scriptsize $a_1$} ($(w11)+(0,0.1)$);
\draw[->,thick] ($(w11)+(0.1,0)$) to node[midway,below,sloped] {\scriptsize $b_1$} ($(c11)+(0,-0.1)$);
\draw[->,thick,dashed] ($(c11)+(0.1,0)$) to node[pos=0.5,below,sloped] {\scriptsize\rule{0pt}{3pt}$\smash{b_1}$} ($(w1)+(0,-0.1)$);
\draw[->,thick] ($(w1)+(0.1,0)$) to node[midway,above,sloped] {\scriptsize $a_2$} ($(c12)+(0,0.1)$);
\draw[->,thick,dashed] ($(c12)+(0.1,0)$) to node[midway,above,sloped] {\scriptsize $a_2$} ($(w12)+(0,0.1)$);
\draw[->,thick] ($(w12)+(-0.1,0)$) to node[midway,below,sloped] {\scriptsize $b_2$} ($(c12)+(0,-0.1)$);
\draw[->,thick,dashed] ($(c12)+(-0.1,0)$) to node[pos=0.4,below,sloped] {\scriptsize\rule{0pt}{3pt}$\smash{b_2}$} ($(w1)+(0,-0.1)$);
\draw[->,thick] ($(w2)-(0.1,0)$) to node[midway,above,sloped] {\scriptsize $a_1$} ($(c21)+(0,0.1)$);
\draw[->,thick,dashed] ($(c21)-(0.1,0)$) to node[midway,above,sloped] {\scriptsize $a_1$} ($(w21)+(0,0.1)$);
\draw[->,thick] ($(w21)+(0.1,0)$) to node[midway,below,sloped] {\scriptsize $b_1$} ($(c21)+(0,-0.1)$);
\draw[->,thick,dashed] ($(c21)+(0.1,0)$) to node[pos=0.5,below,sloped] {\scriptsize\rule{0pt}{3pt}$\smash{b_1}$} ($(w2)+(0,-0.1)$);
\draw[->,thick] ($(w2)+(0.1,0)$) to node[midway,above,sloped] {\scriptsize $a_2$} ($(c22)+(0,0.1)$);
\draw[->,thick,dashed] ($(c22)+(0.1,0)$) to node[midway,above,sloped] {\scriptsize $a_2$} ($(w22)+(0,0.1)$);
\draw[->,thick] ($(w22)+(-0.1,0)$) to node[midway,below,sloped] {\scriptsize $b_2$} ($(c22)+(0,-0.1)$);
\draw[->,thick,dashed] ($(c22)+(-0.1,0)$) to node[pos=0.4,below,sloped] {\scriptsize\rule{0pt}{3pt}$\smash{b_2}$} ($(w2)+(0,-0.1)$);
\draw[densely dotted,thick] (w11) -- +(-0.3,-0.7);
\draw[densely dotted,thick] (w11) -- +(0.3,-0.7);
\draw[densely dotted,thick] (w12) -- +(-0.3,-0.7);
\draw[densely dotted,thick] (w12) -- +(0.3,-0.7);
\draw[densely dotted,thick] (w21) -- +(-0.3,-0.7);
\draw[densely dotted,thick] (w21) -- +(0.3,-0.7);
\draw[densely dotted,thick] (w22) -- +(-0.3,-0.7);
\draw[densely dotted,thick] (w22) -- +(0.3,-0.7);
\end{tikzpicture}}\\
    (each large gray node belongs to $D$, each solid arrow with label $c$ belongs to $R_c$ and each dashed arrow with label $c$ to $S_c$, for $c \in \Sigma_0$).
    Let $\q_w^A$ be the following linear Boolean CQ:
\begin{equation*}
A(u_0) \land \gamma_w(u_0,v_0,\dots,u_n,v_n,u_{n+1}) \land A(u_{n+1}).
\end{equation*}
    The following claim can  readily be verified:
    \begin{proposition}\label{prop:LB:main}
      For every $w \in \Sigma_0^*$, we have \mbox{$w \in B_0$}  iff  $\T_0, \{A(a), D(a)\} \models \q_w^A$.
    \end{proposition}

  The language $B_0$ is, however, not \LOGCFL{}-hard. We now reproduce the definition of the hardest \LOGCFL{} language $\mathcal{L}$ from~\cite{Sudborough:1975:NTC:321906.321913}, which uses $B_0$ as a basis of the construction.
 Let $\Sigma = \Sigma_0 \cup \{ [, ], \# \}$, for distinct symbols $[$, $]$, and $\#$ not in $\Sigma_0$. Then set
  \begin{multline*}
  \mathcal{L} = \bigl\{[x_1 y_1 z_1][x_2 y_2 z_2] \dots [x_k y_k z_k]  \mid  k \geq 1,\\
  x_i \in (\Sigma_0 \cup \{\#\})^* \{\#\} \cup \{\epsilon\} \text{ and }\hspace*{4em}\\
   z_i \in \{ \epsilon \} \cup \{\#\}(\Sigma_0 \cup \{\# \})^*, \text{ for all } i \leq k,
 \text{ and } y_1 y_2 \dots y_k \in B_0 \bigr\}.
  \end{multline*}
  To explain the intuition, following \cite{Sudborough:1975:NTC:321906.321913}, let a string of symbols of the form $[w_1 \# w_2 \# \dots \# w_n]$, where $w_i \in \Sigma^*$ for all $i$, be called a \emph{block} and let each of the substrings $w_i$ be called a \emph{choice}. Then, $\mathcal{L}$ is the set of all strings of blocks such that there exists a sequence of choices, one from each block, which is in the base language $B_0$. The reader should notice that a choice (possibly of the empty string) must be made from each block. For example,
  \begin{align}
    [a_1 a_2 \# b_2 b_1] & \notin {\mathcal L}, \label{eq:not-in-L1}\\
    [a_1 a_2 \# b_2 b_1][b_2 b_1] & \in {\mathcal L},\label{eq:in-L1}\\
    [a_1 a_2 \# b_2 b_1][a_1 b_1] & \notin {\mathcal L},\label{eq:not-in-L2}\\
    [\# a_1 a_2 \# b_2 b_1][a_1 b_1] & \in {\mathcal L}.\label{eq:in-L2}
  \end{align}

We say that a word $w$ over $\Sigma$ is \emph{block-formed} if the following conditions are satisfied:
    \begin{itemize}
      \item[--] the word begins with $[$ and ends with $]$,
      \item[--] after each $[$ there is no $[$ before~$]$;
      \item[--] each non-final $]$ is followed immediately by $[$;
      \item[--] between each pair of matching $[$ and $]$ there is at least one symbol.
    \end{itemize}

   With these definitions at hand, we first describe a log\-space transducer that, given a word $w$ over $\Sigma$,  returns a linear Boolean CQ $\q_w$  with binary predicates $R_c$ and $S_c$, for $c\in \Sigma$,  and unary predicates $A$ and $E$.
If the word $w = c_0 \dots c_n$ is block-formed, then $\q_w$ consists of the following atoms:
    \begin{equation*}
A(u_0) \land \gamma_w(u_0,v_0,\dots,u_n,v_n,u_{n+1}) \land A(u_{n+1}).
    \end{equation*}
Otherwise, the transducer returns a query that consists of a prefix of $A(u_0)\land  \gamma_w(u_0,v_0,\dots,u_n,v_n,u_{n+1})$ and ends in $E(u_i)$, for some $i$, which will indicate an error (as all  queries containing $E$ will be false in $\T_\ddag,\{A(a)\}$).
    It is straightforward to verify that the required transducer can be implemented in~$\mathsf{L}$.

\smallskip

    Let $\T_\ddag$ contain the two axioms~\eqref{eq:LB:main} and the following axioms:
    \begin{align}
\label{eq:LB:A}
A(x) & \to D(x),\\
    \label{eq:LB:1}
       D(x) & \to \exists y \bigl(R_{[}(x,y) \land S_{[}(y,x)\bigr),\\
    \label{eq:LB:b}
       D(x) & \to \exists y \bigl(R_{[}(x,y) \land S_{\#}(y,x) \land \exists z\,\bigl(S_{[}(y,z) \land R_{\#}(z,y) \land F(z)\bigr)\bigr),\\
    \label{eq:LB:n}
       D(x) & \to \exists y \bigl(R_{]}(x,y) \land S_{]}(y,x)\bigr),\\
    \label{eq:LB:e}
       D(x) & \to \exists y \bigl(R_{\#}(x,y) \land S_{]}(y,x) \land \exists z\,\bigl(S_{\#}(y,z) \land R_{]}(z,y) \land F(z)\bigr)\bigr),\\
 \label{eq:LB:c}
       F(x) & \to \exists y\,\bigl(R_{c}(x,y) \land S_{c}(y,x)\bigr), & \text{ for } c \in \Sigma_0 \cup \{\#\}.
    \end{align}
    The four additional branches of the canonical model of $(\T_\ddag, \{A(a)\})$ at each point in $D$ are shown below:\\
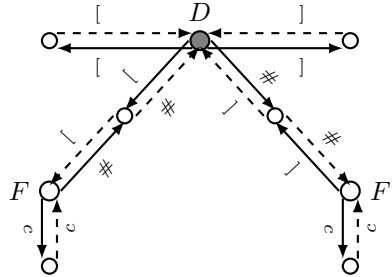
\centerline{\begin{tikzpicture}[nd/.style={inner sep=0pt,minimum size=2.5mm,thick,draw,circle,fill=gray},
nd2/.style={inner sep=0pt,minimum size=2mm,thick,draw,circle,fill=white},>=latex]
\node[nd,label=above:{$D$}] (a) at (0,0) {};
\node[nd2] at (-1,-1) (c1) {};
\node[nd,fill=white,label=left:{$F$}] at (-2,-2) (w1) {};
\node[nd2] at (1,-1) (c2) {};
\node[nd,fill=white,label=right:{$F$}] at (2,-2) (w2) {};
\node[nd2] at (-2,-3) (c11) {};
\node[nd2] at (-2,0) (c12) {};
\node[nd2] at (2,-3) (c21) {};
\node[nd2] at (2,0) (c22) {};
\draw[->,thick] ($(a)-(0.15,0)$) to node[pos=0.7,above,sloped] {\scriptsize $[$} ($(c1)+(0,0.075)$);
\draw[->,thick,dashed] ($(c1)-(0.15,0)$) to node[midway,above,sloped] {\scriptsize $[$} ($(w1)+(0,0.075)$);
\draw[->,thick] ($(w1)+(0.15,0)$) to node[midway,below,sloped] {\scriptsize $\#$} ($(c1)+(0,-0.075)$);
\draw[->,thick,dashed] ($(c1)+(0.15,0)$) to node[pos=0.3,below,sloped] {\scriptsize $\#$} ($(a)+(0,-0.075)$);
\draw[->,thick] ($(a)+(0.15,0)$) to node[pos=0.7,above,sloped] {\scriptsize $\#$} ($(c2)+(0,0.075)$);
\draw[->,thick,dashed] ($(c2)+(0.15,0)$) to node[midway,above,sloped] {\scriptsize $\#$} ($(w2)+(0,0.075)$);
\draw[->,thick] ($(w2)+(-0.15,0)$) to node[midway,below,sloped] {\scriptsize $]$} ($(c2)+(0,-0.075)$);
\draw[->,thick,dashed] ($(c2)+(-0.15,0)$) to node[pos=0.3,below,sloped] {\scriptsize $]$} ($(a)+(0,-0.075)$);
\draw[->,thick] ($(w1)+(-0.1,-0.1)$) to node[midway,below,sloped] {\scriptsize $c$} ($(c11)+(-0.1,0.1)$);
\draw[->,thick,dashed] ($(c11)+(0.1,0.1)$) to node[midway,below,sloped] {\scriptsize $c$} ($(w1)+(0.1,-0.1)$);
\draw[->,thick] ($(w2)+(-0.1,-0.1)$) to node[midway,below,sloped] {\scriptsize $c$} ($(c21)+(-0.1,0.1)$);
\draw[->,thick,dashed] ($(c21)+(0.1,0.1)$) to node[midway,below,sloped] {\scriptsize $c$} ($(w2)+(0.1,-0.1)$);
\draw[->,thick] ($(a)+(-0.1,-0.1)$) to node[pos=0.7,below,sloped] {\scriptsize $[$} ($(c12)+(0.1,-0.1)$);
\draw[->,thick,dashed] ($(c12)+(0.1,0.1)$) to node[pos=0.3,above,sloped] {\scriptsize $[$} ($(a)+(-0.1,0.1)$);
\draw[->,thick] ($(a)+(0.1,-0.1)$) to node[pos=0.7,below,sloped] {\scriptsize $]$} ($(c22)+(-0.1,-0.1)$);
\draw[->,thick,dashed] ($(c22)+(-0.1,0.1)$) to node[pos=0.3,above,sloped] {\scriptsize $]$} ($(a)+(0.1,0.1)$);
\end{tikzpicture}}\\
    (the labels $D$ and $F$ are indicated next to the nodes, and, as before, each solid arrow with label $c$ belongs to $R_c$ and each dashed arrow with label $c$ to $S_c$, for $c \in \Sigma_0$; to avoid clutter, only one pair of $c$-arrows is shown at the bottom).

   Let $\q^D_w$ be defined identically to $\q^A_w$ except that  the two occurrences of $A$ are replaced by~$D$. The following property is established similarly to Proposition~\ref{prop:LB:main}:
    \begin{proposition}\label{prop:LB:xz}
    For any block-formed word \mbox{$w]\in \Sigma^*$},
 \begin{equation*}
  w = [x, \text{ for } x\in (\Sigma_0 \cup \{\#\})^* \{\#\} \cup \{\epsilon\},\qquad
  \text{ iff }\qquad \{\eqref{eq:LB:1}, \eqref{eq:LB:b}, \eqref{eq:LB:c}\},\{D(d)\}\models \q^D_w.
\end{equation*}
    For any block-formed word $[w\in \Sigma^*$,
 \begin{equation*}
w = z], \text{ for  } z\in  \{\epsilon\} \cup \{\#\}(\Sigma_0 \cup \{\#\})^*,\qquad
\text{ iff } \qquad \{\eqref{eq:LB:n}, \eqref{eq:LB:e}, \eqref{eq:LB:c}\}, \{D(d)\}\models \q^D_w.
\end{equation*}
   \end{proposition}

  With these properties established,  it can readily be verified that $\T_\ddag,\{A(a)\} \models \q_w$ iff $w \in \mathcal{L}$.
 Consider a block-formed word $w\in\Sigma^*$. Let $[w_1 \# w_2 \# \dots \# w_n]$ be its $m$-th block and $w_j = y_m$ (that is, $w_j$ is the segment of the $B_0$-word in this block). By Proposition~\ref{prop:LB:xz}, the subtree generated by~\eqref{eq:LB:b} matches the (translation of) $[w_1 \# \dots \# w_{j-1} \#$, whereas the subtree generated by~\eqref{eq:LB:e} matches $\# w_{j+1}\# \dots \# w_n]$. By Proposition~\ref{prop:LB:main}, the $w_j$ itself is mapped into the main tree generated by~\eqref{eq:LB:main}. Note that~\eqref{eq:LB:1} and~\eqref{eq:LB:n} are needed for the case when $j = 1$ and $j = n$, respectively. Finally, observe that (the translation of) $w$ has to be mapped starting from $a$ (the root of the tree) and ending at $a$, and that the tree of the canonical model does not contain concept $E$, so only a block-formed $w$ can be mapped to the canonical model. In particular, $\T_\ddag, \{A(a)\} \not \models \q_w$ for $w$ of~\eqref{eq:not-in-L1} and~\eqref{eq:not-in-L2}, and $\T_\ddag, \{A(a)\} \models \q_w$ for $w$ of~\eqref{eq:in-L1} and~\eqref{eq:in-L2}.
\end{proof}

\subsection{Theorem~\ref{fixed-w1}}
\label{sec:fixed-w1}

\begin{theorem}
There is an ontology $\T_\Box$ such that $\blpr[\T_\Box]$ is $W[1]$-hard.
\end{theorem}

\begin{proof}
The proof is by reduction of the $W[1]$-hard problem \textit{SquareTiling}~\cite{DBLP:series/txtcs/FlumG06}, which is defined as follows:\\[4pt]
\hspace*{1em}{\tabcolsep=3pt\hspace*{-3pt}\begin{tabular}{lp{161mm}}
{\small\textsf{Instance:}} & a set $\mathfrak T$ of tile types painted in colours from a set $\mathfrak C$, a positive integer $k$,\\
{\small\textsf{Parameter:}} & $k$,\\
{\small\textsf{Problem:}} & decide whether $\mathfrak T$ tiles a $k \times k$-grid.
\end{tabular}}\\[3pt]
Suppose $\mathfrak C = \{0, \dots, n\}$, for $n \geq 1$, and $\mathfrak T = \{\mathcal{S}_1, \dots,  \mathcal{S}_m\}$. Denote by $\mathit{right}(t)$, $\mathit{left}(t)$, $\mathit{top}(t)$ and $\mathit{bottom}(t)$ the right, left, top and bottom colour of $\mathcal S_t$, respectively.
%
%
Given a binary predicate name $R$, we denote by $R^i$ a sequence of $i$-many predicates $R$. We represent each colour $c \le n$ by the following two sequences of binary predicates:
\begin{align*}
\mathtt{enc}_c & = P^{3n-c} F^c &&\text{ of length } 3n,\\
\mathtt{cne}_c & = N^{2c} M N^{2(n-c)} && \text{ of length } 2n + 1.
\end{align*}
Examples of $\mathtt{enc}_c$ and $\mathtt{cne}_c$, for $n=4$ and $c=3$, are shown in Figs.~\ref{fig:ontology-path}a and~\ref{fig:ontology-path}d, respectively. Each tile $\mathcal S_t$ is represented by the following sequence of binary predicates:
\begin{equation*}
\mathtt{tile}_t = B\, \mathtt{cne}_{\mathit{right}(t)}\, \mathtt{enc}_{\mathit{left}(t)} \, \mathtt{cne}_{\mathit{top}(t)} \, \mathtt{enc}_{\mathit{bottom}(t)} E\quad \text{ of length } 10n + 4.
\end{equation*}
Since the length of $\mathtt{enc}_c$ does not depend on $c$, we use $|\mathtt{enc}|$ to denote the length of some (any) $\mathtt{enc}_c$, and similarly for $|\mathtt{cne}|$ and $|\mathtt{tile}|$.
We shall also require the following sequences
\begin{align*}
\mathtt{segm} & = \mathtt{tile}_1\, \mathtt{tile}_2 \dots \mathtt{tile}_m\, S && \text{ of length } m\cdot |\mathtt{tile}| + 1,\\
\mathtt{spring} & =  (X Y I)^n && \text{ of length } |\mathtt{enc}|,\\
%
%
\mathtt{probelt} & = I^{|\mathtt{cne}| + |\mathtt{enc}| + 2} \, \mathtt{spring} \, I^{|\mathtt{tile}|+1}I^{2n} M^-\hspace*{-1em} && \text{ of length } 2\cdot |\mathtt{tile}| + 1,\\
\mathtt{probedn} & = I^2 \, \mathtt{spring} \, I^{(|\mathtt{tile}| + 1)k}I^{2n} M^- && \text{ of length } 5n + 3 + (|\mathtt{tile}| + 1)k.
\end{align*}
The first sequence will be called a \emph{segment} and the second a \emph{spring}.

%

The Boolean CQ $\q$ (see Fig.~\ref{fig:q-prime-tiling}) is now defined by taking the following set of atoms (assuming that all variables are existentially quantified):
 \begin{multline*}
\{ A(x_0) \} \cup \mathtt{segm}(x_{0}, x_{1,1}) \cup \hspace*{-1mm}\bigcup_{\substack{2 \leq  i \leq k\\1 \leq j \leq k}} \hspace*{-1mm}  \mathtt{segm}(x_{i-1,j}, x_{i,j}) \cup
  \bigcup_{2 \leq j \leq k} \hspace*{-1mm} \mathtt{segm}(x_{k, j-1}, x_{1,j}) \cup {} \\
  \bigcup_{\substack{2 \leq i \leq k\\1 \leq j \leq k}} \mathtt{probelt}(x_{i,j}, y_{i,j}) \cup \bigcup_{\substack{1 \leq i \leq k\\2 \leq j \leq k}} \mathtt{probedn}(x_{i,j}, z_{i,j}),
\end{multline*}
where $R_1 \dots R_l(x,y)$ stands for $\{R_1(x,x_1), R_2(x_1,x_2), \dots, R_l(x_{l-1},y)\}$ with fresh variables $x_1,\dots,x_{l-1}$.
%
It can be seen that $\q$  is a tree-shaped CQ with $2 (k-1) k$ leaves.

 \begin{figure}
  \centering
\begin{tikzpicture}[>=latex,scale=0.9,nd/.style={inner sep=1pt}]
\begin{scope}[line width=1.5pt]\small
\node (x0) at (12,0.5) {$x_0$};
\foreach \y/\w in {2/1,4/2,6/3,8.5/k} {
\foreach \x/\l in {0/1,3/2,6/3,9/k-1,12/k} {
   \node[nd] (x\l\w) at (\x,\y) {$x_{\l,\w}$};
}
\draw[->] (x1\w) -- (x2\w) node[above,midway] {\scriptsize$\mathtt{segm}$};
\node[nd] (y2\w) at (1.2,\y+0.8) {$y_{2,\w}$};
\draw[Green,->,out=90,in=0] (x2\w) to  node[below,pos=0.8] {\scriptsize$\mathtt{probelt}$} (y2\w);
\draw[->] (x2\w) -- (x3\w) node[below,midway] {\scriptsize$\mathtt{segm}$};
\node[nd] (y3\w) at (4.2,\y+0.8) {$y_{3,\w}$};
\draw[Green,->,out=90,in=0] (x3\w) to node[below,pos=0.8] {\scriptsize$\mathtt{probelt}$} (y3\w) ;
\draw[->] (xk-1\w) -- (xk\w) node[below,midway] {\scriptsize$\mathtt{segm}$};
\node[nd] (yk\w) at (10.2,\y+0.8) {$y_{k,\w}$};
\draw[Green,->,out=90,in=0] (xk\w) to node[below,pos=0.8] {\scriptsize$\mathtt{probelt}$}  (yk\w);
}
\draw[->] (x0) -- (x11) node[above,pos=0.4,sloped] {\scriptsize$\mathtt{segm}$};
\draw[->] (xk1) -- (x12) node[above,pos=0.4,sloped] {\scriptsize$\mathtt{segm}$};
\draw[->] (xk2) -- (x13) node[above,pos=0.4,sloped] {\scriptsize$\mathtt{segm}$};
\draw[-] (xk3) -- (6,7);
\draw[-] (6,7.5) -- (x1k);
\draw[dotted] (6,7.1) -- (6,7.4);
\foreach \x/\l in {0/1,3/2,6/3,9/k-1,12/k} {
\node[nd] (z\l2) at (\x,3.2) {$z_{\l,2}$};
\draw[->,RoyalBlue] (x\l2) -- (z\l2) node[left=3pt,midway,inner sep=0pt,fill=white] {\scriptsize$\mathtt{probedn}$};
}
\foreach \x/\l in {0/1,3/2,6/3,9/k-1,12/k} {
\node[nd] (z\l3) at (\x,5.2) {$z_{\l,3}$};
\draw[->,RoyalBlue] (x\l3) -- (z\l3) node[left=3pt,midway,inner sep=0pt,fill=white] {\scriptsize$\mathtt{probedn}$};
}
\foreach \x/\l in {0/1,3/2,6/3,9/k-1,12/k} {
\node[nd] (z\l k) at (\x,7.7) {$z_{\l,k}$};
\draw[->,RoyalBlue] (x\l k) -- (z\l k) node[left=3pt,midway,inner sep=0pt,fill=white] {\scriptsize$\mathtt{probedn}$};
}
\end{scope}
\end{tikzpicture}
\caption{The structure of the CQ $\q$.
}\label{fig:q-prime-tiling}
\end{figure}
%
%
 %
%

Let $\T_\Box$ be an ontology with the following axioms:
\begin{align*}
  A(x) \to & \ \exists y\,\bigl(BI(x,y)\land  \mathit{Right}(y)\bigr),  &
  A(x) \to& \  \exists u\,\mathit{Sink}(x,u), \\
  \mathit{Right}(x) \to& \ \exists y\,\bigl(\mathit{NI}(x,y)  \land \mathit{Right}(y)\bigr), &
  \mathit{Right}(x) \to & \ \exists z\,\bigl(\mathit{MI}(x,z) \land \mathit{Right}'(z)\bigr),\\
  \mathit{Right}'(x) \to& \ \exists y\,\bigl(\mathit{NI}(x,y) \land \mathit{Right}'(y)\bigr),  &
  \mathit{Right}'(x)\to & \ \mathit{Left}(x),\\
  \mathit{Left}(x) \to& \ \exists y\,\bigl(\mathit{PI}(x,y) \land \mathit{Left}(y)\bigr), &
  \mathit{Left}(x) \to& \ \exists z\,\bigr(\mathit{FI}(x,z) \land \mathit{Left}'(z)\bigr),\\
  \mathit{Left}'(x) \to& \ \exists y\,\bigl(\mathit{FI}(x,y) \land \mathit{Left}'(y)\bigr), &
  \mathit{Left'}(x) \to & \ \mathit{Top}(x),\\
  \mathit{Top}(x) \to& \ \exists y\,\bigl(\mathit{NI}(x,y)  \land \mathit{Top}(y)\bigr), &
  \mathit{Top}(x) \to & \ \exists z\,\bigl(\mathit{MI}(x,z) \land \mathit{Top}'(z)\bigr),\\
  \mathit{Top}'(x) \to& \ \exists y\,\bigl(\mathit{NI}(x,y) \land \mathit{Top}'(y)\bigr), &
  \mathit{Top}'(x) \to & \  \mathit{Bot}(x),\\
  \mathit{Bot}(x) \to& \ \exists y\,\bigl(\mathit{PI}(x,y) \land \mathit{Bot}(y)\bigr), &
  \mathit{Bot}(x) \to & \ \exists z\,\bigl(\mathit{FI}(x,z) \land \mathit{Bot}'(z)\bigr),\\
  \mathit{Bot}'(x) \to& \  \exists y\,\bigl(\mathit{FI}(x,y) \land \mathit{Bot}'(y)\bigr), &
   \mathit{Bot}'(x) \to & \ \exists  z\,\bigl(\mathit{EI}(x,z)  \land  A_2(z)\bigr), \\
   A_2(z) \to & \ \exists x\,\bigl(\mathit{SI}(z, x) \land A(x)\bigr), &
   A_2(z) \to &\ \exists u\,\mathit{Sink}(z,u),\\
  \mathit{Left}'(x)  \to & \ \exists y\,\mathit{X\!\bar{Y}}(x, y), &
  \mathit{Bot}'(x)  \to & \ \exists y\,\mathit{X\!\bar{Y}}(x, y),\\
  %
 P(x,y) \to & \ X(y, x), &
  P(x,y) \to & \ Y(y, x),\\
  \mathit{X\!\bar{Y}}(x, y) \to & \ X(x,y), &
  \mathit{X\!\bar{Y}}(x, y) \to & \ Y(y,x),
\end{align*}
%
%
where $C(x) \to \exists y\,\bigl(Q(x,y) \land D(y)\bigr)$ abbreviates three axioms
\begin{equation*}
C(x) \to \exists y\,Q_D(x,y),\qquad Q_D(x,y) \to Q(x,y)\quad \text{ and }\quad Q_D(x,y) \to D(y).
\end{equation*}
In addition, $\T_\Box$ contains the axioms
\begin{itemize}
\item  $\mathit{QI}(x,y) \to Q(x,y)$ and $\mathit{QI}(x,y)\to I(y,x)$, for all predicates of the form $\mathit{QI}$;

\item  $\mathit{Sink}(x,y) \to Q(x,y)$ and $\mathit{Sink}(x,y) \to Q(y,x)$, for all binary predicates $Q$ except~$I$ and $S$.
\end{itemize}
A path in the canonical model $\mathcal{C}_{\T_\Box, \{ A(a) \}}$ where $\q$ can be homomorphically mapped is shown in Fig.~\ref{fig:ontology-path} sandwiched between $\mathsf{segm}(x_0,x_{1,1}) \cup \mathsf{segm}(x_{1,1},x_{2,1})$ on the left and bottom and $\mathsf{probelt}(x_{2,1},y_{2,1})$ on the right and top (most predicate names are omitted).

\begin{figure}
\centering%
\begin{tikzpicture}[>=latex,nd/.style={circle,draw,fill,inner sep=0pt,minimum size=1.5mm}]
\filldraw[ultra thin,draw=black,fill=gray!20,dashed] (0.3,-0.4) rectangle +(2.7,-2.2);
\node[rotate=90] at (0.8,-1.5) {\small$\mathtt{cne}_c$};
\node at (2.6,-0.7) {\LARGE \textcolor{gray}{$\mathtt{R}$}};
\filldraw[ultra thin,draw=black,fill=gray!5,dashed] (0.3,-2.6) rectangle +(2.7,-2);
\node[rotate=90] at (0.8,-3.6) {\small$\mathtt{enc}$}; 
\node at (2.6,-3.6) {\LARGE \textcolor{gray}{$\mathtt{L}$}};
\filldraw[ultra thin,draw=black,fill=gray!5,dashed] (0.3,-4.6) rectangle +(2.7,-2.4);
\node[rotate=90] at (0.8,-5.8) {\small$\mathtt{cne}$}; 
\node at (2.6,-5.8) {\LARGE \textcolor{gray}{$\mathtt{T}$}};
\filldraw[ultra thin,draw=black,fill=gray!5,dashed] (0.3,-7) rectangle +(2.7,-2);
\node[rotate=90] at (0.8,-8) {\small$\mathtt{enc}$}; 
\node at (2.6,-8) {\LARGE \textcolor{gray}{$\mathtt{B}$}};
\filldraw[ultra thin,draw=black,fill=gray!5,dashed] (3.1,-11.8) rectangle +(1.2,2.7);
\node at (3.7,-11.3) {\small$\mathtt{cne}$}; 
\node at (3.7,-9.4) {\LARGE \textcolor{gray}{$\mathtt{R}$}};
\filldraw[ultra thin,draw=black,fill=gray!20,dashed] (4.3,-11.8) rectangle +(2.3,2.7);
\node at (5.45,-11.3) {\small$\mathtt{enc}_c$}; 
\node at (5,-9.4) {\LARGE \textcolor{gray}{$\mathtt{L}$}};
\filldraw[ultra thin,draw=black,fill=gray!5,dashed] (6.6,-11.8) rectangle +(1.2,2.7);
\node at (7.2,-11.3) {\small$\mathtt{cne}$}; 
\node at (7.2,-9.4) {\LARGE \textcolor{gray}{$\mathtt{T}$}};
\filldraw[ultra thin,draw=black,fill=gray!5,dashed] (7.8,-11.8) rectangle +(1,2.7);
\node at (8.3,-11.3) {\small$\mathtt{enc}$}; 
\node at (8.3,-9.4) {\LARGE \textcolor{gray}{$\mathtt{B}$}};
\node[nd,black,label=right:{$v_0$},label=above:{$A$}] at (1.5,0) {};
\node[nd,Green] at (0.5,0.05) {};
\node  (x01) at (0,-1) {$x_0$};
\draw[thin,->]  (x01) -- (0.4,-0.2);
\node[nd,black,label=below:{$A$}] (w0w1w2) at (10.4,-9.45) {};
\node (lbl) at (9.6,-8.2) {$v_0 w_{1,1} S w_{2,1}S$};
\draw[->,thin] (lbl) -- (w0w1w2);
\node[nd,Green,label=left:{$x_{2,1}$}] at (11.4,-9.45) {};
\node[nd,black,label=above:{\hspace*{2em}$v_0 w_{1,1} S$},label=left:{$A$}] at (2.7,-10.6) {};
\node[nd,Green] at (2.65,-11.6) {};
\node (x11) at (1.5,-11.6) {$x_{1,1}$};
\draw[thin,->] (x11) -- (2.5,-11.6);
\draw[rounded corners=1mm] (4.7,-10) rectangle +(1.9,2);
\node at (5.65,-8.4) {$\small\mathtt{spring}$};
\begin{scope}[thick]
\draw[Green] (0.45,0.05) -- +(-0.65,-0.05) node[midway,above] {$\mathsf{tile}_1$};
\draw[Green,->] (-0.2,0) -- +(0.65,-0.05) node[pos=0.3,below] {$\dots$}; 
\draw[->, OrangeRed] (0.5,-0.05) -- +(0,-0.35);
\draw[->, Green] (0.5,-0.4) -- +(0,-0.4);
\draw[dotted, Green] (0.5,-0.8) -- +(0,-0.2);
\draw[->, Green] (0.5,-1) -- +(0,-0.4);
\draw[->, Violet] (0.5,-1.4) -- +(0,-0.4);
\draw[->, Green] (0.5,-1.8) -- +(0,-0.4);
\draw[->, Green] (0.5,-2.2) -- +(0,-0.4);
%
%
\draw[Green] (0.5,-2.6) -- +(0,-2);
%
\draw[Green] (0.5,-4.6) -- +(0,-2.4);
%
\draw[Green] (0.5,-7) -- +(0,-2);
\draw[->, OrangeRed] (0.5,-9) -- +(0,-0.4);
%
\node[rotate=90,Green] at (0,-4.5) {$\mathsf{tile}_{t_{1,1}}$};
\draw[Green] (0.5,-9.35) -- +(-0.7,-0.05) node[pos=0.7,above] {$\dots$}; 
\draw[Green,->] (-0.2,-9.4) -- +(0.7,-0.05) node[pos=0.4,below] {$\mathsf{tile}_m$};
\draw[OrangeRed,->] (0.5,-9.45) --  (2.6,-11.6) node[midway,sloped,below] {}; 
\draw[Green] (2.65,-11.65) -- +(0.05,-0.65) node[midway,left] {$\mathsf{tile}_1$};
\draw[Green,->] (2.7,-12.3)-- +(0.05,0.65) node[midway,right] {$\dots$}; 
\draw[->, OrangeRed] (2.75,-11.6) -- +(0.35,0);
\draw[Green] (3.1,-11.6) -- +(1.2,0);
%
\draw[->, Green] (4.3,-11.6) -- +(0.4,0);
\draw[dotted,Green] (4.7,-11.6) -- +(0.3,0);
\draw[->, Green] (5,-11.6) -- +(0.4,0);
\draw[->, Violet] (5.4,-11.6) -- +(0.4,0);
\draw[->, Violet] (5.8,-11.6) -- +(0.4,0);
\draw[->, Violet] (6.2,-11.6) -- +(0.4,0);
%
\draw[Green] (6.6,-11.6) -- +(1.2,0);
\draw[Green] (7.8,-11.6) -- +(1,0);
\draw[->, OrangeRed] (8.8,-11.6) -- +(0.4,0);
\node[Green] at (6.25,-12) {$\mathsf{tile}_{t_{2,1}}$};
%
%
\draw[Green] (9.15,-11.6) -- +(0.05,-0.7) node[midway,left] {$\dots$}; 
\draw[Green,->] (9.2,-12.3)-- +(0.05,0.7) node[midway,right] {$\mathsf{tile}_m$};
\draw[OrangeRed,->] (9.25,-11.6) -- +(2.1,2.1) node[midway,sloped,below] {}; 
\end{scope}
\begin{scope}[thick]\scriptsize
\draw[Cerulean,->] (1.45,0) -- ++(-0.7,0) node[pos=0.6,above] {$\mathit{Sink}$};
\draw[->, OrangeRed] (1.5,-0.05) -- +(0,-0.35) node[midway,left] {$B$};
\draw[->, Green] (1.5,-0.4) -- +(0,-0.4) node[midway,right] {$N$};
\draw[dotted, Green] (1.5,-0.8) -- +(0,-0.2);
\draw[->, Green] (1.5,-1) -- +(0,-0.4) node[midway,right] {$N$};
\draw[->, Violet] (1.5,-1.4) -- +(0,-0.4) node[midway,left] {$M$};
\draw[->, Green] (1.5,-1.8) -- +(0,-0.4) node[midway,right] {$N$};
\draw[->, Green] (1.5,-2.2) -- +(0,-0.4) node[midway,right] {$N$};
%
\draw[->, Green] (1.5,-2.6) -- +(0,-0.4) node[midway,right] {$P$};
\draw[dotted, Green] (1.5,-3) -- +(0,-0.2);
\draw[->, Green] (1.5,-3.2) -- +(0,-0.4) node[midway,right] {$P$};
\draw[->, Violet] (1.5,-3.6) -- +(0,-0.4) node[midway,left] {$F$};
\draw[->, Orange] (1.5,-4) -- +(0.4,0) node[pos=1.1,above] {$X$};
\draw[dotted, Violet] (1.5,-4) -- +(0,-0.2);
\draw[->, Violet] (1.5,-4.2) -- +(0,-0.4) node[midway,left] {$F$};
\draw[->, Orange] (1.5,-4.6) -- +(0.4,0) node[pos=1.1,above] {$X$};
%
\draw[->, Green] (1.5,-4.6) -- +(0,-0.4) node[midway,right] {$N$};
\draw[dotted, Green] (1.5,-5) -- +(0,-0.2);
\draw[->, Green] (1.5,-5.2) -- +(0,-0.4) node[midway,right] {$N$};
\draw[->, Violet] (1.5,-5.6) -- +(0,-0.4) node[midway,left] {$M$};
\draw[->, Green] (1.5,-6) -- +(0,-0.4) node[midway,right] {$N$};
\draw[dotted, Green] (1.5,-6.4) -- +(0,-0.2);
\draw[->, Green] (1.5,-6.6) -- +(0,-0.4) node[midway,right] {$N$};
%
\draw[->, Green] (1.5,-7) -- +(0,-0.4) node[midway,right] {$P$};
\draw[dotted, Green] (1.5,-7.4) -- +(0,-0.2);
\draw[->, Green] (1.5,-7.6) -- +(0,-0.4) node[midway,right] {$P$};
\draw[->, Violet] (1.5,-8) -- +(0,-0.4) node[midway,left] {$F$};
\draw[->, Orange] (1.5,-8.4) -- +(0.4,0) node[pos=1.1,above] {$X$};
\draw[dotted, Violet] (1.5,-8.4) -- +(0,-0.2);
\draw[->, Violet] (1.5,-8.6) -- +(0,-0.4) node[midway,left] {$F$};
\draw[->, Orange] (1.5,-9) -- +(0.4,0) node[pos=1.1,above] {$X$};
\draw[->, OrangeRed] (1.5,-9) -- +(0,-0.4) node[midway,left] {$E$};
\draw[->, OrangeRed] (1.5,-9.4) -- +(1.15,-1.15) node[midway,below,sloped] {$S$};
\draw[Cerulean,->] (1.5,-9.4) -- ++(-0.7,0) node[midway,below] {$\mathit{Sink}$};
\draw[Cerulean,->] (2.7,-10.65) -- ++(0,-0.7) node[pos=0.6,below,sloped] {$\mathit{Sink}$};
\draw[->, OrangeRed] (2.75,-10.6) -- +(0.35,0) node[midway,below] {$B$};
\draw[Green] (3.1,-10.6) -- +(1.2,0);
%
\draw[->, Green] (4.3,-10.6) -- +(0.4,0) node[midway,above] {$P$};
\draw[dotted,Green] (4.7,-10.6) -- +(0.3,0);
\draw[->, Green] (5,-10.6) -- +(0.4,0) node[midway,above] {$P$};
\draw[->, Violet] (5.4,-10.6) -- +(0.4,0) node[midway,below] {$F$};
\draw[->, Orange] (5.8,-10.6) -- +(0,0.4) node[pos=1.1,left] {$X$};
\draw[->, Violet] (5.8,-10.6) -- +(0.4,0)  node[midway,below] {$F$};
\draw[->, Orange] (6.2,-10.6) -- +(0,0.4) node[pos=1.1,left] {$X$};
\draw[->, Violet] (6.2,-10.6) -- +(0.4,0) node[midway,below] {$F$};
\draw[->, Orange] (6.6,-10.6) -- +(0,0.4) node[pos=1.1,left] {$X$};
%
\draw[Green] (6.6,-10.6) -- +(1.2,0);
\draw[Green] (7.8,-10.6) -- +(1,0);
\draw[->, OrangeRed] (8.8,-10.6) -- +(0.4,0) node[midway,below] {$E$};
\draw[Cerulean,->] (9.2,-10.6) -- ++(0,-0.7) node[midway,above,sloped] {$\mathit{Sink}$};
\draw[->, OrangeRed] (9.2,-10.6) -- +(1.15,1.15) node[midway,below,sloped] {$S$};
\draw[dotted, OrangeRed] (10.4,-9.4) -- +(0,1);
\begin{scope}
\draw[ultra thin] (9.2,-9.8) -- +(0,1);
\draw[thin,<->] (9.2,-8.9) -- (6.6,-8.9) node[midway,above] {$|\mathtt{cne}| + |\mathtt{enc}| + 1$};
\draw[thin,<-] (4.7,-8.1) -- (4,-8.1);
\draw[thin,->] (4,-8.1) -- (4,-1.8) node[midway,below,sloped] {$|\mathtt{tile}| + 2n + 1$};
\draw[ultra thin] (4.1,-1.8) -- +(-1.8,0);
\draw[ultra thin] (3.3,-2.6) -- +(-1,0);
\draw[thin,->] (3.2,-8.9) -- (3.2,-2.6)  node[midway,below,sloped] {$|\mathtt{tile}| + 1$};
\draw[ultra thin] (4.3,-8.8) -- +(0,-1);
\draw[thin,->] (3.2,-8.9) -- (4.3,-8.9);
\end{scope}
\node[Violet,nd,label=above:{$y_{2,1}$}] at (2.3,-1.4) {};
\draw[<-, Violet] (2.3,-1.45) -- +(0,-0.35) node[midway,right] {$M^-$};
\draw[<-,gray] (2.3,-1.8) -- (2.3,-9.4);
\draw[gray] (2.3,-9.4) -- +(0.4,-0.4);
\draw[gray] (2.7,-9.8) -- (4.7,-9.8);
\draw[black,<-] (4.7,-9.8) -- +(0.4,0);
\draw[black,dotted] (5.1,-9.8) -- +(0.3,0);
\draw[black,<-] (5.4,-9.8) -- +(0.35,0);
\draw[black,<-] (5.85,-9.8) -- +(0.3,0);
\draw[black,<-] (5.75,-9.8) -- +(+0.05,0.5);
\draw[black,->] (5.85,-9.8) -- +(-0.05,0.5);
\draw[black,<-] (6.25,-9.8) -- +(0.3,0);
\draw[black,<-] (6.15,-9.8) -- +(+0.05,0.5);
\draw[black,->] (6.25,-9.8) -- +(-0.05,0.5);
\draw[black,<-] (6.5,-9.8) -- +(+0.05,0.5);
\draw[black,->] (6.6,-9.8) -- +(-0.05,0.5);
\draw[gray,<-] (6.6,-9.8) -- (9.2,-9.8);
\draw[gray,out=90,in=90,<-] (9.2,-9.8) to (11.4,-9.4);
\end{scope}
\begin{scope}[line width=1.2pt,xshift=80mm,yshift=7mm,yscale=0.8]
\coordinate (d1) at (-1.5,-1.5);
\coordinate (d2) at (-1.5,-10.5);
\filldraw[ultra thin,draw=black,fill=gray!15,dashed] (d1) rectangle +(4.75,-9);
%
\draw[ultra thin] (-1.5,-7.5) -- +(4.75,0);
\draw[ultra thin] (-1.5,-8.5) -- +(4.75,0);
\draw[thin,<->] (-0.9,-10.5) -- +(0,2) node[midway,above,sloped] {\scriptsize $2n-2c$};
\draw[thin,<->] (-0.9,-7.5) -- +(0,6) node[midway,above,sloped] {\scriptsize $2c$};
\node at (0,-1) {d)};
\foreach \i in {1,...,6} {
\draw[Green,->] (0,-\i-0.5) -- ++(0,-1) node[left,midway] {\scriptsize $N$};
}
\foreach \i in {7} {
\draw[Violet,->] (0,-\i-0.5) -- ++(0,-1) node[left,midway] {\scriptsize $M$};
}
\foreach \i in {8,...,9} {
\draw[Green,->] (0,-\i-0.5) -- ++(0,-1) node[left,midway] {\scriptsize $N$};
}
\node at (1.25,-1) {e)};
\foreach \i in {1,...,6} {
\draw[Green,->] (1.25,-\i-0.5) -- ++(0,-1) node[left,midway] {\scriptsize $N$} node[midway,right] {\scriptsize$I^-$};
}
\foreach \i in {7} {
\draw[Violet,->] (1.25,-\i-0.5) -- ++(0,-1) node[left,midway] {\scriptsize $M$} node[midway,right] {\scriptsize$I^-$};
}
\foreach \i in {8,...,9} {
\draw[Green,->] (1.25,-\i-0.5) -- ++(0,-1) node[left,midway] {\scriptsize $N$} node[midway,right] {\scriptsize$I^-$};
}
\node at (2.5,-1) {f)};
\foreach \i/\l in {8/I,9/I} {
\draw[gray,<-] (2.5,-\i-0.5) -- ++(0,-1) node[right,midway] {\scriptsize $\l$};
}
\draw[Violet,<-] (2.5,-7.55) -- ++(0,-0.95) node[right,midway] {\scriptsize $M^-$};
\node[Violet,nd,label=above:{$y_{2,1}$}] at (2.5,-7.5) {};
\end{scope}
\draw[gray] (3,-0.4) -- (d1);
\draw[gray] (3,-2.6) -- (d2);
\begin{scope}[line width=1.2pt,xshift=-12mm,yshift=-103mm,xscale=0.95]
\coordinate (e1) at (1.5,-2.1);
\coordinate (e2) at (13.5,-2.1);
\filldraw[ultra thin,draw=black,fill=gray!15,dashed] (e1) rectangle +(12,-4);
 \begin{scope}[line width=1.2pt]
 \node at (1.2,-5.4) {a)};
 \foreach \i in {1,...,9} {
 \draw[->,Green] (\i+0.5,-5.4) -- ++(1,0) node[above,midway] {\scriptsize$P$};
 }
 \foreach \i in {10,...,12} {
 \draw[->,Violet] (\i+0.5,-5.4) -- ++(1,0) node[above,midway] {\scriptsize$F$};
 }
 \node at (1.2,-4) {b)};
 \foreach \i in {1,...,9} {
 \draw[->,Green] (\i+0.5,-4) -- ++(1,0) node[above,midway] {\scriptsize$P$} node[midway,below] {\scriptsize\begin{tabular}{c}$I^-$\\[-2pt]$X^-$\\[-2pt]$Y^-$\end{tabular}};
 }
 \foreach \i in {10,...,12} {
 \draw[->,Violet] (\i+0.5,-4) -- ++(1,0) node[above,midway] {\scriptsize$F$} node[midway,below] {\scriptsize\begin{tabular}{c}$I^-$\end{tabular}};
 \draw[->,YellowOrange] (\i+1.5,-4) -- ++(0,0.7) node [midway,above,sloped] {\scriptsize $X$} node[pos=0.7,below,sloped] {\scriptsize $Y^-$};
 }
 \node at (1.2,-3) {c)};
 \foreach \i in {10,...,12} {
 \draw[->,black] (\i+1.55,-3) -- ++(0,0.7) node [midway,below,sloped] {\scriptsize $X$};
 \draw[->,black] (\i+1.45,-2.3) -- ++(0,-0.7) node [midway,below,sloped] {\scriptsize $Y$};
 \draw[->,black] (\i+1.45,-3) -- ++(-0.9,0) node [midway,below] {\scriptsize $I$};
 }
 \foreach \i/\l in {7/I,8/Y,9/X} { 
 \draw[->,black] (\i+1.5,-3) -- ++(-1,0) node [midway,below] {\scriptsize $\l$};
 }
 %
 %
 %
 \draw[ultra thin] (10.5, -2.1) -- +(0,-4);
 \draw[<->,thin] (10.5,-4.8) -- ++(3,0) node[midway,above] {\scriptsize $c$};
 \draw[<->,thin] (1.5,-5.9) -- ++(12,0) node[midway,above] {\scriptsize $3n$};
 %
 \draw[ultra thin] (7.5,-3) -- +(0,0.9);
 \draw[thin,<->] (7.5,-2.2) -- +(6,0) node[pos=0.3,below] {\scriptsize $3n-2c$};
 \draw[thin,<->] (1.5,-2.7) -- +(6,0) node[midway,below] {\scriptsize $2c$};
 \end{scope}
\end{scope}
\draw[gray] (4.3,-11.8) -- (e1);
\draw[gray] (6.6,-11.8) -- (e2);
\end{tikzpicture}
  \caption{Matching the first two segments, $\mathtt{segm}(x_0,x_{1,1})$ and $\mathtt{segm}(x_{1,1},x_{2,1})$, of $\q$ and $\mathtt{probelt}(x_{2,1}, y_{2,1})$ in  the canonical model $\C$, and
 the magnified fragments for $\mathtt{enc}_c$ and $\mathtt{cne}_c$ with $n=4$ and $c=3$: a) subsequence $\mathtt{enc}_c$ of $\mathtt{tile}_t$ in $\q$; b) a path in $\C$ where $\mathtt{enc}_c$ is mapped; c) matching the $\mathtt{spring}$ subquery of $\mathtt{probelt}$ in  $\C$;
d) subsequence $\mathtt{cne}_c$ of $\mathtt{tile}_t$ in~$\q$; e) a path in $\C$ where $\mathtt{cne}_c$ is mapped; f) mapping variable $y_{2,1}$ of the subquery $\mathtt{probelt}(x_{2,1}, y_{2,1})$ in  $\C$.}\label{fig:ontology-path}
\end{figure}

We show that $\T_\Box, \{ A(a) \} \models \q$ iff $\mathfrak T$ tiles a $k \times k$-grid.
Here, we only prove $(\Rightarrow)$ and leave the converse direction to the reader.
For two sequences $w = R_1 \dots R_l$ and $w' = R'_1 \dots R'_l$ of binary predicate names, we write $w \sqsubseteq w'$  if $\T_\Box \models R_i(x,y) \to R'_i(x,y)$, for all $i$ $(1 \le i \le l)$.

Let $h$ be a homomorphism from $\q$ to $\C = \mathcal{C}_{\T_\Box, \{ A(a) \}}$ such that $h(x_0) = v_0$. Then $v_0 \in A^\C$ and \mbox{$h(x_{1,1}) = v_{1,1}  \in A^\C$} with $v_{1,1}$ of the form $v_0 w_{1,1}S$, for some $w_{1,1}$ that begins with $B$ but does not contain~$S$. Since $(v_0,v_0  \mathit{Sink})\notin S^{\C}$ and $|\mathtt{tile}|$ is even, it follows that there is a unique tile $t_{1,1}$ such that
\begin{itemize}
\item[--] $|w_{1,1}| = |\mathtt{tile}|$ and $w_{1,1} \sqsubseteq \mathtt{tile}_{t_{1,1}}$,
\item[--] the subquery of $\mathtt{segm}(x_0, x_{1,1})$ for the sequence $\mathtt{tile}_1 \dots \mathtt{tile}_{t_{1,1}-1}$ is mapped to the $\mathit{Sink}$ arrow at $v_0$ (i.e., forwards and backwards between $v_0$ and $v_0  \mathit{Sink}$);
\item[--] the subquery for the sequence $\mathtt{tile}_{t_{1,1}+1} \dots \mathtt{tile}_m$ is mapped to a $\mathit{Sink}$ arrow at $v_0 w_{1,1}$ (see Fig.~\ref{fig:ontology-path}).
\end{itemize}
Consider now any subquery $\mathtt{segm}(x_{i-1, 1}, x_{i,1})$, for $2 \leq i \leq k$. By the same argument, we obtain $h(x_{i,1}) = v_{i,1}$, for $v_{i,1}=v_{i-1,1} w_{i,1} S$, and $w_{i,1} \sqsubseteq \mathtt{tile}_{t_{i,1}}$, for a unique $1 \leq t_{i,1} \leq m$. Next, $h(x_{1,2}) = v_{1,2}$ for $v_{1,2} = v_{k,1} w_{1,2}\,S$ with $w_{1,2} \sqsubseteq \mathtt{tile}_{t_{1,2}}$ and, eventually, every subquery $\mathtt{segm}(x_{i-1, j}, x_{i,j})$, for $1 \leq i \leq k$ and $1 \leq j \leq k$, is mapped in such a way that
\begin{equation*}
h(x_{i,j}) = \begin{cases} v_{i-1,j} w_{i,j} S, & \text{ if } i \geq 2,\\ v_{k,j-1} w_{i,j} S, & \text{ if } i = 1, j \geq 2,\\ v_0 w_{1,1} S, & \text{ if } i = 1, j = 1. \end{cases}  \ \ \text{ with } \ \ w_{i,j} \sqsubseteq \mathtt{tile}_{t_{i,j}}, \text{ for a unique } t_{i,j},
\end{equation*}
We prove now that the tiles $\mathcal{S}_{t_{i,j}}$ placed at $(i,j)$ of the $k\times k$-grid form a tiling.

First, we show that $\mathit{left}(t_{i,j}) = \mathit{right}(t_{i-1,j})$, for $2 \leq i \leq k$ and $1\leq j\leq k$. As we observed above, $h(x_{i,j})$ is of the form $v\, w_{i-1,j} \, S \,w_{i,j} \, S$. Consider the subquery $\mathtt{probelt}(x_{i,j}, y_{i,j})$ and recall that
\begin{equation*}
w_{i,j} \sqsubseteq B\, \mathtt{cne}_{\mathit{right}(t_{i,j})} \, \mathtt{enc}_{\mathit{left}(t_{i,j})}  \, \mathtt{cne}_{\mathit{top}(t_{i,j})}  \, \mathtt{enc}_{\mathit{bottom}(t_{i,j})}\, E.
\end{equation*}
By the structure of $\T_\Box$, the subqueries $I^{|\mathtt{enc}| +|\mathtt{cne}| +2}(x_{i,j}, v_{i,j})$ and $\mathtt{spring}(v_{i,j}, u_{i,j})$ of $\mathtt{probelt}(x_{i,j}, y_{i,j})$
are mapped by $h$ in such a way (see Fig.~\ref{fig:ontology-path}) that
\begin{align*}
h(v_{i,j}) & \sqsubseteq v \, w_{i-1,j} \, S\,B \, \mathtt{cne}_{\mathit{right}(t_{i,j})} \, \mathtt{enc}_{\mathit{left}(t_{i,j})},\\
h(u_{i,j}) & \sqsubseteq v \, w_{i-1,j}\, S\, B \, \mathtt{cne}_{\mathit{right}(t_{i,j})} \, P^{2c}, && \text{ for } c = \mathit{left}(t_{i,j}).
\end{align*}
On the other hand,  the last element in $\mathtt{probelt}$ is $M^-$, and so we must have
\begin{equation*}
h(y_{i,j}) \sqsubseteq v\, B \, N^{2c'}, \qquad \text{ for } c' = \mathit{right}(t_{i-1,j}),
\end{equation*}
which is only possible if $\mathit{right}(t_{i-1,j}) =  \mathit{left}(t_{i,j})$; see Fig.~\ref{fig:ontology-path}.

That $\mathit{down}(t_{i,j}) = \mathit{up}(t_{i,j-1})$, for $2 \leq j \leq k$ and $1\leq i \leq k$, is proved similarly by considering the mapping of the subquery $\mathtt{probedn}(x_{i,j}, z_{i,j})$.
\end{proof}

\section{Experiments}

\subsection{Computing rewritings}

We computed 6 types of rewritings for linear
queries similar to those in Example~\ref{ex:rewriting:1}
and a fixed ontology from Example~\ref{ex:rewriting:2}.
The first three rewritings were obtained by running executables of  Rapid~\cite{DBLP:conf/cade/ChortarasTS11}, Clipper~\cite{DBLP:conf/aaai/EiterOSTX12} and Presto~\cite{DBLP:conf/kr/RosatiA10}
with a 15 minute timeout on a desktop machine.
The other three rewritings are rewritings \textsc{Lin}, \textsc{Log} and \textsc{Tw}
described in Sections~\ref{sec:5}, \ref{sec:boundedtw} and \ref{sec:boundedleaf} respectively.

We considered the following three sequences: 
\begin{align*}
\tag{Sequence 1}&RRSRSRSRRSRRSSR,\\
\tag{Sequence 2} &SRRRRRSRSRRRRRR,\\
\tag{Sequence 3} &SRRSSRSRSRRSRRS.
\end{align*}
For each of the three sequences, we consider the line-shaped
queries with 1--15 atoms formed by their prefixes. Table~\ref{tab:size} presents the sizes of the different types of rewritings.

\begin{table*}[t]
\caption{The size (number of clauses) of different types of rewritings for the three sequences of queries ( -- indicates timeout after 15 minutes)}   \centering\tabcolsep=4pt
\scalebox{0.80}{
    \begin{tabular}{|c|r|r|r|r|r|r|r|r|r|r|r|r|r|r|r|r|r|r|}
    \hline
    no. & \multicolumn{6}{|c|}{Sequence~1} & \multicolumn{6}{|c|}{Sequence~2} & \multicolumn{6}{|c|}{Sequence~3} \\
    of & \multicolumn{6}{|c|}{\footnotesize\!$RRSRSRSRRSRRSSR$\!} & \multicolumn{6}{|c|}{\footnotesize\!$SRRRRRSRSRRRRRR$\!} & \multicolumn{6}{|c|}{\footnotesize\!$SRRSSRSRSRRSRRS$\!} \\
    \cline{2-19}
    \!atoms\! & \!Rapid\! & \!\!Clipper\!\! & \!Presto\! & \textsc{Lin}   & \textsc{Log} & \textsc{Tw} & Rapid & \!\!Clipper\!\! &
	\!Presto\! & \textsc{Lin}   & \textsc{Log} & \textsc{Tw}
	& Rapid & \!\!Clipper\!\! & \!Presto\! & \textsc{Lin}   & \textsc{Log} & \textsc{Tw}    \\
    \hline
    1 & 1 & 1          &   5&2 & 1 & 1& 1  & 1 &   5 &2 & 1  & 1& 1  & 1 &   5& 2  & 1  & 1\\
    2 & 1 & 1          &   5&5 & 2 & 0& 2  & 2 &  14 &5 & 4  & 2& 2  & 2 &  14& 5  & 4  & 2\\
    3 & 2 & 2          &  14&8 & 5 & 3& 2  & 2 &  14 &8 & 5  & 3& 2  & 2 &  14& 8  & 5  & 3\\
    4 & 3  & 3         &  19&11& 8 & 4& 2  & 2 &  14 &11& 6  & 3& 4  & 4 &  23& 11 & 8  & 5\\
    5 & 5 & 5          &  24&14& 12& 6& 2  & 2 &  14 &14& 8  & 4& 4  & 4 &  23& 14 & 10 & 6\\
    6 & 7  & 7         &  33&17& 16&10& 2  & 2 &  14 &17& 10 & 4& 8  & 8 &  39& 17 & 15 & 7\\
    7 & 10 & 11        &  49&20& 20&10& 4  & 4 &  23 &20& 13 & 7& 11 & 11&  57& 20 & 18 &14\\
    8 & 13 & 16        &  77&23& 24&14& 6  & 7 &  29 &23& 16 & 7& 18 & 24&  96& 23 & 21 & 8\\
    9 & 13 & 16        &  77&26& 27&15& 10 & 13&  50 &26& 22 &10& 24 & 35& 183& 26 & 27 &10\\
    10    & 26 & 44    & 203&29& 32&16& 14 & 26&  83 &29& 27 &11& 34 & 63& 356& 29 & 33 &17\\
    11    & 39 & 72    & 329&32& 36&16& 14 & 26&  83 &32& 29 &14& 43 &100& 356& 32 & 37 &20\\
    12    & 39 & 126   & 329&35& 40&21& 14 & 26&  83 &35& 33 &18& 56 &302&1028& 35 & 42 &23\\
    13    & -- & 241   & 959&38& 45&24& -- & 30&  83 &38& 35 &20& -- & --&1712& 38 & 46 &25\\
    14    & -- & --    & 959&41& 47&25& -- & 31&  83 &41& 36 &16& -- & --&1712& 41 & 51 &27\\
    15    & -- & -- &   2723&44& 51&22& -- & 30&  83 &44& 37 &15& -- & --&5108& 44 & 52 &29\\
    \hline
    \end{tabular}%
	}
\label{tab:size}%
\end{table*}%

\begin{table*}
  \centering\tabcolsep=6pt
  \caption{Generated datasets}
    \begin{tabular}{|c|r|r|r|r|r|}
    \hline
    dataset & \multicolumn{1}{|c|}{$V$} &\multicolumn{1}{|c|}{$p$} & \multicolumn{1}{|c|}{$q$} & {\tabcolsep=0pt\begin{tabular}{c}avg. degree\\[-2pt] of vertices\end{tabular}} & no.\ of atoms\\
    \hline
    1.ttl & 1\,000  & 0.050  & 0.050  & 50 &61\,498\\
    2.ttl & 5\,000  & 0.002 & 0.004 & 10 &64\,157\\
    3.ttl & 10\,000 & 0.002 & 0.004 & 20 &256\,804\\
    4.ttl & 20\,000 & 0.002 & 0.010  & 40 &1\,027\,028\\
    \hline
    \end{tabular}%
  \label{tab:datasets}%
\end{table*}%

\pagebreak

\subsection{Datasets}\label{sec:datasets}

We used  Erd\"os-R\`enyi random graphs with independent
parameters $V$ (number of vertices), $p$ (probability of an $R$-edge)
and $q$ (probability of concepts $A$ and $B$ at a given vertex).
Note that we intentionally did not introduce any $S$-edges.
The last parameter, the average degree of a vertex,
is $V\cdot p$. Table~\ref{tab:datasets} summarises the parameters of the datasets.


\subsection{Evaluating rewritings}

We evaluated all obtained rewritings
on the datasets in Section~\ref{sec:datasets} using RDFox triplestore~\cite{DBLP:conf/semweb/NenovPMHWB15} with
999-second timeout. The materialisation time and other relevant statistics are given in
Tables~\ref{tab:rdfox}, \ref{tab:rdfoxtwo}, and \ref{tab:rdfoxthree}.

\begin{table*}
  \centering
  \caption{Evaluating rewritings on RDFox - 1}
  \scalebox{0.60}{
    \begin{tabular}{|c|r|r|r|r|r|r|r|r|r|r|r|r|r|r|r|r|}
    \hline
    data- & query  & \multicolumn{7}{c|}{evaluation time (sec)} & \multicolumn{1}{|c|}{no.\ of}
 & \multicolumn{7}{c|}{no.\ of generated tuples} \\\cline{3-9}\cline{11-17}
set   & \multicolumn{1}{|c|}{size} & \multicolumn{1}{|c|}{Rapid} & \multicolumn{1}{|c|}{Clipper} &\multicolumn{1}{|c|}{Presto} & \multicolumn{1}{|c|}{\textsc{Lin}} &  \multicolumn{1}{|c|}{\textsc{Log}} &\multicolumn{1}{|c|}{\textsc{Tw}}&\multicolumn{1}{|c|}{\textsc{Tw*}} & \multicolumn{1}{|c|}{answers} & \multicolumn{1}{|c|}{Rapid} & \multicolumn{1}{|c|}{Clipper} &\multicolumn{1}{|c|}{Presto} & \multicolumn{1}{|c|}{\textsc{Lin}} & \multicolumn{1}{|c|}{\textsc{Log}} & \multicolumn{1}{|c|}{\textsc{Tw}} &\multicolumn{1}{|c|}{\textsc{Tw*}} \\
\hline
      & 1   & 0.021 & 0.019 & 0.034 & 0.049 & 0.017 & 0.016 & 0.01  & 61390   &  61390 & 61390 & 122780 & 61449 & 61390 & 61390 & 61390 \\
      & 2   & 0.675 & 0.694 & 0.706 & 0.898 & 0.505 & 0.652 & 0.698 & 976789  &  976789 & 976789 & 1038179 & 1041822 & 1038179 & 976789 & 976789 \\
      & 3   & 0.058 & 0.053 & 0.125 & 0.013 & 0.112 & 0.01 & 0.012  & 2956    &  2956 & 2956 & 64394 & 3054 & 64394 & 3004 & 3004 \\
      & 4   & 0.204 & 0.201 & 0.314 & 0.087 & 0.675 & 0.76 & 0.12   & 212213  &  212213 & 212213 & 273710 & 283409 & 1314797 & 1189061 & 212272 \\
      & 5   & 0.12 & 0.114 & 0.314 & 0.014 & 0.576 & 0.696 & 0.064  & 2956    &  2956 & 2956 & 64453 & 3150 & 1105636 & 976837 & 3004 \\
      & 6   & 0.266 & 0.248 & 0.685 & 0.093 & 0.266 & 0.768 & 0.124 & 212213  &  212213 & 212213 & 273710 & 292815 & 337479 & 1198455 & 218710 \\
      & 7   & 0.271     & 0.242 &  1.11     & 0.008    & 0.243 &0.687 & 0.05& 2\,956       & 2\,956  & 2\,956     &64453    & 3\,246        & 125\,361    &982797  & 3148  \\
      & 8   & 0.412     & 0.377 &  1.406    & 0.084    & 0.904 &0.944 & 0.186& 212\,213     & 212\,213& 212\,213   &273710   & 302\,221      & 1\,659\,409&1410727 & 431100    \\
      & 9   & 3.117     & 3.337 &  12.713   & 3.376    & 2.941 &2.405 & 1.633& 998\,945     & 998\,945& 998\,945   &1060442  & 2\,927\,979   & 2\,684\,359&2435551 & 1455913    \\
1.ttl & 10  & 1.079     & 1.102 &  18.432   & 0.012    & 0.607 &0.76  & 0.166& 8\,374       & 8\,374  & 10\,760    &69871    & 12\,573       & 1\,178\,714&1203649 & 224057    \\
      & 11  & 2.246     & 1.984 &  48.311   & 0.385    & 0.945 &1.075 & 0.371& 436\,000     & 436\,000& 436\,000   &497497   & 836\,876      & 1\,618\,743&1663534 & 664174    \\
      & 12  & 13.693    & 30.032&  >999     & 8.129    & 6.867 &5.922 & 5.28& 999\,998     & 999\,998& 1\,000\,000&--       & 5\,311\,314   & 4\,439\,352 &3217262 & 2241208    \\
      & 13  & --        & 6.810 &  560.206  & 0.027    & 0.616 &0.946 & 0.274& 20\,985      & --      & 24\,839    &82482    & 38\,200       & 553\,821   &1234421 & 254888  \\
      & 14  & --        & --    &  913.387  & 0.013    & 0.358 &0.819 & 0.27& 0            & --      & --         &61497    & 48        	 & 312\,723   &1201459 & 228307 \\
      & 15  & --        & --    &  >999     & 0.032    & 0.394 &0.994 & 0.33& 2\,000       & --      & --         &--       & 70\,277       & 376\,313    &1417786 & 442579  \\
\hline
      & 1   & 0.02 & 0.022 & 0.039 & 0.02 & 0.019 & 0.017 & 0.008   &64103  & 64103 & 64103 & 128206 & 64125 & 64103 & 64103 & 64103 \\
      & 2   & 0.273 & 0.305 & 0.321 & 0.29 & 0.297 & 0.275 & 0.466  &809731 & 809731 & 809731 & 873834 & 874112 & 873834 & 809731 & 809731 \\
      & 3   & 0.03 & 0.028 & 0.06 & 0.011 & 0.058 & 0.01 & 0.013    &427    & 427 & 427 & 64561 & 489 & 64561 & 458 & 458 \\
      & 4   & 0.057 & 0.054 & 0.103 & 0.032 & 0.448 & 0.315 & 0.035 &8778   & 8778 & 8778 & 72934 & 74004 & 947164 & 818531 & 8800 \\
      & 5   & 0.05 & 0.046 & 0.128 & 0.014 & 0.423 & 0.301 & 0.03   &427    & 427 & 427 & 64583 & 551 & 938875 & 809762 & 458 \\
      & 6   & 0.08 & 0.074 & 0.27 & 0.035 & 0.084 & 0.316 & 0.038   &8778   & 8778 & 8778 & 72934 & 75103 & 77253 & 819648 & 9490 \\
      & 7   & 0.089     & 0.080  & 0.378  & 0.008      & 0.078 &0.295 & 0.024& 427          & 427         & 427       &64583  & 613           & 68\,546    &810647 & 551\\
      & 8   & 0.136     & 0.125  & 0.467  & 0.029      & 0.434 &0.322 & 0.037& 8\,778       & 8\,778      & 8\,778    &72934  & 76\,202       & 1\,085\,362&828448 & 18334\\
      & 9   & 0.202     & 0.254  & 1.179  & 0.369      & 0.554 &0.391 & 0.102& 105\,853     & 105\,853    & 105\,853  &170009 & 1\,020\,363   & 1\,190\,249&933295 & 123190\\
2.ttl & 10  & 0.174     & 0.204  & 2.341  & 0.011      & 0.461 &0.321 & 0.052& 11           & 11          & 438       &64167  & 506           & 943\,097   &819428 & 9354\\
      & 11  & 0.192     & 0.259  & 4.726  & 0.036      & 0.473 &0.336 & 0.053& 651          & 651         & 9\,396    &64807  & 74\,922       & 944\,210   &820354 & 11271\\
      & 12  & 0.244     & 0.699  & 24.778 & 0.396      & 1.034 &0.509 & 0.15& 8\,058       & 8\,058      & 113\,179  &72214  & 1\,004\,735   & 1\,940\,300 &934269 & 124420\\
      & 13  & --        & 0.629  & 20.555 & 0.015      & 0.244 &0.458 & 0.084& 0            & --          & 438       &64156  & 502           & 209\,915   &820373 & 10321\\
      & 14  & --        & --     & 25.243 & 0.014      & 0.153 &0.350 & 0.081& 0            & --          & --        &64156  & 31            & 200\,962   &820106 & 10722\\
      & 15  & --        & --     & 66.916 & 0.032      & 0.172 &0.335 & 0.072& 0            & --          & --        &64156  & 64\,543       & 265\,087   &828884 & 19522\\
\hline
      & 1   & 0.131 & 0.094 & 0.225 & 0.101 & 0.096 & 0.14 & 0.032 & 256699 & 256699 & 256699 & 513398 & 256756 & 256699 & 256699 & 256699 \\
      & 2   & 2.933 & 2.946 & 3.017 & 2.955 & 3.053 & 2.929 & 3.039& 6379932& 6379932 & 6379932 & 6636631 & 6638150 & 6636631 & 6379932 & 6379932\\
      & 3   & 0.206 & 0.175 & 0.519 & 0.03 & 0.499 & 0.029 & 0.034 & 1217   & 1217 & 1217 & 257963 & 1311 & 257963 & 1264 & 1264 \\
      & 4   & 0.399 & 0.424 & 0.927 & 0.171 & 4.003 & 3.419 & 0.231& 67022  & 67022 & 67022 & 323825 & 327716 & 6961626 & 6447011 & 67079 \\
      & 5   & 0.36 & 0.357 & 1.112 & 0.036 & 4.133 & 3.396 & 0.179 & 1217   & 1217 & 1217 & 258020 & 1405 & 6895915 & 6379979 & 1264 \\
      & 6   & 0.632 & 0.57 & 1.806 & 0.169 & 0.836 & 3.425 & 0.228 & 67022  & 67022 & 67022 & 323825 & 331647 & 363640 & 6450931 & 69782 \\
      & 7   & 0.631     & 0.581 &  2.981   & 0.035     & 0.756&3.255 & 0.156& 1\,217       & 1\,217      & 1\,217      &258020 & 1\,499        & 296\,711     &6382460 & 1405   \\
      & 8   & 0.925     & 0.876 &  3.739   & 0.159     & 4.377&3.405 & 0.278& 67\,022      & 67\,022     & 67\,022     &323825 & 335\,578      & 7\,546\,184  &6518010 & 136975   \\
      & 9   & 1.949     & 2.275 &  14.564  & 4.063     & 5.251&4.169 & 1.169& 1\,678\,668  & 1\,678\,668 & 1\,678\,668 &1935471 & 8\,613\,829   & 9\,225\,201 &8196944 & 1815899   \\
3.ttl & 10  & 1.24      & 1.377 &  35.109  & 0.049     & 4.731&3.571 & 0.342& 60           & 60          & 1\,277      &256863 & 1\,389        & 6\,936\,178  &6449555 & 68557   \\
      & 11  & 1.403     & 1.798 &  60.858  & 0.249     & 4.846&3.607 & 0.343& 11\,498      & 11\,498     & 77\,811     &268301 & 341\,459      & 6\,949\,160  &6462905 & 85267   \\
      & 12  & 1.697     & 5.413 &  572.53  & 4.355     &10.128&6.693 & 1.645& 305\,640     & 305\,640    & 1\,951\,654&562443 & 8\,780\,232   & 15\,626\,926  &8438115 & 2058532    \\
      & 13  & --        & 4.382 &  484.969 & 0.082     & 1.762&4.926 & 0.599& 0            & --    	   & 1\,277       &256803 & 1\,377     	 & 917\,117       &6453717 & 72776 \\
      & 14  & --        & --    &  575.487 & 0.063     & 1.115&3.972 & 0.584& 0            & --          & --         &256803 & 47      	     & 850\,309   &6452195 & 73900   \\
      & 15  & --        & --    &  >999    & 0.177     & 1.011&3.585 & 0.501& 0            & --    	   & --           &--    & 257\,974      & 1\,107\,065    &6519217 & 140979 \\
\hline
      & 1   & 0.433 & 0.451 & 1.037 & 0.495 & 0.439 & 0.456 & 0.165       & 1026526 & 1026526 & 1026526 & 2053052 & 1026774 & 1026526 & 1026526 &1026526 \\
      & 2   & 27.549 & 28.088 & 28.329 & 27.011 & 29.532 & 32.331 & 31.34 & 49364886& 49364886 & 49364886 & 50391412 & 50404311 & 50391412 & 49364886 & 49364886\\
      & 3   & 2.067 & 2.409 & 3.657 & 0.159 & 4.087 & 0.161 & 0.162       & 13103   & 13103 & 13103 & 1039882 & 13613 & 1039882 & 13356 & 13356 \\
      & 4   & 4.866 & 5.438 & 9.511 & 1.37 & 38.919 & 31.188 & 2.746      & 1286991 & 1286991 & 1286991 & 2314018 & 2353661 & 52718280 & 50652125 & 1287239 \\
      & 5   & 4.061 & 4.032 & 10.374 & 0.209 & 42.943 & 33.064 & 2.142    & 13103   & 13103 & 13103 & 1040130 & 14119 & 51444898 & 49365139 & 13356 \\
      & 6   & 6.909 & 7.133 & 16.249 & 1.443 & 7.767 & 36.268 & 2.782     & 1286991 & 1286991 & 1286991 & 2314018 & 2393145 & 2952225 & 50691250 & 1313261 \\
      & 7   & 6.614     & 6.277   & 23.7 & 0.243      & 8.586  &29.098& 2.02 & 13\,103      & 13\,103     & 13\,103        &1040130  & 14\,625        & 1\,665\,376  &49391598  & 14115\\
      & 8   & 11.441    & 10.923  & 29.1 & 1.880      & 54.813 &29.426& 3.669& 1\,286\,991  & 1\,286\,991 & 1\,286\,991    &2314018  & 2\,432\,629    & 56\,098\,445 &51978489  & 2600996\\
      & 9   & 46.704    & 50.668  & 193  & 76.169     & 102.055&66.464& 33.63& 58\,753\,514 & 58\,753\,514& 58\,753\,514  &59780541  & 114\,973\,160  & 114\,837\,395&110717131 & 61339643\\
4.ttl & 10  & 14.348    & 15.503  & 462  & 0.375      & 43.347 &30.008& 4.694& 19\,966      & 19\,966     & 33\,014        &1046993  & 35\,359        & 52\,103\,362 &50698955  & 1321716\\
      & 11  & 19.593    & 20.907  & 821  & 2.843      & 44.410 &31.061& 5.319& 1\,872\,159  & 1\,872\,159 & 3\,051\,184    &2899186  & 4\,397\,556    & 53\,986\,724 &52602849  & 3224788\\
      & 12  & 71.354    & 182.499 & >999 & 172.822    & 237.478&179.12& 90.04& 79\,939\,048 & 79\,939\,048& 120\,229\,590 &--     & 199\,083\,489& 242\,500\,074     &189429768 & 140064931\\
      & 13  & --        & 54.497  & >999 & 0.562      & 22.345 &44.427& 7.105& 22\,474      & --          & 53\,717        &--     & 58\,826        & 5\,686\,759    &50759705  & 1382714\\
      & 14  & --        & --      & >999 & 0.550      & 12.462 &36.259& 7.493& 0            & --          & --            &--     & 253            & 4\,356\,739     &50704606  & 1353393\\
      & 15  & --        & --      & >999 & 1.211      & 11.315 &30.709& 7.028& 12\,165      & --          & --             &--     & 1\,064\,542    & 5\,395\,902    &52014512  & 2652797\\
    \hline
    \end{tabular}%
	}
  \label{tab:rdfox}%
\end{table*}%


\begin{table*}
  \centering
  \caption{Evaluating rewritings on RDFox - 2}
  \scalebox{0.60}{
    \begin{tabular}{|c|r|r|r|r|r|r|r|r|r|r|r|r|r|r|r|r|}
    \hline
    data- & query  & \multicolumn{7}{c|}{evaluation time (sec)} & \multicolumn{1}{|c|}{no.\ of}
 & \multicolumn{7}{c|}{no.\ of generated tuples} \\\cline{3-9}\cline{11-17}
set   & \multicolumn{1}{|c|}{size} & \multicolumn{1}{|c|}{Rapid} & \multicolumn{1}{|c|}{Clipper} &\multicolumn{1}{|c|}{Presto} & \multicolumn{1}{|c|}{\textsc{Lin}} &  \multicolumn{1}{|c|}{\textsc{Log}} &\multicolumn{1}{|c|}{\textsc{Tw}} &\multicolumn{1}{|c|}{\textsc{Tw*}}& \multicolumn{1}{|c|}{answers} & \multicolumn{1}{|c|}{Rapid} & \multicolumn{1}{|c|}{Clipper} &\multicolumn{1}{|c|}{Presto} & \multicolumn{1}{|c|}{\textsc{Lin}} & \multicolumn{1}{|c|}{\textsc{Log}} & \multicolumn{1}{|c|}{\textsc{Tw}}&\multicolumn{1}{|c|}{\textsc{Tw*}} \\
\hline
      & 1   & 0.009 & 0.005 & 0.005 & 0.005 & 0.005 & 0.005 & 0.007 & 0  & 0 & 0 & 0 & 48 & 0 & 0 & 0 \\
      & 2   & 0.009 & 0.008 & 0.021 & 0.05 & 0.012 & 0.008 & 0.007 & 59  & 59 & 59 & 61508 & 64406 & 118 & 59 & 59 \\
      & 3   & 0.083 & 0.058 & 0.077 & 0.9 & 0.093 & 0.732 & 0.058 & 3584  & 3584 & 3584 & 65033 & 1092161 & 65033 & 980373 & 3584 \\
      & 4   & 2.363 & 4.049 & 2.301 & 8.32 & 0.11 & 0.723 & 0.073 & 57571  & 57571 & 57571 & 119020 & 2204964 & 119079 & 1034419 & 57630 \\
      & 5   & 97    & 92   &102   & 13.599   & 2    &14.272& 2.718& 59000     & 59000 & 59000 & 120449 & 3265393 & 1097297 &2035848 & 59059\\
      & 6   & >999  & >999 &>999  & 17.882   & 19   &13.881& 42.914& 59000     & --    & --    & --    & 4324393 & 1162212 &2039373 & 62584 \\
      & 7   & 129   & 122  &>999  & 0.384    & 0.25 &0.749& 0.344& 2832      & 2832  & 2832 & --    & 156824 & 132259      &1030122 & 6464    \\
      & 8   & >999  & >999 &>999  & 10.963   & 2    &1.82& 21.399& 55991     & --    & --    & --    & 3352724 & 304347    &1302623 & 268322  \\
      & 9   & 162   & 158  &>999  & 0.395    & 0.21 &0.722& 0.344& 2832      & 2832  & 2832 & --    & 156920 & 187155      &1040255 & 5895    \\
1.ttl & 10  & >999  & >999 &>999  & 11.118   & 2    &12.928& 39.21& 55991     & --    & --    & --    & 3362130 & 1220806  &2104667 & 68937 \\
      & 11  & >999  & >999 &>999  & 20.217   & 4    &14.611&  >999 & 59000     & --    & --    & --    & 5920653 & 2251570 &2342243 & --    \\
      & 12  & >999  & >999 &>999  & 31.648   & 21   &19.079&  >999 & 59000     & --    & --    & --    & 8714382 & 3361965 &4165789 & --    \\
      & 13  & --    & >999 &>999  & 34.395   & 46   &193.512&  >999 & 59000     & --    & --    & --    & 9783393 &3429574 &4198870 & --    \\
      & 14  & --    & >999 &>999  & 39.818   & 223  &190.334&  >999 & 59000     & --    & --    & --    & 10842393&1509563 &4130571 & --   \\
      & 15  & --    & >999 &>999  & 49.391   & 232  &226.827&  >999 & 59000     & --    & --    & --    & 11901393&1594164 &4420495 & --   \\
\hline
      & 1   & 0.007 & 0.007 & 0.005 & 0.007 & 0.007 & 0.005 & 0.004 & 0  & 0 & 0 & 0 & 31 & 0 & 0 & 0 \\
      & 2   & 0.01 & 0.01 & 0.028 & 0.027 & 0.011 & 0.008 & 0.008 & 22 &  22 & 22 & 64147 & 64543 & 44 & 22 & 22 \\
      & 3   & 0.025 & 0.025 & 0.041 & 0.345 & 0.046 & 0.313 & 0.024 & 256  & 256 & 256 & 64381 & 879372 & 64381 & 809987 & 256 \\
      & 4   & 0.135 & 0.136 & 0.169 & 4.798 & 0.055 & 0.297 & 0.023 & 3300  & 3300 & 3300 & 67425 & 9329702 & 67447 & 813053 & 3322 \\
      & 5   & 1.314    & 1.278  &1.824 & 39.195  & 0.513 & 4.714& 0.122& 34474  & 34474 & 34474 & 98599 & 33935400 & 908352    & 9240858 & 34496   \\
      & 6   & 13.597   & 13.652 &19.52 & 119.212 & 0.698 & 4.606& 0.178& 106742 & 106742& 106742 & 170867 & 59117304 & 1044957 & 9313360 & 106998  \\
      & 7   & 1.396    & 1.34   &18.91 & 0.116   & 0.102 & 0.326& 0.028& 248    & 248   & 248 & 64404 & 214761 & 129190        & 815625 & 535       \\
      & 8   & 1.572    & 1.987  &20.58 & 2.518   & 0.095 & 0.364& 0.069& 3478   & 3478  & 3478 & 67634 & 2968573 & 199843      & 825309 & 12300     \\
      & 9   & 1.397    & 1.554  &35.15 & 0.118   & 0.076 & 0.333& 0.033& 248    & 248   & 248 & 64404 & 214823 & 132187        & 813759 & 728       \\
2.ttl & 10  & 1.636    & 2.634  &233   & 2.591   & 0.639 & 4.45 & 0.069& 3478   & 3478  & 3478 & 67634 & 2969672 & 976875      & 9245685 & 4871     \\
      & 11  & 1.677    & 12.024 &895   & 30.575  & 0.98  & 4.434& 0.66& 35382  & 35382 & 35382 & 99538 & 26328037 & 1823608    & 9285127 & 44313  \\
      & 12  & 2.009    & 143    &>999  & 128.532 & 1.756 & 5.666& 7.999& 106895 & 106895& 106895 & --    & 71017728 & 2184441  & 10358119 & 1010563 \\
      & 13  & --       & >999   &>999  & 243.656 & 2.559 &47.098& 5.121& 110000 & --    & --    & --    & 115653199 & 2742932  & 34483363 & 145486 \\
      & 14  & --       & >999   &>999  & 325.755 & 2.866 &50.997& 12.028& 110000 & --    & --    & --    & 151038934 & 1448087 & 35282112 & 111224 \\
      & 15  & --       & >999   &>999  & 433.438 & 26.903&54.518& 133.512& 110000 & --    & --    & --    & 176515562 & 9102348& 35442252 & 118515 \\
\hline
      & 1   & 0.009 & 0.01 & 0.009 & 0.011 & 0.009 & 0.011 & 0.009 & 0 &  0 & 0 & 0 & 47 & 0 & 0 & 0 \\
      & 2   & 0.023 & 0.02 & 0.115 & 0.145 & 0.022 & 0.019 & 0.019 & 57 &  57 & 57 & 256813 & 257974 & 114 & 57 & 57 \\
      & 3   & 0.123 & 0.127 & 0.249 & 3.364 & 0.315 & 3.212 & 0.136 & 1462  & 1462 & 1462 & 258218 & 6668549 & 258218 & 6381394 & 1462 \\
      & 4   & 1.992 & 1.93 & 3.072 & 85.844 & 0.345 & 3.21 & 0.122 & 36260  & 36260 & 36260 & 293016 & 86686553 & 293073 & 6416249 & 36317 \\
      & 5   & 47    & 56   &76.8 & 967    & 7.09 &70.117 & 1.898& 452502 & 452502   & 452502 & 709258 & 187656175 & 7089247 & 86439255 & 452559  \\
      & 6   & >999  & >999 &>999 & >999   & 9.996&73.99  & 3.965& 570000 & -        & --    & --    & --    & 7464849       & 86558158 & 571462     \\
      & 7   & 47    & 51   &>999 & 1.591  & 0.736&3.47   & 0.181& 2125   & 2125      & 2125 & --    & 883690 & 518306       & 6413768  & 3634      \\
      & 8   & 77    & 99   &>999 & 60.365 & 0.667&3.601  & 1.842& 53191  & 53191     & 53191 & --    & 22657990 & 862422    & 6536462  & 120327   \\
      & 9   & 50    & 56   &>999 & 1.885  & 0.473&3.496  & 0.223& 2125   & 2125      & 2125 & --    & 883784 & 553583       & 6419638  & 3446      \\
3.ttl & 10  & 79    & 142  &>999 & 59.019 & 7.999&67.145 & 2.083& 53191  & 53191    & 53191 & --    & 22661921 & 7401781    & 86497805 & 58664  \\
      & 11  & 81    & >999 &>999 & >999   &10.862&68.956 & 50.812& 516631 & 516631& --    & --    & --    & 14275796        & 87027128 & 587987    \\
      & 12  & 116   & >999 &>999 & >999   &26.218&112.098& 306.304& 570000 & 570000& --    & --    & --    & 16280643       & 95308112 & 8298971    \\
      & 13  & --    & >999 &>999 & >999   & 45.19&>999   & 785.247& 570000 & --    & --    & --    & --    & 27255415       & --        & 1026838    \\
      & 14  & --    & >999 &>999 & >999   &74.691&>999   & >999   & 570000 & --    & --    & --    & --    & 9092721        & --        & --        \\
      & 15  & --    & >999 &>999 & >999   &>999  &>999   & >999   & --     & --    & --    & --    & --    & --             & --        & --          \\
\hline
      & 1   & 0.026 & 0.027 & 0.027 & 0.035 & 0.026 & 0.047 & 0.029 & 0 &  0 & 0 & 0 & 253 & 0 & 0 & 0 \\
      & 2   & 0.068 & 0.067 & 0.5 & 0.543 & 0.078 & 0.069 & 0.07 & 248 &  248 & 248 & 1027022 & 1040241 & 496 & 248 & 248 \\
      & 3   & 0.992 & 0.99 & 1.483 & 33.62 & 1.98 & 30.768 & 0.976 & 12651  & 12651 & 12651 & 1039425 & 51050537 & 1039425 & 49377537 & 12651 \\
      & 4   & 60.836 & 69.126 & 65.671 & M   & 2.175 & 30.532 & 1.272 & 609193 & 609193 & 609193 & 1635967 & --      & 1636215 & 49974327 & 609441 \\
      & 5   & >999    & >999 &>999   & >999   & 85  &>999  & 60.335  &4947136 & --& --    & --    & --    & 55339044   & --       & 4947384  \\
      & 6   & >999    & >999 &>999   & >999   & 287 &>999  & 261.562  &4960000 & --& --    & --    & --    & 56390837   & --       & 4972651  \\
      & 7   & >999    & >999 &>999   & 63     & 5   &31.839& 3.118  &62572   & --& --    & -- & 10949093 & 2141879    & 50070886 & 75476 \\
      & 8   & >999    & >999 &>999   & >999   & 13  &37.121& 273.336  &2435666 & --& --    & --    & --    & 6151203    & 53696984 & 3723153 \\
      & 9   & >999    & >999 &>999   & 61     & 5   &31.899& 5.725  &62572   & --& --    & -- & 10949599 & 2739031    & 50050255 & 76176 \\
4.ttl & 10  & >999    & >999 &>999   & >999   & 131 &>999  & 319.902  &2435666 & --& --    & --    & --    & 58829172   & --       & 2487953  \\
      & 11  & >999    & >999 &>999   & M      & 214 &>999 & --   & 4960000 & --& --    & --    & --    & 111363802  & --       &   --   \\
      & 12  & >999    & >999 &>999   & M      & >999&>999 & --   & --     & -- & --    & --    & --    & --         & --       &   --   \\
      & 13  & --      & >999 &>999   & M      & >999&>999 & --   & --     & -- & --    & --    & --    & --         & --       &   --   \\
      & 14  & --      & >999 &>999   & M      & >999&>999 & --   & --     & -- & --    & --    & --    & --         & --       &   --   \\
      & 15  & --      & >999 &>999   & M      & >999&>999 & --   & --     & -- & --    & --    & --    & --         & --       &   --   \\
    \hline
    \end{tabular}%
	}
  \label{tab:rdfoxtwo}%
\end{table*}%

\begin{table*}
  \centering
  \caption{Evaluating rewritings on RDFox - 3}
  \scalebox{0.65}{
    \begin{tabular}{|c|r|r|r|r|r|r|r|r|r|r|r|r|r|r|r|r|}
    \hline
    data- & query  & \multicolumn{7}{c|}{evaluation time (sec)} & \multicolumn{1}{|c|}{no.\ of}
 & \multicolumn{7}{c|}{no.\ of generated tuples} \\\cline{3-9}\cline{11-17}
set   & \multicolumn{1}{|c|}{size} & \multicolumn{1}{|c|}{Rapid} & \multicolumn{1}{|c|}{Clipper}& \multicolumn{1}{|c|}{Presto} & \multicolumn{1}{|c|}{\textsc{Lin}} &  \multicolumn{1}{|c|}{\textsc{Log}} &\multicolumn{1}{|c|}{\textsc{Tw}} & \multicolumn{1}{|c|}{\textsc{Tw*}} & \multicolumn{1}{|c|}{answers} & \multicolumn{1}{|c|}{Rapid} & \multicolumn{1}{|c|}{Clipper} & \multicolumn{1}{|c|}{Presto}& \multicolumn{1}{|c|}{\textsc{Lin}} & \multicolumn{1}{|c|}{\textsc{Log}} & \multicolumn{1}{|c|}{\textsc{Tw}} & \multicolumn{1}{|c|}{\textsc{Tw*}} \\
\hline
&1 & 0.004 & 0.003 & 0.003 & 0.004 & 0.003 &         0.021&0.003 & 0    & 0 & 0 & 0 & 48 & 0                               & 0     &0 \\
&2 & 0.006 & 0.006 & 0.017 & 0.022 & 0.008 &         0.014&0.005 & 59   & 59 & 59 & 61508 & 64406 & 118                    & 59    &59 \\
&3 & 0.053 & 0.06 & 0.065 & 0.849 & 0.087  &          0.69&0.053 & 3584 & 3584 & 3584 & 65033 & 1092161 & 65033           & 980373 &3584 \\
&4 & 0.012 & 0.01 & 0.074 & 0.01 & 0.009   &         0.008&0.008 & 2    & 2 & 2 & 61499 & 3176 & 168                      & 109    &109 \\
&5 & 0.011 & 0.009 & 0.07 & 0.008 & 0.009  &         0.008&0.009 & 0    & 0 & 0 & 61497 & 48 & 166                        & 59     &59 \\
&6 & 0.018 & 0.015 & 0.139 & 0.023 & 0.09  &         0.677&0.055 & 2    & 2 & 2 & 61499 & 64560 & 65203                  & 980434  &3704 \\
&7 & 0.017 & 0.015 & 0.145 & 0.01 & 0.087  &          0.68&0.057 & 0    & 0 & 0 & 61497 & 144 & 65190                     & 980480 &3691 \\
1.ttl&8 & 0.025 & 0.034 & 0.339 & 0.026    & 0.044 & 0.009&0.008 & 2    & 2 & 135 & 61499 & 73966 & 129565                & 170    &286 \\
&9 & 0.025 & 0.034 & 0.433 & 0.01 & 0.034  &         0.009&0.008 & 0    & 0 & 2 & 61497 & 240 & 65530                     & 109    &214 \\
&10& 0.035 & 0.086 & 0.549 & 0.029 & 0.026 &         0.015&0.015 & 2    & 2 & 135 & 61499 & 83372 & 67690                 & 12950  &13114 \\
&11& 0.034 & 0.086 & 4.445 & 1.164 & 0.54  &         0.765&0.221 & 133  & 0 & 2 & 61630 & 1684864 & 1095576             & 1227962  &251278 \\
&12& 0.048 & 0.223 & 4.877 & 0.013 & 0.137 &         0.699&0.062 & 2    & 2 & 135 & 61499 & 4082 & 192211               & 983694   &4115 \\
&13& -     & -     & 13.007 & 0.141 & 0.153&          0.79&0.175 & 133  & -     & -     & 61630 & 380205 & 226874         & 1228297&229396 \\
&14& -     & -     & 382.922 & 3.738 &0.878&         1.166&0.318 & 1967 & -       & -       & 63464 & 3842746 & 1282299  & 1809813 &270081 \\
&15& -     & -      & 307.184 & 0.017& 0.36&         0.771&0.224 & 11   & -      & -      & 61508 & 16542 & 242156      & 1228610  &252720 \\
\hline
&1 & 0.004 & 0.004 & 0.004 & 0.004 & 0.004 &       0.004 &0.003 & 0    & 0 & 0 & 0 & 31 & 0                                 & 0    &0 \\
&2 & 0.006 & 0.006 & 0.02 & 0.023 & 0.009  &        0.006&0.006 & 22   & 22 & 22 & 64147 & 64543 & 44                     & 22     &22 \\
&3 & 0.022 & 0.019 & 0.04 & 0.339 & 0.045  &        0.29 &0.019 & 256  & 256 & 256 & 64381 & 879372 & 64381               & 809987 &256 \\
&4 & 0.013 & 0.011 & 0.047 & 0.01 & 0.01   &        0.009&0.009 & 0    & 0 & 0 & 64156 & 490 & 75                        & 53      &53 \\
&5 & 0.012 & 0.011 & 0.044 & 0.008 & 0.009 &        0.01 &0.009 & 0    & 0 & 0 & 64156 & 31 & 75                          & 22     &22 \\
&6 & 0.02 & 0.016 & 0.081 & 0.027 & 0.042  &        0.304&0.021 & 0    & 0 & 0 & 64156 & 64543 & 64456               & 810009      &300 \\
2.ttl&7 & 0.018 & 0.015 & 0.094 & 0.011 & 0.041 & 0.297  &0.024 & 0    & 0 & 0 & 64156 & 93 & 64465                     & 810040   &309 \\
&8 & 0.025 & 0.036 & 0.182 & 0.027 & 0.053 &     0.01    &0.009 & 0    & 0 & 0 & 64156 & 65642 & 129037                  & 75      &119 \\
&9 & 0.026 & 0.037 & 0.215 & 0.012 & 0.03  &      0.009  &0.009 & 0    & 0 & 0 & 64156 & 155 & 64611                     & 53      &106 \\
&10& 0.038 & 0.091 & 0.327 & 0.028 & 0.029 &     0.014   &0.013 & 0    & 0 & 0 & 64156 & 66741 & 65120                 & 1393      &1468 \\
&11& 0.036 & 0.09 & 1.467 & 0.345 & 0.358  &      0.314  &0.055 & 0    & 0 & 0 & 64156 & 906286 & 879949                & 818896   &9218 \\
&12& 0.052 & 0.268 & 1.868 & 0.014 & 0.106 &     0.294   &0.03  & 0    & 0 & 0 & 64156 & 494 & 193096                  & 810468    &385 \\
&13& -     & -     & 4.579 & 0.032 & 0.119 &     0.359   &0.051 & 0    & -     & -     & 64156 & 74216 & 193944         & 819532   &10495 \\
&14& -     & -     & 26.213 & 0.38 & 0.454 &     0.37    &0.123 & 0    & -       & -       & 64156 & 995998 & 1008466    & 819319  &9523 \\
&15& -     & -      & 26.689 & 0.017&0.209 &   0.352     &0.063 & 0    & -      & -      & 64156 & 502 & 198540         & 819067   &9293 \\
\hline
&1 & 0.009 & 0.009 & 0.009 & 0.01 & 0.007  &       0.008&0.008 & 0    & 0 & 0 & 0 & 47 & 0 & 0                              0        & 0\\
&2 & 0.019 & 0.017 & 0.104 & 0.111 & 0.02  &       0.017&0.016 & 57   & 57 & 57 & 256813 & 257974 & 114                   & 57       & 57\\
&3 & 0.11 & 0.135 & 0.233 & 3.244 & 0.274  &       3.109&0.113 & 1462 & 1462 & 1462 & 258218 & 6668549 & 258218            & 6381394 & 1462 \\
&4 & 0.038 & 0.034 & 0.277 & 0.034 & 0.026 &      0.027 &0.028 & 0    & 0 & 0 & 256803 & 1314 & 161                        & 104     & 104 \\
&5 & 0.036 & 0.036 & 0.275 & 0.025 & 0.024 &      0.031 &0.027 & 0    & 0 & 0 & 256803 & 47 & 161                           & 57     & 57  \\
&6 & 0.063 & 0.056 & 0.663 & 0.128 & 0.298 &      3.122 &0.133 & 0    & 0 & 0 & 256803 & 257974 & 258379                   & 6381451 & 1576 \\
3.ttl&7 & 0.061 & 0.062 & 0.709 & 0.032 & 0.287  & 3.101&0.132 & 0    & 0 & 0 & 256803 & 141 & 258369                      & 6381498 & 1566 \\
&8 & 0.094 & 0.153 & 1.425 & 0.138 & 0.297 &        0.03&0.027 & 0    & 0 & 0 & 256803 & 261905 & 516433                   & 161     & 275 \\
&9 & 0.098 & 0.15 & 1.819 & 0.037 & 0.156  &        0.03&0.027 & 0    & 0 & 0 & 256803 & 235 & 258660                     & 104      & 208\\
&10& 0.143 & 0.399 & 2.478 & 0.15 & 0.148  &       0.049&0.048 & 0    & 0 & 0 & 256803 & 265836 & 259504                  & 5473     & 5634\\
&11& 0.141 & 0.368 & 12.374 & 3.343 & 3.315&       3.397&0.384 & 0    & 0 & 0 & 256803 & 6866425 & 6670079                & 6452693  & 72865\\
&12& 0.21 & 1.136 & 15.915 & 0.051 & 0.576 &       3.133&0.171 & 0    & 0 & 0 & 256803 & 1326 & 773580                    & 6382718  & 1730\\
&13& -     & -    & 35.05 & 0.194 & 0.623  &       3.521&0.341 & 0    & -     & -     & 256803 & 329484 & 776449          & 6451652  & 74135\\
&14& -     & -    & 399.257 & 3.948 & 3.982&       3.344&0.558 & 0    & -       & -       & 256803 & 8461907 & 7190771    & 6463879  & 74581\\
&15& -     & -    & 388.289 & 0.06 & 1.34  &       3.213&0.378 & 0    & -      & -      & 256803 & 1377 & 803755          & 6452448  & 73026\\
\hline
&1 & 0.026 & 0.025 & 0.025 & 0.039 & 0.025 &      0.024&0.025 & 0    & 0 & 0 & 0 & 253 & 0                                       & 0  & 0        \\
&2 & 0.064 & 0.069 & 0.471 & 0.522 & 0.064 &       0.07&0.064 & 248  & 248 & 248 & 1027022 & 1040241 & 496                   & 248    & 248    \\
&3 & 0.929 & 0.938 & 1.404 & 28.325 & 1.857&     28.103&0.945 & 12651& 12651 & 12651 & 1039425 & 51050537 & 1039425        & 49377537 & 12651  \\
&4 & 0.198 & 0.173 & 1.617 & 0.157 & 0.095 &      0.129&0.135 & 4    & 4 & 4 & 1027031 & 13800 & 753                       & 505      & 505  \\
&5 & 0.182 & 0.174 & 1.617 & 0.144 & 0.094 &      0.143&0.138 & 0    & 0 & 0 & 1027027 & 253 & 749                           & 248    & 248    \\
&6 & 0.327 & 0.312 & 4.729 & 0.64 & 1.913  &     28.148&1     & 4    & 4 & 4 & 1027031 & 1040479 & 1040182                 & 49377789 & 13151  \\
4.ttl&7 & 0.308 & 0.325 & 4.721 & 0.222 & 1.98  &27.908&1.106 & 0    & 0 & 0 & 1027027 & 759 & 1040183                     & 49378038 & 13152  \\
&8 & 0.504 & 0.778 & 9.217 & 0.675 & 1.278 &      0.158&0.129 & 4    & 4 & 236 & 1027031 & 1079963 & 2080575                 & 757    & 1249    \\
&9 & 0.522 & 0.835 & 12.456 & 0.266 & 0.705&       0.14&0.131 & 0    & 0 & 4 & 1027027 & 1265 & 1041493                       & 505   & 1002     \\
&10& 0.782 & 2.174 & 15.698 & 0.738 & 0.66 &      0.288&0.253 & 4    & 4 & 236 & 1027031 & 1119447 & 1055223                 & 52295  & 53040    \\
&11& 0.76 & 2.077 & 93.286 & 30.477&30.641 &     29.507&3.476 & 232  & 0 & 4 & 1027259 & 54927712 & 51065747               & 50689528 & 1325139  \\
&12& 1.083 & 6.03 & 114.063 & 0.354 & 3.362&     28.046&1.329 & 4    & 4 & 236 & 1027031 & 15222 & 3107857                & 49391554  & 14314 \\
&13& -     & -     & 253.131 & 1.64 & 3.442&     30.217&3.913 & 232  & -     & -     & 1027259 & 2499217 & 3173640        & 50730474  & 1353321 \\
&14& -     & -     & >999 & 74.607 & 35.483&     30.531&5.52  & 10972& -       & -       & -    & 117902759 & 53931133    & 52556376  & 1368984 \\
&15& -     & -      & >999 & 0.454 & 10.929&     29.497&3.763 & 1    & -      & -      & -    & 35953 & 3754770           & 50690218  & 1326126 \\
    \hline
    \end{tabular}%
	}
  \label{tab:rdfoxthree}%
\end{table*}%

\subsection{Discussion}

Note that the  three types of rewritings suggested in this paper give rise to three different rewriting strategies for linear queries.
Let us compare how the execution time depends on the exact rewriting strategy. We see in Table \ref{tab:rdfox} that for most queries in Sequence 1 the
\textsc{Lin} rewriting shows the best performance, while for Sequences 2 and 3 algorithms  \textsc{Log} and  \textsc{Tw*}
are the winners (Tables \ref{tab:rdfoxtwo} and \ref{tab:rdfoxthree}).  Note also that even within a single sequence the results may vary with the number of atoms.

All three rewriting algorithms are based upon a common idea: given a query, pick a point (or a set of points) that would split the query into subqueries,
then rewrite these subqueries recursively, and then include rules that join the results into the rewriting of the initial query. However, there is
a liberty in the choice of this point, and our rewritings are essentially different in this strategy. Thus, different rewritings generate NDL programs which are related to
each other like different execution plans for CQs. Taking into account that we use highly unbalanced data (empty $S$ versus dense $R$) and that
RDFox just materialises  all of the predicates of the program without using magic sets or optimising the program before executions, the performance naturally depends
on how we split the query into subqueries in the rewriting algorithm.

In the paper, we described three simple complexity-motivated splitting strategies. Our experiments show that none of them is always the best and the execution time 
may be dramatically improved by using an `adaptable' splitting strategy which would work similarly to a query execution
planner in database management systems and use statistical information about the data to generate a quickly executable NDL program.

The difference in performance between different types of optimal
rewritings made us investigate its causes.
For example, we noticed that the $\textsc{Tw}$-rewriting of the query with 3 atoms of Sequence~3
\begin{align*}
G (x,y) &\leftarrow S (x,z) \land P_{13} (z, y), \\
P_{13}(x,y) &\leftarrow R (x, z)\land R (z, y),\\
G (x,y) &\leftarrow A_P(x)\land R(x, y)
\end{align*}
takes as long as 28 seconds to execute on the fourth dataset because
it needs so much time to materialise $P_{13}$, which has around
$6\cdot 10^6$ triples. On the other hand, if we remove this predicate
by substituting its definition into the first rule, we obtain
the rewriting
\begin{align*}
G (x,y) &\leftarrow S (x, z)\land R(x,v)\land R (v, y),\\
G (x,y) &\leftarrow A_P (x)\land R(x, y),
\end{align*}
which is executed in 0.945 seconds. This substitution could
be done automatically by a clever NDL engine, but not performed by RDFox.
Thus, we made an attempt to `improve' the  $\textsc{Tw}$-rewriting by getting rid
in this fashion of all predicates that are defined by a single rule and occur
not more than twice in the bodies of the rules.
However, though the rewriting $\textsc{Tw*}$ thus obtained shows a much better performance on Sequences~1 and~3 (see Tables \ref{tab:rdfox} and \ref{tab:rdfoxthree}),
it is not always so on Sequence 2 (Table \ref{tab:rdfoxtwo}). This observation suggests that our rewriting could
be executed faster on a more advanced NDL engine than RDFox which would
carry out such substitutions depending on the cardinality of EDBs.

\bibliographystyle{abbrv}


\begin{thebibliography}{10}

\bibitem{Abitebouletal95}
S.~Abiteboul, R.~Hull, and V.~Vianu.
\newblock {\em Foundations of Databases}.
\newblock Addison-Wesley, 1995.

\bibitem{DBLP:books/cu/ArenasBLM2014}
M.~Arenas, P.~Barcel{\'{o}}, L.~Libkin, and F.~Murlak.
\newblock {\em Foundations of Data Exchange}.
\newblock Cambridge University Press, 2014.

\bibitem{Arora&Barak09}
S.~Arora and B.~Barak.
\newblock {\em Computational Complexity: A Modern Approach}.
\newblock Cambridge University Press, New York, NY, USA, 1st edition, 2009.

\bibitem{BCMNP03}
F.~Baader, D.~Calvanese, D.~McGuinness, D.~Nardi, and P.~Patel-Schneider,
  editors.
\newblock {\em The Description Logic Handbook: {T}heory, Implementation and
  Applications}.
\newblock Cambridge University Press, 2003.

\bibitem{DBLP:conf/lics/BienvenuKP15}
M.~Bienvenu, S.~Kikot, and V.~V. Podolskii.
\newblock Tree-like queries in {OWL 2 QL:} succinctness and complexity results.
\newblock In {\em Proc.\ of the 30th Annual {ACM/IEEE} Symposium on Logic in
  Computer Science, LICS 2015}, pages 317--328. {IEEE} Computer Society, 2015.

\bibitem{DBLP:conf/ijcai/BienvenuOSX13}
M.~Bienvenu, M.~Ortiz, M.~Simkus, and G.~Xiao.
\newblock Tractable queries for lightweight description logics.
\newblock In {\em Proc.\ of the 23nd Int.\ Joint Conf.\ on Artificial
  Intelligence (IJCAI 2013)}, pages 768--774. IJCAI/AAAI, 2013.

\bibitem{DBLP:journals/pvldb/BursztynGM16}
D.~Bursztyn, F.~Goasdou{\'{e}}, and I.~Manolescu.
\newblock Teaching an {RDBMS} about ontological constraints.
\newblock {\em {PVLDB}}, 9(12):1161--1172, 2016.

\bibitem{DBLP:conf/pods/CalauttiGP15}
M.~Calautti, G.~Gottlob, and A.~Pieris.
\newblock Chase termination for guarded existential rules.
\newblock In {\em Proc.\ of the 34th {ACM} Symposium on Principles of Database
  Systems, {PODS 2015}}, pages 91--103, 2015.

\bibitem{DBLP:journals/ws/CaliGL12}
A.~Cal\`{\i}, G.~Gottlob, and T.~Lukasiewicz.
\newblock A general datalog-based framework for tractable query answering over
  ontologies.
\newblock {\em Journal of Web Semantics}, 14:57--83, 2012.

\bibitem{DBLP:journals/ai/CaliGP12}
A.~Cal\`{\i}, G.~Gottlob, and A.~Pieris.
\newblock Towards more expressive ontology languages: The query answering
  problem.
\newblock {\em Artificial Intelligence}, 193:87--128, 2012.

\bibitem{DBLP:journals/semweb/CalvaneseGLLPRRRS11}
D.~Calvanese, G.~{De Giacomo}, D.~Lembo, M.~Lenzerini, A.~Poggi,
  M.~Rodriguez-Muro, R.~Rosati, M.~Ruzzi, and D.~F. Savo.
\newblock The {MASTRO} system for ontology-based data access.
\newblock {\em Semantic Web}, 2(1):43--53, 2011.

\bibitem{CDLLR07}
D.~Calvanese, G.~De~Giacomo, D.~Lembo, M.~Lenzerini, and R.~Rosati.
\newblock Tractable reasoning and efficient query answering in description
  logics: the {{\textit{DL-Lite}}} family.
\newblock {\em Journal of Automated Reasoning}, 39(3):385--429, 2007.

\bibitem{DBLP:journals/tcs/ChekuriR00}
C.~Chekuri and A.~Rajaraman.
\newblock Conjunctive query containment revisited.
\newblock {\em Theoretical Computer Science}, 239(2):211--229, 2000.

\bibitem{DBLP:conf/cade/ChortarasTS11}
A.~Chortaras, D.~Trivela, and G.~Stamou.
\newblock Optimized query rewriting for {OWL 2 QL}.
\newblock In {\em Proc.\ of CADE-23}, volume 6803 of {\em LNCS}, pages
  192--206. Springer, 2011.

\bibitem{DBLP:journals/jacm/Cook71}
S.~A. Cook.
\newblock Characterizations of pushdown machines in terms of time-bounded
  computers.
\newblock {\em Journal of the {ACM}}, 18(1):4--18, 1971.

\bibitem{DBLP:journals/jair/GrauHKKMMW13}
B.~{Cuenca Grau}, I.~Horrocks, M.~Kr{\"o}tzsch, C.~Kupke, D.~Magka, B.~Motik,
  and Z.~Wang.
\newblock Acyclicity notions for existential rules and their application to
  query answering in ontologies.
\newblock {\em Journal of Artificial Intelligence Research (JAIR)},
  47:741--808, 2013.

\bibitem{DBLP:journals/csur/DantsinEGV01}
E.~Dantsin, T.~Eiter, G.~Gottlob, and A.~Voronkov.
\newblock Complexity and expressive power of logic programming.
\newblock {\em ACM Computing Surveys}, 33(3):374--425, 2001.

\bibitem{di2013optimizing}
F.~Di~Pinto, D.~Lembo, M.~Lenzerini, R.~Mancini, A.~Poggi, R.~Rosati, M.~Ruzzi,
  and D.~F. Savo.
\newblock Optimizing query rewriting in ontology-based data access.
\newblock In {\em Proc.\ of the 16th Int.\ Conf.\ on Extending Database
  Technology (EDBT 2013)}, pages 561--572. ACM, 2013.

\bibitem{DBLP:books/daglib/0029346}
A.~Doan, A.~Y. Halevy, and Z.~G. Ives.
\newblock {\em Principles of Data Integration}.
\newblock Morgan Kaufmann, 2012.

\bibitem{DBLP:conf/aaai/EiterOSTX12}
T.~Eiter, M.~Ortiz, M.~{\v{S}}imkus, T.-K. Tran, and G.~Xiao.
\newblock Query rewriting for {Horn-SHIQ} plus rules.
\newblock In {\em Proc.\ of the 26th AAAI Conf.\ on Artificial Intelligence
  (AAAI 2012)}, pages 726--733. AAAI, 2012.

\bibitem{DBLP:journals/tcs/FellowsHRV09}
M.~R. Fellows, D.~Hermelin, F.~A. Rosamond, and S.~Vialette.
\newblock On the parameterized complexity of multiple-interval graph problems.
\newblock {\em Theoretical Computer Science}, 410(1):53--61, 2009.

\bibitem{DBLP:series/txtcs/FlumG06}
J.~Flum and M.~Grohe.
\newblock {\em Parameterized Complexity Theory}.
\newblock Texts in Theoretical Computer Science. An {EATCS} Series. Springer,
  2006.

\bibitem{optique}
M.~Giese, A.~Soylu, G.~Vega-Gorgojo, A.~Waaler, P.~Haase, E.~Jim{\'e}nez-Ruiz,
  D.~Lanti, M.~Rezk, G.~Xiao, {\"O}.~{\"O}z{\c{c}}ep, and R.~Rosati.
\newblock Optique: Zooming in on big data.
\newblock {\em IEEE Computer}, 48(3):60--67, 2015.

\bibitem{DBLP:conf/icalp/GogaczM14}
T.~Gogacz and J.~Marcinkowski.
\newblock All-instances termination of chase is undecidable.
\newblock In {\em Proc.\ of the 41st Int.\ Colloquium Automata, Languages, and
  Programming (ICALP 2014), Part {II}}, volume 8573 of {\em Lecture Notes in
  Computer Science}, pages 293--304. Springer, 2014.

\bibitem{DBLP:journals/ai/GottlobKKPSZ14}
G.~Gottlob, S.~Kikot, R.~Kontchakov, V.~V. Podolskii, T.~Schwentick, and
  M.~Zakharyaschev.
\newblock The price of query rewriting in ontology-based data access.
\newblock {\em Artificial Intelligence}, 213:42--59, 2014.

\bibitem{DBLP:conf/icalp/GottlobLS99}
G.~Gottlob, N.~Leone, and F.~Scarcello.
\newblock Computing {LOGCFL} certificates.
\newblock In {\em Proc. of the 26th Int. Colloquium on Automata, Languages and
  Programming (ICALP-99)}, volume 1644 of {\em Lecture Notes in Computer
  Science}, pages 361--371. Springer, 1999.

\bibitem{DBLP:conf/icde/GottlobOP11}
G.~Gottlob, G.~Orsi, and A.~Pieris.
\newblock Ontological queries: Rewriting and optimization.
\newblock In {\em Proc.\ of ICDE 2011}, pages 2--13. IEEE Computer Society,
  2011.

\bibitem{gottlob2014query}
G.~Gottlob, G.~Orsi, and A.~Pieris.
\newblock Query rewriting and optimization for ontological databases.
\newblock {\em ACM Transactions on Database Systems (TODS)}, 39(3):25, 2014.

\bibitem{DBLP:journals/siamcomp/Greibach73}
S.~A. Greibach.
\newblock The hardest context-free language.
\newblock {\em {SIAM} J. Comput.}, 2(4):304--310, 1973.

\bibitem{huf52}
D.~A. Huffman.
\newblock A method for the construction of minimum-redundancy codes.
\newblock {\em Proceedings of the Institute of Radio Engineers},
  40(9):1098--1101, 1952.

\bibitem{DBLP:conf/semweb/Jimenez-RuizKZH15a}
E.~Jim{\'{e}}nez{-}Ruiz, E.~Kharlamov, D.~Zheleznyakov, I.~Horrocks, C.~Pinkel,
  M.~G. Skj{\ae}veland, E.~Thorstensen, and J.~Mora.
\newblock {BootOX}: Bootstrapping {OWL} 2 ontologies and {R2RML} mappings from
  relational databases.
\newblock In {\em Proc.\ of the {ISWC} 2015 Posters {\&} Demonstrations Track
  at the 14th Int.\ Semantic Web Conf.\ (ISWC-2015)}, volume 1486 of {\em CEUR
  Workshop Proceedings}. CEUR-WS, 2015.

\bibitem{DBLP:journals/ai/KaminskiNG16}
M.~Kaminski, Y.~Nenov, and B.~{Cuenca Grau}.
\newblock Datalog rewritability of {D}isjunctive {D}atalog programs and
  non-{H}orn ontologies.
\newblock {\em Artificial Intelligence}, 236:90--118, 2016.

\bibitem{DBLP:conf/semweb/KharlamovHJLLPR15}
E.~Kharlamov, D.~Hovland, E.~Jim{\'{e}}nez{-}Ruiz, D.~Lanti, H.~Lie, C.~Pinkel,
  M.~Rezk, M.~G. Skj{\ae}veland, E.~Thorstensen, G.~Xiao, D.~Zheleznyakov, and
  I.~Horrocks.
\newblock Ontology based access to exploration data at {S}tatoil.
\newblock In {\em Proc.\ of the 14th Int.\ Semantic Web Conf.\ (ISWC 2015),
  Part II}, volume 9367 of {\em Lecture Notes in Computer Science}, pages
  93--112. Springer, 2015.

\bibitem{LICS14}
S.~Kikot, R.~Kontchakov, V.~Podolskii, and M.~Zakharyaschev.
\newblock On the succinctness of query rewriting over shallow ontologies.
\newblock In {\em Proc.\ of the Joint Meeting of the 23rd {EACSL} Annual Conf.\
  on Computer Science Logic {(CSL 2014)} and the 29th Annual {ACM/IEEE}
  Symposium on Logic in Computer Science (LICS 2014)}, pages 57:1--57:10.
  {ACM}, 2014.

\bibitem{DBLP:conf/icalp/KikotKPZ12}
S.~Kikot, R.~Kontchakov, V.~V. Podolskii, and M.~Zakharyaschev.
\newblock Exponential lower bounds and separation for query rewriting.
\newblock In {\em Proc.\ of the 39th Int.\ Colloquium on Automata, Languages
  and Programming (ICALP 2012)}, volume 7392 of {\em Lecture Notes in Computer
  Science}, pages 263--274. Springer, 2012.

\bibitem{DBLP:conf/dlog/KikotKZ11}
S.~Kikot, R.~Kontchakov, and M.~Zakharyaschev.
\newblock On (in)tractability of {OBDA} with {OWL 2 QL}.
\newblock In {\em Proc.\ of the 24th Int.\ Workshop on Description Logics (DL
  2011)}, volume 745, pages 224--234. CEUR-WS, 2011.

\bibitem{DBLP:conf/kr/KikotKZ12}
S.~Kikot, R.~Kontchakov, and M.~Zakharyaschev.
\newblock Conjunctive query answering with {OWL~2~QL}.
\newblock In {\em Proc.\ of the 13th Int.\ Conf.\ on Principles of Knowledge
  Representation and Reasoning (KR 2012)}, pages 275--285. AAAI, 2012.

\bibitem{Koch:2006:PQT:1142351.1142382}
C.~Koch.
\newblock Processing queries on tree-structured data efficiently.
\newblock In {\em Proc.\ of the 25th ACM SIGMOD-SIGACT-SIGART Symposium on
  Principles of Database Systems (PODS 2006)}, pages 213--224. ACM, 2006.

\bibitem{DBLP:journals/semweb/KonigLMT15}
M.~K{\"{o}}nig, M.~Lecl{\`{e}}re, M.-L. Mugnier, and M.~Thomazo.
\newblock Sound, complete and minimal {UCQ}-rewriting for existential rules.
\newblock {\em Semantic Web}, 6(5):451--475, 2015.

\bibitem{KR10our}
R.~Kontchakov, C.~Lutz, D.~Toman, F.~Wolter, and M.~Zakharyaschev.
\newblock The combined approach to query answering in {DL-Lite}.
\newblock In {\em Proc.\ of the 12th Int.\ Conf.\ on Principles of Knowledge
  Representation and Reasoning (KR 2010)}, pages 247--257. AAAI Press, 2010.

\bibitem{DBLP:conf/semweb/KontchakovRRXZ14}
R.~Kontchakov, M.~Rezk, M.~Rodriguez{-}Muro, G.~Xiao, and M.~Zakharyaschev.
\newblock Answering {SPARQL} queries over databases under {OWL} 2 {QL}
  entailment regime.
\newblock In {\em Proc.\ of the 13th Int. Semantic Web Conf. (ISWC 2014), Part
  {I}}, volume 8796 of {\em Lecture Notes in Computer Science}, pages 552--567.
  Springer, 2014.

\bibitem{Lenzerini13}
M.~Lenzerini.
\newblock Ontology-based data management.
\newblock {\em ACM SIGMOD Blog}, May 2013.

\bibitem{kyrie2}
J.~Mora, R.~Rosati, and {\'O}.~Corcho.
\newblock {Kyrie2:} query rewriting under extensional constraints in {ELHIO}.
\newblock In {\em Proc.\ of the 13th Int. Semantic Web Conf.\ (ISWC 2014)},
  volume 8796 of {\em Lecture Notes in Computer Science}, pages 568--583.
  Springer, 2014.

\bibitem{profiles}
B.~Motik, B.~Cuenca~Grau, I.~Horrocks, Z.~Wu, A.~Fokoue, and C.~Lutz.
\newblock {\em {OWL} 2 {W}eb {O}ntology {L}anguage Profiles}.
\newblock {W3C} {R}ecommendation, 2012.
\newblock Available at \url{http://www.w3.org/TR/owl2-profiles/}.

\bibitem{DBLP:conf/semweb/NenovPMHWB15}
Y.~Nenov, R.~Piro, B.~Motik, I.~Horrocks, Z.~Wu, and J.~Banerjee.
\newblock {RDFox}: {A} highly-scalable {RDF} store.
\newblock In {\em Proc.\ of the 14th Int.\ Semantic Web Conf. (ISWC 2015), Part
  {II}}, volume 9367 of {\em Lecture Notes in Computer Science}, pages 3--20.
  Springer, 2015.

\bibitem{DBLP:conf/dlog/Perez-UrbinaMH09}
H.~P{\'e}rez-Urbina, B.~Motik, and I.~Horrocks.
\newblock A comparison of query rewriting techniques for {DL-Lite}.
\newblock In {\em Proc.\ of the 22nd Int.\ Workshop on Description Logics (DL
  2009)}, volume 477 of {\em CEUR Workshop Proceedings}. CEUR-WS, 2009.

\bibitem{Perez-Urbina12}
H.~P{\'e}rez-Urbina, E.~Rodr{\'\i}guez-D{\'\i}az, M.~Grove, G.~Konstantinidis,
  and E.~Sirin.
\newblock Evaluation of query rewriting approaches for {OWL 2}.
\newblock In {\em Proc.\ of SSWS+HPCSW 2012}, volume 943 of {\em CEUR Workshop
  Proceedings}. CEUR-WS, 2012.

\bibitem{PicalausaV:sparql-2011}
F.~Picalausa and S.~Vansummeren.
\newblock What are real {SPARQL} queries like?
\newblock In {\em Proc.\ of the Int.\ Workshop on Semantic Web Information
  Management (SWIM)}. ACM, 2011.

\bibitem{PLCD*08}
A.~Poggi, D.~Lembo, D.~Calvanese, G.~De~Giacomo, M.~Lenzerini, and R.~Rosati.
\newblock Linking data to ontologies.
\newblock {\em Journal on Data Semantics}, X:133--173, 2008.

\bibitem{DBLP:conf/semweb/Rodriguez-MuroKZ13}
M.~Rodriguez{-}Muro, R.~Kontchakov, and M.~Zakharyaschev.
\newblock Ontology-based data access: Ontop of databases.
\newblock In {\em Proc.\ of the 12th Int.\ Semantic Web Conf.\ (ISWC 2013),
  Part {I}}, volume 8218 of {\em Lecture Notes in Computer Science}, pages
  558--573. Springer, 2013.

\bibitem{DBLP:conf/dlog/Rodriguez-MuroKZ13}
M.~Rodriguez{-}Muro, R.~Kontchakov, and M.~Zakharyaschev.
\newblock Query rewriting and optimisation with database dependencies in
  {Ontop}.
\newblock In {\em Informal Proc.\ of the 26th Int.\ Workshop on Description
  Logics (DL 2013)}, volume 1014 of {\em CEUR Workshop Proceedings}, pages
  917--929. CEUR-WS, 2013.

\bibitem{DBLP:conf/esws/Rosati12}
R.~Rosati.
\newblock Prexto: Query rewriting under extensional constraints in {DL-Lite}.
\newblock In {\em Proc.\ of the 9th Extended Semantic Web Conf.\ (EWSC 2012)},
  volume 7295 of {\em Lecture Notes in Computer Science}, pages 360--374.
  Springer, 2012.

\bibitem{DBLP:conf/kr/RosatiA10}
R.~Rosati and A.~Almatelli.
\newblock Improving query answering over {DL-Lite} ontologies.
\newblock In {\em Proc.\ of the 12th Int.\ Conf.\ on Principles of Knowledge
  Representation and Reasoning (KR 2010)}, pages 290--300. AAAI Press, 2010.

\bibitem{DBLP:conf/semweb/SequedaAM14}
J.~F. Sequeda, M.~Arenas, and D.~P. Miranker.
\newblock {OBDA:} query rewriting or materialization? {In} practice, both!
\newblock In {\em Proc.\ of the 13th Int. Semantic Web Conf.\ (ISWC 2014), Part
  I}, volume 8796 of {\em Lecture Notes in Computer Science}, pages 535--551.
  Springer, 2014.

\bibitem{soylu2016}
A.~Soylu, M.~Giese, E.~Jimenez-Ruiz, G.~Vega-Gorgojo, and I.~Horrocks.
\newblock Experiencing optiquevqs: A multi-paradigm and ontology-based visual
  query system for end users.
\newblock {\em Universal Access in the Information Society}, 15(1):129--152,
  2016.

\bibitem{Sudborough:1975:NTC:321906.321913}
I.~H. Sudborough.
\newblock A note on tape-bounded complexity classes and linear context-free
  languages.
\newblock {\em Journal of the {ACM}}, 22(4):499--500, Oct. 1975.

\bibitem{sudborough78}
I.~H. Sudborough.
\newblock On the tape complexity of deterministic context-free languages.
\newblock {\em Journal of the {ACM}}, 25(3):405--414, 1978.

\bibitem{DBLP:conf/ijcai/Thomazo13}
M.~Thomazo.
\newblock Compact rewritings for existential rules.
\newblock In {\em Proc.\ of the 23rd Int.\ Joint Conf.\ on Artificial
  Intelligence ({IJCAI} 2013)}. {IJCAI/AAAI}, 2013.

\bibitem{venetis2016rewriting}
T.~Venetis, G.~Stoilos, and V.~Vassalos.
\newblock Rewriting minimisations for efficient ontology-based query answering.
\newblock In {\em Proc.\ of the 28th Int.\ Conf.\ on Tools with Artificial
  Intelligence (ICTAI 2016)}, pages 1095--1102. IEEE, 2016.

\bibitem{DBLP:journals/jcss/Venkateswaran91}
H.~Venkateswaran.
\newblock Properties that characterize {LOGCFL}.
\newblock {\em Journal of Computer and System Sciences}, 43(2):380--404, 1991.

\bibitem{DBLP:conf/vldb/Yannakakis81}
M.~Yannakakis.
\newblock Algorithms for acyclic database schemes.
\newblock In {\em Proc.\ of the 7th Int.\ Conf.\ on Very Large Data Bases
  (VLDB)}, pages 82--94. {IEEE} Computer Society, 1981.

\end{thebibliography}

\end{document}